\documentclass[conference]{IEEEtran}
\ifCLASSINFOpdf
\else
\fi
\usepackage{amsfonts}
\usepackage{amsmath,amsthm}
\usepackage{graphicx,epstopdf,epsfig,multirow,epic,bm}
\usepackage{color}
\usepackage{multicol}
\usepackage{makeidx,showidx}
\usepackage{ifthen}
\usepackage{mathrsfs}
\usepackage{colortbl}
\usepackage{multicol}
\usepackage{latexsym}
\usepackage{amssymb}
\usepackage{algorithm}
\usepackage{algorithmic}
\usepackage{longtable}
\usepackage{array}
\usepackage{enumerate}
\usepackage{paralist}
\usepackage{fancyhdr}
\usepackage{CJK}

\theoremstyle{plain}
\newtheorem{thm}{Theorem}
\newtheorem{cor}[thm]{Corollary}
\newtheorem{lem}[thm]{Lemma}
\newtheorem{prop}[thm]{Proposition}
\theoremstyle{definition}
\newtheorem{defn}{Definition}
\theoremstyle{definition}
\newtheorem{rem}{Remark}

\theoremstyle{definition}

\theoremstyle{definition}

\theoremstyle{definition}

\newcommand{\qqed}{\hfill $\Box$}

\begin{document}
%
% paper title
% Titles are generally capitalized except for words such as a, an, and, as,
% at, but, by, for, in, nor, of, on, or, the, to and up, which are usually
% not capitalized unless they are the first or last word of the title.
% Linebreaks \\ can be used within to get better formatting as desired.
% Do not put math or special symbols in the title.
\title{Using Chinese Characters To Generate Text-Based Passwords For Information Security}

% author names and affiliations
% use a multiple column layout for up to three different
% affiliations

\author{\IEEEauthorblockN{Bing \textsc{Yao}$^{1,8}$,  Yarong \textsc{Mu}$^{1}$,  Yirong \textsc{Sun}$^{1}$, Hui \textsc{Sun}$^{1,4}$, Xiaohui \textsc{Zhang}$^{1}$, Hongyu \textsc{Wang}$^{4}$,  Jing \textsc{Su}$^{4}$\\ Mingjun \textsc{Zhang}$^{2,\dagger}$,  Sihua \textsc{Yang}$^{2}$,  Meimei  \textsc{Zhao}$^{3}$,  Xiaomin \textsc{Wang}$^{4}$, Fei \textsc{Ma}$^{4}$\\ Ming \textsc{Yao}$^{5}$, Chao \textsc{Yang}$^{6}$, Jianming \textsc{Xie}$^{7}$}
\IEEEauthorblockA{{1} College of Mathematics and Statistics,
 Northwest Normal University, Lanzhou, 730070, CHINA}
 \IEEEauthorblockA{{2} School of Information Engineering, Lanzhou University of Finance and Economics, Lanzhou, 730030, CHINA}
\IEEEauthorblockA{{3} College of Science, Gansu Agricultural University, Lanzhou 730070, CHINA}
\IEEEauthorblockA{{4} School of Electronics Engineering and Computer Science, Peking University, Beijing, 100871, CHINA}
\IEEEauthorblockA{{5} Department of Information Process and Control Engineering, Lanzhou Petrochemical College of\\ Vocational Technology, Lanzhou, 730060, CHINA}
\IEEEauthorblockA{{6} School of Mathematics, Physics \& Statistics, Shanghai University of Engineering Science, Shanghai, 201620, CHINA}
\IEEEauthorblockA{{7} Department of Mathematics of Lanzhou City University, Lanzhou, 730070, CHINA}
\IEEEauthorblockA{{8} School of Electronics and Information Engineering, Lanzhou Jiaotong University, Lanzhou, 730070, CHINA\\
$^\dagger$ Corresponding authors: shuxue1998@163.com}

% <-this % stops an unwanted space
\thanks{Manuscript received December 8, 2017; revised December 26, 2017.
Corresponding author: Bing Yao, email: yybb918@163.com.}}

% conference papers do not typically use \thanks and this command
% is locked out in conference mode. If really needed, such as for
% the acknowledgment of grants, issue a \IEEEoverridecommandlockouts
% after \documentclass

% for over three affiliations, or if they all won't fit within the width
% of the page, use this alternative format:
%
%\author{\IEEEauthorblockN{Michael Shell\IEEEauthorrefmark{1},
%Homer Simpson\IEEEauthorrefmark{2},
%James Kirk\IEEEauthorrefmark{3},
%Montgomery Scott\IEEEauthorrefmark{3} and
%Eldon Tyrell\IEEEauthorrefmark{4}}
%\IEEEauthorblockA{\IEEEauthorrefmark{1}School of Electrical and Computer Engineering\\
%Georgia Institute of Technology,
%Atlanta, Georgia 30332--0250\\ Email: see http://www.michaelshell.org/contact.html}
%\IEEEauthorblockA{\IEEEauthorrefmark{2}Twentieth Century Fox, Springfield, USA\\
%Email: homer@thesimpsons.com}
%\IEEEauthorblockA{\IEEEauthorrefmark{3}Starfleet Academy, San Francisco, California 96678-2391\\
%Telephone: (800) 555--1212, Fax: (888) 555--1212}
%\IEEEauthorblockA{\IEEEauthorrefmark{4}Tyrell Inc., 123 Replicant Street, Los Angeles, California 90210--4321}}

% use for special paper notices
%\IEEEspecialpapernotice{(Invited Paper)}

% make the title area
\maketitle

% As a general rule, do not put math, special symbols or citations
% in the abstract
\begin{abstract}
Graphical passwords (GPWs) are in many areas of the current world, in which the two-dimensional code has been applied successfully nowadays. Topological graphic passwords (Topsnut-gpws) are a new type of cryptography, and they differ from the existing GPWs. A Topsnut-gpw consists of two parts: one is a topological structure (graph), and one is a set of discrete elements (a graph labelling, or coloring), the topological structure connects these discrete  elements together to form an interesting ``story''. It is not easy to remember passwords made up of longer bytes for many Chinese people. Chinese characters are naturally
topological structures and have very high information density, especially, Chinese couplets form natively public keys and private keys in authentication. Our idea is to transform  Chinese characters into computer and electronic equipments with touch screen by speaking, writing and keyboard for forming  Hanzi-gpws (one type of Topsnut-gpws). We will translate Chinese characters into graphs (Hanzi-graphs), and apply mathematical techniques (graph labellings) to construct Hanzi-gpws, and next using Hanzi-gpws  produces text-based passwords (TB-paws) with longer bytes as desired. We will explore a new topic of encrypting networks by means of algebraic groups, called graphic groups (Ablian additive finite group), and construct several kinds of self-similar Hanzi-networks, as well as some approaches for the encryption of networks, an important topic of researching information security. The stroke order of writing Chinese characters motivates us to study directed Hanzi-gpws based on directed graphs. We will introduce flawed graph labellings on disconnected Hanzi-graphs such that each Hanzi-gpw with a flawed graph labelling can form a set of connected Topsnut-gpws, like a group. Moreover, we introduce connections between different graphic groups that can be used to encrypt networks based on community partition.\\[4pt]
\end{abstract}

% no keywords

% For peer review papers, you can put extra information on the cover
% page as needed:
% \ifCLASSOPTIONpeerreview
% \begin{center} \bfseries EDICS Category: 3-BBND \end{center}
% \fi
%
% For peerreview papers, this IEEEtran command inserts a page break and
% creates the second title. It will be ignored for other modes.
\IEEEpeerreviewmaketitle

\pagestyle{fancy}
\pagestyle{plain}

\section{Introduction and preliminary}

Security of cyber and information  is facing more challenges  and thorny problems in today's world. There may exist such situation: a protection by the virtue of AI (artificial intelligence) resists attackers equipped by AI in current networks. We have to consider: How to overcome various attacker equipped by AI's tools?

The origin of AI was generally acknowledged in Dartmouth Conference in 1956. In popularly,  AI is defined as: ``\emph{Artificial intelligence (AI) is a branch of computer science. It attempts to understand the essence of intelligence and produce a new kind of intelligence machine that can respond in a similar way to human intelligence. Research in this field includes robots, language recognition, image recognition, natural language processing and expert systems}.''

In fact, the modern AI can be divided into two parts, namely, ``artificial'' and ``intelligence''. It is difficult for computer to learn ``qualitative change independent of quality'' in terms of learning and ``practice''. It is difficult for them to go directly from on ``quality'' to another ``quality'' or from one ``concept'' to another ``concept''. Because of this, practice here is not the same practice as human beings, since the process of human practice includes both experience and creation.

For the above statement on AI, we cite an important sentence: \emph{A key signature of human intelligence is the ability to make `infinite use of finite means'} (\cite{Humboldt-W-1999-1836} in 1836; \cite{Chomsky-N-1965} in 1965), as the beginning of an article entitled  ``\emph{Relational inductive biases, deep learning, and graph networks}'' by Battaglia \emph{et al.} in \cite{Battaglia-27-authors-arXiv1806-01261v2}. They have pointed out: ``\emph{in which a small set of elements (such as words) can be productively composed in limitless ways (such as into new sentences)}'', and they argued that \emph{combinatorial generalization} must be a top priority for AI to achieve human-like abilities, and that structured representations and computations are key to realizing this object. As an example of supporting `infinite use of finite means', \emph{self-similarity} is common phenomena between a part of a complex system and the whole of the system.

Yao \emph{et al.} in \cite{Yao-Sun-Zhang-Mu-Wang-Xu-2018} have listed some advantages of Chinese characters. Wang Lei, a teacher and researcher of Shenyang Institute of Education, stepped on the stage of ``I am a speaker'' and explained the beauty of Chinese characters as: (1) Chinese characters are pictographs, and each one of Chinese characters represents a meaning, a truth, a culture, a spirit. (2) Chinese characters are naturally topological structures. (3) The biggest advantage of Chinese characters is that the information density is very high. (4) Chinese characters is their inheritance and stability. Chinese characters are picturesque in shape, beautiful in sound and complete in meaning. It is concise, efficient and vivid, and moreover it is the most advanced written language in the world.

\subsection{Researching background}

The existing \emph{graphical passwords} (GPWs) were investigated for a long time (Ref. \cite{Suo-Zhu-Owen-2005, Biddle-Chiasson-van-Oorschot-2009, Gao-Jia-Ye-Ma-2013}). As an alternation, Wang \emph{et al.} in \cite{Wang-Xu-Yao-2016} and \cite{Wang-Xu-Yao-Key-models-Lock-models-2016} present a new-type of graphical passwords, called \emph{topological graphic passwords} (Topsnut-gpws), and show their Topsnut-gpws differing from the existing GPWs. A Topsnut-gpw consists of two parts: one is a topological structure (graph), and one is a set of discrete elements (here, a graph labelling, or a coloring), the topological structure connects these discrete  elements together to form an interesting ``story'' for easily remembering. Graphs of graph theory are ubiquitous in the real world, representing objects and their relationships such as social networks, e-commerce networks, biology networks and traffic networks and many areas of science such as Deep Learning, Graph Neural Network, Graph Networks (Ref. \cite{Battaglia-27-authors-arXiv1806-01261v2} and \cite{LeCun-Bengio-Hinton-2015}).  Topsnut-gpws based on techniques of graph theory, in recent years, have been investigated fast and produce abundant fruits (Ref. \cite{SUN-ZHANG-YAO-IMCEC-2018,ZHANG-SUN-YAO-Liu-IMCEC-2018,ZHANG-SUN-YAO-IMCEC-2018}).

As examples, two Topsnut-gpws is shown in Fig.\ref{fig:example-1} (b) and (c).

\begin{figure}[h]
\centering
\includegraphics[height=4.2cm]{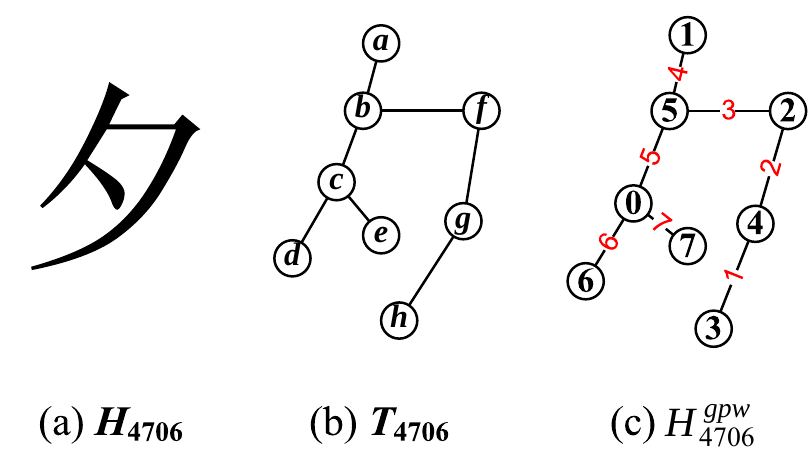}\\
\caption{\label{fig:example-1} {\small (a) A simplified Chinese character $H_{4706}$ defined in \cite{GB2312-80}; (b) a mathematical model $T_{4706}$ of $H_{4706}$; (c) another mathematical model $H^{gpw}_{4706}$ of $H_{4706}$.}}
\end{figure}

\begin{figure}[h]
\centering
\includegraphics[height=4.6cm]{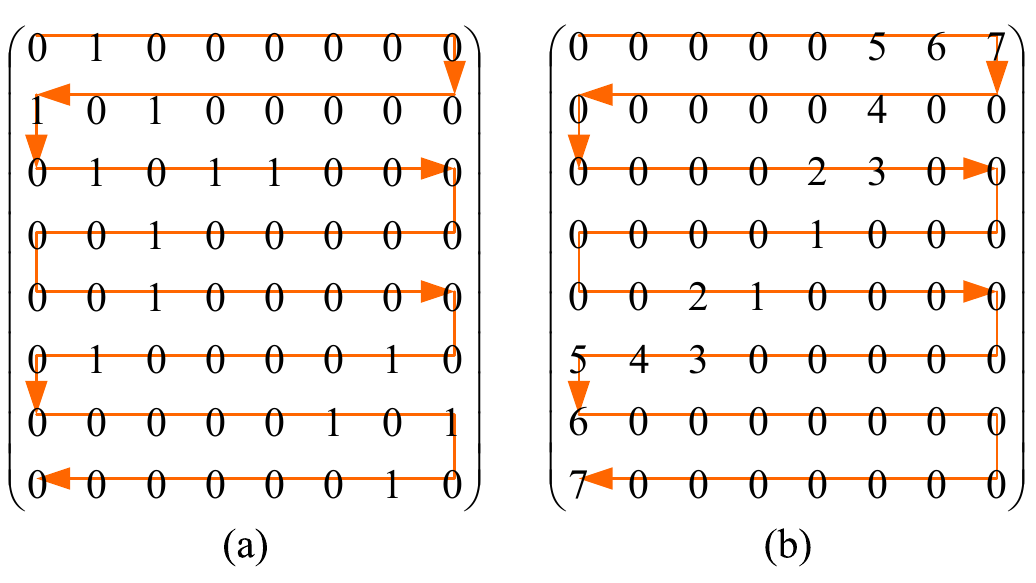}\\
\caption{\label{fig:example-1-matrix} {\small (a) A popular matrix $A(T_{4706})$ of $T_{4706}$; (b) the Hanzi-matrix $A^*(H^{gpw}_{4706})$ of $H^{gpw}_{4706}$.}}
\end{figure}

There are many advantages of Topsnut-gpws, such as, the space of Topsnut-gpws is large enough such that the decrypting Topsnut-gpws will be terrible and horrible if using current computer. In graph theory, Cayley's formula (Ref. \cite{Bondy-2008})
\begin{equation}\label{eqa:Cayley-Formula}
\tau (K_n)=n^{n-2}
\end{equation}
pointed that the number $\tau (K_n)$ of spanning trees (tree-like Topsnut-gpws) of a \emph{complete graph} (network) $K_n$ is non-polynomial, so Topsnut-gpws are \emph{computationally security}; Topsnut-gpws are suitable for people who need not learn new rules and are allowed to use their private knowledge in making Topsnut-gpws for the sake of remembering easily; Topsnut-gpws, very often, run fast in communication networks because they are saved in computer by popular matrices rather than pictures; Topsnut-gpws are suitable for using mobile equipments with touch screen and speech recognition; Topsnut-gpws can generate quickly text-based passwords (TB-paws) with bytes as long as desired, but these TB-paws can not reconstruct the original Topsnut-gpws, namely, it is \emph{irreversible}; many mathematical conjectures (NP-problems) are related with Topsnut-gpws such that they are really \emph{provable security}.

The idea of ``translating Chinese characters into Topsnut-gpws'' was first proposed in \cite{Yao-Sun-Zhang-Mu-Wang-Xu-2018}. Topsnut-gpws were made by \emph{Hanzi-graphs} are called \emph{Hanzi-gpws} (Ref. \cite{MU-YAO-2018-11, MU-YAO-2018-22, MU-ZHANG-YAO-2018}), see a Hanzi-graph $T_{4706}$ and a Hanzi-gpw $H^{gpw}_{4706}$ are shown in Fig.\ref{fig:example-1} (b) and (c). By the narrowed line under the Hanzi-matrix $A^*(T_{4706})$ shown in Fig.\ref{fig:example-1-matrix}(a), we get a text-based password (TB-paw) as follows
$${
\begin{split}
A(T_{4706})=&01000000000001010101100000000100\\
&00100000010000100000010101000000
\end{split}}
$$
and furthermore we obtain another TB-paw
$${
\begin{split}
A(H^{gpw}_{4706})=&00000567004000000000230000010000\\
&00210000000003456000000000000007
\end{split}}
$$
along the narrowed line under the Hanzi-matrix $A^*(H^{gpw}_{4706})$ shown in Fig.\ref{fig:example-1-matrix}(b). There are efficient algorithms for writing $A(T_{4706})$ and $A(H^{gpw}_{4706})$ from the Hanzi-matrices. It is not difficult to see there are at least $(64)!$ TB-paws made by two matrices $A^*(T_{4706})$ and $A^*(H^{gpw}_{4706})$, respectively.

There are many unsolved problems in graph theory, which can persuade people to believe that Topsnut-gpws can withdraw cipher's attackers, such a famous example is: ``\emph{If a graph $R$ with the maximum $R(s,k)-1$ vertices has no a complete graph $K_s$ of $s$ vertices and an independent set of $k$ vertices, then we call $R$ a \emph{Ramsey graph} and $R(s,k)-1$ a \emph{Ramsey number}. As known, it is a terrible job for computer to find Ramsey number $R(5,5)$, although we have known $46\leq R(5,5) \leq49$}''. Joel Spencer said:``\emph{Erd\"{o}s asks us to imagine an alien force, vastly more powerful than us, landing on Earth and demanding the value of $R(5,5)$ or they will destroy our planet. In that case, he claims, we should marshal all our computers and all our mathematicians and attempt to find the value. But suppose, instead, that they ask for $R(6,6)$. In that case, he believes, we should attempt to destroy the aliens}''.

\subsection{Researching tasks}

Although Yao \emph{et al.}  \cite{Yao-Sun-Zhang-Mu-Wang-Xu-2018} have proposed Hanzi-graphs and Hanzi-gpws, however, we think that their junior work is just a beginning on Topsnut-gpws made by the idea of ``Hanzi-graphs puls graph labellings''.

Our goals are: (1) To design passwords of Chinese characters by voice inputting, hand inputting into computers and  mobile equipments with touch screen; (2) to make more complex TB-paws for encrypting electronic files, or encrypting networks.

In technique, we will introduce how to construct mathematical models of Chinese characters, called \emph{Hanzi-graphs}, and then use Hanzi-graphs and graph labelling/colorings to build up Hanzi-graph passwords, called \emph{Hanzi-gpws}. Then, several types of Hanzi-matrices will be defined for producing TB-paws. Moreover, we will explore to encrypt dynamic networks, such as deterministic networks, scale-free networks, self-similar networks, and so on.

In producing TB-paws from Hanzi-gpws, we can get TB-paws  with hundreds bytes. As known,   brute-force attacks work by calculating every possible combination that could make up a password and testing it to see if it is the correct password. As the password's length increases, the amount of time, on average, to find the correct password increases exponentially. AES  (Advanced Encryption Standard) permits the use of 256-bit keys.  How many possible combinations of $2^{256}$ (or 256-bit)  encryption are there?  There are

115,792,089,237,316,195,423,570,985,008,687,907,853,

269,984,665,640,564,039,457,584,007,913,129,639,936\\
(78 digits) possible combinations for 256-bit keys \cite{256-bit-encryption}. Breaking a symmetric 256-bit key by brute force requires $2^{128}$ times more computational power than a 128-bit key. Fifty supercomputers that could check a billion billion ($10^{18}$) AES keys per second (if such a device could ever be made) would, in theory, require about $3\times 10^{51}$ years to exhaust the 256-bit key space, cited from ``Brute-force attack'' in Wikipedia.

\subsection{Preliminaries: terminology, notation and definitions}

Undefined labelling definitions, terminology and algorithms mentioned here are cited from \cite{Bondy-2008} and \cite{Gallian2018}. The following terminology and notation will be used in this article:
\begin{asparaenum}[$\ast$ ]
\item Hanzis (Chinese characters) mentioned here are listed in GB2312-80 encoding of Chinese characters, in which there are 6763 simplified Chinese characters and 682 signs (another Chinese encoding is GBK,  formed in Oct. 1995, containing 21003 simplified Chinese characters and 883 signs, \cite{GB2312-80}).
\item A $(p,q)$-graph $G$ has $p$ vertices (nodes) and $q$ edges (links), notations $V(G)$ and $E(G)$ are the sets of  vertices and edges of $G$, respectively.
\item The number of elements of a set $X$ is called \emph{cardinality}, denoted as $|X|$.
\item The set of neighbors of a vertex $x$ is denoted as $N_{ei}(x)$, and the number of elements of the set $N_{ei}(x)$ is denoted as $|N_{ei}(x)|$, also, $|N_{ei}(x)|$ is called the \emph{degree} of the vertex $x$, very often, write $\textrm{deg}(x)=|N_{ei}(x)|$.
\item A vertex $u$ is called a ``leaf'' if its degree $\textrm{deg}(u)=|N_{ei}(u)|=1$.
\item A subgraph $H$ of a graph $G$ is called a \emph{vertex-induced subgraph} over a subset $S$ of $V(G)$ if $V(H)=S$ and $u,v\in S$ for any $uv\in E(H)$. Very often, we write this subgraph as $H=G[S]$.
\item An \emph{edge-induced graph} $G[E\,^*]$ over an edge subset $E\,^*$ of $E(G)$ is a subgraph having its edge set $E\,^*$ and its vertex set $V(G[E\,^*])\subseteq V(G)$ containing two ends of every edge of $E\,^*$.
\end{asparaenum}

We will use various labelling techniques of graph theory in this article.
\begin{defn}\label{defn:define-labelling-basic}
\cite{Yao-Sun-Zhang-Mu-Sun-Wang-Su-Zhang-Yang-Yang-2018arXiv} A \emph{labelling} $h$ of a graph $G$ is a mapping $h:S\subseteq V(G)\cup E(G)\rightarrow [a,b]$ such that $h(x)\neq h(y)$ for any pair of elements $x,y$ of $S$, and write the label set $h(S)=\{h(x): x\in S\}$. A \emph{dual labelling} $h'$ of a labelling $h$ is defined as: $h'(z)=\max h(S)+\min h(S)-h(z)$ for $z\in S$. Moreover, $h(S)$ is called the \emph{vertex label set} if $S=V(G)$, $h(S)$ the \emph{edge label set} if $S=E(G)$, and $h(S)$ the \emph{universal label set} if $S=V(G)\cup E(G)$. Furthermore, if $G$ is a bipartite graph with vertex bipartition $(X,Y)$, and holds $\max h(X)<\min h(Y)$, we call $h$ a \emph{set-ordered labelling} of $G$.\qqed
\end{defn}

We use a notation $S^2$ to denote the set of all subsets of a set $S$. For instance, $S=\{a,b,c\}$, so $S^2$ has its own elements: $\{a\}$, $\{b\}$, $\{c\}$, $\{a,b\}$, $\{a,c\}$, $\{b,c\}$ and $\{a,b,c\}$. The empty set $\emptyset$ is not allowed to belong to $S^2$ hereafter. We will use set-type of labellings defined in the following Definition \ref{defn:set-labelling}.

\begin{defn}\label{defn:set-labelling}
\cite{Yao-Sun-Zhang-Mu-Sun-Wang-Su-Zhang-Yang-Yang-2018arXiv} Let $G$ be a $(p,q)$-graph $G$. We have:

(i) A \emph{set mapping} $F: V(G)\cup E(G)\rightarrow [0, p+q]^2$ is called a \emph{total set-labelling} of $G$ if $F(x)\neq F(y)$ for distinct elements $x,y\in V(G)\cup E(G)$.

(ii) A \emph{vertex set mapping} $F: V(G) \rightarrow [0, p+q]^2$ is called a \emph{vertex set-labelling} of $G$ if $F(x)\neq F(y)$ for distinct vertices $x,y\in V(G)$.

(iii) An \emph{edge set mapping} $F: E(G) \rightarrow [0, p+q]^2$ is called an \emph{edge set-labelling} of $G$ if $F(uv)\neq F(xy)$ for distinct edges $uv, xy\in E(G)$.

(iv) A \emph{vertex set mapping} $F: V(G) \rightarrow [0, p+q]^2$ and a proper edge mapping $g: E(G) \rightarrow [a, b]$ are called a \emph{v-set e-proper labelling $(F,g)$} of $G$ if $F(x)\neq F(y)$ for distinct vertices $x,y\in V(G)$ and two edge labels $g(uv)\neq g(wz)$ for distinct edges $uv, wz\in E(G)$.

(v) An \emph{edge set mapping} $F: E(G) \rightarrow [0, p+q]^2$ and a proper vertex mapping $f: V(G) \rightarrow [a,b]$ are called an \emph{e-set v-proper labelling $(F,f)$} of $G$ if $F(uv)\neq F(wz)$ for distinct edges $uv, wz\in E(G)$ and two vertex labels $f(x)\neq f(y)$ for distinct vertices $x,y\in V(G)$.\qqed
\end{defn}

\section{Translating Chinese characters into graphs}

Hanzis, also Chinese characters, are expressed in many forms, such as: font, calligraphy, traditional Chinese characters, simplified Chinese characters, brush word, \emph{etc.} As known, \emph{China Online Dictionary} includes Xinhua Dictionary, Modern Chinese Dictionary, Modern Idiom Dictionary, Ancient Chinese Dictionary, and other  12 dictionaries total, China Online Dictionary contains about 20950 Chinese characters; 520,000 words; 360,000 words (28,770 commonly used words); 31920 idioms; 4320 synonyms; 7690 antonyms; 14000 allegorical sayings; 28070 riddles; and famous aphorism 19420.

\subsection{Two types of Chinese characters}

In general, there are two type of Chinese characters used in the world, one is called \emph{traditional Chinese characters} and another one \emph{simplified Chinese characters}, see Fig.\ref{fig:jian-fan-zi}. We, very often, call a traditional  Chinese characters or a simplified Chinese characters as a \emph{Hanzi} (Chinese character).

The stroke number of a Hanzi $S^{CC}$ is less than that of the traditional Chinese character $T^{CC}$ corresponding with $S^{CC}$, in general. We can compute the difference of two strokes of two-type Chinese characters $S^{CC}$ and $T^{CC}$, denoted as $D(H)=s(T^{CC})-s(S^{CC})$. For example, $D(H_{13})=s(T^{CC}_{13})-s(S^{CC}_{13})=24-3=21$, where the Hanzi $T^{CC}_{13}=S^{CC}_{13}$ is shown in Fig.\ref{fig:jian-fan-zi}(13). And, $D(H_{3})=s(T^{CC}_{3})-s(S^{CC}_{3})=23-3=20$, where the Hanzi $T^{CC}_{3}=S^{CC}_{3}$ is shown in Fig.\ref{fig:jian-fan-zi}(3).

\begin{figure}[h]
\centering
\includegraphics[height=2.2cm]{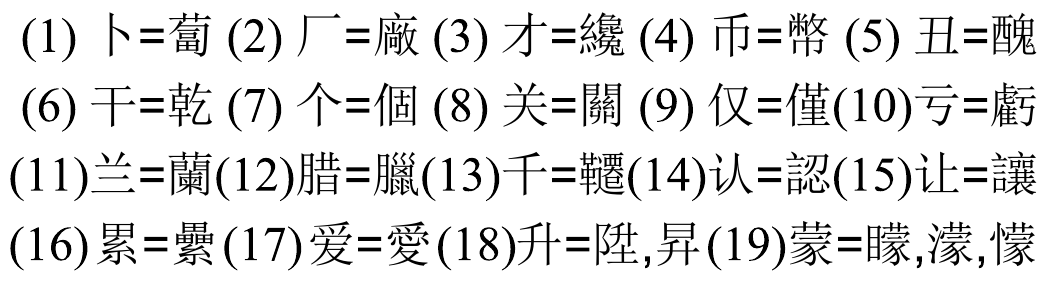}\\
\caption{\label{fig:jian-fan-zi} {\small Left is a simplified Chinese character, and Right is a traditional Chinese character in each of equations above.}}
\end{figure}

Some Hanzis are no distinguishing about traditional Chinese characters and simplified Chinese characters.

\begin{figure}[h]
\centering
\includegraphics[height=2.8cm]{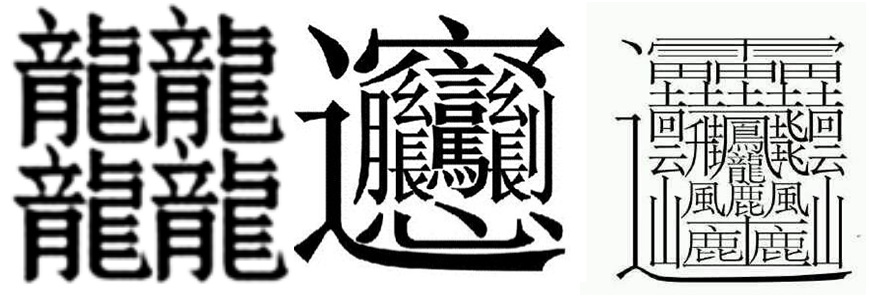}\\
\caption{\label{fig:largest-strokes} {\small Three Chinese characters with more strokes, Left has 64 strokes, Middle has 56 strokes.}}
\end{figure}

\subsection{Different fonts of Hanzis}

There are four fonts in printed Hanzis. In Fig.\ref{fig:different-font}, we give four basic fonts: \emph{Songti}, \emph{Fangsong}, \emph{Heiti} and \emph{Kaiti}. Clearly, there are differences in some printed Hanzis. These differences will be important for us when we build up mathematical models of Hanzis.

\begin{figure}[h]
\centering
\includegraphics[height=4.2cm]{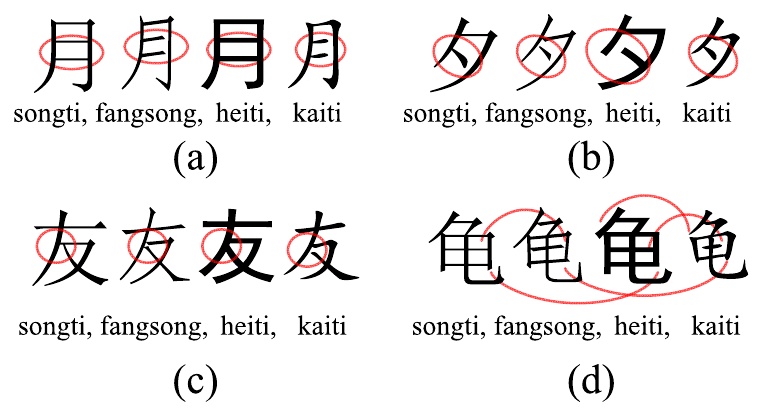}\\
\caption{\label{fig:different-font} {\small Difference between four fonts in printed Hanzis.}}
\end{figure}

\subsection{Matching behaviors of Hanzi-graphs}

\vskip 0.4cm

\subsubsection{Dui-lians, also, Chinese couplets} In Chinese culture, a sentence, called ``\emph{Shang-lian}'', has its own matching sentence, named as ``\emph{Xia-lian}'', and two sentences Shang-lian and Xia-lian form a \emph{Chinese couplet}, refereed as  ``\emph{Dui-lian}'' in Chinese. The sentence (a) of Fig.\ref{fig:duilian-1} is a Shang-lian, and the sentence (b) of Fig.\ref{fig:duilian-1} is a Xia-lian of the Shang-lian (a). We can use Dui-lians to design Topsnut-gpws. For example, we can consider the Shang-lian (a) shown in Fig.\ref{fig:duilian-1} as a public key, the Xia-lian (b) shown in Fig.\ref{fig:duilian-1} as a private key, and the Dui-lian (c) as the authentication. Moreover, the Dui-lian (c) can be made as a public key, and it has its own matching Dui-lian (d) as a private key, we have the authentication (e) of two Dui-lians (c) and (d). However, Dui-lians have their complex, for instance, the Shang-lian (f) shown in Fig.\ref{fig:duilian-1} has over $70,000$ candidate private keys. As known, a Dui-lian ``Chongqing Yonglian'' written by Xueyi Long has 1810 Hanzis. Other particular Chinese couplets are shown in Fig.\ref{fig:duilian-2} and Fig.\ref{fig:duilian-3}.

\begin{figure}[h]
\centering
\includegraphics[height=4.6cm]{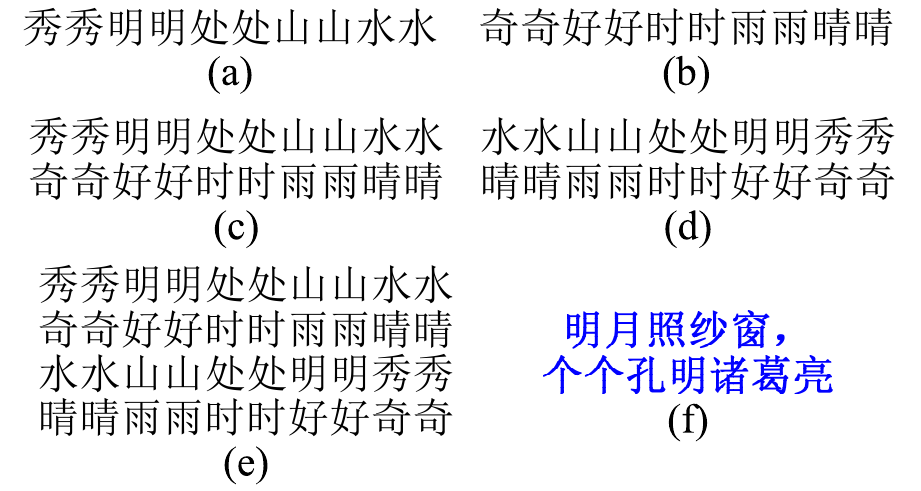}\\
\caption{\label{fig:duilian-1} {\small Couplets: (a) is a public key; (b) is a private key; (c) is the authentication of the public key (a) and the private key (b); (d) is a private key matching with the public key (c); (e) is the authentication of (c) and (d); (f) is a famous public key having no matching, although there are over $70,000$ candidate private keys.}}
\end{figure}

\begin{figure}[h]
\centering
\includegraphics[height=5.8cm]{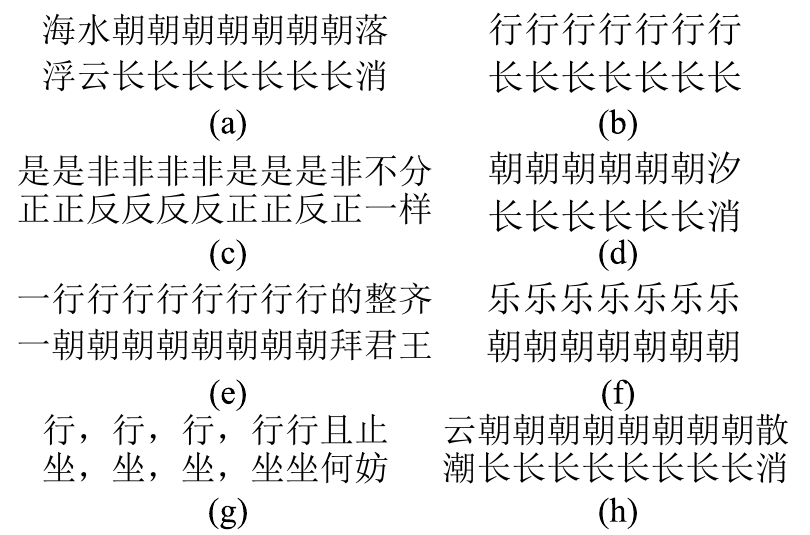}\\
\caption{\label{fig:duilian-2} {\small One Hanzi (Chinese character) may appear two or more times in a Dui-lian (couplet).}}
\end{figure}

\begin{figure}[h]
\centering
\includegraphics[width=9cm]{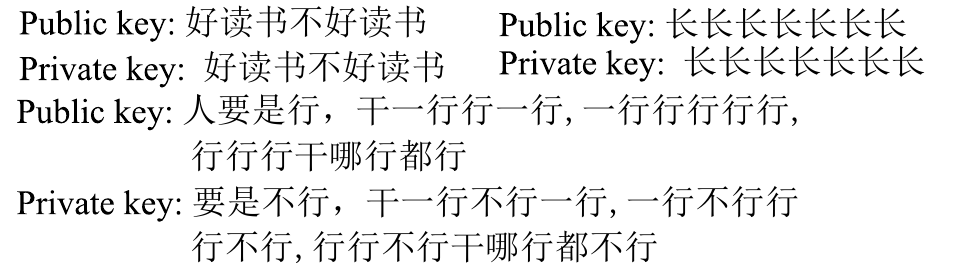}\\
\caption{\label{fig:duilian-3} {\small The same Hanzis in a couplet.}}
\end{figure}

\vskip 0.4cm

\subsubsection{Conundrums in Chinese} Chinese riddles (also ``Miyu'') are welcomed by Chinese people, and Chinese riddles appear in many where and actions of China. (see Fig. \ref{fig:riddles-miyu})

\begin{figure}[h]
\centering
\includegraphics[height=6.2cm]{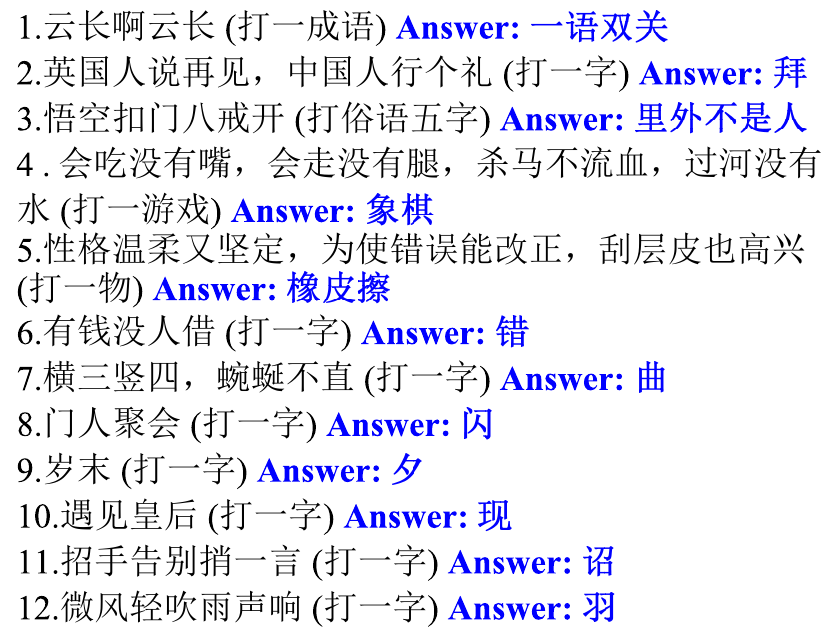}\\
\caption{\label{fig:riddles-miyu} {\small Twelves Chinese conundrums.}}
\end{figure}

\subsubsection{Chinese Xie-hou-yu} ``\emph{Xie-hou-yu}'' is a two-part allegorical saying, of which the first part, always stated, is descriptive, while the second part, sometimes unstated, carries the message (see Fig.\ref{fig:xiehouyu-2-part}).

\begin{figure}[h]
\centering
\includegraphics[height=5.8cm]{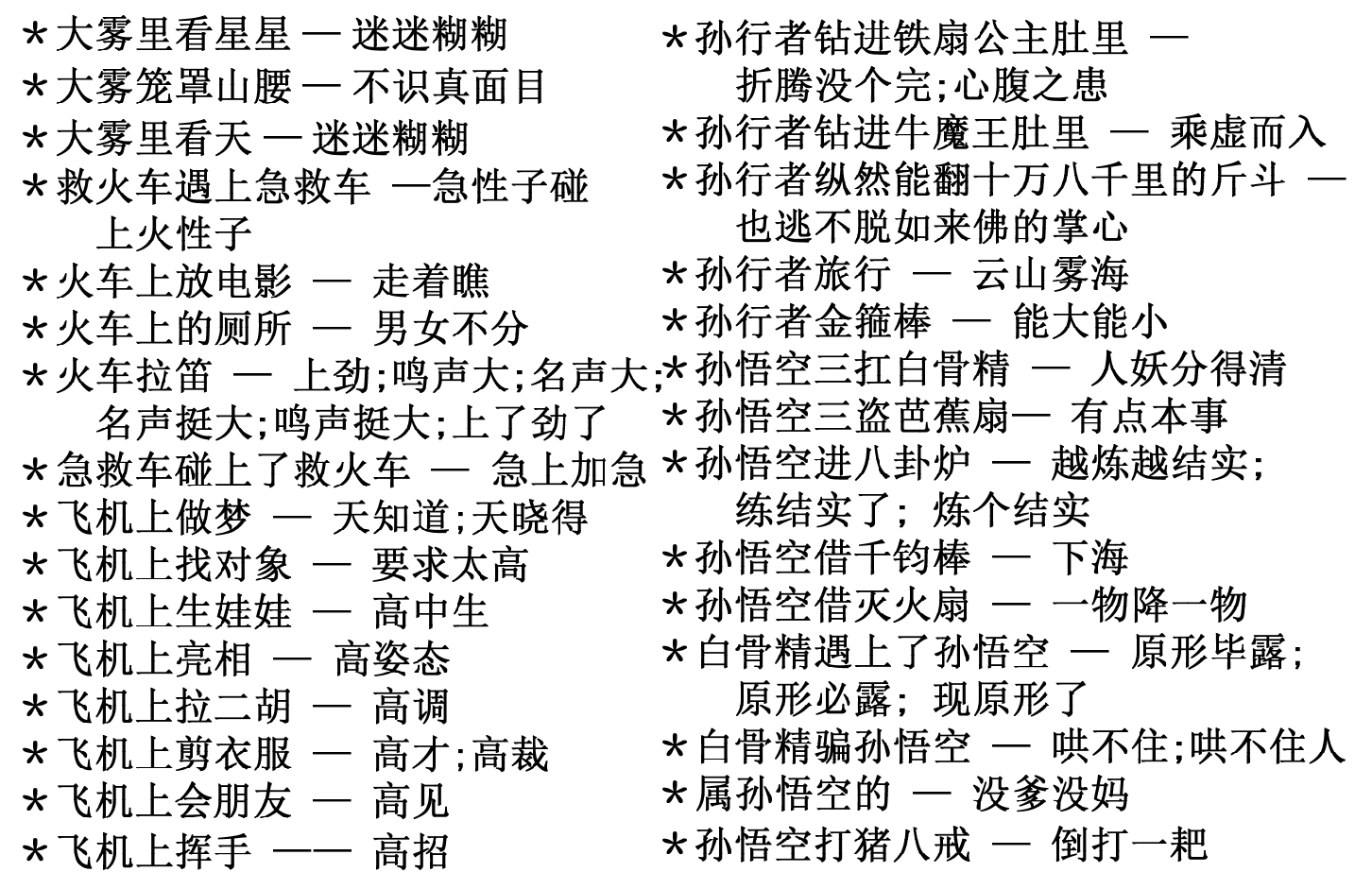}\\
\caption{\label{fig:xiehouyu-2-part} {\small Some Xie-hou-yus in Chinese.}}
\end{figure}

\subsubsection{Chinese tongue twisters} Chinese tongue twisters are often applied in Chinese comic dialogue (cross talk), which are popular in China. (see Fig. \ref{fig:tongue-twisters})

\begin{figure}[h]
\centering
\includegraphics[height=6cm]{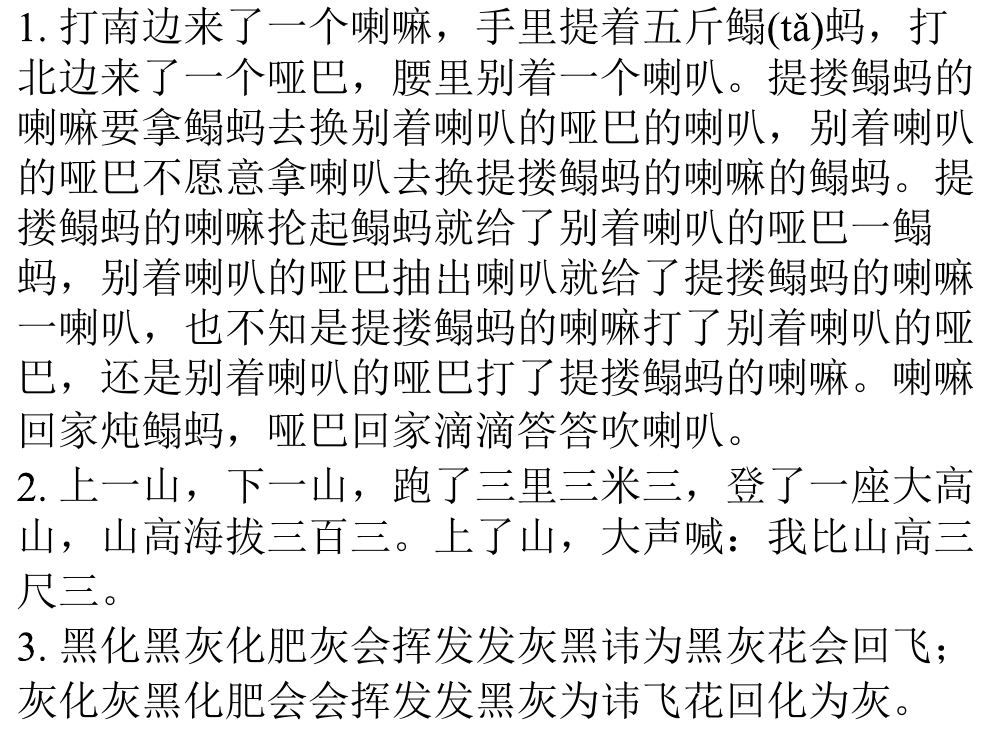}\\
\caption{\label{fig:tongue-twisters} {\small Three Chinese tongue twisters.}}
\end{figure}

\subsubsection{Understanding by insight, homonyms} Such examples shown in Fig.\ref{fig:homonyms}.
\begin{figure}[h]
\centering
\includegraphics[width=8.6cm]{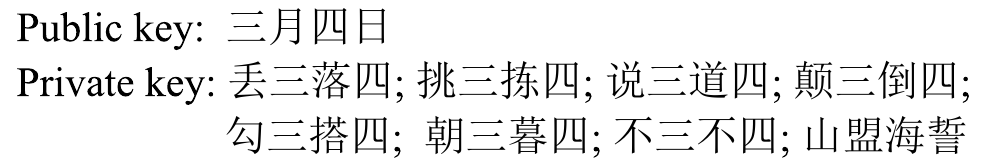}\\
\caption{\label{fig:homonyms} {\small Private keys obtained by homonyms, or understanding by insight.}}
\end{figure}

\subsubsection{Same pronunciation, same Pianpang} In Fig.\ref{fig:same-pronunciation-pianpang}, we can see eight Hanzis with the same pronunciation shown in Fig.\ref{fig:same-pronunciation-pianpang} (a) and ten Hanzis with the same Pianpang shown in Fig.\ref{fig:same-pronunciation-pianpang} (b). Moreover, all Hanzia have the same pronunciation ``ji'' in a famous Chinese paragraph  shown in Fig.\ref{fig:ji-ji-ji}.

\begin{figure}[h]
\centering
\includegraphics[height=3cm]{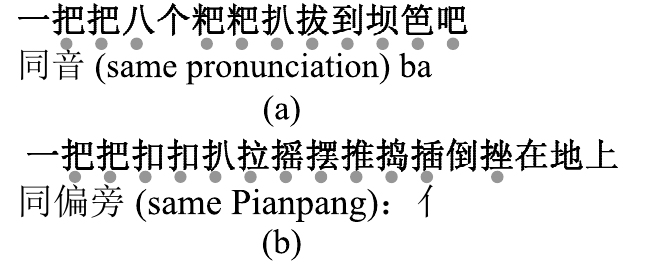}\\
\caption{\label{fig:same-pronunciation-pianpang} {\small (a) Ten Hanzis with the same pronunciation; (b) twelves Hanzis with the same Pianpang.}}
\end{figure}

\begin{figure}[h]
\centering
\includegraphics[width=8cm]{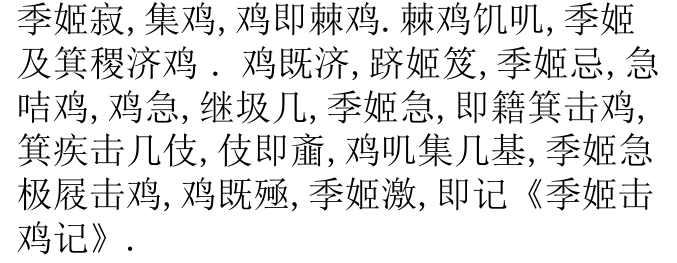}\\
\caption{\label{fig:ji-ji-ji} {\small All Hanzia have the same pronunciation ``ji''.}}
\end{figure}

\subsubsection{Chinese dialects} (also, ``Fangyan'') One Chinese word may have different replacements in local dialects of Chinese. For example, father, daddy can be substituted as Fig.\ref{fig:dialect-11}. And, different expressions of a sentence ``Daddy, where do you go?'' is shown in Fig.\ref{fig:dialect-22}.

\begin{figure}[h]
\centering
\includegraphics[width=8cm]{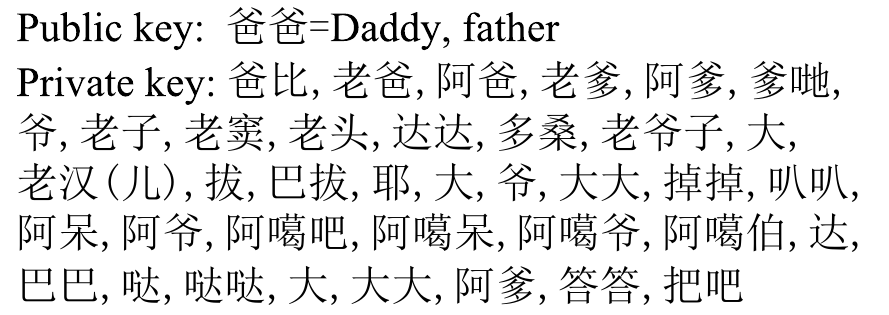}\\
\caption{\label{fig:dialect-11} {\small Father, daddy in Chinese dialects.}}
\end{figure}

\begin{figure}[h]
\centering
\includegraphics[width=8cm]{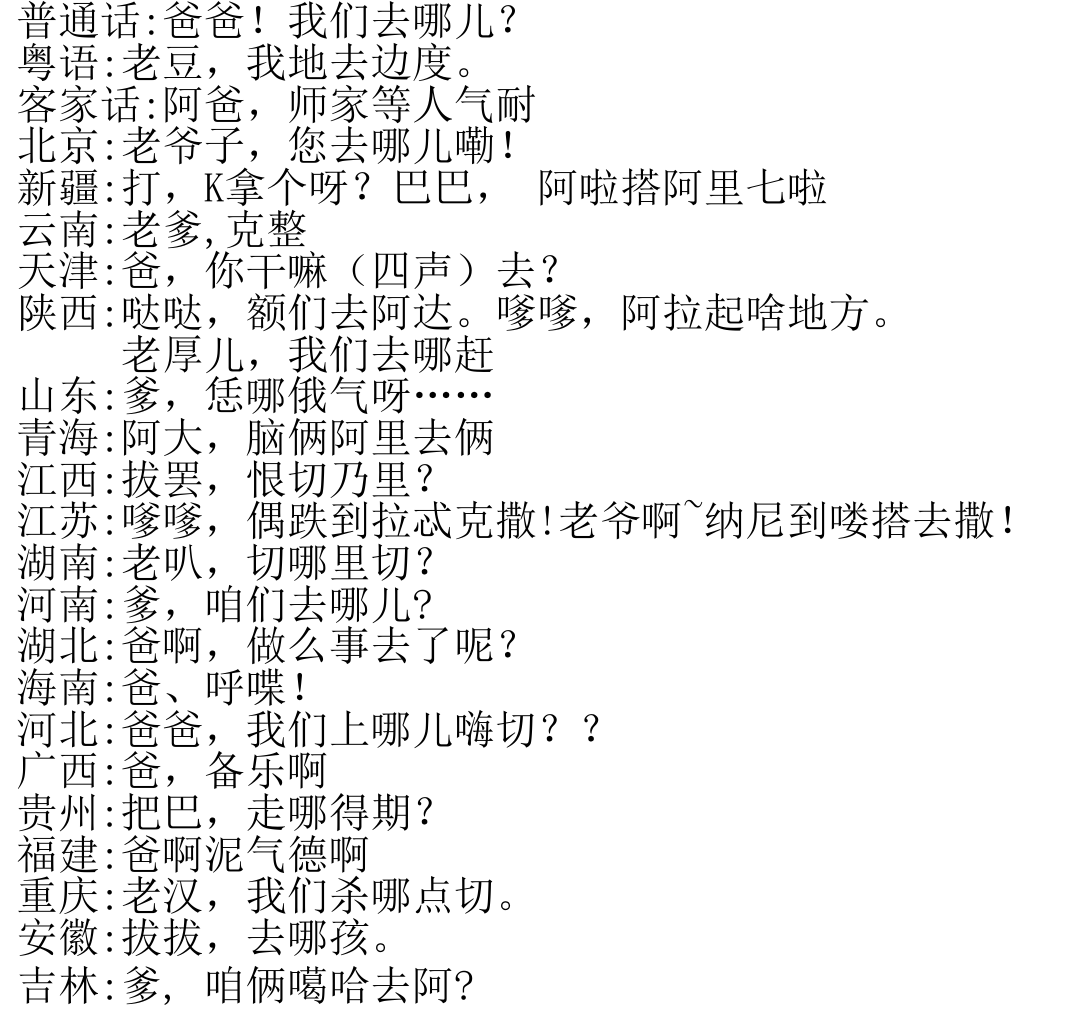}\\
\caption{\label{fig:dialect-22} {\small Daddy, where do you go?}}
\end{figure}

\subsubsection{Split Hanzis, building Hanzis}

An example is shown in Fig.\ref{fig:split-building-words} (a) for illustrating ``split a word into several words'', and Fig.\ref{fig:split-building-words} (b) is for building words by a given word.

\begin{figure}[h]
\centering
\includegraphics[height=8.6cm]{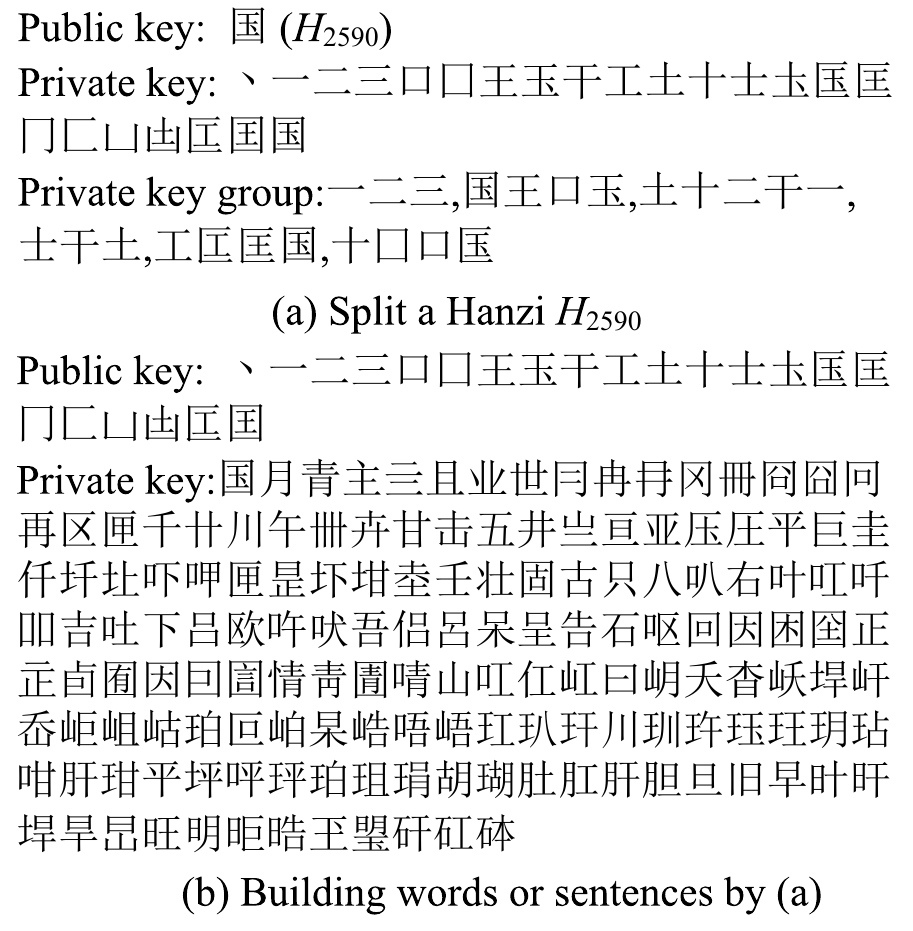}\\
\caption{\label{fig:split-building-words} {\small (a) Split a word into several words; (b) building words by a group of words obtained from splitting a given word.}}
\end{figure}

\subsubsection{Explaining Hanzis} See examples are shown in Fig.\ref{fig:explain-words}.

\begin{figure}[h]
\centering
\includegraphics[width=8.2cm]{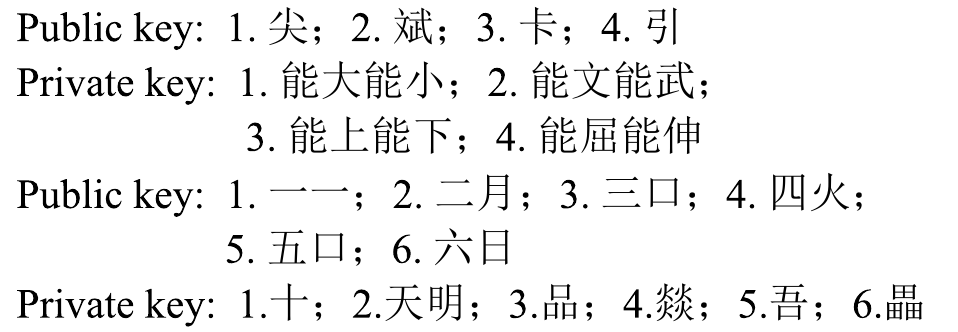}\\
\caption{\label{fig:explain-words} {\small Explaining words.}}
\end{figure}

\subsubsection{Tang poems}

As known, there are at least 5880195 Tang poems in China (see Fig.\ref{fig:Tang-poems}).
\begin{figure}[h]
\centering
\includegraphics[width=8.2cm]{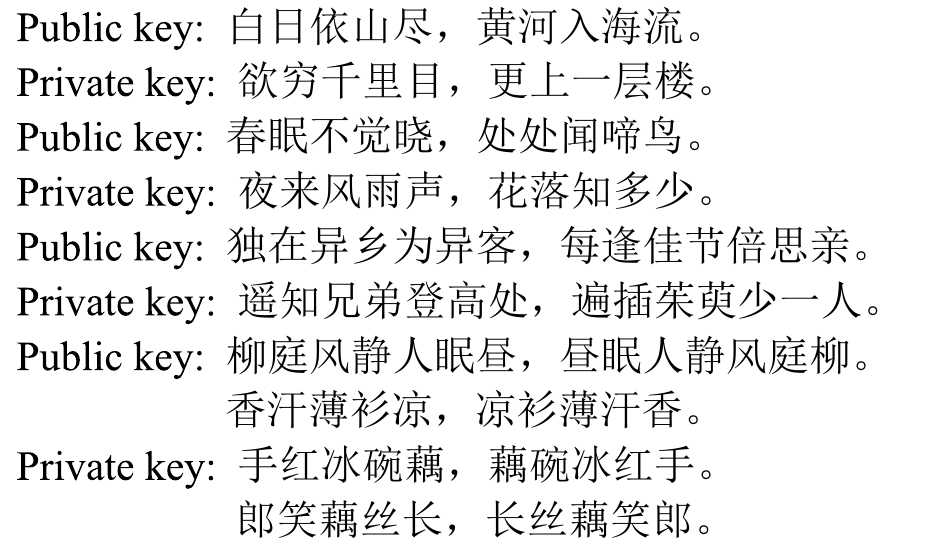}\\
\caption{\label{fig:Tang-poems} {\small Tang poems.}}
\end{figure}

\subsubsection{Idioms and Hanzi idiom-graphs}

A Hanzi idiom-graph $H$ (see Fig.\ref{fig:2-idioms}) is one labelled with Hanzi idioms by a vertex labelling $f$, two vertices $u,v$ are joined by an edge labelled with $f(uv)=f(u)\cap f(v)$.

\begin{figure}[h]
\centering
\includegraphics[width=8cm]{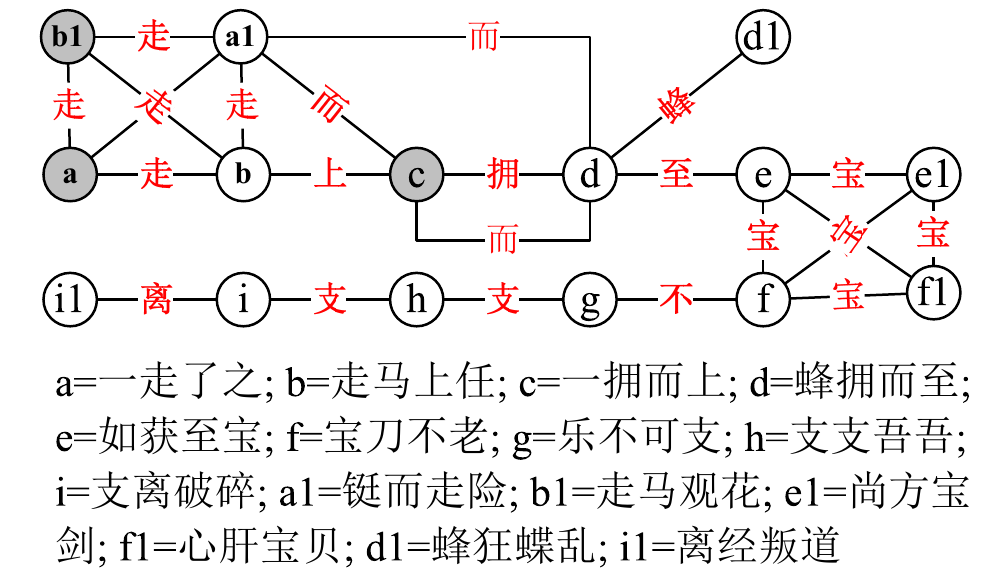}\\
\caption{\label{fig:2-idioms} {\small A Hanzi idiom-graph.}}
\end{figure}

\vskip 0.4cm

\subsubsection{Traditional Chinese characters are complex than Simplified  Chinese characters} Expect the stroke number of a traditional Chinese character is greater than that of a simplified  Chinese character, the usage of some traditional Chinese characters, also, is not unique, such examples are shown in Fig.\ref{fig:two-vs-comparine}.

\begin{figure}[h]
\centering
\includegraphics[width=8cm]{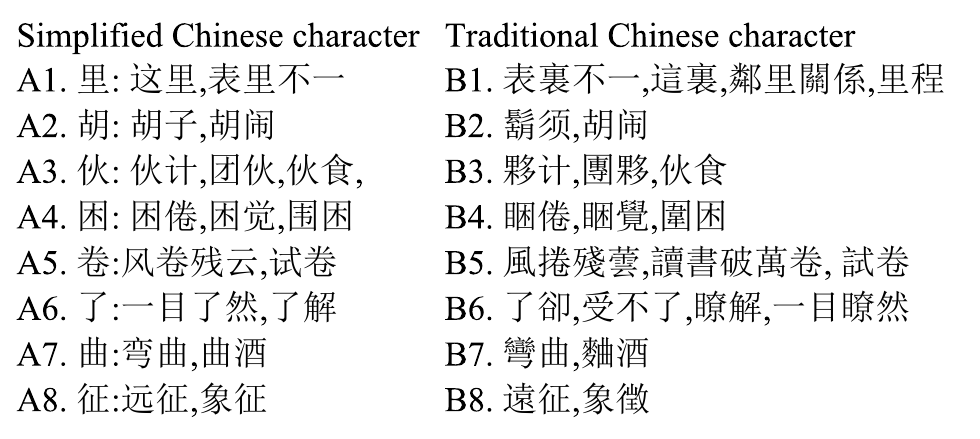}\\
\caption{\label{fig:two-vs-comparine} {\small The usage of a traditional Chinese character B$x$ with $x\in [1,8]$ is not unique.}}
\end{figure}

\vskip 0.4cm

\subsubsection{Configuration in Hanzis}

\begin{asparaenum}[$\ast$ ]
\item Symmetry means that Hanzis posses horizontal symmetrical structures, or vertical symmetrical structures, or two directional symmetries. We select some Hanzis having symmetrical structures in Fig.\ref{fig:symmetrical-overlapping} (a), (b), (c) and (f).
\item Overlapping Hanzis. See some overlapping Hanzis shown in Fig.\ref{fig:symmetrical-overlapping} (d), (e), (f) and (g). Moreover, in Fig.\ref{fig:symmetrical-overlapping} (g), a Hanzi $H_{4311}$ (read `shu\={a}ng') (\emph{2-overlapping Hanzi}) is consisted of two Hanzis $H_{5154}$ (read `y\`{o}u'), and another $H_{5458}$ (read `ru\`{o}') (\emph{3-overlapping Hanzi}) is consisted of three Hanzis $H_{5154}$. Moreover, four Hanzi $H_{5154}$ construct a Hanzi (read as `zhu\'{o}', \emph{4-overlapping Hanzi}).
\end{asparaenum}

\begin{figure}[h]
\centering
\includegraphics[height=5.4cm]{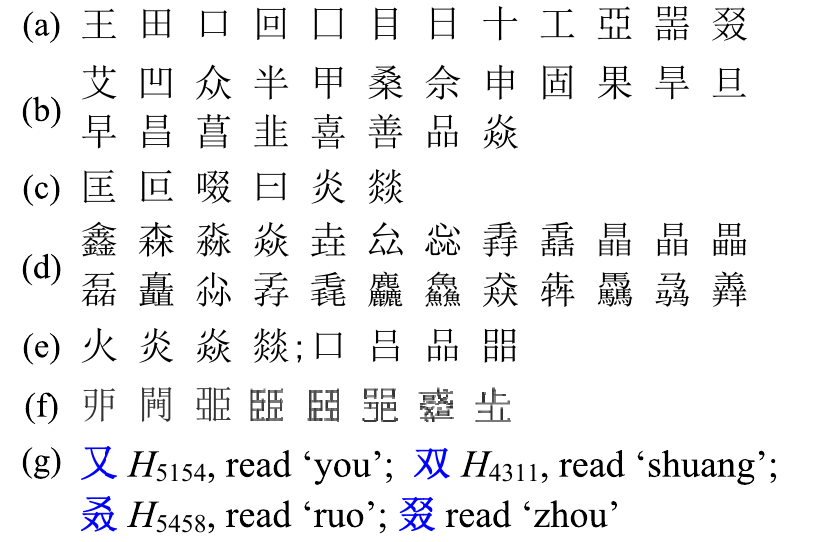}\\
\caption{\label{fig:symmetrical-overlapping} {\small Hanzis with symmetrical structure and overlapping structure.}}
\end{figure}

\subsection{Mathematical models of Hanzis}

We will build up mathematical models of Hanzis, called Hanzi-graphs, in this subsection.

\vskip 0.4cm

\subsubsection{The existing expressions of Hanzis}

In fact, a Hanzi has been expressed in the way: (1) a ``\emph{pinyin}'' in oral communication, for example, the pinyin ``r\'{e}n'' means ``man'', but it also stands for other 12 Hanzis at least (see Fig.\ref{fig:same-pinyin-ren}(a)); (2) a word with four English letters and numbers of $0,1,2,\dots ,8,9$, for instance, ``r\'{e}n''=4EBA (see Fig.\ref{fig:same-pinyin-ren}(b), also called a \emph{code}); (3) a number code ``4043'' defined in ``GB2312-80 Encoding of Chinese characters'' \cite{GB2312-80}S, which is constituted by $0,1,2,\dots ,8,9$ (see Fig.\ref{fig:same-pinyin-ren}(c)).

Clearly, the above three ways are not possible for making passwords with bytes as long as desired. We introduce the fourth way, named as Topsnut-gpw, see an example shown in Fig.\ref{fig:example-1}(c).

\begin{figure}[h]
\centering
\includegraphics[height=2.8cm]{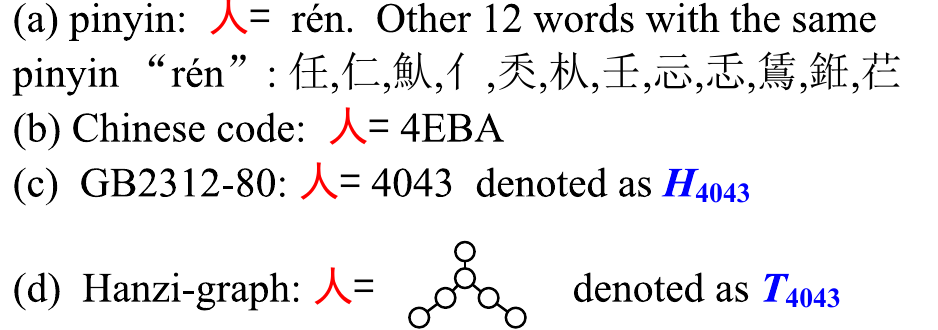}\\
\caption{\label{fig:same-pinyin-ren} {\small Four expressions of a Hanzi $H_{4043}$ (= man).}}
\end{figure}

As known, Hanzi-graphs are saved in computer by popular matrices, see a Hanzi-graph $T_{4706}$ shown in Fig.\ref{fig:example-1} (b) and its  matrix $A(T_{4706})$ shown in Fig.\ref{fig:example-1-matrix} (a).

\begin{figure}[h]
\centering
\includegraphics[height=5cm]{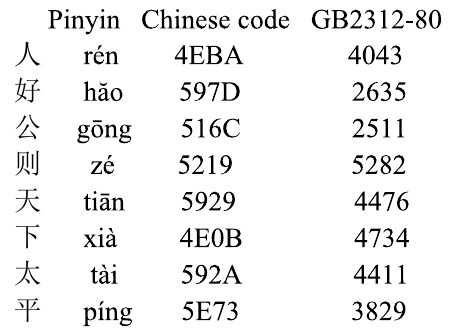}\\
\caption{\label{fig:rrhg-ztxtp} {\small Three substituted expressions of eight Hanzis.}}
\end{figure}

\begin{figure}[h]
\centering
\includegraphics[height=3.6cm]{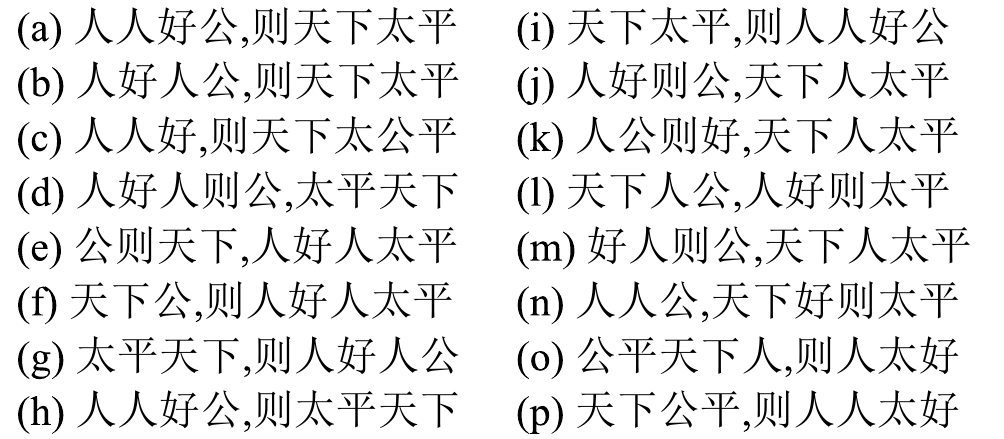}\\
\caption{\label{fig:rrhg-GB2312-80} {\small The meaningful paragraphs obtained from nine Hanzis of GB2312-80.}}
\end{figure}

In Fig.\ref{fig:rrhg-GB2312-80}, we use two expressions (a1) and (a2) to substitute a Chinese sentence (a), that is, (a)=(a1), or (a)=(a2). By this method, we have

\begin{flushleft}
(a1) $H_{4043}H_{4043}H_{2635}H_{2511}H_{5282}H_{4476}H_{4734}H_{4411}H_{3829}$;\\
(a2) $404340432635251152824476473444113829$.

(b1) $H_{4043}H_{2635}H_{4043}H_{2511}H_{5282}H_{4476}H_{4734}H_{4411}H_{3829}$;\\
(b2) $404326354043251152824476473444113829$

(c1) $H_{4043}H_{4043}H_{2635}H_{5282}H_{4476}H_{4734}H_{4411}H_{2511}H_{3829}$;\\
(c2) $404340432635528244764734441125113829$

(d1) $H_{4043}H_{2635}H_{4043}H_{5282}H_{2511}H_{4411}H_{3829}H_{4476}H_{4734}$;\\
(d2) $404326354043528225114411382944764734$

(e1) $H_{2511}H_{5282}H_{4476}H_{4734}H_{4043}H_{2635}H_{4043}H_{4411}H_{3829}$;\\
(e2) $251152824476473440432635404344113829$

(f1) $H_{4476}H_{4734}H_{2511}H_{5282}H_{4043}H_{2635}H_{4043}H_{4411}H_{3829}$;\\
(f2) $447647342511528240432635404344113829$

(g1) $H_{4411}H_{3829}H_{4476}H_{4734}H_{5282}H_{4043}H_{2635}H_{4043}H_{2511}$;\\
(g2) $441138294476473452824043263540432511$

(h1) $H_{4043}H_{4043}H_{2635}H_{2511}H_{5282}H_{4411}H_{3829}H_{4476}H_{4734}$;\\
(h2) $404340432635251152824411382944764734$

(i1) $H_{4476}H_{4734}H_{4411}H_{3829}H_{5282}H_{4043}H_{4043}H_{2635}H_{2511}$;\\
(i2) $447647344411382952824043404326352511$

(j1) $H_{4043}H_{2635}H_{5282}H_{2511}H_{4476}H_{4734}H_{4043}H_{4411}H_{3829}$;\\
(j2) $404326355282251144764734404344113829$

(k1) $H_{4043}H_{2511}H_{5282}H_{2635}H_{4476}H_{4734}H_{4043}H_{4411}H_{3829}$;\\
(k2) $404325115282263544764734404344113829$

(l1) $H_{4476}H_{4734}H_{4043}H_{2511}H_{4043}H_{2635}H_{5282}H_{4411}H_{3829}$;\\
(l2) $447647344043251140432635528244113829$
\end{flushleft}

Fig.\ref{fig:rrhg-GB2312-80} shows some permutations of nine Hanzis $H_{4043}$, $H_{4043}$, $H_{2635}$, $H_{2511}$, $H_{5282}$, $H_{4476}$, $H_{4734}$, $H_{4411}$, $H_{3829}$. In fact, there are about $9!=362,880$ permutations made by these nine Hanzis. If a paragraph was made by a fixed group $W$ of 50 Hanzis, then we may have about $50!\approx 2^{214}$ paragraphs made by the same group $W$. So, we have enough large space of Hanzi-graphs for making Hanzi-gpws.

\vskip 0.4cm

\subsubsection{Basic rules for Hanzi-graphs} For the task of building mathematical models of Hanzis, called \emph{Hanzi-graphs}, we give some rules for transforming Hanzis into Hanzi-graphs.

\begin{asparaenum}[\underline{Rule-}1 ]
\item \textbf{Stroke rule.} It is divided into several parts by the strokes of Hanzis. Some examples are shown in Fig.\ref{fig:pianpang-1} and Fig.\ref{fig:pianpang-2} based on ``Pianpang'' of Hanzis.

\item \textbf{Crossing and overlapping rules.} Hanzi-graphs are obtained by the crossing and overlapping rules (see Fig.\ref{fig:pianpang-2} (a) and Fig.\ref{fig:cross-overlapping}).

\begin{figure}[h]
\centering
\includegraphics[height=4cm]{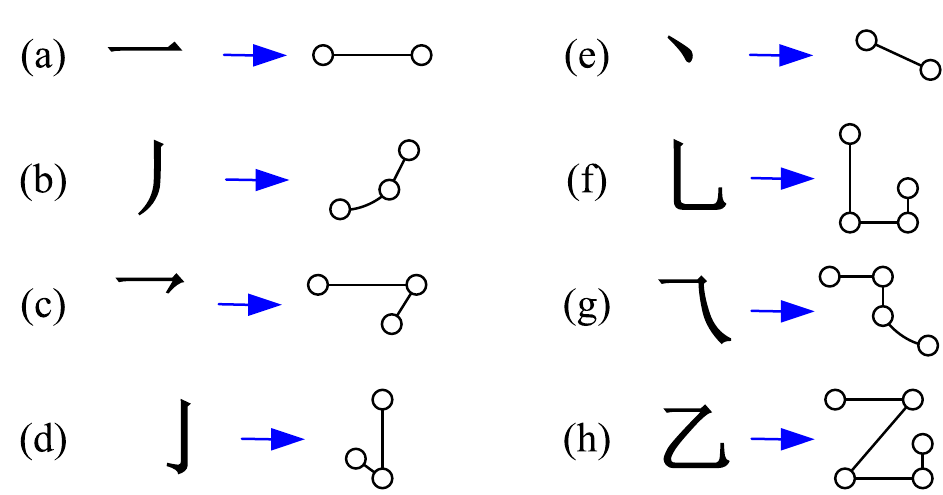}\\
\caption{\label{fig:pianpang-1} {\small Hanzi-graphs with one stroke, in which Hanzi-graphs (b), (c) and (d) can be considered as one from the topology of view. So, (a) and (e) are the same Hanzi-graph.}}
\end{figure}

\begin{figure}[h]
\centering
\includegraphics[height=5cm]{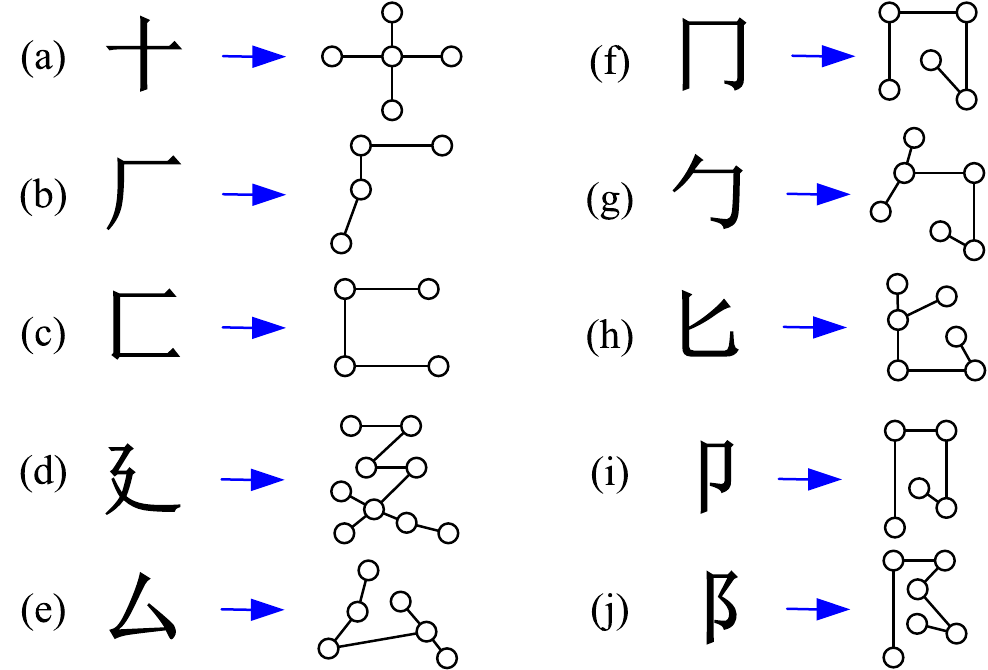}\\
\caption{\label{fig:pianpang-2} {\small Hanzi-graphs with two strokes. According to the topology of view, Hanzi-graphs (b), (c) and Hanzi-graphs (f) and (g) shown in Fig.\ref{fig:pianpang-1} can be considered as one; Hanzi-graphs (e), (f), (i) and Hanzi-graphs (h) shown in Fig.\ref{fig:pianpang-1} are the same; and Hanzi-graphs (g) and (h) are the same.}}
\end{figure}

\begin{figure}[h]
\centering
\includegraphics[height=4.4cm]{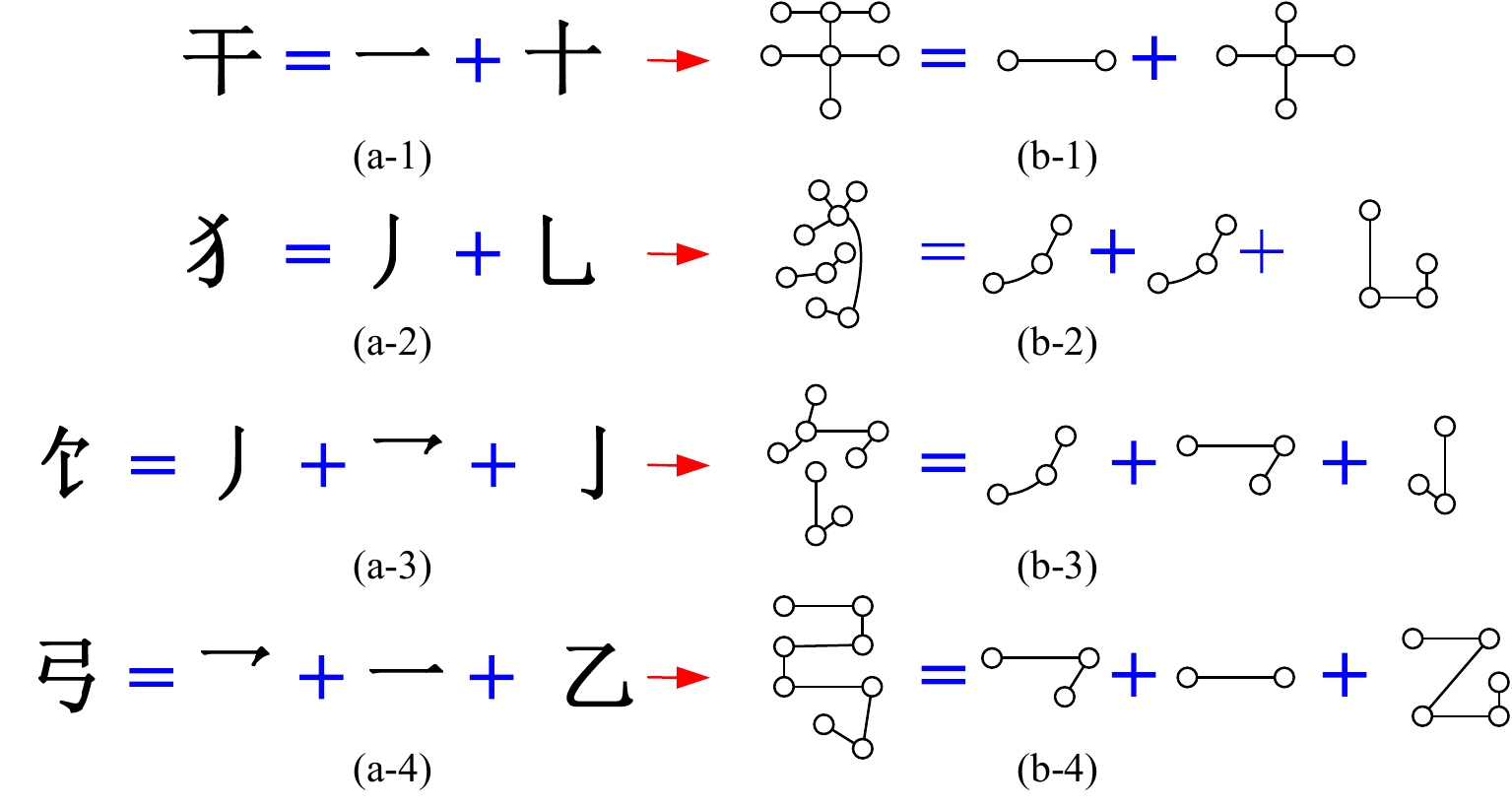}\\
\caption{\label{fig:cross-overlapping} {\small Hanzi-(a-$k$) is transformed into Hanzi-graph-(b-$k$) with $k\in[1,4]$.}}
\end{figure}

\item \textbf{Set-orderedable rule.} We abide for fonts: songti, fangsong, heiti and kaiti to construct Hanzi-graphs that admit set-ordered graceful labellings (see Definition \ref{defn:proper-bipartite-labelling-ongraphs}). Some examples are shown in Fig.\ref{fig:split-standard-1} and Fig.\ref{fig:split-standard-2}.
\item \textbf{No odd-cycles.} We restrict our Hanzi-graphs have no odd-cycles for the guarantee of set-ordered graceful labellings (see Fig.\ref{fig:split-standard-2}). There are over 6763 Hanzis in \cite{GB2312-80}, and we have 3500 Hanzis in frequently used. So it is not an easy job to realize the set-ordered gracefulness of the Hanzi-graphs in \cite{GB2312-80}. Clearly, the 0-rotatable gracefulness of the Hanzi-graphs in \cite{GB2312-80} will be not slight, see Definition \ref{defn:graceful-0-rotatable}.
\end{asparaenum}

\begin{figure}[h]
\centering
\includegraphics[height=4.2cm]{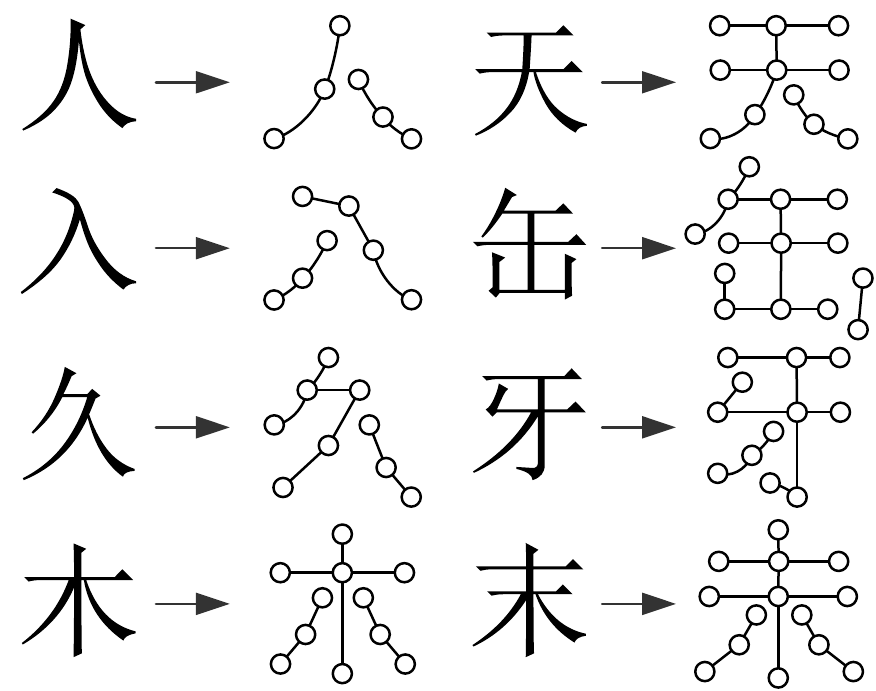}\\
\caption{\label{fig:split-standard-1} {\small First group of mathematical models of Hanzis components and radicals.}}
\end{figure}

\begin{figure}[h]
\centering
\includegraphics[height=6.8cm]{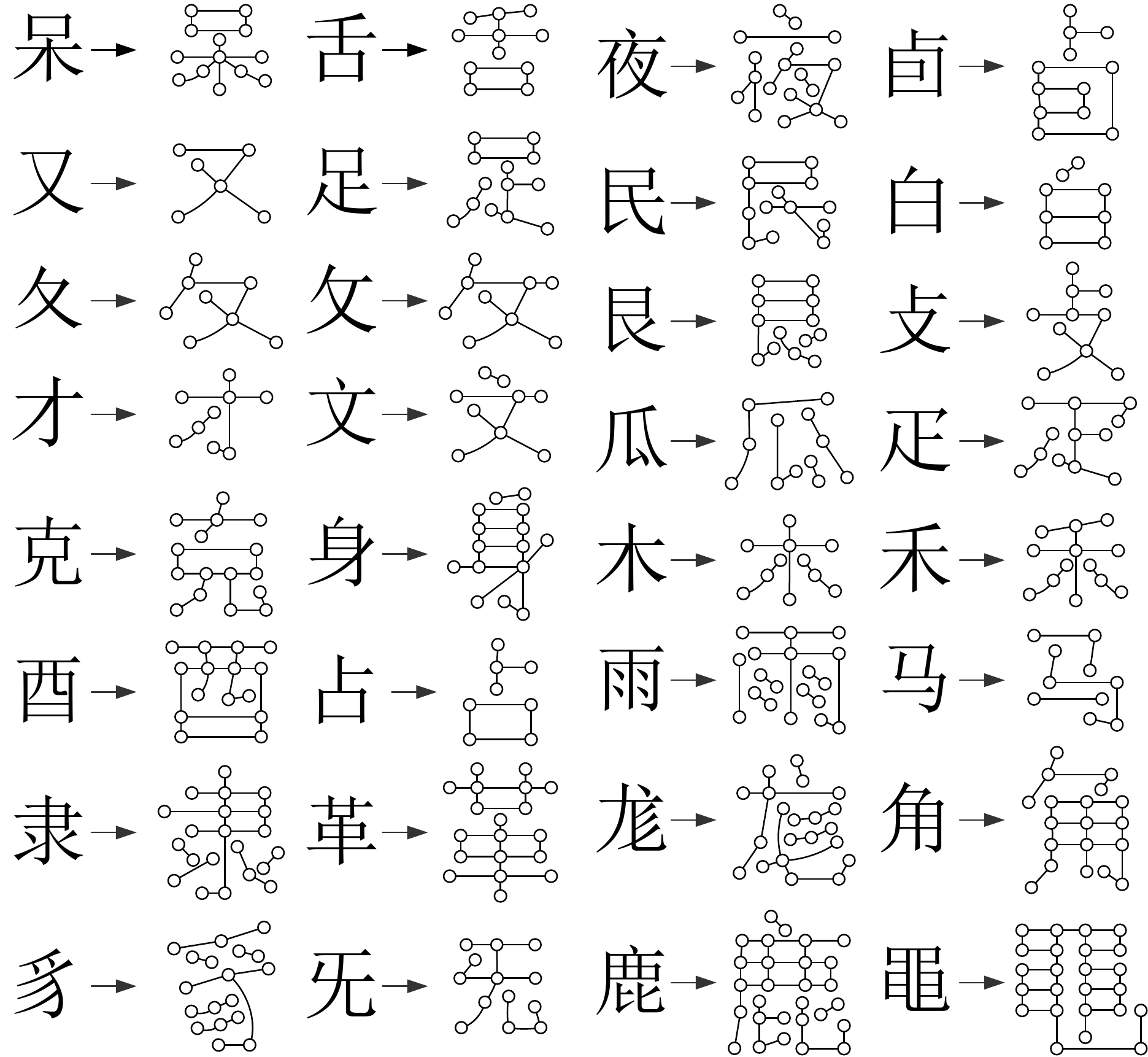}\\
\caption{\label{fig:split-standard-2} {\small Second group of mathematical models of Hanzis components and radicals.}}
\end{figure}

A group of Hanzi-graphs made by Rule-$k$ with $k\in[1,4]$ is shown in Fig.\ref{fig:rrhg-topological}. If a Hanzi-graph is disconnected, and has $k$ components, we refer to it as a \emph{$k$-Hanzi-graph} directly.

\begin{figure}[h]
\centering
\includegraphics[height=3.2cm]{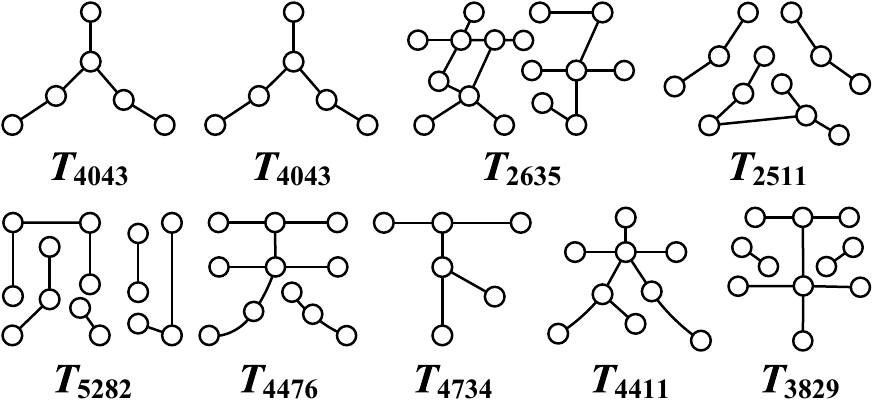}\\
\caption{\label{fig:rrhg-topological} {\small The topological structure (Hanzi-graphs) of Hanzis shown in Fig.\ref{fig:rrhg-ztxtp}.}}
\end{figure}

\subsection{Space of Hanzi-graphs}

A list of commonly used Hanzis in modern Chinese was issued by The State Language Work Committee and The State Education Commission in 1988, with a total of 3500 characters. The commonly used part of the Hanzis with a coverage rate of 97.97\% is about 2500 characters. This means that the commonly used 2500 characters can help us to make a vast space of Hanzi-graphs.

For example, the probability of a Hanzi appearing just once in a Chinese paragraph is a half, so the space of paragraphs made in Hanzis contains at lest $2^{2500}\sim 2^{6700}$ elements, which is far more than the number of sands on the earth. It is known that the number of sands on the earth is about $8\cdot 10^{22}\sim 1.3\cdot 10^{23}$, or about $2^{76.1}\sim 2^{77}$, someone estimates the number of sands on the earth as $5.1 \times 10^{18}\approx 2^{62}$.

\section{Mathematical techniques}

Since some Topsnut-gpws were made by graph coloting/labellings, we show the following definitions of graph coloting/labellings for easily reading and quickly working.

\subsection{Known labellings}

\begin{defn}\label{defn:edge-magic-total-graceful-labelling}
\cite{Yao-Mu-Sun-Zhang-Wang-Su-2018} An \emph{edge-magic total graceful labelling} $g$ of a $(p,q)$-graph $G$ is defined as: $g: V(G)\cup E(G)\rightarrow [1,p+q]$ such that $g(x)\neq g(y)$ for any two elements $x,y\in V(G)\cup E(G)$, and $g(uv)+|g(u)-g(v)|=k$ with a constant $k$ for each edge $uv\in E(G)$. Moreover, $g$ is \emph{super} if $\max g(E(G))<\min g(V(G))$ (or $\max g(V(G))<\min g(E(G))$).\qqed
\end{defn}

In Definition \ref{defn:proper-bipartite-labelling-ongraphs} we restate several known
labellings that can be found in \cite{Gallian2018}, \cite{Bing-Yao-Cheng-Yao-Zhao2009}, \cite{Zhou-Yao-Chen-Tao2012, Zhou-Yao-Chen2013} and \cite{Yao-Cheng-Yao-Zhao-2009}. We write $f(V(G))=\{f(u):u\in V(G)\}$ and $f(E(G))=\{f(uv):uv\in E(G)\}$ hereafter.

\begin{defn} \label{defn:proper-bipartite-labelling-ongraphs}
Suppose that a connected $(p,q)$-graph $G$ admits a mapping $\theta:V(G)\rightarrow \{0,1,2,\dots \}$. For edges $xy\in E(G)$
the induced edge labels are defined as $\theta(xy)=|\theta(x)-\theta(y)|$. Write $\theta(V(G))=\{\theta(u):u\in V(G)\}$,
$\theta(E(G))=\{\theta(xy):xy\in E(G)\}$. There are the following constraints:
\begin{asparaenum}[(a)]
\item \label{Proper01} $|\theta(V(G))|=p$.
\item \label{Proper02} $|\theta(E(G))|=q$.
\item \label{Graceful-001} $\theta(V(G))\subseteq [0,q]$, $\min \theta(V(G))=0$.
\item \label{Odd-graceful-001} $\theta(V(G))\subset [0,2q-1]$, $\min \theta(V(G))=0$.
\item \label{Graceful-002} $\theta(E(G))=\{\theta(xy):xy\in E(G)\}=[1,q]$.
\item \label{Odd-graceful-002} $\theta(E(G))=\{\theta(xy):xy\in E(G)\}=\{1,3,5,\dots ,2q-1\}=[1,2q-1]^o$.
\item \label{Set-ordered} $G$ is a bipartite graph with the bipartition
$(X,Y)$ such that $\max\{\theta(x):x\in X\}< \min\{\theta(y):y\in
Y\}$ ($\max \theta(X)<\min \theta(Y)$ for short).
\item \label{Graceful-matching} $G$ is a tree containing a perfect matching $M$ such that
$\theta(x)+\theta(y)=q$ for each edge $xy\in M$.
\item \label{Odd-graceful-matching} $G$ is a tree having a perfect matching $M$ such that
$\theta(x)+\theta(y)=2q-1$ for each edge $xy\in M$.
\end{asparaenum}

Then we have: a \emph{graceful labelling} $\theta$ satisfies (\ref{Proper01}), (\ref{Graceful-001}) and (\ref{Graceful-002}); a \emph{set-ordered
graceful labelling} $\theta$ holds (\ref{Proper01}), (\ref{Graceful-001}), (\ref{Graceful-002}) and (\ref{Set-ordered}) true;
a \emph{strongly graceful labelling} $\theta$ holds (\ref{Proper01}), (\ref{Graceful-001}), (\ref{Graceful-002}) and
(\ref{Graceful-matching}) true; a \emph{strongly set-ordered graceful labelling} $\theta$ holds (\ref{Proper01}), (\ref{Graceful-001}), (\ref{Graceful-002}), (\ref{Set-ordered}) and (\ref{Graceful-matching}) true. An \emph{odd-graceful labelling} $\theta$ holds (\ref{Proper01}), (\ref{Odd-graceful-001}) and (\ref{Odd-graceful-002}) true; a \emph{set-ordered odd-graceful labelling} $\theta$ holds (\ref{Proper01}), (\ref{Odd-graceful-001}), (\ref{Odd-graceful-002})
and (\ref{Set-ordered}) true; a \emph{strongly odd-graceful labelling} $\theta$ holds (\ref{Proper01}), (\ref{Odd-graceful-001}),
(\ref{Odd-graceful-002}) and (\ref{Odd-graceful-matching}) true; a \emph{strongly set-ordered odd-graceful labelling} $\theta$ holds (\ref{Proper01}), (\ref{Odd-graceful-001}), (\ref{Odd-graceful-002}), (\ref{Set-ordered}) and (\ref{Odd-graceful-matching}) true.\qqed
\end{defn}

\begin{defn}\label{defn:total-graceful-labelling}
A \emph{total graceful labelling} $f$ of a $(p,q)$-graph $G$ is defined as: $f: V(G)\cup E(G)\rightarrow [1,p+q]$ such that $f(uv)=|f(u)-f(v)|$ for each edge $uv\in E(G)$, and $f(x)\neq f(y)$ for any two elements $x,y\in V(G)\cup E(G)$. Moreover, $f$ is super if $\max f(E(G))<\min f(V(G))$ (or $\max f(V(G))<\min f(E(G))$).\qqed
\end{defn}

\begin{defn} \label{defn:connections-several-labellings}
Let $G$ be a $(p, q)$-graph having $p$ vertices and $q$ edges, and let $S_{k,d}=\{k, k + d, \dots , k +(q-1)d\}$ for integers $k\geq 1$ and $d\geq 1$.

(1) \cite{Gallian2018} A \emph{felicitous labelling} $f$ of $G$ holds: $f(V(G))\subseteq [0,q]$, $f(x)\neq f(y)$ for distinct $x,y\in V(G)$ and $f(E(G))=\{f(uv)=f(u)+f(v)~ (\bmod ~q): uv\in E(G)\}=[0,q-1]$; and furthermore, $f$ is \emph{super} if
$f(V(G))=[0,p-1]$.

(2) \cite{S-M-Hegde2000} A \emph{$(k,d)$-graceful labelling} $f$ of $G$ holds $f(V(G))\subseteq [0, k + (q-1)d]$, $f(x)\neq f(y)$ for distinct $x,y\in V(G)$ and $f(E(G))=\{|f(u)-f(v)|;\ uv\in E(G)\}=S_{k,d}$. Especially, a $(k,1)$-graceful labelling is also a $k$-graceful labelling.

(3) \cite{Gallian2018} An \emph{edge-magic total labelling} $f$ of $G$ holds $f(V(G)\cup E(G))=[1,p+q]$ such that for any edge $uv\in E(G)$, $f(u)+f(uv)+f(v)=c$ where the magic constant $c$ is a fixed integer; and furthermore $f$ is \emph{super} if $f(V(G))=[1,p]$.

(4) \cite{Gallian2018} A $(k,d)$-\emph{edge antimagic total labelling} $f$ of $G$ holds $f(V(G)\cup E(G))=[1,p+q]$ and
$\{f(u)+f(uv)+f(v): uv\in E(G)\}=S_{k,d}$, and furthermore $f$ is \emph{super} if $f(V(G))=[1,p]$.

(5) \cite{Zhou-Yao-Chen2013} An \emph{odd-elegant labelling} $f$ of $G$ holds $f(V(G))\subset [0,2q-1]$, $f(u)\neq f(v)$ for distinct $u,v\in V(G)$, and $f(E(G))=\{f(uv)=f(u)+f(v)~(\bmod~2q):uv\in E(G)\}$ $=$ $\{1,3,5,\dots, 2q-1\}=[1,2q-1]^o$.

(6) \cite{Acharya-Hegde1990} A labeling $f$ of $G$ is said to be $(k, d)$-\emph{arithmetic} if $f(V(G))\subseteq [0, k+(q-1)d]$, $f(x)\neq f(y)$ for distinct $x,y\in V(G)$ and $\{f(u)+f(v): uv\in E(G)\}=S_{k,d}$.

(7) \cite{Gallian2018} A \emph{harmonious labelling} $f$ of $G$ holds $f(V(G))\subseteq [0, q-1]$, $\min f(V(G))=0$ and $f(E(G))=\{f(uv)=f(u)+f(v)~ (\bmod ~q): uv\in E(G)\}=[0,q-1]$ such that (i) if $G$ is not a tree, $f(x)\neq f(y)$ for distinct $x,y\in V(G)$; (ii) if $G$ is a tree, $f(x)\neq f(y)$ for distinct $x,y\in V(G)\setminus\{w\}$, and $f(w)=f(x_0)$ for some $x_0\in V(G)\setminus\{w\}$.\qqed
\end{defn}

\begin{defn}\label{defn:relaxed-Emt-labelling}
\cite{Yao-Sun-Zhang-Mu-Sun-Wang-Su-Zhang-Yang-Yang-2018arXiv} Let $f:V(G)\cup E(G)\rightarrow [1,p+q]$ be a total labelling of a $(p,q)$-graph $G$. If there is a constant $k$ such that $f(u)+f(uv)+f(v)=k$, and each edge $uv$ corresponds another edge $xy$ holding $f(uv)=|f(x)-f(y)|$, then we name $f$ as a \emph{relaxed edge-magic total labelling} (relaxed Emt-labelling) of $G$ (called a \emph{relaxed Emt-graph}).\qqed
\end{defn}

\begin{defn}\label{defn:Oemt-labelling}
\cite{Yao-Sun-Zhang-Mu-Sun-Wang-Su-Zhang-Yang-Yang-2018arXiv} Suppose that a $(p,q)$-graph $G$ admits a vertex labelling $f:V(G) \rightarrow [0,2q-1]$ and an edge labelling $g:E(G)\rightarrow [1,2q-1]^o$. If there is a constant $k$ such that $f(u)+g(uv)+f(v)=k$ for each edge $uv\in E(G)$, and $g(E(G))=[1,2q-1]^o$, then we refer to $(f,g)$ as an \emph{odd-edge-magic matching labelling} (Oemm-labelling) of $G$ (called an \emph{Oemm-graph}). \qqed
\end{defn}

\begin{defn}\label{defn:relaxed-Oemt-labelling}
\cite{Yao-Sun-Zhang-Mu-Sun-Wang-Su-Zhang-Yang-Yang-2018arXiv} Suppose that a $(p,q)$-graph $G$ admits a vertex labelling $f:V(G)\rightarrow [0,2q-1]$ and an edge labelling $g:E(G)\rightarrow [1,2q-1]^o$, and let $s(uv)=|f(u)-f(v)|-f(uv)$ for $uv\in E(G)$. If (i) each edge $uv$ corresponds an edge $u'v'$ such that $g(uv)=|f(u')-f(v')|$; (ii) and there exists a constant $k'$ such that each edge $xy$ has a matching edge $x'y'$ holding $s(xy)+s(x'y')=k'$ true; (iii) there exists a constant $k$ such that $f(uv)+|f(u)-f(v)|=k$ for each edge $uv\in E(G)$. Then we call $(f,g)$ an \emph{ee-difference odd-edge-magic matching labelling} (Eedoemm-labelling) of $G$ (called a \emph{Eedoemm-graph}).\qqed
\end{defn}

\begin{defn}\label{defn:6C-labelling}
\cite{Yao-Sun-Zhang-Mu-Sun-Wang-Su-Zhang-Yang-Yang-2018arXiv} A total labelling $f:V(G)\cup E(G)\rightarrow [1,p+q]$ for a bipartite $(p,q)$-graph $G$ is a bijection and holds:

(i) (e-magic) $f(uv)+|f(u)-f(v)|=k$;

(ii) (ee-difference) each edge $uv$ matches with another edge $xy$ holding $f(uv)=|f(x)-f(y)|$ (or $f(uv)=2(p+q)-|f(x)-f(y)|$);

(iii) (ee-balanced) let $s(uv)=|f(u)-f(v)|-f(uv)$ for $uv\in E(G)$, then there exists a constant $k'$ such that each edge $uv$ matches with another edge $u'v'$ holding $s(uv)+s(u'v')=k'$ (or $2(p+q)+s(uv)+s(u'v')=k'$) true;

(iv) (EV-ordered) $\min f(V(G))>\max f(E(G))$ (or $\max f(V(G))<\min f(E(G))$, or $f(V(G))\subseteq f(E(G))$, or $f(E(G))\subseteq f(V(G))$, or $f(V(G))$ is an odd-set and $f(E(G))$ is an even-set);

(v) (ve-matching) there exists a constant $k''$ such that each edge $uv$ matches with one vertex $w$ such that $f(uv)+f(w)=k''$, and each vertex $z$ matches with one edge $xy$ such that $f(z)+f(xy)=k''$, except the \emph{singularity} $f(x_0)=\lfloor \frac{p+q+1}{2}\rfloor $;

(vi) (set-ordered) $\max f(X)<\min f(Y)$ (or $\min f(X)>\max f(Y)$) for the bipartition $(X,Y)$ of $V(G)$.

We refer to $f$ as a \emph{6C-labelling}.\qqed
\end{defn}

\begin{defn}\label{defn:Dgemm-labelling}
\cite{Yao-Sun-Zhang-Mu-Sun-Wang-Su-Zhang-Yang-Yang-2018arXiv} Suppose that a $(p,q)$-graph $G$ admits a vertex labelling $f:V(G)\rightarrow [0,p-1]$ and an edge labelling $g:E(G)\rightarrow [1,q]$, and let $s(uv)=|f(u)-f(v)|-g(uv)$ for $uv\in E(G)$. If there are: (i) each edge $uv$ corresponds an edge $u'v'$ such that $g(uv)=|f(u')-f(v')|$ (or $g(uv)=p-|f(u')-f(v')|$); (ii) and there exists a constant $k''$ such that each edge $xy$ has a matching edge $x'y'$ holding $s(xy)+s(x'y')=k''$ true; (iii) there exists a constant $k$ such that $|f(u)-f(v)|+f(uv)=k$ for each edge $uv\in E(G)$; (iv) there exists a constant $k'$ such that each edge $uv$ matches with one vertex $w$ such that $f(uv)+f(w)=k'$, and each vertex $z$ matches with one edge $xy$ such that $f(z)+f(xy)=k'$, except the \emph{singularity} $f(x_0)=0$. Then we name $(f,g)$ as an \emph{ee-difference graceful-magic matching labelling} (Dgemm-labelling) of $G$ (called a \emph{Dgemm-graph}).\qqed
\end{defn}

\begin{defn}\label{defn:difference-sum-labelling}
\cite{Yao-Sun-Zhang-Mu-Sun-Wang-Su-Zhang-Yang-Yang-2018arXiv} Let $f:V(G)\rightarrow [0,q]$ be a labelling of a $(p,q)$-graph $G$, and let $$S_{um}(G,f)=\sum_{uv\in E(G)}|f(u)-f(v)|,$$ we say $f$ to be a \emph{difference-sum labelling}. Find two extremum $\max_f S_{um}(G,f)$ (profit) and $\min_f S_{um}(G,f)$ (cost) over all difference-sum labellings of $G$.\qqed
\end{defn}

\begin{defn}\label{defn:felicitous-sum-labelling}
\cite{Yao-Sun-Zhang-Mu-Sun-Wang-Su-Zhang-Yang-Yang-2018arXiv} Let $f:V(G)\rightarrow [0,q]$ be a labelling of a $(p,q)$-graph $G$, and let $$F_{um}(G,f)=\sum_{uv\in E(G)}[f(u)+f(v)~(\bmod ~q+1)],$$ we call $f$ a \emph{felicitous-sum labelling}. Find two extremum $\max_f F_{um}(G,f)$ and $\min_f F_{um}(G,f)$ over all felicitous-sum labellings of $G$.\qqed
\end{defn}

Motivated from Definition \ref{defn:difference-sum-labelling} and Definition \ref{defn:felicitous-sum-labelling}, we design:
\begin{defn}\label{defn:three-new-sum-labelling}
$^*$ A connected $(p,q)$-graph $G$ admits a labelling $f:V(G)\cup E(G)\rightarrow [1,p+q]$, such that $f(x)\neq f(w)$ for any pair of elements $x,w\in V(G)\cup E(G)$. We have the following sums:
\begin{equation}\label{eqa:ve-11}
F_{ve}(G,f)=\sum _{uv\in E(G)}\big [f(uv)+|f(u)-f(v)|\big ],
\end{equation}
\begin{equation}\label{eqa:ve-22}
F_{|ve|}(G,f)=\sum _{uv\in E(G)}\big |f(u)+f(v)-f(uv)\big |,
\end{equation}
and
\begin{equation}\label{eqa:ve-magic}
F_{magic}(G,f)=\sum _{uv\in E(G)}\big [f(u)+f(uv)+f(v)\big ]=kq.
\end{equation}

We call $f$ to be: (1) a \emph{ve-sum-difference labelling} of $G$ if it holds (\ref{eqa:ve-11}) true; (2)  a \emph{ve-difference labelling} of $G$ if it holds (\ref{eqa:ve-22}) true; (3) a \emph{k-edge-average labelling} of $G$ if it holds (\ref{eqa:ve-magic}) true.\qqed
\end{defn}

Find these six extremum $\min_fF_{ve}(G,f)$, $\max_fF_{ve}(G,f)$, $\min_fF_{|ve|}(G,f)$, $\max_fF_{|ve|}(G,f)$, $\min_fF_{magic}(G,f)$ and $\max_fF_{magic}(G,f)$ over all $\varphi$-labellings of $G$, where $\varphi\in \{$ve-sum-difference, ve-difference, k-edge-average$\}$.

\begin{defn} \label{defn:set-ordered-felicitous}
\cite{Zhang-Yao-Wang-Wang-Yang-Yang-2013} Let $(X, Y)$ be the bipartition of a bipartite $(p, q)$-graph $G$. If $G$ admits a felicitous labelling $f$ such that $\max\{f(x):x\in X\}<b=\min\{f(y):y\in Y\}$, then we refer to $f$ as a \emph{set-ordered felicitous labelling} and $G$ a \emph{set-ordered felicitous graph}, and write this case as $f(X)<f(Y)$, and moreover $f$ is called an \emph{optimal set-ordered felicitous labelling} if $f(E(G))=[b,b+q-1]$ and $f(E(G))(\bmod~q)=[0,q-1]$.\qqed
\end{defn}

\begin{defn}\label{defn:edge-odd-graceful-labelling}
\cite{Yao-Zhang-Sun-Mu-Sun-Wang-Wang-Ma-Su-Yang-Yang-Zhang-2018arXiv} A $(p,q)$-graph $G$ admits an \emph{edge-odd-graceful total labelling} $h:V(G)\rightarrow [0,q-1]$ and $h:E(G)\rightarrow [1,2q-1]^o$ such that $$\{h(u)+h(uv)+h(v):uv\in E(G)\}=[a,b]$$ with $b-a+1=q$.\qqed
\end{defn}

\begin{defn}\label{defn:multiple-meanings-vertex-labelling}
\cite{Yao-Zhang-Sun-Mu-Sun-Wang-Wang-Ma-Su-Yang-Yang-Zhang-2018arXiv} A $(p,q)$-graph $G$ admits a \emph{multiple edge-meaning vertex labelling} $f:V(G)\rightarrow [0,p-1]$ such that (1) $f(E(G))=[1,q]$ and $f(u)+f(uv)+f(v)=$a constant $k$; (2) $f(E(G))=[p,p+q-1]$ and $f(u)+f(uv)+f(v)=$a constant $k'$; (3) $f(E(G))=[0,q-1]$ and $f(uv)=f(u)+f(v)~(\bmod~q)$; (4) $f(E(G))=[1,q]$ and $|f(u)+f(v)-f(uv)|=$a constant $k''$; (5) $f(uv)=$an odd number for each edge $uv\in E(G)$ holding $f(E(G))=[1,2q-1]^o$, and $\{f(u)+f(uv)+f(v):uv\in E(T)\}=[a,b]$ with $b-a+1=q$.\qqed
\end{defn}

\begin{defn}\label{defn:graceful-odd-graceful-total-set-labelling}
\cite{Yao-Zhang-Sun-Mu-Sun-Wang-Wang-Ma-Su-Yang-Yang-Zhang-2018arXiv} A $(p,q)$-graph $G$ admits a vertex set-labelling $f:V(G)\rightarrow [1,q]^2$~(or $[1,2q-1]^2)$, and induces an edge set-labelling $f'(uv)=f(u)\cap f(v)$. If we can select a \emph{representative} $a_{uv}\in f'(uv)$ for each edge label set $f'(uv)$ with $uv \in E(G)$ such that $$\{a_{uv}:~uv\in E(G)\}=[1,q],~ (\textrm{or}~[1,2q-1]^o),$$ we then call $f$ a \emph{graceful-intersection (or an odd-graceful-intersection) total set-labelling} of $G$.\qqed
\end{defn}

\begin{defn}\label{defn:graph-graceful-group-labelling}
\cite{Yao-Zhang-Sun-Mu-Sun-Wang-Wang-Ma-Su-Yang-Yang-Zhang-2018arXiv} Let $F_{n}(H,h)$ be an every-zero graphic group. A $(p,q)$-graph $G$ admits a \emph{graceful group-labelling} (or an \emph{odd-graceful group-labelling}) $F:V(G)\rightarrow F_{n}(H,h)$ such that each edge $uv$ is labelled by $F(uv)=F(u)\oplus F(v)$ under a \emph{zero} $H_k$, and $F(E(G))=\{F(uv):uv \in E(G)\}=\{H_1,H_2,\dots ,H_q\}$ (or $F(E(G))=\{F(uv):uv \in E(G)\}=\{H_1,H_3,\dots ,H_{2q-1}\}$).\qqed
\end{defn}

\begin{defn}\label{defn:perfect-odd-graceful-labelling}
\cite{Yao-Zhang-Sun-Mu-Sun-Wang-Wang-Ma-Su-Yang-Yang-Zhang-2018arXiv} Let $f$ be an odd-graceful labelling of a $(p,q)$-graph $G$, such that $f(V(G))\subset [0,2q-1]^o$ and $f(E(G))=[1,2q-1]^o$. If $\{|a-b|:~a,b\in f(V(G))\}=[1,p]$, then $f$ is called a \emph{perfect odd-graceful labelling} of $G$.\qqed
\end{defn}

\begin{defn}\label{defn:perfect-varepsilon-labelling}
\cite{Yao-Zhang-Sun-Mu-Sun-Wang-Wang-Ma-Su-Yang-Yang-Zhang-2018arXiv} Suppose that a $(p,q)$-graph $G$ admits an $\varepsilon$-labelling $h: V(G)\rightarrow S\subseteq [0,p+q]$. If $\{|a-b|:~a,b\in f(V(G))\}=[1,p]$, we call $f$ a \emph{perfect $\varepsilon$-labelling} of $G$.\qqed
\end{defn}

\begin{defn}\label{defn:image-labelling}
\cite{Yao-Zhang-Sun-Mu-Sun-Wang-Wang-Ma-Su-Yang-Yang-Zhang-2018arXiv} Let $f_i:V(G)\rightarrow [a,b]$ be a labelling of a $(p,q)$-graph $G$ and let each edge $uv\in E(G)$ have its own label as $f_i(uv)=|f_i(u)-f_i(v)|$ with $i=1,2$. If each edge $uv\in E(G)$ holds $f_1(uv)+f_2(uv)=k$ true, where $k$ is a positive constant, we call $f_1$ and $f_2$ are a matching of \emph{image-labellings}, and $f_i$ a \emph{mirror-image} of $f_{3-i}$ with $i=1,2$.\qqed
\end{defn}

\begin{defn} \label{defn:twin-k-d-harmonious-labellings}
\cite{Yao-Zhang-Sun-Mu-Sun-Wang-Wang-Ma-Su-Yang-Yang-Zhang-2018arXiv} A $(p,q)$-graph $G$ admits two $(k,d)$-harmonious labellings $f_i:V(G)\rightarrow X_0\cup X_{k,d}$ with $i=1,2$, where $X_0=\{0,d,2d, \dots ,(q-1)d\}$ and $X_{k,d}=\{k,k+d,k+2d, \dots ,k+(q-1)d\}$, such that each edge $uv\in E(G)$ is labelled as $f_i(uv)-k=[f_i(u)+f_i(v)-k~(\textrm{mod}~qd)]$ with $i=1,2$. If $f_1(uv)+f_2(uv)=2k+(q-1)d$, we call $f_1$ and $f_2$ a matching of \emph{$(k,d)$-harmonious image-labellings} of $G$.\qqed
\end{defn}

\begin{defn} \label{defn:twin-k-d-harmonious-labellings}
\cite{Yao-Zhang-Sun-Mu-Sun-Wang-Wang-Ma-Su-Yang-Yang-Zhang-2018arXiv} A $(p,q)$-graph $G$ admits a $(k,d)$-labelling $f$, and another $(p',q')$-graph $H$ admits another $(k,d)$-labelling $g$. If $(X_0\cup X_{k,d})\setminus f(V(G)\cup E(G))=g(V(H)\cup E(H))$ with $X_0=\{0,d,2d, \dots ,(q-1)d\}$ and $X_{k,d}=\{k,k+d,k+2d, \dots ,k+(q-1)d\}$, then $g$ is called a \emph{complementary $(k,d)$-labelling} of $f$, and both $f$ and $g$ are a matching of \emph{twin $(k,d)$-labellings} of $G$.\qqed
\end{defn}

\begin{defn}\label{defn:odd-6C-labelling}
\cite{Yao-Zhang-Sun-Mu-Sun-Wang-Wang-Ma-Su-Yang-Yang-Zhang-2018arXiv} A $(p,q)$-graph $G$ admits a total labelling $f:V(G)\cup E(G)\rightarrow [1,4q-1]$. If this labelling $f$ holds:

(i) (e-magic) $f(uv)+|f(u)-f(v)|=k$, and $f(uv)$ is odd;

(ii) (ee-difference) each edge $uv$ matches with another edge $xy$ holding $f(uv)=2q+|f(x)-f(y)|$, or $f(uv)=2q-|f(x)-f(y)|$,  $f(uv)+f(xy)=2q$;

(iii) (ee-balanced) let $s(uv)=|f(u)-f(v)|-f(uv)$ for $uv\in E(G)$, then there exists a constant $k'$ such that each edge $uv$ matches with another edge $u'v'$ holding $s(uv)+s(u'v')=k'$ (or $(p+q+1)+s(uv)+s(u'v')=k'$) true;

(iv) (EV-ordered) $\max f(V(G))<\min f(E(G))$, and $\{|a-b|:a,b\in f(V(G))\}=[1,2q-1]$;

(v) (ve-matching) there are two constants $k_1,k_2$ such that each edge $uv$ matches with one vertex $w$ such that $f(uv)+f(w)=k_1~(\textrm{or }k_2)$;

(vi) (set-ordered) $\max f(X)<\min f(Y)$ for the bipartition $(X,Y)$ of $V(G)$.

We call $f$ an \emph{odd-6C-labelling} of $G$.\qqed
\end{defn}

A parameter sequence is defined as follows:
$$\{(k_i,d_i)\}^m_1=\{(k_1,d_1), (k_2,d_2),\dots ,(k_m,d_m)\},$$
and let $$S(k_i,d_i)^q_1=\{k_i, k_i + d_i, \dots , k_i +(q-1)d_i\}$$ be a recursive set for integers $k_i\geq 1$ and $d_i\geq 1$. Thereby, we have a \emph{Topsnut-gpw sequence} $\{G_{(k_i,d_i)}\}^m_1$ made by some integer sequence $\{(k_i,d_i)\}^m_1$ and a $(p, q)$-graph $G$, where each Topsnut-gpw $G_{(k_i,d_i)}\cong G$, and each Topsnut-gpw $G_{(k_i,d_i)}\in \{G_{(k_i,d_i)}\}^m_1$ admits one labelling of four parameter labellings defined in Definition \ref{defn:3-parameter-labellings}.

\begin{defn} \label{defn:3-parameter-labellings}
\cite{Gallian2016} (1) A \emph{$(k_i,d_i)$-graceful labelling} $f$ of $G_i$
hold $f(V(G_i))\subseteq [0, k_i + (q-1)d_i]$, $f(x)\neq f(y)$ for
distinct $x,y\in V(G_i)$ and $\pi(E(G_i))=\{|\pi(u)-\pi(v)|;\ uv\in
E(G_i)\}=S(k_i,d_i)^q_1$.

(2) A labelling $f$ of $G_i$ is said to be $(k_i,d_i)$-\emph{arithmetic} if $f(V(G_i))\subseteq [0, k_i+(q-1)d_i]$, $f(x)\neq f(y)$ for distinct $x,y\in V(G_i)$ and $\{f(u)+f(v): uv\in E(G_i)\}=S(k_i,d_i)^q_1$.

(3) A $(k_i,d_i)$-\emph{edge antimagic total labelling} $f$ of $G_i$ hold $f(V(G_i)\cup E(G_i))=[1,p+q]$ and
$\{f(u)+f(v)+f(uv): uv\in E(G_i)\}=S(k_i,d_i)^q_1$, and furthermore $f$ is \emph{super} if $f(V(G_i))=[1,p]$.

(4) A \emph{$(k_i,d_i)$-harmonious labelling} of a $(p,q)$-graph $G_i$ is defined by a mapping $h:V(G)\rightarrow [0,k+(q-1)d_i]$ with $k_i,d_i\geq 1$, such that $f(x)\neq f(y)$ for any pair of vertices $x,y$ of $G$, $h(u)+h(v)(\bmod^*~qd_i)$ means that $h(uv)-k=[h(u)+h(v)-k](\bmod~qd_i)$ for each edge $uv\in E(G)$, and the edge label set $h(E(G))=S(k_i,d_i)^q_1$ holds true.\qqed
\end{defn}

We can see the complex of a Topsnut-gpw sequence $\{G_{(k_i,d_i)}\}^m_1$ as:

(i) $\{(k_i,d_i)\}^m_1$ is a random sequence or a sequence with many restrictions.

(ii) $G_{(k_i,d_i)}\cong G$ is a regularity.

(iii) Each $G_{(k_i,d_i)}\in \{G_{(k_i,d_i)}\}^m_1$ admits randomly one labelling in Definition \ref{defn:3-parameter-labellings}.

(iv) Each $G_{(k_i,d_i)}\in \{G_{(k_i,d_i)}\}^m_1$ has its matching $H_{(k_i,d_i)}\in \{H_{(k_i,d_i)}\}^m_1$ under the meaning of image-labelling, inverse labelling and twin labelling, and so on.

\vskip 0.4cm

The goal of applying Topsnut-gpw sequences is for encrypting graphs/networks. Moreover, we have

\begin{defn} \label{defn:Topsnut-gpw-sequences-graph-labellings}
\cite{Yao-Zhang-Sun-Mu-Sun-Wang-Wang-Ma-Su-Yang-Yang-Zhang-2018arXiv} Let $\{(k_i,d_i)\}^m_1$ be a sequence with integers $k_i\geq 0$ and $d_i\geq 1$, and $G$ be a $(p,q)$-graph with $p\geq 2$ and $q\geq 1$. We define a labelling $F:V(G)\rightarrow \{G_{(k_i,d_i)}\}^m_1$, and $F(u_iv_j)=(|k_i-k_j|, ~d_i+d_j~(\textrm{mod}~M))$ with $F(u_i)=G_{(k_i,d_i)}$ and $F(v_j)=G_{(k_j,d_j)}$ for each edge $u_iv_j\in E(G)$. Then

(1) If $\{|k_i-k_j|:~u_iv_j\in E(G)\}=[1,2q-1]^o$ and $\{d_i+d_j~(\textrm{mod}~M)):~u_iv_j\in E(G)\}=[0,2q-3]^o$, we call $F$ a \emph{twin odd-type graph-labelling} of $G$.

(2) If $\{|k_i-k_j|:~u_iv_j\in E(G)\}=[1,q]$ and $\{d_i+d_j~(\textrm{mod}~M)):~u_iv_j\in E(G)\}=[0,2q-3]^o$, we call $F$ a \emph{graceful odd-elegant graph-labelling} of $G$.

(3) If $\{|k_i-k_j|:~u_iv_j\in E(G)\}$ and $\{d_i+d_j~(\textrm{mod}~M)):~u_iv_j\in E(G)\}$ are generalized Fibonacci sequences, we call $F$ a \emph{twin Fibonacci-type graph-labelling} of $G$.\qqed
\end{defn}

\begin{defn}\label{defn:twin-labellings}
\cite{Wang-Xu-Yao-2017-Twin2017} Suppose $f:V(G)\rightarrow [0,2q-1]$ is an odd-graceful labelling of a $(p,q)$-graph $G$ and $g:V(H)\rightarrow [1,2q]$ is a labelling of another $(p',q')$-graph $H$ such that each edge $uv\in E(H)$ has its own label defined as $h(uv)=|h(u)-h(v)|$ and the edge label set $f(E(H))=[1,2q-1]^o$. We say $(f,g)$ to be a \emph{twin odd-graceful labelling}, $H$ a \emph{twin odd-graceful matching} of $G$. \qqed
\end{defn}

\begin{defn}\label{defn:Edge-magic-total-graph-labelling}
\cite{Yao-Zhang-Sun-Mu-Sun-Wang-Wang-Ma-Su-Yang-Yang-Zhang-2018arXiv} Let $M_{pg}(p,q)$ be the set of maximal planar graphs $H_i$ of $i+3$ vertices with $i\in [1,p+q]$, where each face of each planar graph $H_i$ is a triangle. We use a \emph{total labelling} $f$ to label the vertices and edges of a $(p,q)$-graph $G$ with the elements of $M_{pg}(p,q)$, such that $i+ij+j=k$ (a constant), where $f(u_i)=H_i$, $f(u_iv_j)=H_{ij}$ and $f(v_j)=H_j$ for each edges $u_iv_j\in E(G)$. We say $f$ an \emph{edge-magic total planar-graph labelling} of $G$ based on $M_{pg}(p,q)$.\qqed
\end{defn}

\begin{defn}\label{defn:multiple-graph-matching}
\cite{Yao-Sun-Zhang-Mu-Sun-Wang-Su-Zhang-Yang-Yang-2018arXiv} If a $(p,q)$-graph $G$ admits a vertex labelling $f: V(G)\rightarrow [0, p-1]$, such that $G$ can be decomposed into edge-disjoint graphs $G_1,G_2,\dots ,G_m$ with $m\geq 2$ and $E(G)=\bigcup^m_{i=1}E(G_i)$ and $E(G_i)\cap E(G_j)=\emptyset $ for $i\neq j$, and each graph $G_i$ admits a proper labelling $f_i$ induced by $f$. We call $G$ a \emph{multiple-graph matching partition}, denoted as $G=\odot_f\langle G_i\rangle ^m_1$.\qqed
\end{defn}

\begin{thm}\label{thm:set-ordered-matchings-10-labellings}
\cite{Yao-Sun-Zhang-Mu-Sun-Wang-Su-Zhang-Yang-Yang-2018arXiv} If a tree $T$ admits a set-ordered graceful labelling $f$, then $T$ matches with a multiple-tree matching partition $\oplus_f\langle T_i\rangle ^m_1$ with $m\geq 10$.
\end{thm}

\subsection{Flawed graph labellings}

We will show a \emph{(flawed) set-ordered graceful labelling} for each single Hanzi-graph, in other word, the set-ordered graceful labelling is one of standard labellings/colorings in Hanzi-gpws (see examples shown in Fig.\ref{fig:rrhg-set-graceful}). An example from Fig.\ref{fig:rrhg-flawed-set-graceful}, Fig.\ref{fig:rrhg-set-graceful-0} and Fig.\ref{fig:rrhg-set-graceful-1} is used to illustrate our \emph{flawed type of graph labellings}.

\begin{figure}[h]
\centering
\includegraphics[height=3.8cm]{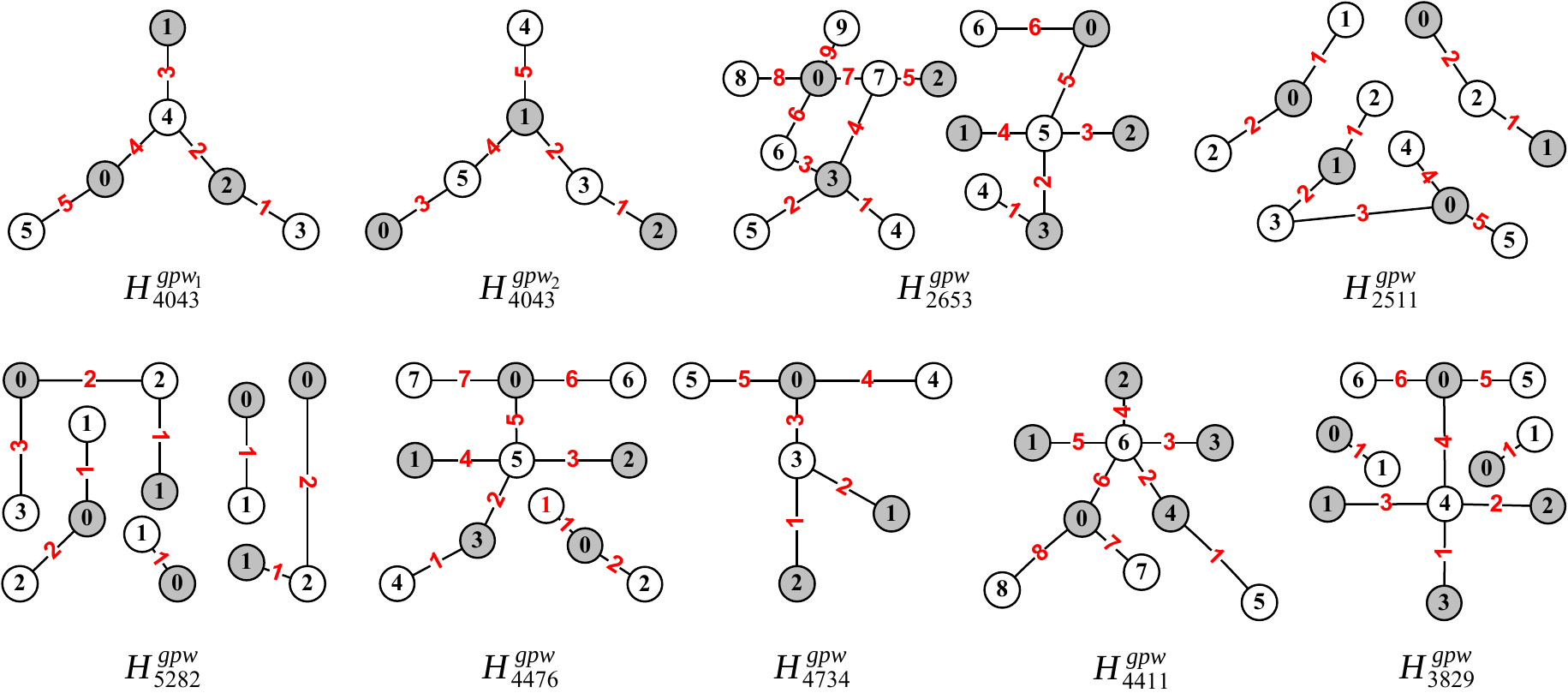}\\
\caption{\label{fig:rrhg-set-graceful} {\small Each connected Hanzi-graph of a Hanzi-graph group $T$ shown in Fig.\ref{fig:rrhg-topological} admits a set-ordered graceful labelling.}}
\end{figure}
In Fig.\ref{fig:rrhg-set-graceful}, the Hanzi-graph group $T$ is
$$T=T_{4043}T_{4043}T_{2635}T_{2511}T_{5282}T_{4476}T_{4734}T_{4411}T_{3829},$$
so we get a Hanzi-gpw group $H$ as follows
$$H=H^{gpw_1}_{4043}H^{gpw_2}_{4043}H^{gpw}_{2635}H^{gpw}_{2511}H^{gpw}_{5282}H^{gpw}_{4476}H^{gpw}_{4734}H^{gpw}_{4411}H^{gpw}_{3829}.$$
Join them by edges for producing a connected graph $T^*=T+E^*$, and then Fig.\ref{fig:rrhg-set-graceful-0} shows us a set-ordered graceful labelling $f^*$ of $T^*$ by the set-ordered graceful labellings of the Hanzi-graph group $T$, and moreover we get a \emph{flawed set-ordered graceful labelling} $f$ of the Hanzi-graph group $T$ by $f(V(T))=f^*(V(T^*))$, $f(E(T))=f^*(E(T^*))\setminus f^*(E^*)$ (see Fig.\ref{fig:rrhg-flawed-set-graceful}).

\begin{figure}[h]
\centering
\includegraphics[height=3cm]{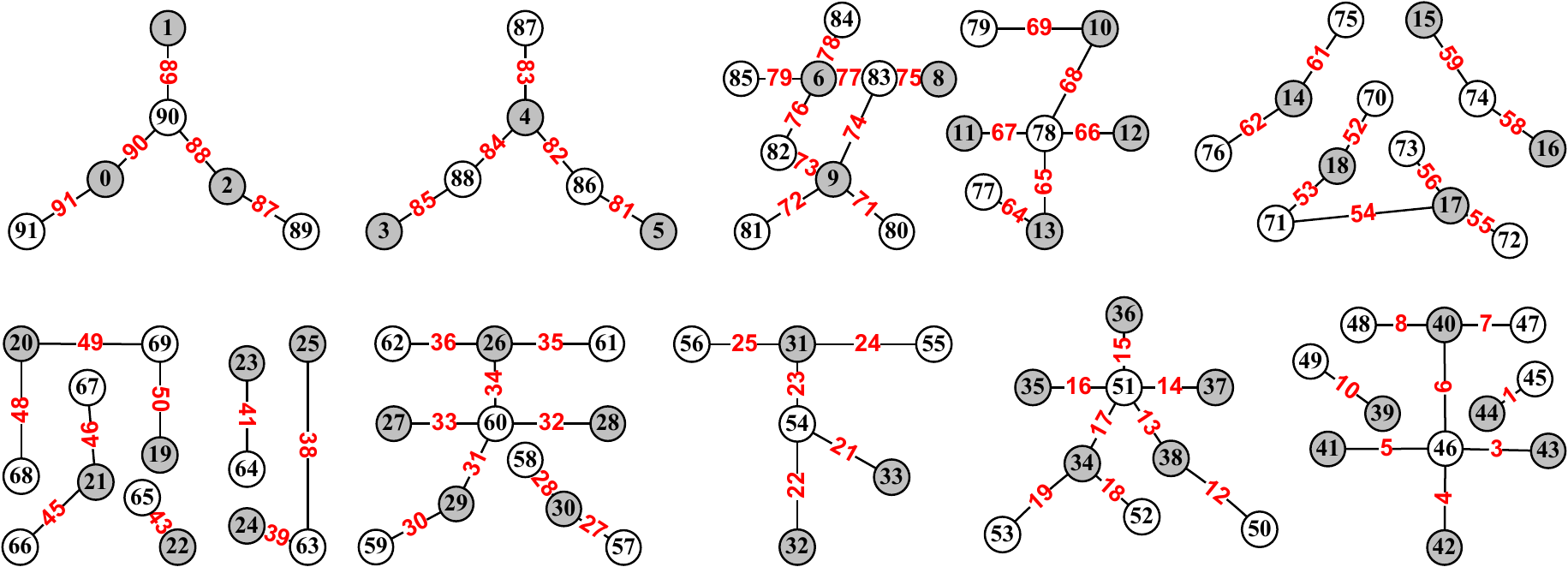}\\
\caption{\label{fig:rrhg-flawed-set-graceful} {\small A Hanzi-graph group $T$ shown in Fig.\ref{fig:rrhg-set-graceful} admits a \emph{flawed set-ordered graceful labelling}.}}
\end{figure}

\begin{figure}[h]
\centering
\includegraphics[height=4.1cm]{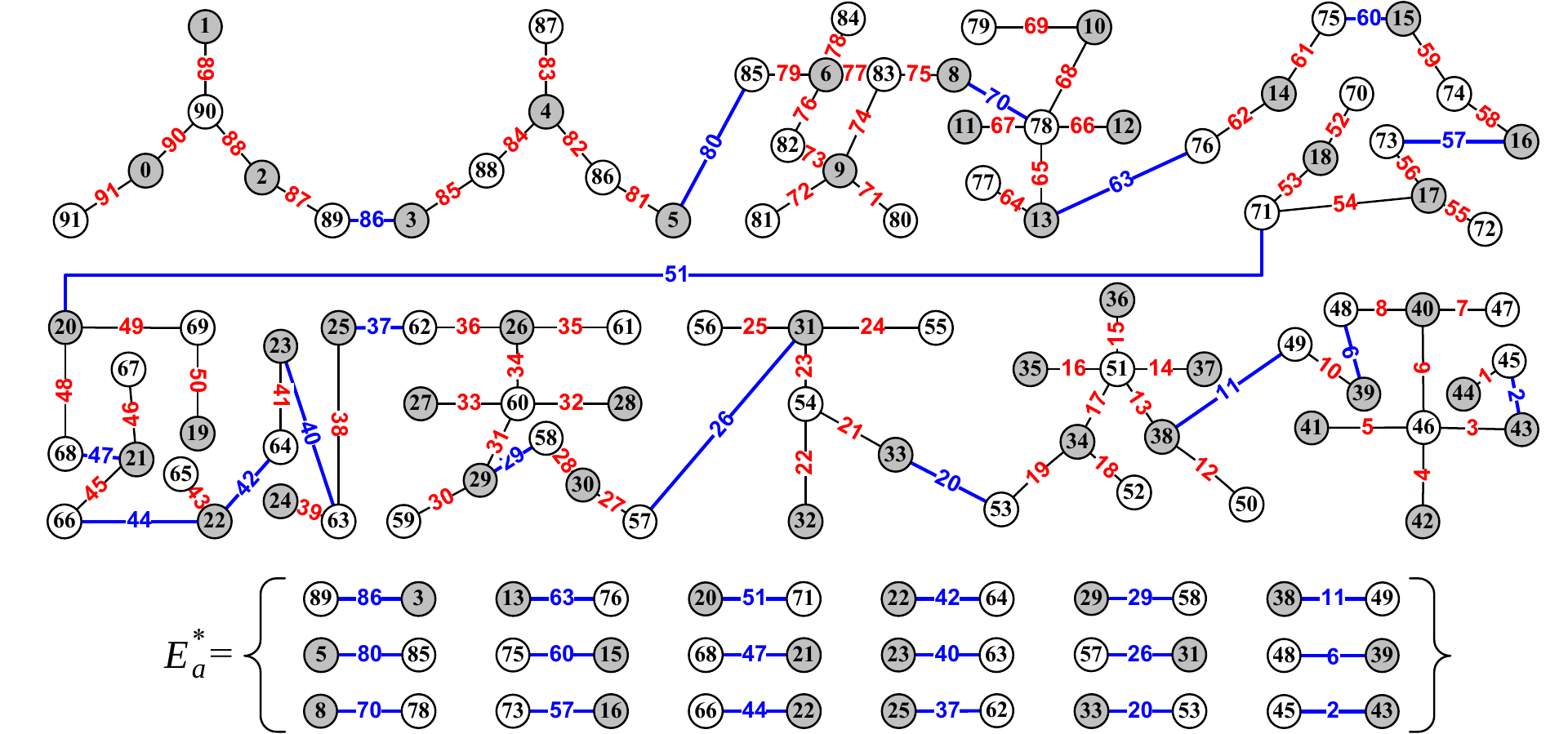}\\
\caption{\label{fig:rrhg-set-graceful-0} {\small A connected graph $T+E^*_a$ based on Fig.\ref{fig:rrhg-flawed-set-graceful} admits a set-ordered graceful labelling.}}
\end{figure}

\begin{figure}[h]
\centering
\includegraphics[height=4.1cm]{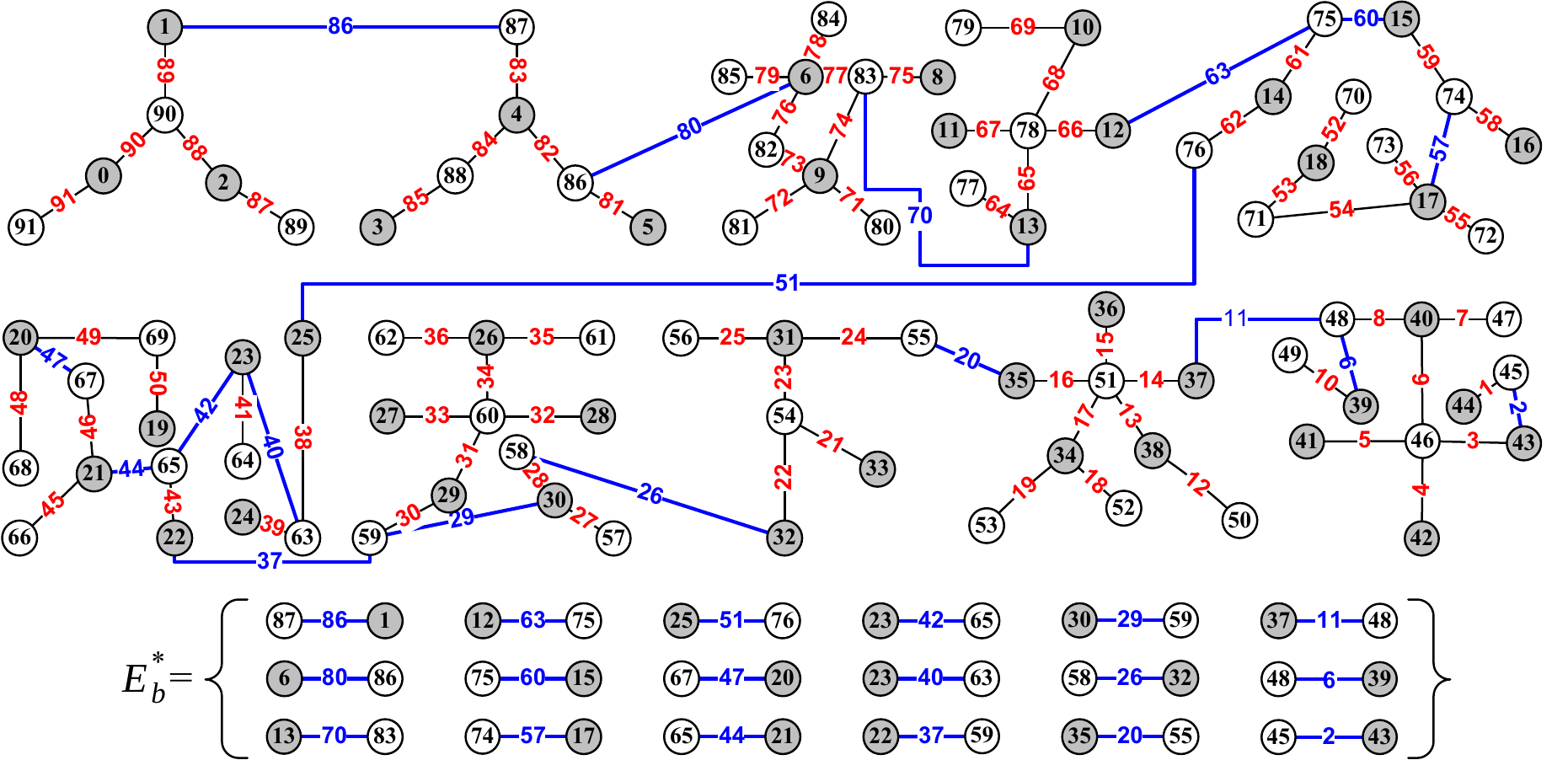}\\
\caption{\label{fig:rrhg-set-graceful-1} {\small Another connected graph $T+E^*_b$ based on Fig.\ref{fig:rrhg-flawed-set-graceful} admits a set-ordered graceful labelling.}}
\end{figure}

In general, we have a result shown in the following:
\begin{thm}\label{defn:group-flawed-labellings}
Let $G_1,G_2,\dots, G_m$ be disjoint connected graphs, and $E^*$ be an edge set such that each edge $uv$ of $E^*$ has one end $u$ in some $G_i$ and another end $v$ is in some $G_j$ with $i\neq j$, and $E^*$ joins $G_1,G_2,\dots, G_m$ together to form a connected graph $H$, denoted as $H=E^*+G$, where $G=\bigcup^m_{i=1}G_i$. We say $G=\bigcup^m_{i=1}G_i$ to be a disconnected $(p,q)$-graph with $p=|V(H)|=\sum^m_{i=1}|V(G_i)|$ and $q=|E(H)|-|E^*|=(\sum^m_{i=1}|E(G_i)|)-|E^*|$. If $H$ admits an $\alpha$-labelling shown in the following
\begin{asparaenum}[\emph{\textrm{Flawed}}-1. ]
\item $\alpha$ is a graceful labelling, or a set-ordered graceful labelling, or graceful-intersection total set-labelling, or a graceful group-labelling.
\item $\alpha$ is an odd-graceful labelling, or a set-ordered odd-graceful labelling, or an edge-odd-graceful total labelling, or an odd-graceful-intersection total set-labelling, or an odd-graceful group-labelling, or a perfect odd-graceful labelling.
\item $\alpha$ is an elegant labelling, or an odd-elegant labelling.
\item $\alpha$ is an edge-magic total labelling, or a super edge-magic total labelling, or super set-ordered edge-magic total labelling, or an edge-magic total graceful labelling.
\item $\alpha$ is a $(k,d)$-edge antimagic total labelling, or a $(k, d)$-arithmetic.
\item $\alpha$ is a relaxed edge-magic total labelling.
\item $\alpha$ is an odd-edge-magic matching labelling, or an ee-difference odd-edge-magic matching labelling.
\item $\alpha$ is a 6C-labelling, or an odd-6C-labelling.
\item $\alpha$ is an ee-difference graceful-magic matching labelling.
\item $\alpha$ is a difference-sum labelling, or a felicitous-sum labelling.
\item $\alpha$ is a multiple edge-meaning vertex labelling.
\item $\alpha$ is a perfect $\varepsilon$-labelling.
\item $\alpha$ is an image-labelling, or a $(k,d)$-harmonious image-labelling.
\item $\alpha$ is a twin $(k,d)$-labelling, or a twin Fibonacci-type graph-labelling, or a twin odd-graceful labelling.
\end{asparaenum}
Then $G$ admits a \emph{flawed $\alpha$-labelling} too.\qqed
\end{thm}

\vskip 0.2cm

The above Theorem \ref{defn:group-flawed-labellings} shows us over thirty-three flawed graph labellings.

\begin{thm}\label{thm:flawed-graceful-labelling-forest}
Let $G=\bigcup^m_{i=1} G_i$ be a forest with components $G_1,G_2,\dots ,G_m$, where each $G_i$ is a $(p_i,p_i-1)$-tree with bipartition $(X_i,Y_i)$ and admits a set-ordered graceful labelling $f_i$ holding $\max f_i(X_i)<\min f_i(Y_i)$ with $i\in [1,m]$ true. Then $G$ admits a flawed set-ordered graceful labelling.
\end{thm}
\begin{proof}
Suppose that $(X_i,Y_i)$: $X_i=\{x_{i,j}:~j\in[1,s_{i}]\}$, and $Y_i=\{y_{i,j}:~j\in[1,t_{i}]\}$. Let $M=\sum ^{m}_{k=1}s_{k}$ and $N=\sum ^{m}_{k=1}t_{k}$. Furthermore, $0\leq \max f_i(X_i)<\min f_i(x_{i,j+1})\leq s_{i}-1$ for $j\in[1,s_{i}]$ and $s_{i}\leq f_i(y_{i,j})<f_i(y_{i,j+1})\leq s_{i}+t_{i}-1=p_i-1$ for $j\in[1,t_{i}]$ with $i\in[1,m]$.

We set a new labelling $g$ as follows:

(1) $g(x_{1,j})=f_1(x_{1,j})$ with $j\in[1,s_{1}]$.

(2) $g(x_{i,j})=f_i(x_{i,j})+\sum ^{i-1}_{k=1}s_{k}$ for $j\in[1,s_{i}]$ and $i\in[2,m]$.

(3) $g(y_{m,j})=f_m(y_{m,j})+M-s_m$ with $j\in[1,t_{m}]$.

(4) $g(y_{\ell,j})=f_{\ell}(y_{\ell,j})+M+\sum ^{m-\ell}_{k=1}t_{k}$ for $j\in[1,t_{\ell}]$ and $\ell\in[1,m-1]$.

Clearly,
\begin{equation}\label{eqa:flawed-labellings-00}
0=g(x_{1,1})\leq g(u)<g(v)\leq g(y_{1,t_{1}})=M+N-1
\end{equation} for $u\in \bigcup^m_{i=1} X_i$ and $v\in \bigcup^m_{i=1} Y_i$.

Notice that $f_i(E(G_i))=[1,p_i-1]=\{f_{i}(y_{i,a})-f_{i}(x_{i,b}):~x_{i,b}y_{i,a}\in E(G_i)\}$. We compute edge labels $g(x_{i,b}y_{i,a})=g(y_{i,a})-g(x_{i,b})$ as follows: For each edge $x_{m,b}y_{m,a}\in E(G_m)$, we have
\begin{equation}\label{eqa:c3xxxxx}
{
\begin{split}
&\quad g(x_{m,b}y_{m,a})=g(y_{m,a})-g(x_{m,b})\\
&=[f_m(y_{m,a})+M-s_m]-\left [f_m(x_{m,b})+\sum ^{m-1}_{k=1}s_{k}\right ]\\
&=f_m(y_{m,a})-f_m(x_{m,b}),
\end{split}}
\end{equation}
so $g(E(G_m))=[1,p_m-1]$. Next, for each edge $x_{i,b}y_{i,a}\in E(G_i)$ with $i\in [1,m-1]$, we can compute
\begin{equation}\label{eqa:c3xxxxx}
{
\begin{split}
&\quad g(x_{i,b}y_{i,a})=g(y_{i,a})-g(x_{i,b})\\
&=\left [f_{i}(y_{i,j})+M+\sum ^{m-i}_{k=1}t_{k}\right ]-\left [f_i(x_{i,j})+\sum ^{i-1}_{k=1}s_{k}\right ]\\
&=f_{i}(y_{i,j})-f_i(x_{i,j})+\sum ^{m}_{k=i}s_{k}+\sum ^{m-i}_{k=1}t_{k}.
\end{split}}
\end{equation}
So, we obtain
\begin{equation}\label{eqa:c3xxxxx}
{
\begin{split}
&\quad g(E(G_i))=[1+M_i, p_i-1+M_i]
\end{split}}
\end{equation} where $M_i=\sum ^{m}_{k=i}s_{k}+\sum ^{m-i}_{k=1}t_{k}$. Thereby, we get
\begin{equation}\label{eqa:c3xxxxx}
g(E(G))=\bigcup ^m_{i=1}g(E(G_i))=\left [1,\bigcup^m_{i=1}(p_i-1)\right ]\setminus g(E^*)
\end{equation}

We claim that $g$ is a flawed set-ordered graceful labelling of the forest $G$.
\end{proof}

\begin{cor}\label{thm:flawed-graceful-labelling-graph}
Let $G$ be a disconnected graph with components $H_1,H_2,\dots ,H_m$, where each $H_i$ is a connected bipartite $(p_i,q_i)$-graph with bipartition $(X_i,Y_i)$ and admits a set-ordered graceful labelling $f_i$ holding $\max f_i(X_i)<\min f_i(Y_i)$ with $i\in [1,m]$. Then $G$ admits a flawed set-ordered graceful labelling.
\end{cor}

If $T=\bigcup ^m_{i=1}T_i$ is a forest having disjoint trees $T_1,T_2,\dots ,T_m$. Does $T$ admit a flawed set-ordered graceful labelling if and only if each tree $T_i$ admits a set-ordered graceful labelling with $i\in [1,m]$? Unfortunately, we have a counterexample for this question shown in Fig.\ref{fig:a-counterexample}, in which $T=T_1\cup T_2$, also, $T=(T_1\ominus T_2)-uv$ admits flawed set-ordered graceful labellings, however, $T_i$ does not admit any set-ordered graceful labelling with $i=1,2$.

In \cite{Yao-Liu-Yao-2017}, the authors have proven the following mutually equivalent labellings:

\begin{thm} \label{thm:connections-several-labellings}
\cite{Yao-Liu-Yao-2017} Let $T$ be a tree on $p$ vertices, and let $(X,Y)$ be its
bipartition of vertex set $V(T)$. For all values of integers $k\geq 1$ and $d\geq 1$, the following assertions are mutually equivalent:

$(1)$ $T$ admits a set-ordered graceful labelling $f$ with $\max f(X)<\min f(Y)$.

$(2)$ $T$ admits a super felicitous labelling $\alpha$ with $\max \alpha(X)<\min \alpha(Y)$.

$(3)$ $T$ admits a $(k,d)$-graceful labelling $\beta$ with
$\beta(x)<\beta(y)-k+d$ for all $x\in X$ and $y\in Y$.

$(4)$ $T$ admits a super edge-magic total labelling $\gamma$ with $\max \gamma(X)<\min \gamma(Y)$ and a magic constant $|X|+2p+1$.

$(5)$ $T$ admits a super $(|X|+p+3,2)$-edge antimagic total labelling $\theta$ with $\max \theta(X)<\min \theta(Y)$.

$(6)$ $T$ has an odd-elegant labelling $\eta$ with $\eta(x)+\eta(y)\leq 2p-3$ for every edge $xy\in E(T)$.

$(7)$ $T$ has a $(k,d)$-arithmetic labelling $\psi$ with $\max \psi(x)<\min \psi(y)-k+d\cdot |X|$ for all $x\in X$ and $y\in Y$.

$(8)$ $T$ has a harmonious labelling $\varphi$ with $\max \varphi(X)<\min \varphi(Y\setminus \{y_0\})$ and $\varphi(y_0)=0$.
\end{thm}

We have some results similarly with that in Theorem \ref{thm:connections-several-labellings} about flawed graph labellings as follows:

\begin{thm} \label{thm:connection-flawed-labellings}
Suppose that $T=\bigcup ^m_{i=1}T_i$ is a forest having disjoint trees $T_1,T_2,\dots ,T_m$, and $(X,Y)$ be its vertex
bipartition. For all values of integers $k\geq 1$ and $d\geq 1$, the following assertions are mutually equivalent:
\begin{asparaenum}[F-1. ]
\item $T$ admits a flawed set-ordered graceful labelling $f$ with $\max f(X)<\min f(Y)$;
\item $T$ admits a flawed set-ordered odd-graceful labelling $f$ with $\max f(X)<\min f(Y)$;

\item $T$ admits a flawed set-ordered elegant labelling $f$ with $\max f(X)<\min f(Y)$;

\item $T$ has a flawed odd-elegant labelling $\eta$ with $\eta(x)+\eta(y)\leq 2p-3$ for every edge $xy\in E(T)$.

\item $T$ admits a super flawed felicitous labelling $\alpha$ with $\max \alpha(X)<\min \alpha(Y)$.

\item $T$ admits a super flawed edge-magic total labelling $\gamma$ with $\max \gamma(X)<\min \gamma(Y)$ and a magic constant $|X|+2p+1$.
\item $T$ admits a super flawed $(|X|+p+3,2)$-edge antimagic total labelling $\theta$ with $\max \theta(X)<\min \theta(Y)$.

\item $T$ has a flawed harmonious labelling $\varphi$ with $\max \varphi(X)<\min \varphi(Y\setminus \{y_0\})$ and $\varphi(y_0)=0$.
\end{asparaenum}
\end{thm}

We present some equivalent definitions with parameters $k,d$ for flawed $(k,d)$-labellings.

\begin{thm} \label{thm:flawed-(k,d)-labellings}
Let $T=\bigcup ^m_{i=1}T_i$ be a forest having disjoint trees $T_1,T_2,\dots ,T_m$, and $V(T)=X\cup Y$. For all values of two integers $k\geq 1$ and $d\geq 1$, the following assertions are mutually equivalent:
\begin{asparaenum}[P-1. ]
\item $T$ admits a flawed set-ordered graceful labelling $f$ with $\max f(X)<\min f(Y)$.

\item $T$ admits a flawed $(k,d)$-graceful labelling $\beta$ with $\max \beta(x)<\min \beta(y)-k+d$ for all $x\in X$ and $y\in Y$.

\item $T$ has a flawed $(k,d)$-arithmetic labelling $\psi$ with $\max \psi(x)<\min \psi(y)-k+d\cdot |X|$ for all $x\in X$ and $y\in Y$.
\item $T$ has a flawed $(k,d)$-harmonious labelling $\varphi$ with $\max \varphi(X)<\min \varphi(Y\setminus \{y_0\})$ and $\varphi(y_0)=0$.
\end{asparaenum}
\end{thm}

\begin{rem}\label{thm:problem-00}
It is interesting to discover new (flawed)  $(k,d)$-labellings for making sequence labellings.
\end{rem}

\subsection{Rotatable labellings}

\begin{defn}\label{defn:mf-graceful-mf-odd-graceful}
$^*$ For any vertex $u$ of a connected and bipartite $(p,q)$-graph $G$, there exist a vertex labelling $f:V(G)\rightarrow [0,q]$ (or $f:V(G)\rightarrow [0,2q-1]$) such that (i) $f(u)=0$; (ii) $f(E(G))=\{f(xy)=|f(x)-f(y)|: ~xy\in E(G)\}=[1,q]$ (or $f(E(G))=[1,2q-1]^o$); (iii) the bipartition $(X,Y)$ of $V(G)$ holds $\max f(X)<\min f(Y)$. Then we say $G$ admits a \emph{$0$-rotatable set-ordered system of (odd-)graceful labellings}, abbreviated as \emph{$0$-rso-graceful system} (\emph{$0$-rso-odd-graceful system}).\qqed
\end{defn}

We can develop Definition \ref{defn:mf-graceful-mf-odd-graceful} to other definitions of (flawed) $0$-rotatable set-ordered system of $\varepsilon$-labellings. However, many Hanzi-graphs (resp. general graphs) do not admit a $0$-rso-graceful system (or $0$-rso-odd-graceful system), even simpler tree-like Hanzi-graphs, for example, a Hanzi-graph made by the Hanzi $H_{4043}$ does not admit a $0$-rso-graceful system. For real application, we can do some operations on those Hanzi-graphs refusing $0$-rso-graceful systems (or $0$-rso-odd-graceful systems). In Fig.\ref{fig:mf-labelling-11}, we can show that (a-1), (a-2), (b-1) and (b-2) admit $0$-rso-graceful systems or $0$-rso-odd-graceful systems.

\begin{figure}[h]
\centering
\includegraphics[height=4.8cm]{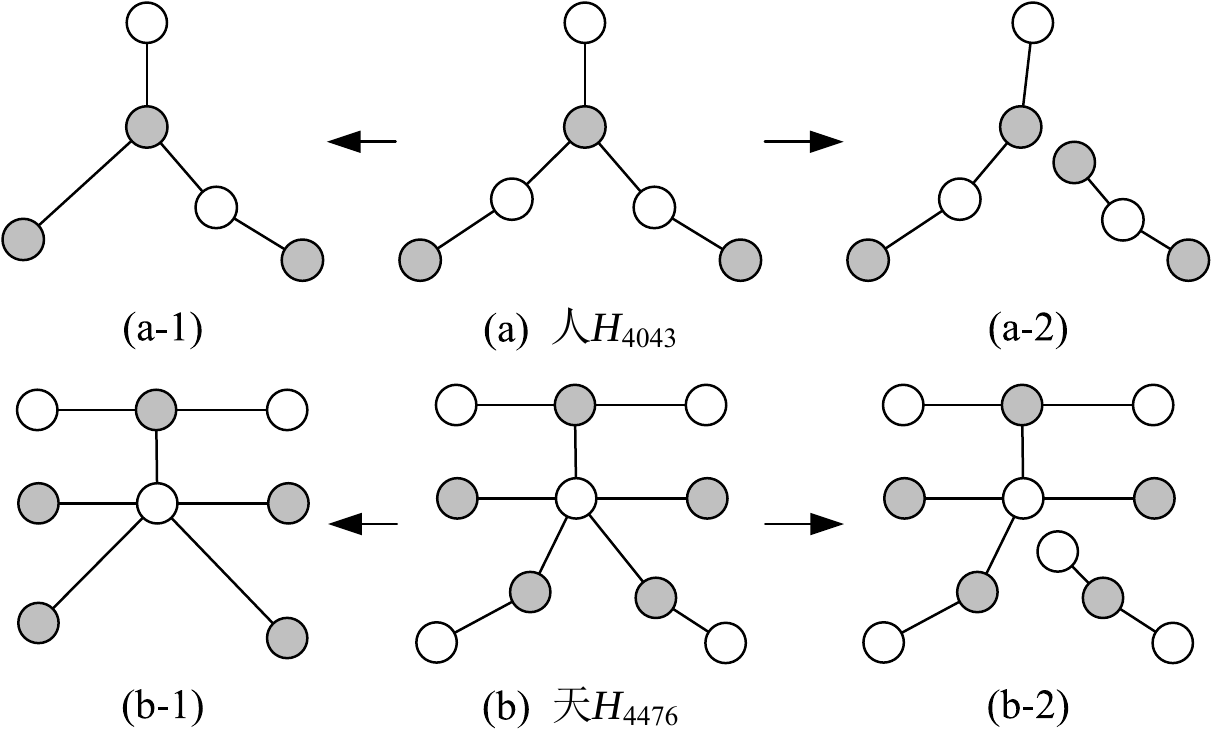}\\
\caption{\label{fig:mf-labelling-11} {\small (a) and (b) do not admit $0$-rso-graceful systems, but (a-1), (a-2), (b-1) and (b-2) admit $0$-rso-graceful systems.}}
\end{figure}

\begin{lem}\label{thm:symmetric-tree}
If a tree $T$  admits a $0$-rotatable system of (odd-)graceful labellings, then its symmetric tree $T\ominus T'$ admits a $0$-rotatable set-ordered system of (odd-)graceful labellings.
\end{lem}
\begin{proof} Let $f$ be a graceful labelling of a tree $T$ having $p$ vertices, and $(X,Y)$ be the bipartition of vertex set of $T$, so $V(T)=X\cup Y$. We take a copy $T'$ of $T$, correspondingly, $(X',Y')$ is the bipartition of vertex set of $T'$, namely, $V(T)=V(T')$ and $X=X',~Y=Y'$. We join a vertex $x_0$ of $T$ with its image $x'_0$ of $T'$ by an edge, and the resultant tree is just a  symmetric tree $T\ominus T'$ with its vertex bipartition $V(T\ominus T')=(X\cup Y',Y\cup X')$. Next, we define a labelling $g$ for the tree $T\ominus T'$ in the way: $g(x)=f(x)$ with $x\in X\subset X\cup Y'$; and  for each $w'\in Y'$ we set $g(w')=f(w)$ with $w\in Y$ where $w'$ is the image of $w$;  for each $x'\in X'$ we let $g(x')=f(x)+p$ with $x\in X$, where $x'$ is the image of $x$; and for each $y'\in Y\subset \cup X'$ we have $g(y')=f(y)+p$ with $y\in Y$. Clearly, $g(X\cup Y')=f(V(T))$ and $g(Y\cup X')=\{f(w)+p:~w\in V(T)\}$ prove that  $\max g(X\cup Y')<\min g(Y\cup X')$. Furthermore, $g(E(T'))=[1,p-1]$, $g(x_0x'_0)=p$ and $g(E(T))=[1+p,2p-1]$. Thereby, we claim that $g$ is a set-ordered graceful labelling of the tree $T\ominus T'$.

By Definition \ref{defn:mf-graceful-mf-odd-graceful}, each vertex $u$ of the  tree $T$  admitting a $0$-rotatable system of (odd-)graceful labellings can be labelled with $f(u)=0$ by some graceful labelling $f$ of $T$. Thereby, this  vertex $u$ (or its image $u'$) can be labelled with $g(u)=0$ by some set-ordered graceful labelling $g$ of the tree $T\ominus T'$. The proof about odd-graceful labellings is very similar with that above, we omit it.
\end{proof}

By Lemma \ref{thm:symmetric-tree}, we can prove the following results (see an example shown in Fig.\ref{fig:a-counterexample}):

\begin{thm}\label{thm:0-rotatable-set-ordered-system00}
Suppose that a connected and bipartite $(p,q)$-graph $G$ admits a $0$-rotatable set-ordered system of (odd-)graceful labellings. Then the edge symmetric graph $G\ominus G'$ admits a $0$-rotatable set-ordered system of (odd-)graceful labellings too, where $G'$ is a copy of $G$, and $G\ominus G'$ is obtained by joining a vertex of $G$ with its image in $G'$ by an edge.
\end{thm}

\begin{figure}[h]
\centering
\includegraphics[height=8cm]{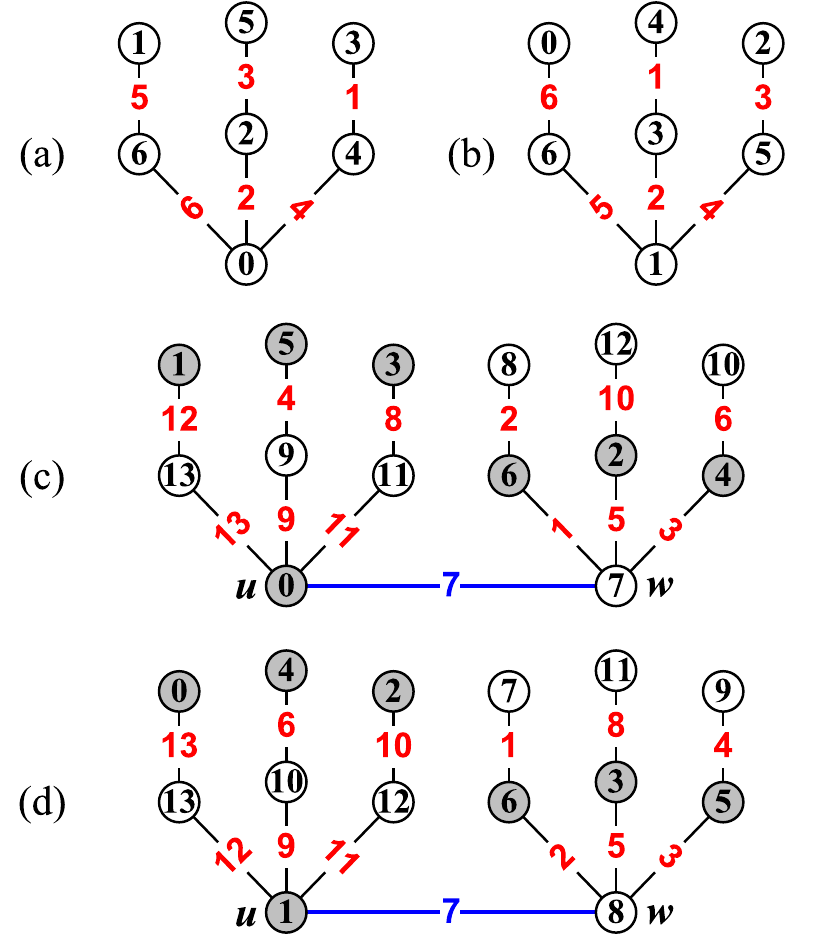}\\
\caption{\label{fig:a-counterexample} {\small Two trees $T_1,T_2$ shown in (a) and (b) admit $0$-rotatable system of graceful labellings, but set-ordered. A tree $T_1\ominus T_2$ admits a $0$-rotatable system of graceful labellings.}}
\end{figure}

\begin{thm}\label{thm:0-rotatable-set-ordered-system11}
There are infinite graphs admit $0$-rotatable set-ordered systems of (odd-)graceful labellings.
\end{thm}

\begin{rem}\label{thm:CCCCCC}
Definition \ref{defn:mf-graceful-mf-odd-graceful} can generalized to other labellings, such as edge-magic total labelling, elegant/odd-elegant labellings, felicitous labelling, $(k,d)$-graceful labelling, edge antimagic total labelling, $(k, d)$-arithmetic, harmonious labelling, odd-edge-magic matching labelling, relaxed edge-magic total labelling, 6C-labelling, odd-6C-labelling, and so on.
\end{rem}

\subsection{Pan-labelling and problems}

We restate a pan-labelling definition (Ref. \cite{Yao-Sun-Zhang-Mu-Sun-Wang-Su-Zhang-Yang-Yang-2018arXiv} and \cite{Yao-Mu-Sun-Zhang-Wang-Su-2018}) as follows:

\begin{defn}\label{defn:sequence-labelling}
\cite{Yao-Mu-Sun-Zhang-Wang-Su-2018} Let $G$ be a $(p,q)$-graph, and let $A_M=\{a_i\}_1^M$ and $B_q=\{b_j\}_1^q$ be two monotonic increasing sequences of non-negative numbers with $M\geq p$. There are the following restrict conditions:
\begin{asparaenum}[Seq-1. ]
\item \label{condi:vertex-mapping} A vertex mapping $f:V(G)\rightarrow A_M$ such that $f(u)\neq f(v)$ for distinct vertices $u,v\in V(G)$.
\item \label{condi:total-mapping} A total mapping $g:V(G)\cup E(G) \rightarrow A_M\cup B_q$ such that $g(x)\neq g(y)$ for distinct elements $x,y\in V(G)\cup E(G)$.
\item \label{condi:induced-edge-label} An induced edge label $f(uv)=O(f(u),f(v))$ for $uv\in E(G)$.
\item \label{condi:total-equation} An $F$-equation $F(g(u),g(uv),g(v))=0$ holds true.
\item \label{condi:total-equation-edge} An $E$-equation $E(f(u),f(uv),f(v))=0$ holds true for an edge labelling $f:E(G)\rightarrow B_q$.
\item \label{condi:not-full-1} $f(E(G))\subseteq B_q$.
\item \label{condi:not-full-2} $g(V(G)\cup E(G))\subseteq A_M\cup B_q$.
\item \label{condi:half-full} $f(V(G))\subseteq A_M$ and $f(E(G))=B_q$.
\item \label{condi:all-full} $f(V(G))=A_M$ and $f(E(G))=B_q$.
\end{asparaenum}
\quad We refer to $f$ as:
\begin{asparaenum}[(1) ]
\item a \emph{sequence-$(A_M,B_q)$ labelling} if Seq-\ref{condi:vertex-mapping} and Seq-\ref{condi:induced-edge-label} hold true;

\item a \emph{sequence-$(A_M,B_q)$ total labelling} if Seq-\ref{condi:total-mapping}, Seq-\ref{condi:induced-edge-label} and Seq-\ref{condi:not-full-2} hold true;

\item a \emph{full sequence-$(A_M,B_q)$ labelling} if Seq-\ref{condi:total-mapping}, Seq-\ref{condi:total-equation} and Seq-\ref{condi:half-full} hold true;

\item a \emph{graceful sequence-$(A_M,B_q)$ labelling} if Seq-\ref{condi:vertex-mapping}, Seq-\ref{condi:induced-edge-label} and Seq-\ref{condi:all-full} hold true;

\item a \emph{total sequence-$(A_M,B_q)$ labelling} if Seq-\ref{condi:total-mapping} and Seq-\ref{condi:total-equation} hold true;

\item a \emph{sequence-$(A_M,B_q)$ $F$-total graceful labelling} if Seq-\ref{condi:total-mapping}, Seq-\ref{condi:total-equation} and Seq-\ref{condi:all-full} hold true;

\item a \emph{sequence-$(A_M,B_q)$ mixed labelling} if Seq-\ref{condi:induced-edge-label}, Seq-\ref{condi:total-equation-edge} and Seq-\ref{condi:half-full} hold true.\qqed
\end{asparaenum}
\end{defn}

If two sets $A_M=\{a_i\}_1^M$ and $B_q=\{b_j\}_1^q$ defined in Definition \ref{defn:sequence-labelling} correspond a graph labelling admitted by graphs, we say $(A_M, B_q)$ a \emph{graph-realized sequence matching}.

\begin{rem}\label{thm:problem-11}
It may be interesting to consider a problem proposed in \cite{Yao-Sun-Zhang-Mu-Sun-Wang-Su-Zhang-Yang-Yang-2018arXiv}: Determine the conditions for $A_M$ and $B_q$ and any $a\in A_M$ corresponds two numbers $a^*\in A_M$, $b^*\in B_q$ such that $b^*=|a-a^*|$. Then determine such sequence pair $(A_M,B_q)$ defined in Definition \ref{defn:sequence-labelling} such that the sequence type of graph labellings defined in Definition \ref{defn:sequence-labelling} hold true on some graphs. If any $b\in B_q$ corresponds two numbers $a',a''\in A_M$ holding $b=|a'-a''|$ true, then we can find at least a forest $T$ admitting a graceful sequence-$(A_M,B_q)$ labelling defined in Definition \ref{defn:sequence-labelling}.
\end{rem}

\section{Construction of Hanzi-gpws}

Our goal is to translate a paragraph written by Hanzis into Hanzi-gpws. First, we will discuss two topics: One is decomposing Hanzi-graphs, another is constructing Hanzi-graphs; second we show approaches for building up Hanzi-gpws.

\subsection{Decomposing graphs into Hanzi-graphs}

We introduce some operations on planar graphs, since Hanzi-graphs are planar.

\begin{figure}[h]
\centering
\includegraphics[height=2.6cm]{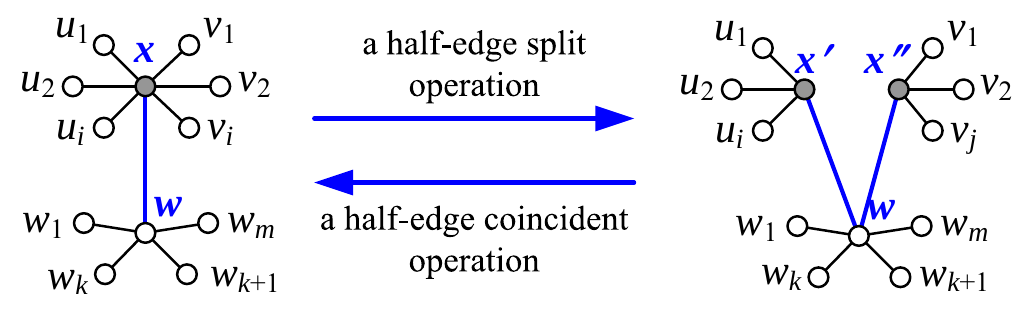}\\
\caption{\label{fig:split-operation-half-edge} {\small A scheme for a half-edge split
operation from left to right, and a half-edge coincident operation from right to left.}}
\end{figure}

\begin{defn}\label{defn:split-operation-combinatoric}
Let $xw$ be an edge of a $(p,q)$-graph $G$, such that $N(x)=\{w,u_1,u_2,\dots ,u_i, v_1,v_2,\dots ,v_j\}$ and $N(w)=\{x,w_1,w_2,\dots ,w_m\}$, and $xw\in E(G)$.
\begin{asparaenum}[Op-1. ]
\item A \emph{half-edge split operation} is defined by deleting the edge $xw$, and then splitting the vertex $x$ into two vertices $x',x''$ and joining $x'$ with these vertices $w,u_1,u_2,\dots ,u_i$, and finally joining $x''$ with these vertices $w,v_1,v_2,\dots ,v_j$. The resultant graph is denoted as $G \wedge ^{1/2}xw$, named as a \emph{half-edge split graph}, and $N(x')\cap N(x'')=\{w\}$ in $H$. (see Fig.\ref{fig:split-operation-half-edge})
\item A \emph{half-edge coincident operation} is defined as: Suppose that $N(x')\cap N(x'')=\{w\}$, we coincide $x'$ with $x''$ into one, denoted as $x=(x',x'')$, such that $N(x)=N(x')\cap N(x'')$, that is, delete one multiple edge. The resultant graph is denoted as $G(x'w\odot x'w)$, called a \emph{half-edge coincident graph}. (see Fig.\ref{fig:split-operation-half-edge})
\item \cite{Yao-Sun-Zhang-Mu-Sun-Wang-Su-Zhang-Yang-Yang-2018arXiv} A \emph{vertex-split operation} is defined in the way: Split $x$ into two vertices $x',x''$ such that $N(x')=\{w,u_1,u_2,\dots ,u_i\}$ and $N(x'')=\{v_1,v_2,\dots ,v_j\}$ with $N(x')\cap N(x'')=\emptyset$; the resultant graph is written as $G\wedge x$, named as a \emph{vertex-split graph}. (see Fig.\ref{fig:split-operation-vertex-edge} from (a) to (b))
\item \cite{Yao-Sun-Zhang-Mu-Sun-Wang-Su-Zhang-Yang-Yang-2018arXiv} A \emph{vertex-coincident operation} is defined by coinciding two vertices $x'$ and $x''$ in to one $x=(x',x'')$ such that $N(x)=N(x')\cup N(x'')$; the resultant graph is written as $G(x'\odot x'')$, called a \emph{vertex-coincident graph}. (see Fig.\ref{fig:split-operation-vertex-edge} from (b) to (a))

\begin{figure}[h]
\centering
\includegraphics[height=5.6cm]{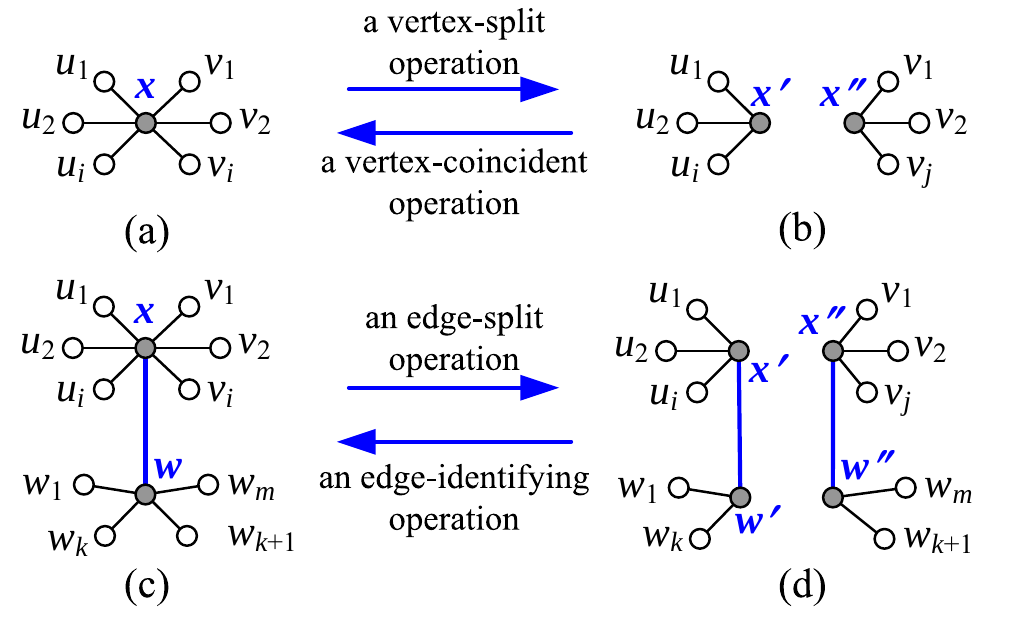}\\
\caption{\label{fig:split-operation-vertex-edge} {\small (a) A \emph{vertex-coincident graph} $G(x'\odot x'')$ obtained by a vertex-coincident operation from (b) to (a); (b) a \emph{vertex-split graph} $G\wedge x$ obtained by a vertex-split
operation from (a) to (b); (c) an \emph{edge-coincident graph} $G(x'w'\odot x''w'')$ obtained by an edge-coincident operation from (d) to (c); (d) an \emph{edge-split graph} $G\wedge xw$ obtained by an edge-split operation from (c) to (d).}}
\end{figure}

\item \cite{Yao-Sun-Zhang-Mu-Sun-Wang-Su-Zhang-Yang-Yang-2018arXiv} An \emph{edge-split operation} is defined as: Split the edge $xw$ into two edges $x'w'$ and $x''w''$ such that $N(x')=\{w',u_1,u_2,\dots ,u_i\}$ and $N(x'')=\{w'',v_1,v_2,\dots ,v_j\}$, $N(w')=\{x',w_1,w_2,\dots ,w_k\}$ and $N(w'')=\{x'',w_{k+1},w_{k+2},\dots ,w_{m}\}$; the resultant graph is written as $G\wedge xw$, named as an \emph{edge-split graph}, with $N(w')\cap N(x')=\emptyset$, $N(w')\cap N(x'')=\emptyset$, and $N(x')\cap N(w')=\emptyset$, $N(x')\cap N(w'')=\emptyset$. (see Fig.\ref{fig:split-operation-vertex-edge} from (c) to (d))
\item \cite{Yao-Sun-Zhang-Mu-Sun-Wang-Su-Zhang-Yang-Yang-2018arXiv} An \emph{edge-coincident operation} is defined by coinciding two edges $x'w'$ and $x''w''$ into one edge $xw$ such that $N(x)=N(x')\cup N(x'')\cup \{w=(w',w'')\}$ and $N(w)=N(w')\cup N(w'')\cup \{x=(x',x'')\}$; the resultant graph is written as $G(x'w'\odot x''w'')$, called an \emph{edge-coincident graph}. (see Fig.\ref{fig:split-operation-vertex-edge} from (d) to (c))
\item \cite{Bondy-2008} In Fig.\ref{fig:split-operation-contract-subdivision}, an \emph{edge-contracting operation} is shown as: Delete the edge $xy$ first, and then coincide $x$ with $y$ into one vertex $w=(x,y)$ such that $N(w)=[N(x)\setminus \{y\}]\cup [N(y)\setminus \{x\}]$. The resultant graph is denoted as $G\triangleleft xy$, named as an \emph{edge-contracted graph}.
 \item \cite{Bondy-2008} In Fig.\ref{fig:split-operation-contract-subdivision}, an \emph{edge-subdivided operation} is defined in the way: Split the vertex $w$ into two vertices $x,y$, and join $x$ with $y$ by a new edge $xy$, such that $N(x)\cap N(y)=\emptyset$, $y\in N(x)$, $x\in N(y)$ and $N(w)=[N(x)\setminus \{y\}]\cup [N(y)\setminus \{x\}]$. The resultant graph is denoted as $G\triangleright w$, called an \emph{edge-subdivided graph}.\qqed

\begin{figure}[h]
\centering
\includegraphics[height=1.8cm]{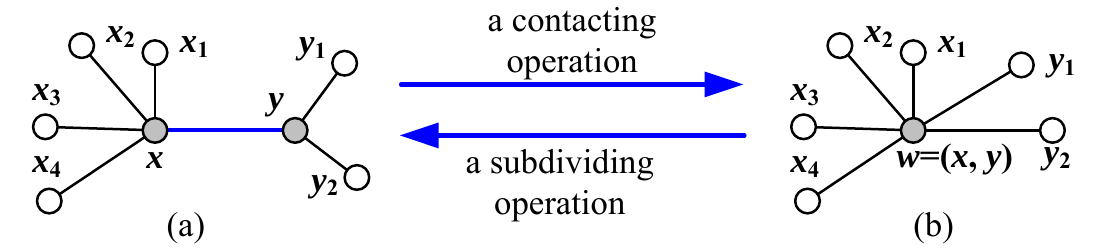}\\
\caption{\label{fig:split-operation-contract-subdivision} {\small (a) An \emph{edge-subdivided graph} $G\triangleright w$ obtained by subdividing a vertex $w=(x,y)$ into an edge $xy$ from (b) to (a); (b) an \emph{edge-contracted graph} $G\triangleleft xy$ obtained by contracting an edge $xy$ to a vertex $(x,y)$ from (a) to (b).}}
\end{figure}
\end{asparaenum}
\end{defn}

\begin{figure}[h]
\centering
\includegraphics[height=7.4cm]{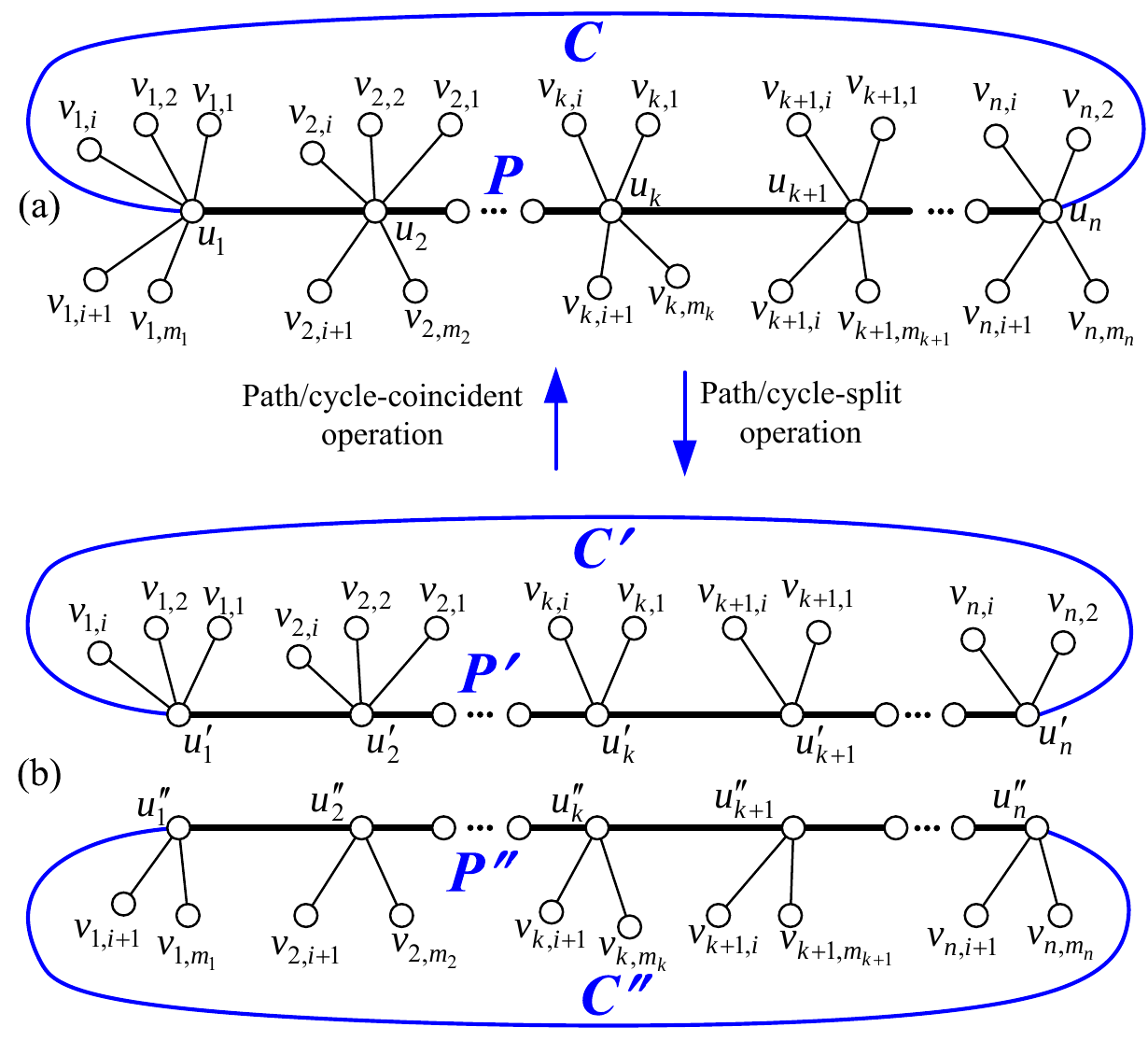}\\
\caption{\label{fig:split-operation-path-cycle} {\small (a) A path $P=u_1u_2\cdots u_n$ in black, and a cycle $C=P+u_1u_n$ in black and blue; (b) two cycles $C'=P'+u'_1u'_n$ and $C''=P''+u''_1u''_n$ obtained by a series of half-edge split operations.}}
\end{figure}

A simple example of using a series of half-edge split operations is shown in Fig.\ref{fig:use-half-edge-operation}. Furthermore, we obtain a half-edge split graph $H=G \wedge ^{1/2}\{u_iu_{i+1}\}^{n-1}_1$ shown in Fig.\ref{fig:split-operation-path-cycle} (b), where $P=u_1u_2\cdots u_n$, and $H$ contains two cycles $C'=P'+u'_1u'_n$ and $C''=P''+u''_1u''_n$, two paths $P'$ and $P''$ by a series of half-edge split operations.

\begin{figure}[h]
\centering
\includegraphics[height=6.8cm]{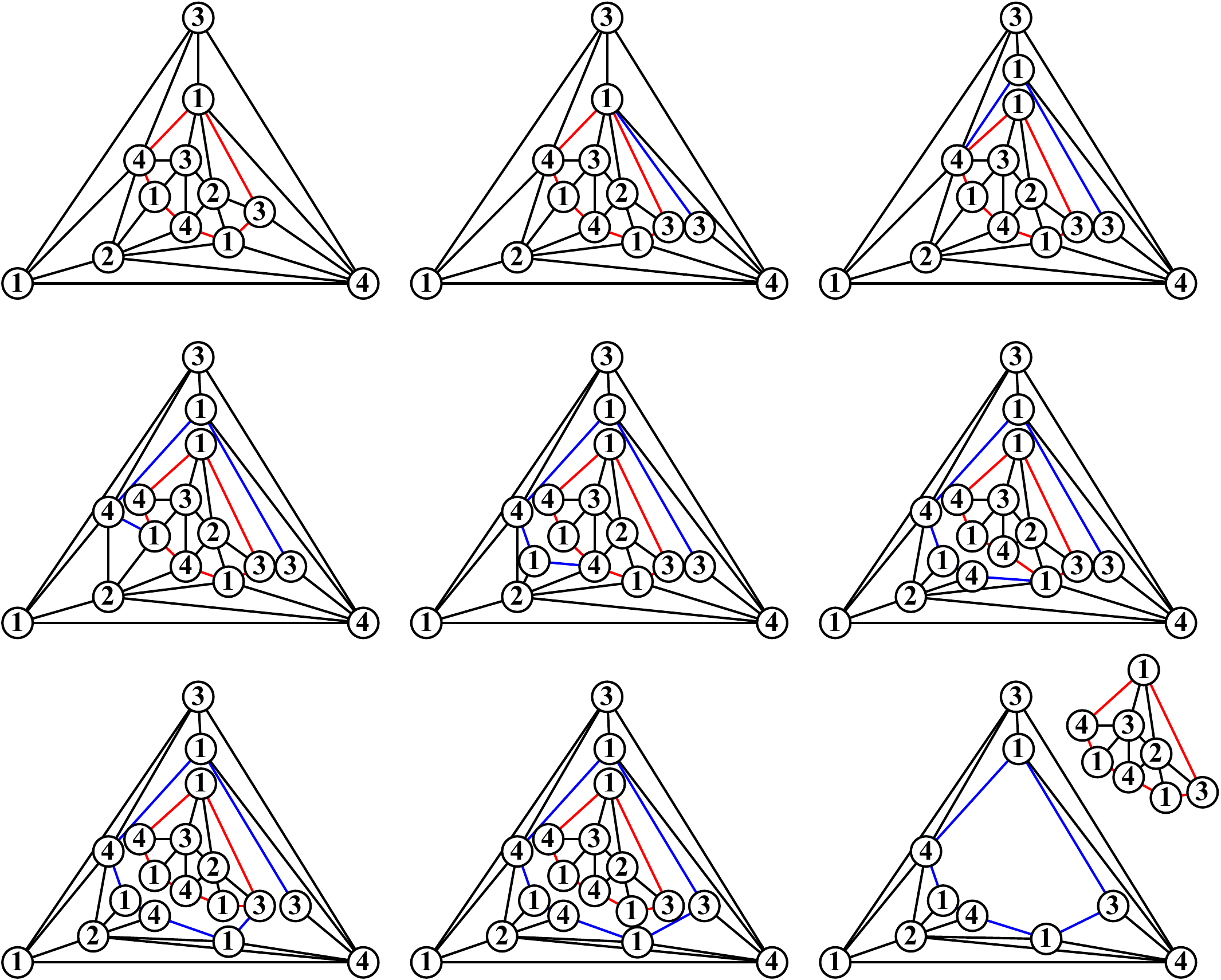}\\
\caption{\label{fig:use-half-edge-operation} {\small An example for illustrating the half-edge split operation.}}
\end{figure}

By the operations defined in Definition \ref{defn:split-operation-combinatoric} and the induction, we obtain the following result:

\begin{thm} \label{thm:decomposed-into-Hanzi-graphs}
Any simple graph $G$ can be decomposed into Hanzi-graphs $G_1,G_2,\dots, G_m$ with $G=\bigcup^m_{i=1}G_i$, $E(G)=\bigcup^m_{i=1}E(G_i)$ and $E(G_i)\cap E(G_j)=\emptyset $ if $i\neq j$.
\end{thm}

Let $D_{Hg}(G)$ be the smallest number of Hanzi-graphs obtained by decomposing a simple graph $G$. In general, there exists some simple graph $G$ such that $D_{Hg}(G)<D_{Hg}(H)$ for some proper subgraph $H$ of $G$. We present an example shown in Fig.\ref{fig:tree-56-55-yu}.

\begin{figure}[h]
\centering
\includegraphics[height=7cm]{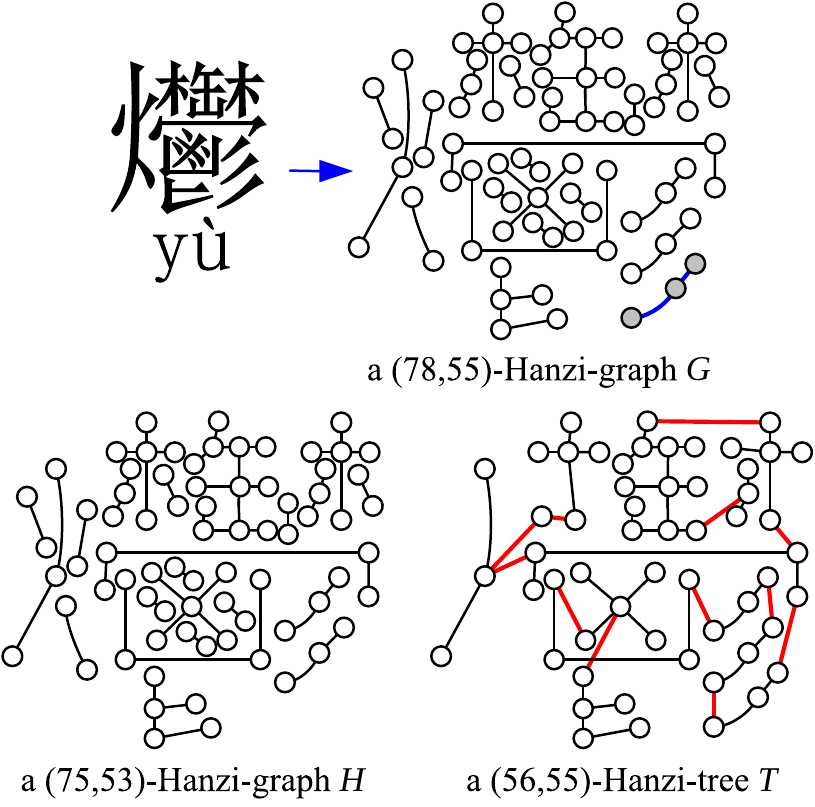}\\
\caption{\label{fig:tree-56-55-yu} {\small An example for illustrating $D_{Hg}(G)<D_{Hg}(H)$, since $H$ is not a Hanzi, $D_{Hg}(G)=1$, but $D_{Hg}(H)\geq 2$.}}
\end{figure}

In Fig.\ref{fig:tree-56-55-yu}, $G$ is a $(78,55)$-Hanzi-graph obtained from a Hanzi $H_{7229}$ with code 7229 in \cite{GB2312-80}, $H$ is a proper subgraph of $G$, and $T$ is a tree obtained from $G$. Clearly, the tree $T$ can be decomposed into a Hanzi-graph $G$ by the vertex-split operation. Moreover, the tree $T$ admits a set-ordered graceful labelling, since it is a caterpillar. However, $D_{Hg}(G)=1<2=D_{Hg}(H)$, since there does not exist one Hanzi having its Hanzi-graph to be $H$.

We show an \emph{algorithm} for making a flawed set-ordered graceful labelling of the $(78,55)$-Hanzi-graph $G$ by Fig.\ref{fig:tree-yu-flawed-graceful}:

\textbf{CATERPILLAR-CONSTRUCTION algorithm}

\textbf{Step 1.} Rearrange the caterpillar-like components $H_1,H_2,\dots ,H_n$ of a disconnected graph $G$ by $H_{i_1},H_{i_2},\dots ,H_{i_n}$, see (a) in Fig.\ref{fig:tree-yu-flawed-graceful};

\textbf{Step 2.} Join the components $H_{i_1},H_{i_2},\dots ,H_{i_n}$ of $G$ by new edges for obtaining a caterpillar $T'=\ominus ^m_{j=1}H_{j_1}$ with the new edge set $E^*=\{u_su_t: u_s\in E(H_{i_s}),u_t\in E(H_{i_t})\}$, see (b) in Fig.\ref{fig:tree-yu-flawed-graceful};

\textbf{Step 3.} Give a graceful set-ordered labelling of $T'$, see (c) in Fig.\ref{fig:tree-yu-flawed-graceful};

\textbf{Step 4.} Delete the added edges and then we get a flawed set-ordered graceful labelling of $G$, see (d) in Fig.\ref{fig:tree-yu-flawed-graceful}.

It is noticeable, we have two or more caterpillars like $T'$ by different permutations $H_{i_1},H_{i_2},\dots ,H_{i_n}$ of $H_1,H_2,\dots ,H_n$ of the Hanzi-graph $G$ and distinct joining ways to join $H_{i_1},H_{i_2},\dots ,H_{i_n}$ together by new edges. Thus, we claim that the Hanzi-graph $G$ admits two or more flawed set-ordered graceful labellings, which enables us to make  more complex TB-paws.

Fig.\ref{fig:tree-yu-flawed-graceful} (a) shows a Hanzi-graph $H_{yu}$ (read `y\`{u}') with 23 components $J_1,J_2,\dots ,J_{23}$, so we have $(23)!$ different permutations $J_{i_1}J_{i_2}\dots J_{i_{23}}$ of $J_1J_2\dots J_{23}$, which distribute us $(23)!$ Hanzi-gpws having flawed set-ordered graceful labellings, here, $(23)!\approx 2^{74.5}$. Thereby, Hanzis can produce various larger scale spaces of Hanzi-gpws.

\begin{figure}[h]
\centering
\includegraphics[height=9cm]{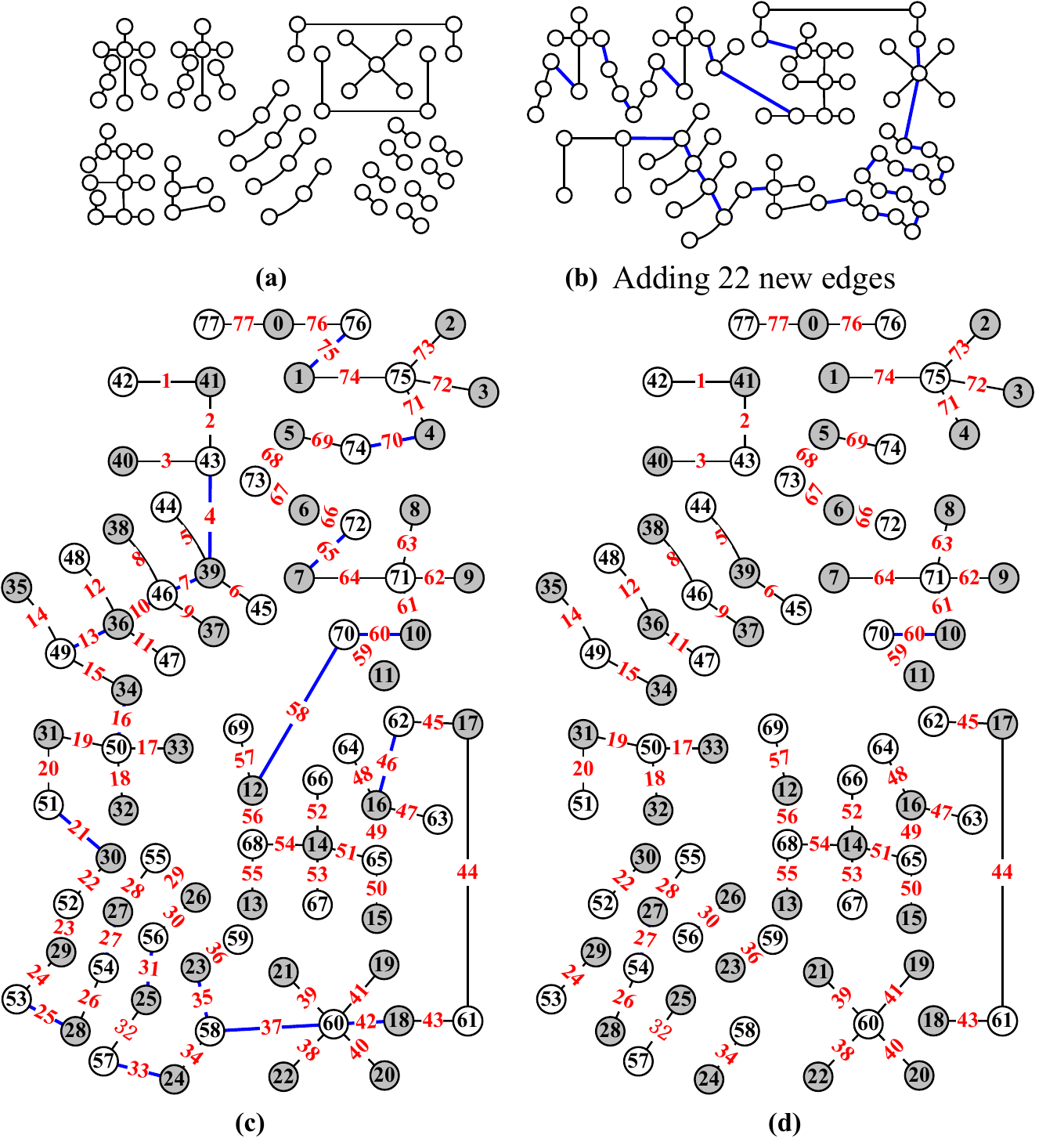}\\
\caption{\label{fig:tree-yu-flawed-graceful} {\small A scheme for illustrating the procedure of finding a flawed set-ordered graceful labelling of the $(78,55)$-Hanzi-graph $G$ shown in Fig.\ref{fig:tree-56-55-yu}.}}
\end{figure}

Our split operations can be used to decompose graphs and form new colorings/labellings. In \cite{Wang-Wang-Yao2019-Euler-Split} and \cite{Wang-Ma-Yao2019-spltting-connectivity}, the authors investigate the v-split  and e-split connectivity of graphs/networks. They define two new connectivities as follows:

\textbf{The v-divided connectivity.} A \emph{v-divided $k$-connected graph} $H$ holds: $H\wedge V^*$ (or $H\wedge \{x_i\}^{k}_1$) is disconnected, where $V^*=\{x_1,x_2,\dots,x_k\}$ is a subset of $V(H)$, each component $H_j$ of $H\wedge \{x_i\}^{k}_1$ has at least a vertex $w_j\not \in V^*$, $|V(H\wedge \{x_i\}^{k}_1)|=k+|V(H)|$ and $|E(H\wedge \{x_i\}^{k}_1)|=|E(H)|$. The smallest number of $k$ for which $H\wedge \{x_i\}^{k}_1$ is disconnected is called the \emph{v-divided connectivity} of $H$, denoted as $\kappa_{d}(H)$ (see an example shown in Fig.\ref{fig:5-degree-connectivity-new}).

\textbf{The e-divided connectivity.} An \emph{e-divided $k$-connected graph} $H$ holds: $H\wedge \{e_i\}^{k}_1$ (or $H\wedge E^*$) is disconnected, where $E^*=\{e_1,e_2,\dots,e_k\}$ is a subset of $E(H)$, each component $H_j$ of $H\wedge \{e_i\}^{k}_1$ has at least a vertex $w_j$ being not any end of any edge of $E^*$, $|V(H\wedge \{e_i\}^{k}_1)|=2k+|V(H)|$ and $|E(H\wedge \{e_i\}^{k}_1)|=k+|E(H)|$. The smallest number of $k$ for which $H\wedge \{e_i\}^{k}_1$ is disconnected is called the \emph{e-divided connectivity} of $H$, denoted as $\kappa'_{d}(H)$ (see an example shown in  Fig.\ref{fig:5-degree-connectivity-new}).

\begin{figure}[h]
\centering
\includegraphics[width=8.6cm]{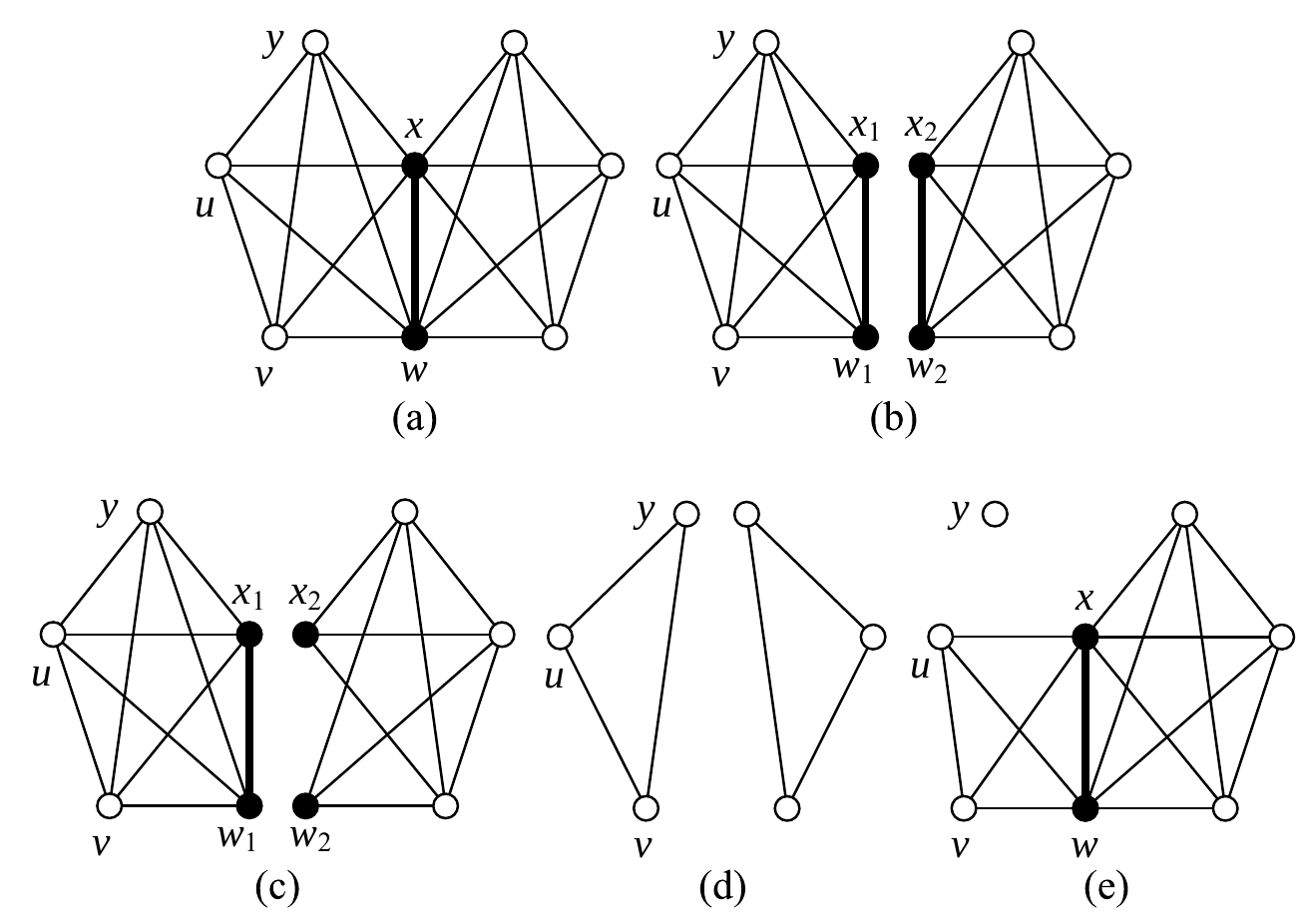}
\caption{\label{fig:5-degree-connectivity-new}{\small  (a) A graph $H$ with minimum degree $\delta(H)=4$;  (b) an e-divided graph $H\wedge xw$ with $\kappa'_{d}(H)=1$; (c) a v-divided graph $H\wedge \{x,w\}$ with $\kappa_{d}(H)=2$; (d) a v-deleted graph $H-\{x,w\}$ with $\kappa(H)=2$; (e) an e-deleted graph $H-\{yx,yw,yu,yv\}$ with $\kappa'(H)=4$.}}
\end{figure}

Recall that the minimum degree $\delta(H)$, the vertex connectivity $\kappa(H)$ and the edge connectivity $\kappa'(H)$ of a simple graph $G$ hold the following inequalities \cite{Bondy-2008}
\begin{equation}\label{eqa:popular-connectivity}
\kappa(H)\leq \kappa'(H)\leq \delta(H)
\end{equation}
true. Unfortunately, we  do not have the inequalities (\ref{eqa:popular-connectivity}) about the minimum degree $\delta(H)$, the v-divided connectivity $\kappa_{d}(H)$ and the e-divided connectivity $\kappa'_{d}(H)$ for a simple graph $H$. Moreover, we have
\begin{lem}\label{thm:equivalent-proof}
A graph $G$ is $k$-connected if and only if it is v-divided $k$-connected, namely, $\kappa_{d}(H)=\kappa(H)$.
\end{lem}

\begin{thm}\label{thm:lemma-equivalent-proof}
If a $k$-connected graph has a property related with its $k$-connectivity, so do a v-divided $k$-connected graph.
\end{thm}

\begin{thm}\label{thm:vertex-divided-vs-e-divided}
Any connected graph $G$ holds the inequalities $\kappa'_{d}(G)\leq \kappa_{d}(G)\leq 2\kappa'_{d}(G)$ true, and the boundaries are reachable.
\end{thm}

\begin{thm}\label{thm:Euler-graph-vs-hamilton-cycle}
Suppose that a connected graph $G$ has a subset $X$ holding $G-X$ to be not connected and to have the most number $n(G-X)=n_{dis}(G)$ of components if and only if each component of $G-X$ is a complete graph.
\end{thm}

As the application of the v-divided and v-coincident operations, the article  \cite{Wang-Ma-Yao2019-spltting-connectivity} shows

\begin{thm}\label{thm:equivalent-Euler-graph}
A simple graph $G$ of $n$ edges is a connected Euler's graph if and only if

(E-1) it can be divided into a cycle $C_n$ by a series of vertex divided operations;

(E-2) its overlapping kennel graph $H$ holds diameter $D(H)\leq 2$ and no vertex of $H$ is adjacent to two vertices of odd-degrees in $H$, simultaneously.
\end{thm}

\subsection{Tree-like structure of Hanzi-gpws}

The sentence ``tree-like structure'' means also ``linear structure''. We have: linear edge-joining structure, linear vertex-combined structure and linear ve-mixed structure.

We present some techniques for building up Hanzi-gpws in the following:

1. Suppose that each of disjoint graphs $G_1,G_2, \dots , G_m$ is connected. We join a vertex $u_i$ of $G_i$ with some vertex $u_j$ of $G_j$ for $i\neq j$ by a new edge, and denote this edge by $u_iu_j$ (called a \emph{joining edge}), such that each graph $G_t$ is joined with another graph $G_s$ with $s,t\in [1,m]$, and the resultant graph is denoted as $G=\ominus ^m_{i=1}G_i$, and the new edge set is denoted as $E^*=\{u_iu_j: u_i\in E(G_i),u_j\in E(G_j)\}$. We refer to this procedure as a \emph{linear edge-joining operation} if $|E^*|=m-1$.

2. If a Hanzi-graph $H$ is disconnected, and it has its components $H_1,H_2, \dots , H_n$, so we can do a linear edge-joining operation to $H$, such that the resultant graph $H^*$ is connected.

3. We combine a vertex $u_i$ of $G_i$ and somr vertex $u_j$ of $G_j$ with $i\neq j$ into one $w_{i,j}=(u_i,u_j)$, we say that $G_i$ is combined with $G_j$, such that each graph $G_t$ is combined with another graph $G_s$ with $s,t\in [1,m]$ and $s \neq t$. The resultant graph is denoted as $G=\odot ^m_1G_i$, we name the procedure of obtaining $G=\odot ^m_1G_i$ a \emph{linear vertex-combined operation} if the overlapping vertex set $\{w_{i,j}=(u_i,u_j)\}$ has just $m-1$ vertices.

We say that each linear edge-joining operation corresponds a linear vertex-combined operation, and call them as a \emph{linear matching operation}, since a linear vertex-combined operation can be defined as: Suppose that $G=\ominus ^m_{i=1}G_i$ is obtained by doing a linear edge-joining operation to disjoint connected Hanzi-graphs $G_1,G_2, \dots , G_m$, and $E^*$ is the set of all joining edges. We contract each joining edge of $E^*$ to a vertex (see Fig.\ref{fig:split-operation-contract-subdivision}), so two ends of the edge are coincided into one, called a \emph{combined vertex}, and the resultant graph is just $\odot ^m_{i=1}G_i$ defined above. Conversely, we can subdivide each combined vertex $x$ of $\odot ^m_{i=1}G_i$ (see Fig.\ref{fig:split-operation-contract-subdivision}), which is the result of combining a vertex $x_i$ of $G_i$ and a vertex $x_j$ of $G_j$ with $i\neq j$ into one, to be an edge $x_ix_j$. Thereby, the result of subdividing all combined vertices of $\odot ^m_{i=1}G_i$ is just $\ominus ^m_{i=1}G_i$ defined above. The operation mentioned here are defined in Definition \ref{defn:split-operation-combinatoric}.

\begin{thm} \label{thm:contract-result}
Suppose $G=\ominus ^m_{i=1}G_i$ obtained by doing a series of edge-joining operations on disjoint connected graphs $G_1,G_2, \dots , G_m$. We contract each $G_i$ to a vertex, the resultant graph is a tree if and only if the joining edge set $E^*$ of $G$ has only $(m-1)$ edges.
\end{thm}

\begin{thm} \label{thm:contract-subdivide}
Suppose $G=\odot ^m_{i=1}G_i$ obtained by doing a series of vertex-combined operations on disjoint connected graphs $G_1,G_2, \dots , G_m$. We subdivide each combined vertex $x$ of $\odot ^m_{i=1}G_i$ by an edge $x_ix_j$ with $x_i\in V(G_i)$ and $x_j\in V(G_j)$, so we get $G\wedge \{x\}=\ominus ^m_{i=1}G_i$, where $\{x\}$ is the set of combined vertices. Then $G=\odot ^m_{i=1}G_i$ is obtained by doing linear vertex-combined operations on disjoint connected graphs $G_1,G_2, \dots , G_m$ if and only if $G\wedge \{x\}=\ominus ^m_{i=1}G_i$ is a result of doing linear edge-joining operations on disjoint connected graphs $G_1,G_2, \dots , G_m$.
\end{thm}

\subsection{Non-tree-like structure of Hanzi-gpws}

Let $\{G_i\}^m_1=\{G_1,G_2,\dots ,G_m\}$. We join a vertex $u\in V(G_i)$ with some vertex $v\in V(G_j)$ by an edge $uv$ with $i,j\in [1,m]$ and $i\neq j$, the resultant graph is denoted as $H=\ominus \langle \{G_i\}^m_1\rangle $, named as a \emph{compound graph}, and particularly, $H=\ominus \langle G_1, G_2\rangle $ if $m=2$. If any pair of $G_i$ and $G_j$ is joined by a unique edge $uv$ with $u\in V(G_i)$ and $v\in V(G_j)$, then we say the compound graph $H$ to be \emph{simple}, otherwise $H$ to be \emph{multiple}. We call $H$ a non-tree-like structure if $H$ contains a cycle $C=u_{i_1}u_{i_2}\cdots u_{i_k}u_{i_1}$ with $k\geq 3$ and $u_{i_j}\in V(G_{i_j})$ and $j\in [1,k]$, as well as $G_{i_j}\in \{G_i\}^m_1$ with $i_j\neq i_t$ if $j\neq t$.

\subsection{Graphs made by Hanzi-gpws}

\begin{defn}\label{defn:graceful-k-rotatable}
A $(p,q)$-graph $G$ is \emph{graceful $k$-rotatable} for a constant $k\in [0,q]$ if any vertex $u$ of $G$ can be labelled as $f(u)=k$ by some graceful labelling $f$ of $G$.\qqed
\end{defn}

We can generalize Definition \ref{defn:graceful-k-rotatable} to other labellings as follows:
\begin{defn}\label{defn:graceful-0-rotatable}
(\cite{Gallian2018, Wang-Xu-Yao-2016, Wang-Xu-Yao-Key-models-Lock-models-2016, Wang-Xu-Yao-2017-Twin2017,Yao-Sun-Zhang-Mu-Sun-Wang-Su-Zhang-Yang-Yang-2018arXiv, Yao-Zhang-Sun-Mu-Sun-Wang-Wang-Ma-Su-Yang-Yang-Zhang-2018arXiv}) A $(p,q)$-graph $G$ is \emph{$\varepsilon$-$k$-rotatable} for a constant $k\in [a,b]$ if any vertex $u$ of $G$ can be labelled as $f(u)=k$ by some $\varepsilon$-labelling $f$, where $\varepsilon\in \{$odd-graceful, twin odd-graceful, elegant, odd-elegant, twin odd-elegant, felicitous, edge-magic total, edge-magic total graceful, 6C-labelling, etc.$\}$.\qqed
\end{defn}

In fact, determining the $0$-rotatable gracefulness of trees is not slight, even for caterpillars (Ref. \cite{Zhou-Yao-Cheng-2011}). A hz-$\varepsilon$-graph $G$ (see Fig.\ref{fig:graceful-transformation-1} and Fig.\ref{fig:hz-varepsilon-graph}) has its vertices corresponding Hanzi-gpws, and two vertices of $G$ are adjacent to each other if the corresponding Hanzi-gpws $H^{gpw}_{u_i}$ having its $\varepsilon$-labelling $f_i$ and topological structure $T_{u_i}$ with $i=1,2$ can be transformed to each other by the following rules, where $f_1$ and $f_2$ are the same-type of labellings, for example, $f_1$ and $f_2$ are flawed graceful, and so on.

\begin{asparaenum}[\textbf{Rule}-1. ]
\item \label{item:dual-labelling} The labelling of $H^{gpw}_{u_i}$ is the dual labelling of the labelling of $H^{gpw}_{u_{3-i}}$ with $i=1,2$, so $T_{u_1}\cong T_{u_2}$, here $T_{u_i}$ is the topological structure of Hanzi-gpw $H^{gpw}_{u_i}$ with $i=1,2$.
\item \label{item:adding-reducing-v-e} $H^{gpw}_{u_i}$ is obtained by adding (reducing) vertices and edges to $H^{gpw}_{u_{3-i}}$, such that $f_i$ can be deduced by $f_{3-i}$ with $i=1,2$.
\item \label{item:complete-same} $H^{gpw}_{u_1}=H^{gpw}_{u_2}$, namely, $f_1=f_2$ and $T_{u_1}\cong T_{u_2}$.
\item \label{item:image} $H^{gpw}_{u_i}$ is an image of Hanzi-gpw $H^{gpw}_{u_{3-i}}$ with $i=1,2$.
\item \label{item:inverse} $H^{gpw}_{u_i}$ is an inverse of Hanzi-gpw $H^{gpw}_{u_{3-i}}$ with $i=1,2$.
\item \label{item:same-structure-different-lab} $T_{u_1}\cong T_{u_2}$, but $f_1\neq f_2$.
\end{asparaenum}

\begin{figure}[h]
\centering
\includegraphics[height=10cm]{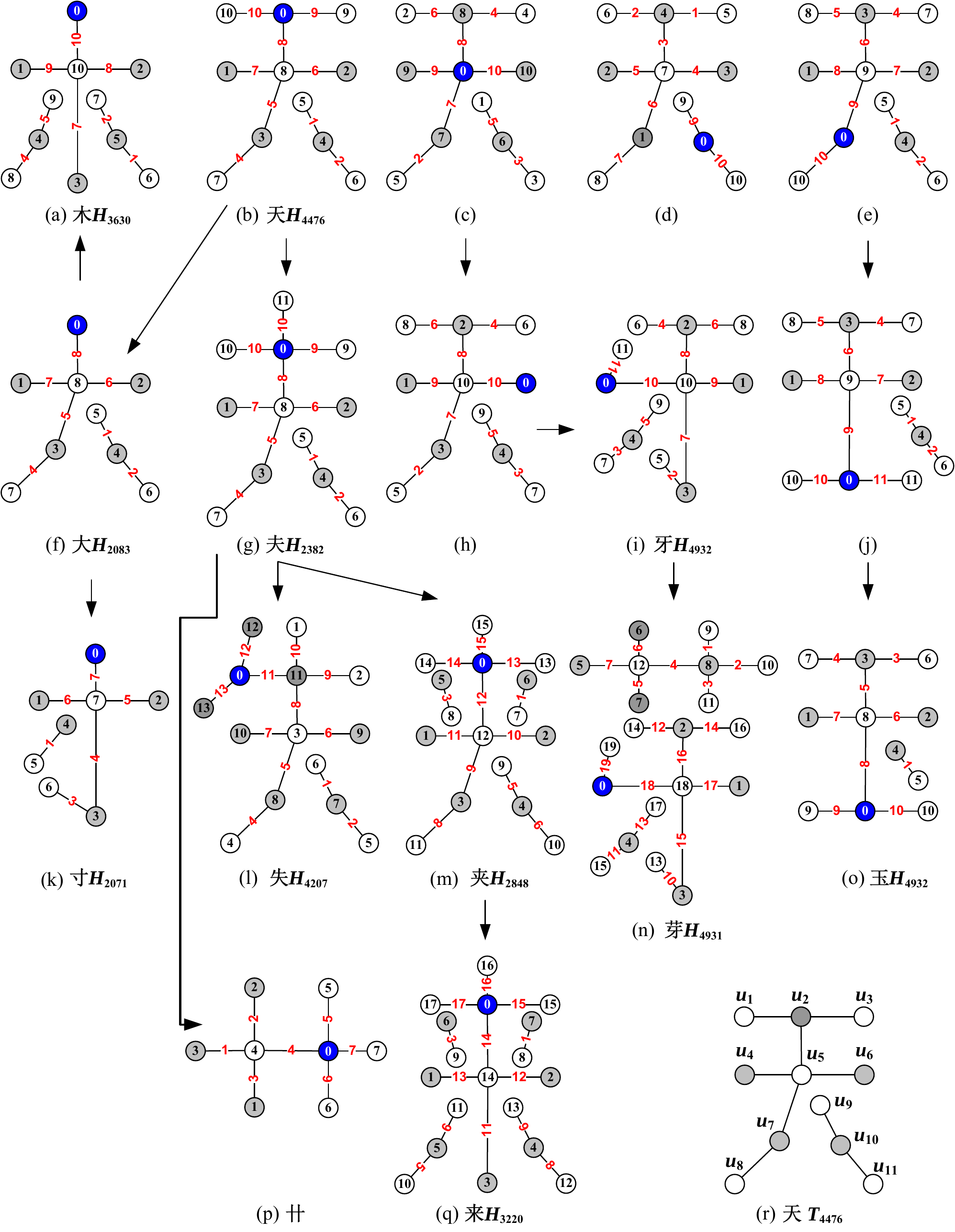}\\
\caption{\label{fig:graceful-transformation-1} {\small A derivative hanzi-system built on a Hanzi $H_{4476}$.}}
\end{figure}

A \emph{generalized Hanzi-gpw system} shown in Fig.\ref{fig:graceful-transformation-1} is based on a Hanzi $H_{4476}$ with its topological structure $T_{4476}$ shown in Fig.\ref{fig:graceful-transformation-1} (r). Based on the Hanzi $H_{4476}$, four Hanzi-gpws $H^{gpw_j}_{4476}$ with $j\in[1,4]$ are (b)$=H^{gpw_1}_{4476}$, (c)$=H^{gpw_2}_{4476}$, (d)$=H^{gpw_3}_{4476}$ and (e)$=H^{gpw_4}_{4476}$ respectively, which tell us that the Hanzi $H_{4476}$ admits \emph{flawed graceful $o$-rotatable labellings}, since each vertex of the Hanzi $H_{4476}$ can be labelled with $0$ by some graceful labelling of $H_{4476}$ under \textbf{Rule}-\ref{item:dual-labelling}. Because each Hanzi-gpws $H^{gpw_j}_{4476}$ admit a \emph{flawed set-ordered graceful labelling} $f_j$ with $j\in[1,4]$, so $f_j(V(H^{gpw_j}_{4476}))=[0,10]$ for $j\in[1,4]$ according to Definition \ref{defn:group-flawed-labellings}.

\begin{figure}[h]
\centering
\includegraphics[height=3cm]{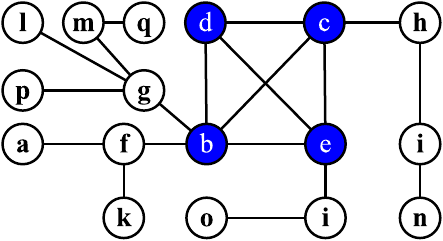}\\
\caption{\label{fig:hz-varepsilon-graph} {\small A hz-$\varepsilon$-graph made by the derivative hanzi-system of a Hanzi $H_{4476}$ shown in Fig.\ref{fig:graceful-transformation-1}, where $\varepsilon$ is the graceful labelling.}}
\end{figure}

Clearly, our hz-$\varepsilon$-graphs are Hanzi-gpws too. In Fig.\ref{fig:graceful-transformation-1}, we can see the following facts:

\begin{asparaenum}[(1) ]
\item Two flawed set-ordered graceful labellings in Fig.\ref{fig:graceful-transformation-1}(c) and (h) are \emph{dual} to each other. Again, adding a new vertex and a new edge makes (i)$=H^{gpw}_{4932}$; then (n)$=H^{gpw}_{493}$ is obtained by adding 8 new vertices and 7 new edges to (i).

\item By \textbf{Rule}-\ref{item:adding-reducing-v-e}, $H^{gpw}_{4476} \rightarrow H^{gpw}_{2328}$ from (b) to (g) after adding a vertex and an edge to $H^{gpw}_{4476}$.

\item Do \textbf{Rule}-\ref{item:dual-labelling} to (g)$=H^{gpw}_{2328}$, and adding two vertices and two edges produces (l)$=H^{gpw}_{4207}$ by \textbf{Rule}-\ref{item:adding-reducing-v-e}.

\item Deleting two vertices and two edges from (b)$=H^{gpw}_{4476}$, under \textbf{Rule}-\ref{item:adding-reducing-v-e}, so obtaining (f)$=H^{gpw}_{2083}$; and deleting a vertex and an edge of (f)$=H^{gpw}_{2083}$ produces (k)$=H^{gpw}_{2071}$; and do \textbf{Rule}-\ref{item:adding-reducing-v-e} two times for getting (a)$=H^{gpw}_{3630}$.

\item The Hanzi-gpw (q)$=H^{gpw}_{3220}$ is the result of adding four new vertices and two new edges, and deleting an old vertex and an old edge to (m)$=H^{gpw}_{2848}$.
\end{asparaenum}

\begin{figure}[h]
\centering
\includegraphics[height=9.2cm]{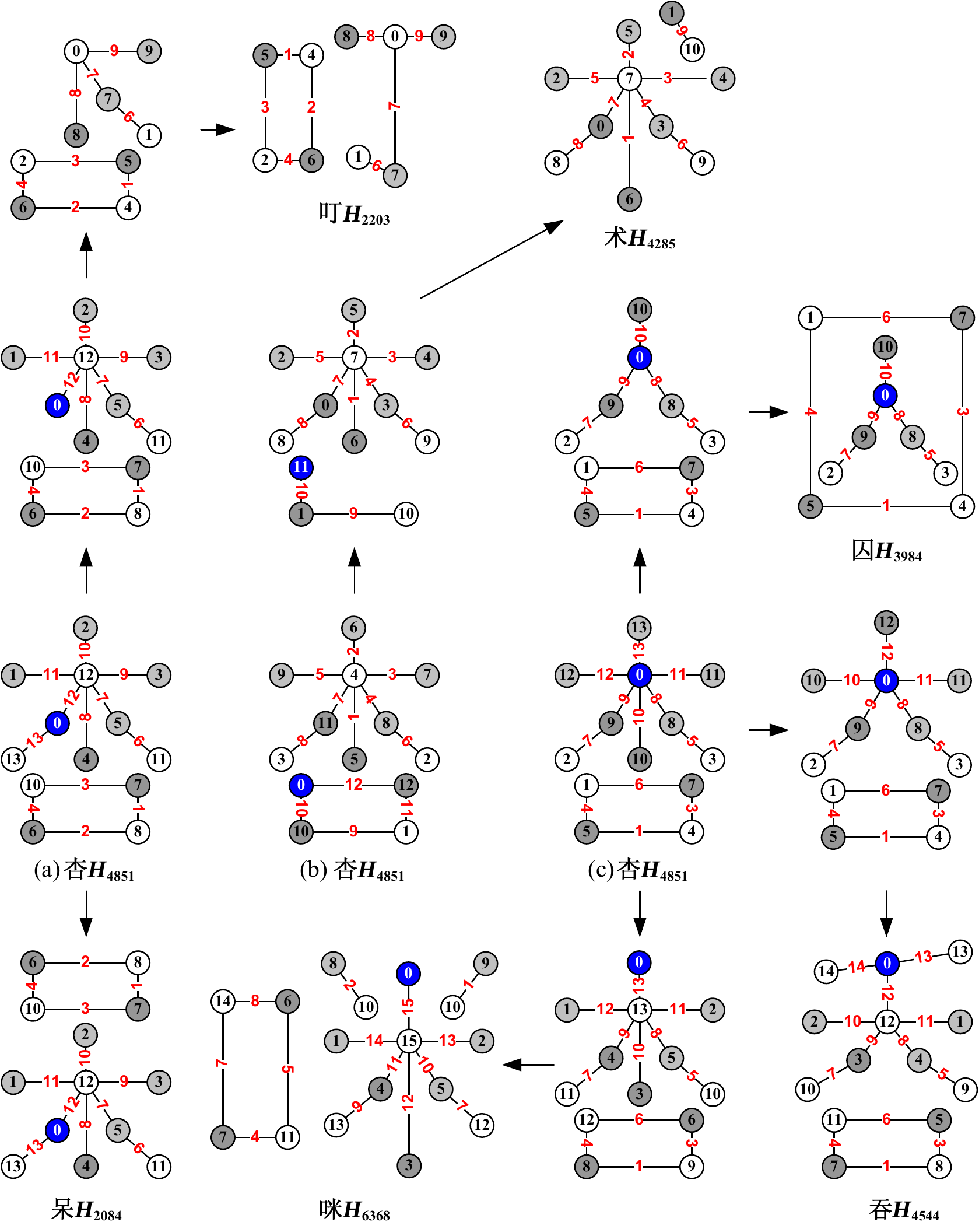}\\
\caption{\label{fig:graceful-transformation-2} {\small Hanzi-gpws made by Hanzis lebelled with $H_{abcd}$ and pseudo-Hanzis with no labels.}}
\end{figure}

\subsection{Estimating space of Hanzi-gpws}

Estimating the space of Hanzi-gpws will be facing the following basic problems:

\begin{asparaenum}[\textrm{Space}-1. ]
\item Estimating the space of Hanzi-graphs.

\item Find out graph labellings that are admitted by Hanzi-graphs.
\end{asparaenum}

As known, there are over 200 graph labellings introduced in the famous survey \cite{Gallian2018}. Meanwhile, new graph labellings emerge everyday.

We present an example in Fig.\ref{fig:1-group-more-flawed-graceful}, in which there are seven connected components in a group of three Hanzis-gpws $H^{gpw}_{4043}$, $H^{gpw}_{2511}$ and $H^{gpw}_{3829}$. Clearly, we have $7!=5040$ flawed set-ordered graceful labellings on $H^{gpw}_{4043}$, $H^{gpw}_{2511}$ and $H^{gpw}_{3829}$, which can deduce $7!=5040$ flawed set-ordered $\varepsilon$-labellings, here, the flawed set-ordered $\varepsilon$-labelling is equivalent to the flawed set-ordered graceful labelling in Definition \ref{defn:group-flawed-labellings}.

\begin{figure}[h]
\centering
\includegraphics[height=10cm]{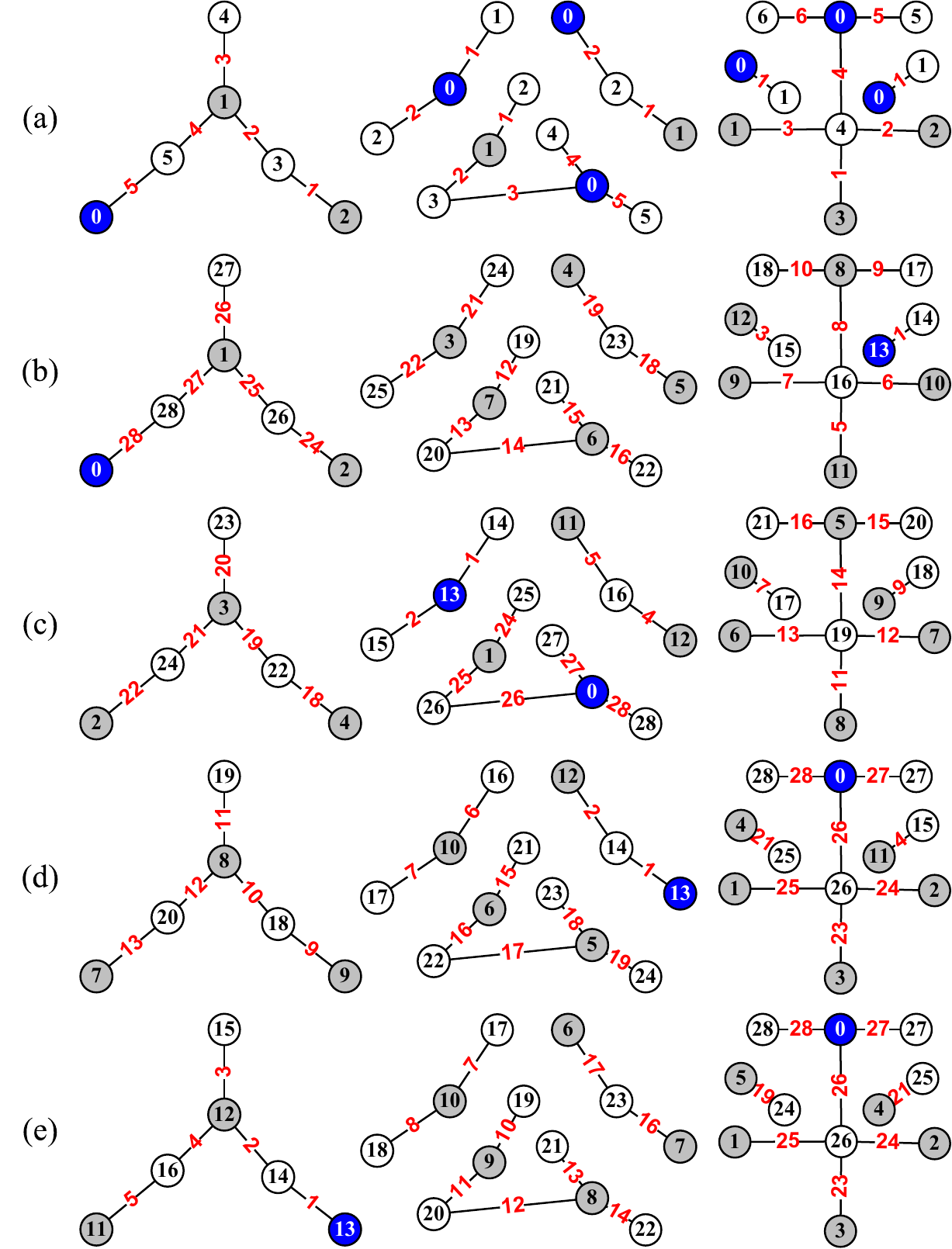}\\
\caption{\label{fig:1-group-more-flawed-graceful} {\small (a) A group of three Hanzis-gpws with their seven connected components having set-ordered graceful labellings; (b)-(e) four different flawed set-ordered graceful labellings.}}
\end{figure}

Two examples shown in Fig.\ref{fig:rrhg-set-graceful-0} and Fig.\ref{fig:rrhg-set-graceful-1} motivate us to find such edge sets $E^*$ to join disjoint connected graphs $G_1,G_2,\dots, G_m$, such that the resultant graph $H=E^*+G$, where $G=\bigcup ^m_{i=1}G_i$, admits a set-ordered graceful labelling $f$. So, $G=\bigcup^m_{i=1}G_i$ admits a flawed set-ordered graceful labelling defined in Definition \ref{defn:group-flawed-labellings}. Notice that

(1) $E(G)=E(H)\setminus E^*$ with $|E(H)|=91$,

(2) $V(G)=V(H)=X\cup Y$ such that $x\in X$ and $y\in Y$ for each edge $xy\in E(G)$ or $xy\in E(H)$, and ``set-ordered'' means $\max f(X)<\min f(Y)$.

Hence, we have the edge label set $f(E^*)=\{f(x_iy_i)=f(y_i)-f(x_i): ~x_iy_i\in E^*\}=\{2$, $6$, $11$, $20$, $26$, $29$, $37$, $40$, $42$, $44$, $47$, $51$, $57$, $60$, $63$, $70$, $80$, $86\}$ with $18=|E^*|$. Thereby, we can set: $(2+k_1)-k_1=2$, $(6+k_2)-k_2=6$, $(11+k_3)-k_3=11$, $(20+k_4)-k_4=20$, $(26+k_5)-k_5=26$, $(29+k_6)-k_6=29$, $(37+k_7)-k_7=37$, $(40+k_8)-k_8=40$, $(42+k_9)-k_9=42$, $(44+k_{10})-k_{10}=44$, $(47+k_{11})-k_{11}=47$, $(51+k_{12})-k_{12}=51$, $(57+k_{13})-k_{13}=57$, $(60+k_{14})-k_{14}=60$, $(63+k_{15})-k_{15}=63$, $(70+k_{16})-k_{16}=70$, $(80+k_{17})-k_{17}=80$, $(86+k_{18})-k_{18}=86$. Each $k_i\in [a_i,b_i]$ with $i\in [1,18]$, for example, $k_{18}\in [0,5]$, $k_{17}\in [0,11]$, $k_{16}\in [0,21]$, $k_{15}\in [0,32]$, etc. We have deduced $k_i\in [0,91-f(x_iy_i)]=[0,91-f(y_i)+f(x_i)]$ with $i\in [1,18]$.

We refer to the set $f(E^*)$ as an \emph{$f$-base} under the set-ordered graceful labelling $f$ of $H$. Each group $L=\{k_{i_1},k_{i_2},\dots ,k_{i_{18}}\}$ with $k_i\in [a_i,b_i]$ for $i\in [1,18]$ yields an edge set $E(E^*,L)$ to form a connected graph $H_L$ having a set-ordered graceful labelling $f_L$ and $V(G)=V(H_L)$, as well as $E(G)=E(H)\setminus E(E^*,L)$ and $|E(H_L)|=91$ based on the $f$-base. We are ready to present a generalization of flawed type labellings.

\begin{defn}\label{defn:generalization-flawed-labellings-00}
$^*$ Let $G_1,G_2,\dots, G_m$ be disjoint connected graphs, let $G=\bigcup^m_{i=1}G_i$ and let $E_j$ with $j\in [1,n]$ be an edge set such that each edge $uv$ of $E_j$ has one end $u$ in some $G_i$ and another end $v$ is in some $G_j$ with $i\neq j$, and $E_j$ joins $G_1,G_2,\dots, G_m$ together to form a connected graph $H_j$, denoted as $H_j=E_j+G$. We say $G=\bigcup^m_{i=1}G_i$ to be a disconnected $(p,q)$-graph with $p=|V(H_j)|=\sum^m_{i=1}|V(G_i)|$ and $q=|E(H_j)|-|E_j|=(\sum^m_{i=1}|E(G_i)|)-|E_j|$.

Suppose that $H_k$ admits a set-ordered graceful labelling $f_k$ with $k\in [1,n_k]$, where $n_k$ is the number of set-ordered graceful labellings admitted by $H_k$, correspondingly, $G$ admits a flawed set-ordered graceful labelling $g_k$ with $k\in [1,n_k]$, such that $\max f_k(X)<\min f_k(Y)$ and $\max g_k(X)<\min g_k(Y)$ with $V(G)=V(H_k)=X\cup Y$. We call the edge label set $f_k(E_k)=\{f_k(x_{k,i}y_{k,i})=f_k(y_{k,i})-f_k(x_{k,i}): ~x_{k,i}y_{k,i}\in E_k,~i\in [1,|E_k|]\}$ as an \emph{$f_k$-base}. Since $[f_k(x_{k,i}y_{k,i})+k_s]-k_s=f_k(x_{k,i}y_{k,i})$ as each $k_s\in [a_s,b_s]$ ($a_s\in f_k(X)$ and $b_s\in f_k(Y)$) with $s\in [1,|E_k|]$, each new edge set $E^s_k=\{u_{s,i}v_{s,i}:~i\in [1,|E_k|]\}$ holding $f_k(E^s_k)=f_k(E_k)$ induces a connected adding-edge graph $H^s_k=E^s_k+G$.\qqed
\end{defn}

In Definition \ref{defn:generalization-flawed-labellings-00}, it is easy to see $m-1=\min\{|E_j|:j\in [1,n]\}$, however, determining $\max \{|E_j|:j\in [1,n]\}$ seems to be not slight. Since there are many connected graphs $H^s_k=E^s_k+G$ induced by each $f_k$-base based on $H_k$ shown in Definition \ref{defn:generalization-flawed-labellings-00}. Let $N_{um}(H^s_k)$ be the number of set-ordered graceful labellings admitted by $H^s_k$, and so $G$ admits $n_k\cdot |E_k|\cdot N_{um}(H^s_k)$ set-ordered graceful labellings at least.

\begin{figure}[h]
\centering
\includegraphics[height=6.2cm]{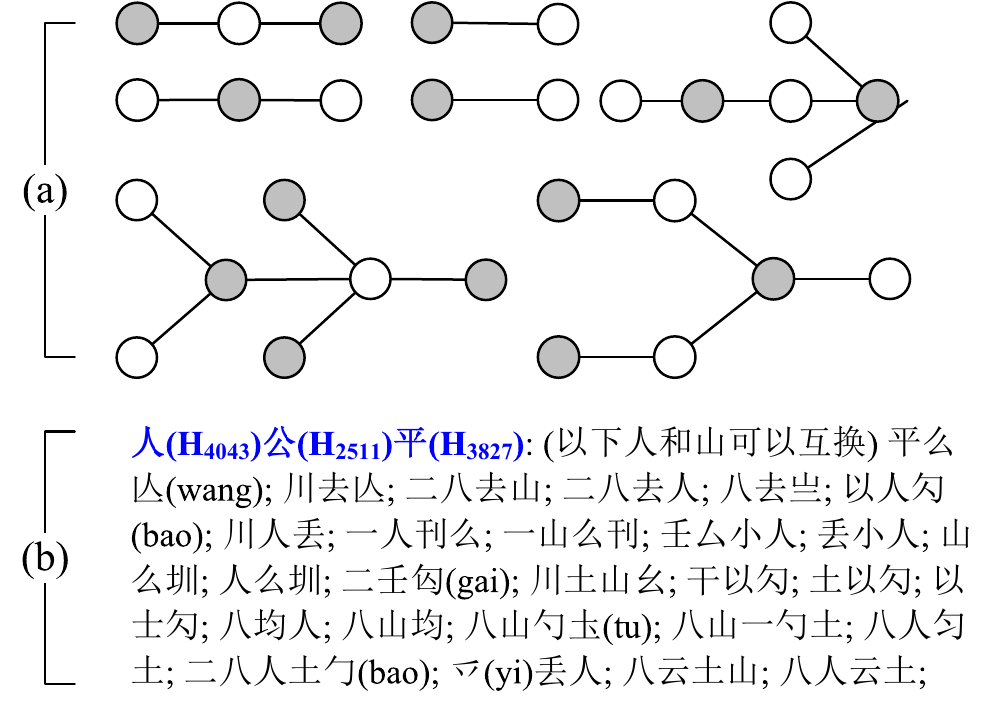}\\
\caption{\label{fig:1-group-more-flawed-graceful} {\small Paragraphs can be transformed into Hanzi-graphs having the same components with that of three Hanzis $H_{4043}$, $H_{2511}$ and $H_{3829}$. (a) A public key; (b) possible private keys.}}
\end{figure}

\subsection{Self-growable Hanzi-gpws}

Some Hanzi-gpws can grow to many Topsnut-gpws. An example shown in Fig.\ref{fig:self-growing-11} is just a \emph{self-growable Hanzi-gpw}, in which $H^{gpw}_{4476}$ is as a public key, $H^{gpw}_{4585}$ is as a private key, the authentication is shown in (e). And the graph $G_1=H^{gpw_0}_{4476}$ shown in (a) admits a flawed set-ordered graceful labelling. It is not hard to see $G_1\rightarrow G_2\rightarrow \cdots \rightarrow G_m$, each $G_i$ admits a flawed set-ordered graceful labelling. Since $G_i$ is a proper subgraph of $G_{i+1}$, we name the sequence $\{G_i\}^m_1$ as a \emph{self-growing Hanzi-gpw sequence}.

\begin{figure}[h]
\centering
\includegraphics[height=7cm]{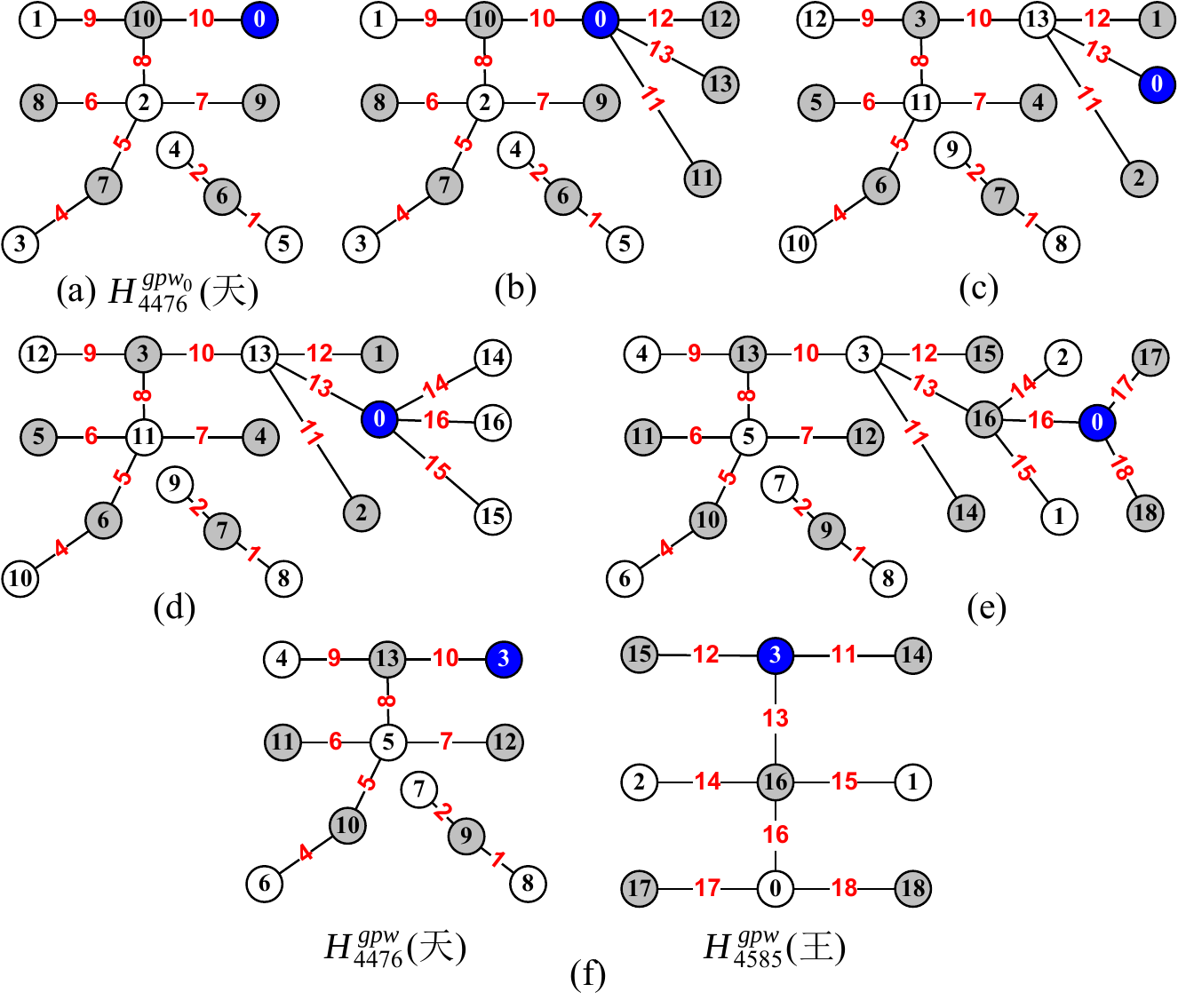}\\
\caption{\label{fig:self-growing-11} {\small Two Hanzi-gpws ($H_{4476}$, $H_{4585}$) shown in (f) are made from (a) to (e).}}
\end{figure}

Observe examples shown in Fig.\ref{fig:self-growing-11} and Fig.\ref{fig:self-growing-22}, we can define a new labelling in Definition \ref{defn:joint-flawed-labellings} as follows:
\begin{defn}\label{defn:joint-flawed-labellings}
$^*$ Let $H=\bigcup^m_{i=1}G_i$ with subgraphs $G_1,G_2,\dots, G_m$, and let $H$ admit a flawed set-ordered graceful labelling $h$. If each $G_i$ admits a flawed set-ordered graceful labelling $h_i$ induced by $h$, such that $|h_i(E(G_i))\cap h_{i+1}(E(G_{i+1}))|=1$ for $i\in [1,m-1]$, then we say $h$ a \emph{flawed jointly set-ordered graceful labelling} of $H$ (see examples shown in Fig.\ref{fig:self-growing-11} and Fig.\ref{fig:self-growing-22}).\qqed
\end{defn}

Comparing Definition \ref{defn:generalization-flawed-labellings-00} with Definition \ref{defn:joint-flawed-labellings}, we can claim
\begin{thm}\label{thm:flawed-jointly-set-ordered}
A graph  $G$ defined in Definition \ref{defn:generalization-flawed-labellings-00} admits a flawed set-ordered graceful labelling if and only if $H$ defined in Definition \ref{defn:joint-flawed-labellings} admits a flawed jointly set-ordered graceful labelling.
\end{thm}

\begin{figure}[h]
\centering
\includegraphics[height=7.6cm]{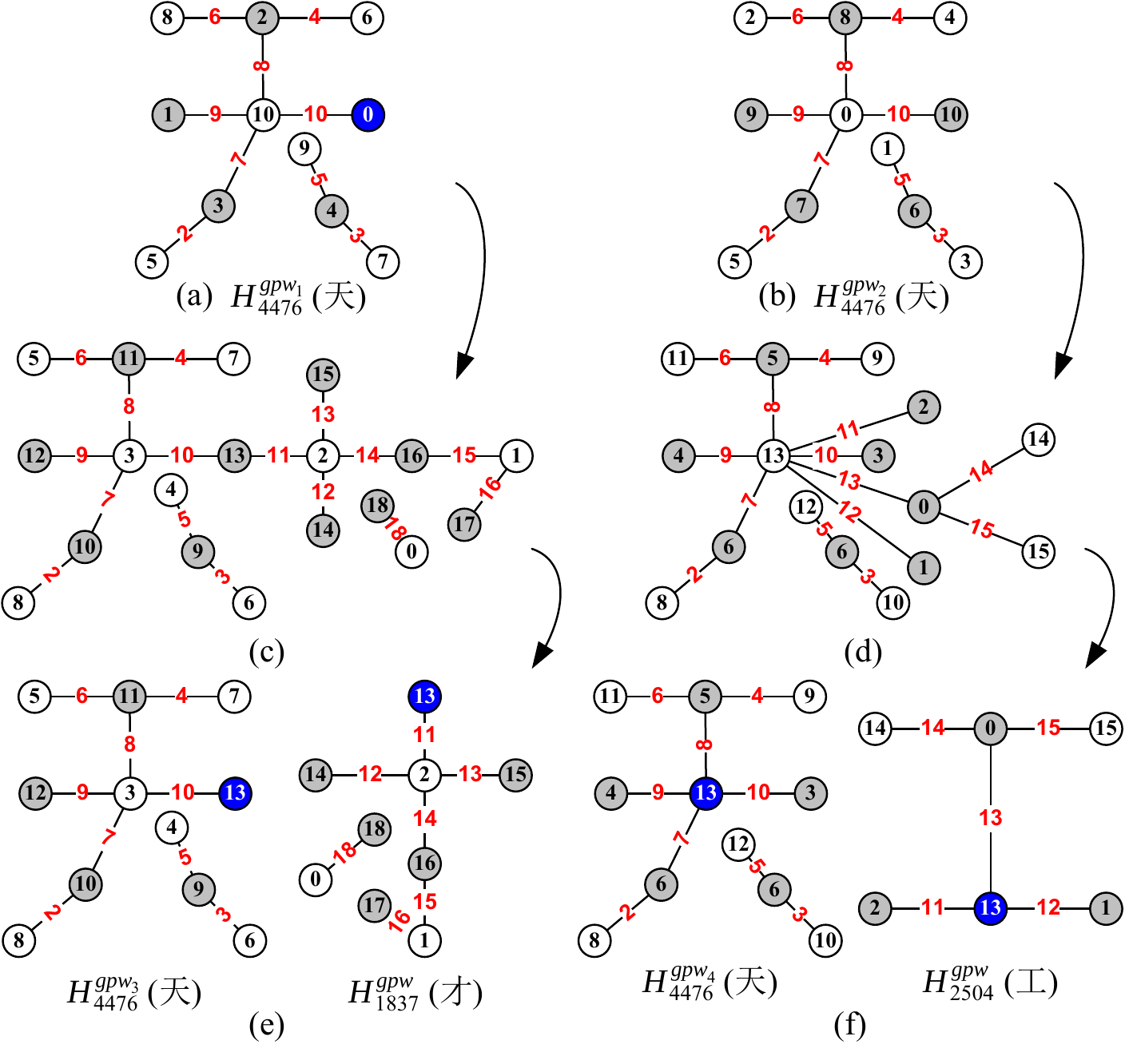}\\
\caption{\label{fig:self-growing-22} {\small Two examples for self-growing Hanzi-gpws admitting flawed set-ordered graceful labellings.}}
\end{figure}

\section{Producing text-based passwords from Hanzi-gpws}

In \cite{Yao-Zhang-Sun-Mu-Sun-Wang-Wang-Ma-Su-Yang-Yang-Zhang-2018arXiv}, some techniques were introduced for making TB-paws from Topsnut-gpws. We will produce TB-paws from Hanzi-gpws in this section.

We are facing the following tasks:

(1) How many techniques are there for generating TB-paws from Hanzi-gpws? And can these techniques be translated into efficient algorithms?

(2) How long bytes of those TB-paws made by Hanzi-gpws are there? As known, the largest prime is $M_{77232917}=2^{77232917}-1$, which has 23,249,425 (twenty-three million two hundred and forty-nine thousand four hundred and twenty-five) bytes long.

(3) How to reconstruct Hanzi-gpws by TB-paws?

\subsection{Matrix expressions of Hanzi-gpws}

As known, each $(p,q)$-graph $G$ of graph theory has itself \emph{adjacency matrix} $A(G)=(a_{ij})_{p\times p}$, where $a_{ij}=1$ if vertex $i$ is adjacent with vertex $j$, otherwise $a_{ij}=0$, as well as $a_{ii}=0$. Here, we adjust $A(G)$ as $A(G,f)=(f(a_{ij}))_{p\times p}$ for a (an odd-)graceful labelling $f$ of $G$ by defining $f(a_{ij})=|f(i)-f(j)|$ if edge $ij\in E(G)$, otherwise $f(a_{ij})=0$ if edge $ij\not\in E(G)$, and $f(a_{ii})=0$. We refer to $A(G,f)$ as an \emph{adjacency edge-value matrix} of $G$.

Motivated from the matrix expression of graphs, such as \emph{incidence matrices} of graphs, we have the following Topsnut-matrix definition:

\begin{defn}\label{defn:Topsnut-matrix}
(\cite{Sun-Zhang-Zhao-Yao-2017, Yao-Sun-Zhao-Li-Yan-2017, Yao-Zhang-Sun-Mu-Sun-Wang-Wang-Ma-Su-Yang-Yang-Zhang-2018arXiv}) A \emph{Topsnut-matrix} $A_{vev}(G)$ of a $(p,q)$-graph $G$ is defined as
\begin{equation}\label{eqa:a-formula}
\centering
{
\begin{split}
A_{vev}(G)&= \left(
\begin{array}{ccccc}
x_{1} & x_{2} & \cdots & x_{q}\\
e_{1} & e_{2} & \cdots & e_{q}\\
y_{1} & y_{2} & \cdots & y_{q}
\end{array}
\right)_{3\times q}\\
&=(X~W~Y)^{-1}_{3\times q}
\end{split}}
\end{equation}\\
where
\begin{equation}\label{eqa:three-vectors}
{
\begin{split}
&X=(x_1 ~ x_2 ~ \cdots ~x_q), W=(e_1 ~ e_2 ~ \cdots ~e_q)\\
&Y=(y_1 ~ y_2 ~\cdots ~ y_q),
\end{split}}
\end{equation}
where each edge $e_i$ has its own two ends $x_i$ and $y_i$ with $i\in [1,q]$; and $G$ has another \emph{Topsnut-matrix} $A_{vv}(G)$ defined as $A_{vv}(G)=(X ~ Y)^{-1}_{2\times q}$, where $X,Y$ are called \emph{vertex-vectors}, $W$ \emph{edge-vector}.\qqed
\end{defn}

By these two new matrices $A(G,f)=(f(a_{ij}))_{p\times p}$ and $A_{vev}(G)=(X~W~Y)^{-1}_{3\times q}$, we can make more complex TB-paws.

\begin{lem}\label{thm:Topsnut-matrix-vs-graph}
Any matrix $(X~W~Y)^{-1}_{3\times q}$ corresponds a graph $G$ with $V(G)=(XY)^*$ and $E(G)=W^*$, where \begin{equation}\label{eqa:Topsnut-matrix-vs-graph}
{
\begin{split}
&X=(x_1 ~ x_2 ~ \cdots ~x_q), W=(e_1 ~ e_2 ~ \cdots ~e_q)\\
&Y=(y_1 ~ y_2 ~\cdots ~ y_q),
\end{split}}
\end{equation} are three vectors of real numbers, and $(XY)^*$ is the set of different elements in $X$ and $Y$, and $W^*$ is the set of different elements in $W$.
\end{lem}

\begin{lem}\label{thm:set-ordered-matrix}
If $\max \{x_i:i\in [1,q]\}<\min \{y_j:j\in [1,q]\}$ in a matrix $A_{vev}(G)=(X~W~Y)^{-1}_{3\times q}$ of a $(p,q)$-graph $G$ if and only if $G$ has no odd-cycle.
\end{lem}

We name a Topsnut-matrix $A_{vev}(G)=(X~W~Y)^{-1}_{3\times q}$ of a $(p,q)$-graph $G$ with $\max X<\min Y$ as a \emph{set-ordered Topsnut-matrix}. And, $A_{vev}(G)$ is \emph{$X$-increasing} if $x_i\leq x_{i+1}$ with $i\in[1,q-1]$, $A_{vev}(G)$ is \emph{edge-ordered} if $e_j\leq e_{j+1}$ with $j\in[1,q-1]$. For the simplicity of statement, we regard $(XY)^*=X\cup Y$ and $W^*=W$ in $A_{vev}(G)$ hereafter.

\subsection{Operations on Topsnut-matrices \cite{Yao-Zhang-Sun-Mu-Sun-Wang-Wang-Ma-Su-Yang-Yang-Zhang-2018arXiv}}

Suppose that $A_{vev}(G)=(X~W~Y)^{-1}_{3\times q}$ is a Topsnut-matrix of a $(p,q)$-graph $G$. There are some operations on Topsnut-matrices based on Definition \ref{defn:Topsnut-matrix} in the following.

\begin{asparaenum}[Mo-1. ]
\item \textbf{Compound operation.} We define a particular Topsnut-matrix by an edge $e_i=x_iy_i$ of $G$ as $A(e_i)=(x_i~e_i~y_i)^{-1}$, and set a \emph{compound operation} ``$\odot$'' between these Topsnut-matrices $A(e_i)$ with $i\in [1,q]$. Hence, we get
\begin{equation}\label{eqa:edge-topsnut-matrix}
{
\begin{split}
A(e_i)\odot A(e_j)&= \left(
\begin{array}{ccccc}
x_{i}\\
e_{i}\\
y_{i}
\end{array}
\right)\odot
\left(
\begin{array}{ccccc}
x_{j}\\
e_{j}\\
y_{j}
\end{array}
\right)
\\
&= \left(
\begin{array}{ccccc}
x_{i}&x_{j}\\
e_{i}&e_{j}\\
y_{i}&y_{j}
\end{array}
\right)_{3\times 2}\\
&=(x_i~e_i~y_i)^{-1}\odot (x_j~e_j~y_j)^{-1},
\end{split}}
\end{equation}
so we can rewrite the Topsnut-matrix $A_{vev}(G)$ of $G$ in another way
\begin{equation}\label{eqa:dge-topsnut-matrix-0}
{
\begin{split}
A_{vev}(G)=\odot ^q_{i=1}A(e_i).
\end{split}}
\end{equation}

For a group of disjoint Topsnut-gpws $G_1,G_2,\dots ,G_m$, we have a $(p,q)$-graph $G=\bigcup ^m_{i=1}G_i$, where $p=\sum^m_{i=1}|V(G_i)|$ and $q=\sum^m_{i=1}|E(G_i)|$. Thereby we get a Topsnut-matrix of $G$ as
\begin{equation}\label{eqa:G}
A_{vev}(G)=\odot ^m_{i=1}A(G_{j_i})
\end{equation}
where $G_{j_1}G_{j_2}\dots G_{j_m}$ is a permutation of $G_{1}G_{2}\dots G_{m}$.

\item \textbf{Joining operation.} We have a vev-type TB-paw obtained by a \emph{joining operation ``$\uplus$''} as follows:
\begin{equation}\label{eqa:TB-paw-by-dges}
{
\begin{split}
D_{vev}(G)&=\uplus^q_{i=1}D_{vev}(e_i)=\uplus^q_{j=1}x_{i_j1}e_{i_j}y_{i_j}\\
&=x_{i_1}e_{i_1}y_{i_1}x_{i_2}e_{i_2}y_{i_2}\cdots x_{i_q}e_{i_q}y_{i_q}.
\end{split}}
\end{equation}
where $x_{i_1}e_{i_1}y_{i_1}x_{i_2}e_{i_2}y_{i_2}\cdots x_{i_q}e_{i_q}y_{i_q}$ is a permutation of $x_{1}e_{1}y_{1}x_{2}e_{2}y_{2}\cdots x_{q}e_{q}y_{q}$.

\item \textbf{Column-exchanging operation.} We exchange the positions of two columns $(x_i~e_i~y_i)^{-1}$ and $(x_j~e_j~y_j)^{-1}$ in $A_{vev}(G)$, so we get another Topsnut-matrix $A'_{vev}(G)$. In mathematical symbol, the \emph{column-exchanging operation} $c_{(i,j)}(A_{vev}(G))=A'_{vev}(G)$ is defined by
$${
\begin{split}
&\quad c_{(i,j)}(x_1 ~ x_2 ~ \cdots ~\textcolor[rgb]{0.00,0.00,1.00}{x_i}~ \cdots ~\textcolor[rgb]{0.00,0.00,1.00}{x_j}~ \cdots ~x_q)\\
&=(x_1 ~ x_2 ~ \cdots ~\textcolor[rgb]{0.00,0.00,1.00}{x_j}~ \cdots ~\textcolor[rgb]{0.00,0.00,1.00}{x_i}~ \cdots ~x_q),
\end{split}}$$
$$
{
\begin{split}
&\quad c_{(i,j)}(e_1 ~ e_2 ~ \cdots ~\textcolor[rgb]{0.00,0.00,1.00}{e_i}~ \cdots ~\textcolor[rgb]{0.00,0.00,1.00}{e_j}~ \cdots ~e_q)\\
&=(e_1 ~ e_2 ~ \cdots ~\textcolor[rgb]{0.00,0.00,1.00}{e_j}~ \cdots ~\textcolor[rgb]{0.00,0.00,1.00}{e_i}~ \cdots ~e_q),
\end{split}}
$$

and
$$
{
\begin{split}
&\quad c_{(i,j)}(y_1 ~ y_2 ~ \cdots ~\textcolor[rgb]{0.00,0.00,1.00}{y_i}~ \cdots ~\textcolor[rgb]{0.00,0.00,1.00}{y_j}~ \cdots ~y_q)\\
&=(y_1 ~ y_2 ~ \cdots ~\textcolor[rgb]{0.00,0.00,1.00}{y_j}~ \cdots ~\textcolor[rgb]{0.00,0.00,1.00}{y_i}~ \cdots ~y_q).
\end{split}}
$$
\item \textbf{XY-exchanging operation.} We exchange the positions of $x_i$ and $y_i$ of the $i$th column of $A_{vev}(G)$ by an \emph{XY-exchanging operation} $l_{(i)}$ defined as:
$${
\begin{split}
&\quad l_{(i)}(x_1 ~ x_2 ~ \cdots x_{i-1}~\textcolor[rgb]{0.00,0.00,1.00}{x_i}~x_{i+1} \cdots ~x_q)\\
&=(x_1 ~ x_2 ~ \cdots x_{i-1}~\textcolor[rgb]{0.00,0.00,1.00}{y_i}~x_{i+1} \cdots ~x_q)
\end{split}}
$$
and
$${
\begin{split}
&\quad l_{(i)}(y_1 ~ y_2 ~ \cdots y_{i-1}~\textcolor[rgb]{0.00,0.00,1.00}{y_i}~y_{i+1} \cdots ~y_q)\\
&=(y_1 ~ y_2 ~ \cdots y_{i-1}~\textcolor[rgb]{0.00,0.00,1.00}{x_i}~y_{i+1} \cdots ~y_q),
\end{split}}
$$
the resultant matrix is denoted as $l_{(i)}(A_{vev}(G))$.
\end{asparaenum}

\begin{figure}[h]
\centering
\includegraphics[height=3cm]{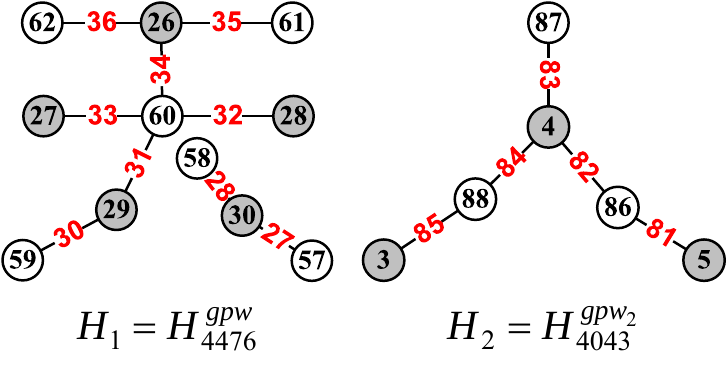}\\
\caption{\label{fig:TB-haoren} {\small A Hanzi-gpw $H_1\cup H_2$ made by two Hanzi-gpws $H^{gpw}_{4476}$ and $H^{gpw_2}_{4043}$ shown in Fig.\ref{fig:rrhg-set-graceful-0}.}}
\end{figure}

{\small
\begin{equation}\label{eqa:tian-ren-11}
\centering
A(H_{1})= \left(
\begin{array}{cccccccccc}
57 & 58 & 29 & 29& 28& 60& 60& 61& 62\\
27 & 28 & 30 & 31& 32& 33& 34& 35& 36\\
30 & 30 & 59 & 60& 60& 27& 26& 26& 26
\end{array}
\right)
\end{equation}

\begin{equation}\label{eqa:tian-ren-22}
A(H_2)= \left(
\begin{array}{cccccc}
5 & 4 & 87 & 88& 88\\
81 & 82 & 83 & 84& 85\\
86 & 86 & 4 & 4& 3
\end{array}
\right)
\end{equation}
}

Two matrices (\ref{eqa:tian-ren-11}) and (\ref{eqa:tian-ren-22}) are two Topsnut-matrices of two Hanzi-gpws $H_1=H^{gpw}_{4476}$ and $H_2=H^{gpw}_{4043}$ shown in Fig.\ref{fig:TB-haoren}, respectively. Now, we do a column-exchanging operation $c_{(1,4)}$ on the Topsnut-matrix $A(H_2)$ in (\ref{eqa:tian-ren-22}) as in the following (\ref{eqa:tian-ren-33}):

{\small
\begin{equation}\label{eqa:tian-ren-33}
c_{(1,4)}(A(H_2))= \left(
\begin{array}{cccccc}
\textbf{\textcolor[rgb]{0.00,0.00,1.00}{88}} & 4 & 87 & \textbf{\textcolor[rgb]{0.00,0.00,1.00}{5}}& 88\\
\textbf{\textcolor[rgb]{0.00,0.00,1.00}{84}} & 82 & 83 & \textbf{\textcolor[rgb]{0.00,0.00,1.00}{81}}& 85\\
\textbf{\textcolor[rgb]{0.00,0.00,1.00}{4}} & 86 & 4 & \textbf{\textcolor[rgb]{0.00,0.00,1.00}{86}}& 3
\end{array}
\right)
\end{equation}\\
}
And, we do an XY-exchanging operation $l_{(3)}$ on the Topsnut-matrix $A(H_2)$ in (\ref{eqa:tian-ren-22}) for getting the following Topsnut-matrix (\ref{eqa:tian-ren-44}):

{\small
\begin{equation}\label{eqa:tian-ren-44}
l_{(3)}(A(H_2))= \left(
\begin{array}{cccccc}
5 & 4 & \textbf{\textcolor[rgb]{0.00,0.00,1.00}{4}} & 88& 88\\
81 & 82 & 83 & 84& 85\\
86 & 86 & \textbf{\textcolor[rgb]{0.00,0.00,1.00}{87}} & 4& 3
\end{array}
\right)
\end{equation}
}

A result of a mixed operation of column-exchanging and XY-exchanging operations on the Topsnut-matrix $A(H_2)$ in (\ref{eqa:tian-ren-22}) is as in (\ref{eqa:tian-ren-55}):

{\small
\begin{equation}\label{eqa:tian-ren-55}
{
\begin{split}
l_{(3)}c_{(1,4)}(A(H_2))&=l_{(3)}\left(\begin{array}{cccccc}
\textbf{\textcolor[rgb]{0.00,0.00,1.00}{88}} & 4 & 87 & \textbf{\textcolor[rgb]{0.00,0.00,1.00}{5}}& 88\\
\textbf{\textcolor[rgb]{0.00,0.00,1.00}{84}} & 82 & 83 & \textbf{\textcolor[rgb]{0.00,0.00,1.00}{81}}& 85\\
\textbf{\textcolor[rgb]{0.00,0.00,1.00}{4}} & 86 & 4 & \textbf{\textcolor[rgb]{0.00,0.00,1.00}{86}}& 3
\end{array}\right)\\
&=\left(\begin{array}{cccccc}
\textbf{\textcolor[rgb]{0.00,0.00,1.00}{88}} & 4 & \textbf{\textcolor[rgb]{0.00,0.00,1.00}{4}} & \textbf{\textcolor[rgb]{0.00,0.00,1.00}{5}}& 88\\
\textbf{\textcolor[rgb]{0.00,0.00,1.00}{84}} & 82 & 83 & \textbf{\textcolor[rgb]{0.00,0.00,1.00}{81}}& 85\\
\textbf{\textcolor[rgb]{0.00,0.00,1.00}{4}} & 86 & \textbf{\textcolor[rgb]{0.00,0.00,1.00}{87}} & \textbf{\textcolor[rgb]{0.00,0.00,1.00}{86}}& 3
\end{array}
\right)
\end{split}}
\end{equation}
}

{\small
\begin{equation}\label{eqa:tian-ren-together}
{
\begin{split}
&\quad A(H_1\odot H_2)=A(H_1)\odot A(H_2)\\
&= \left(
\begin{array}{l}
57 ~ 58 ~ 29 ~ 29~ 28~ 60~ 60~ 61~ 62~5\quad 4 ~~ 87 ~ 88~ 88\\
27 ~ 28 ~ 30 ~ 31~ 32~ 33~ 34~ 35~ 36~81 ~ 82 ~ 83 ~ 84~ 85\\
30 ~ 30 ~ 59 ~ 60~ 60~ 27~ 26~ 26~ 26~86 ~ 86 ~ 4\quad 4~ ~3
\end{array}
\right)
\end{split}}
\end{equation}
}
Clearly,
$${
\begin{split}
A(H_1\odot H_2)&=A(H_1)\odot A(H_2)\\
&\neq A(H_2)\odot A(H_1)=A(H_2\odot H_1).
\end{split}}$$

Now, we do a series of column-exchanging operations $c_{(i_k,j_k)}$ with $k\in [1,m]$, and a series of XY-exchanging operations $l_{(i_s)}$ with $s\in [1,n]$ to $A_{vev}(G)$, the resultant matrix is written by $T_{(c,l)}(A_{vev}(G))$.
\begin{lem}\label{thm:2-trees-matrices-isomorphic}
Suppose $T$ and $H$ are trees of $q$ edges. If $T_{(c,l)}(A_{vev}(T))=A_{vev}(H)$, then these two trees may be isomorphic to each other, or may not be isomorphic to each other.
\end{lem}

Notice that two Topsnut-gpws  $G$ and $Q$  are labelled graphs. If $T_{(c,l)}(A_{vev}(G))=A_{vev}(Q)$ can induce $G\cong Q$, but it  is not a solution of the \emph{Graph Isomorphic Problem} in graph theory. The column-exchanging operation and the XY-exchanging operation tell us that a $(p,q)$-graph $G$ may have many Topsnut-matrices according to the labellings admitted by $G$. Moreover, we have the number of TB-paws obtained from a Topsnut-matrix of a Topsnut-gpw below.

\begin{thm}\label{thm:number-tb-paws-a-labelling}
Suppose that a $(p,q)$-graph $G$ admits a labelling $f:V(G)\rightarrow [a,b]$ and $e_i=f(u_iv_i)=F(f(u_i),f(v_i))=(x_i,y_i)$, then there are $2q+q!$ Topsnut-matrices $A^s_{vev}(G)$ with $s\in [1,2q+q!]$, and a Topsnut-matrix $A^s_{vev}(G)$ can produce $(3q)!$ TB-paws, in which each TB-paw is a permutation of $x_{1}e_{1}y_{1}x_{2}e_{2}y_{2}\cdots x_{q}e_{q}y_{q}$.
\end{thm}

For proving Theorem \ref{thm:number-tb-paws-a-labelling}, we can see that there are $q!$ edge-permutations of $e_{1}e_{2}\cdots e_{q}$, and for a fixed edge-permutation there are $2q$ arrangements of $X_{i}$ and $Y_i$, so we have $2q+q!$ Topsnut-matrices of the $(p,q)$-graph $G$ with a labelling $f$ stated in Theorem \ref{thm:number-tb-paws-a-labelling}.

\subsection{Basic ways for producing TB-paws from Topsnut-matrices}

If a string $T_b(G)=a_1a_2\dots a_{3q}$ with $a_i\in X\cup W\cup Y$, such that $a_1a_2\dots a_{3q}$ is a permutation of $x_{1}e_{1}y_{1}x_{2}e_{2}y_{2}\cdots x_{q}e_{q}y_{q}$, we say $T_b(G)$ a \emph{text-based password} (TB-paw) from a Hanzi-matrix $A_{vev}(G)=(X~W~Y)^{-1}$ of a $(p,q)$-graph $G$ (see Definition \ref{defn:Topsnut-matrix}). The string $T^{-1}_b(G)=a_{3q}a_{3q-1}\dots a_{2}a_{1}$ is called the \emph{inverse} of the string $T_b(G)$.

We have two types of TB-paws from a Hanzi-matrix $A_{vev}(G)$:

\textbf{Type-1.} If any $a_ia_{i+1}\in T_b(G)$ holds one of the following cases shown in (\ref{eqa:k-lines-TB-paw})
\begin{equation}\label{eqa:k-lines-TB-paw}
{
\begin{split}
&a_ia_{i+1}=x_ix_{i+1}, a_ia_{i+1}=e_ie_{i+1}, a_ia_{i+1}=y_iy_{i+1},\\
&a_ia_{i+1}=x_ie_{i+1}, a_ia_{i+1}=e_ix_{i+1},\\
&a_ia_{i+1}=e_{i+1}x_i, a_ia_{i+1}=x_{i+1}e_i,\\
&a_ia_{i+1}=x_ie_{i}, a_ia_{i+1}=e_ix_{i},\\
&a_ia_{i+1}=x_iy_{i+1}, a_ia_{i+1}=y_ix_{i+1},\\
&a_ia_{i+1}=y_ie_{i}, a_ia_{i+1}=e_iy_{i},\\
&a_ia_{i+1}=e_iy_{i+1}, a_ia_{i+1}=y_ie_{i+1}\\
&a_ia_{i+1}=y_{i+1}e_i, a_ia_{i+1}=e_{i+1}y_i
\end{split}}
\end{equation}
then we say this $T_b(G)$ a \emph{$1$-line TB-paw}.

\textbf{Type-2.} If $T_b(G)$ can be cut into $k$ segments, namely $T_b(G)=\uplus^k_{j=1}T^j_b$ with $T^1_b=a_1a_2\dots a_{t_1}$, $T^2_b=a_{t_1+1}a_{t_1+2}\dots a_{t_1+t_2}$, and for $j\in [2,k]$ we have
\begin{equation}\label{eqa:c3xxxxx}
T^j_b=a_{t_1+\cdots +t_{j-1}+1}a_{t_1+\cdots +t_{j-1}+2}\dots a_{t_1+\cdots +t_{j-1}+t_j}
\end{equation}
where $3q=\sum ^k_{j=1}t_j$ with $t_j\geq 2$ and $k\geq 2$, and any $a_ia_{i+1}\in T^j_b$ holds one of the cases in (\ref{eqa:k-lines-TB-paw}), then we say $T_b(G)$ to be a \emph{$k$-line TB-paw}.

Now, we show several examples for illustrating Type-1 and Type-2 introduced above.

\textbf{Example 1.} We can write a $1$-line TB-paw $T^{(1)}_b(H_1\odot H_2)$ by the Topsnut-matrix $A(H_1\odot H_2)=A(H_1)\odot A(H_2)$ in the form (\ref{eqa:tian-ren-together}) in the following
$${
\begin{split}
T^{(1)}_b(H_1\odot H_2)=&57582929286060616254878888\\
&8584838281363534333231302827\\
&3030596060272626268686443
\end{split}}
$$
Furthermore, a $3$-line TB-paw $T^{(3)}_b(H_1\cup H_2)$ is based on the Topsnut-matrix $A(H_1\cup H_2)=A(H_1)\odot A(H_2)$ in the form (\ref{eqa:tian-ren-together}) as follows
$${
\begin{split}
T^{(3)}_b(H_1\cup H_2)=&57582929286060616254878888\\
&2728303132333435368182838485\\
&3030596060272626268686443
\end{split}}
$$ such that
$$T^{(3)}_b(H_1\cup H_2)=D_1\uplus D_2\uplus D_3$$
with

$D_1=57582929286060616254878888$,

$D_2=2728303132333435368182838485$ and

$D_3=3030596060272626268686443$.

\begin{figure}[h]
\centering
\includegraphics[height=12cm]{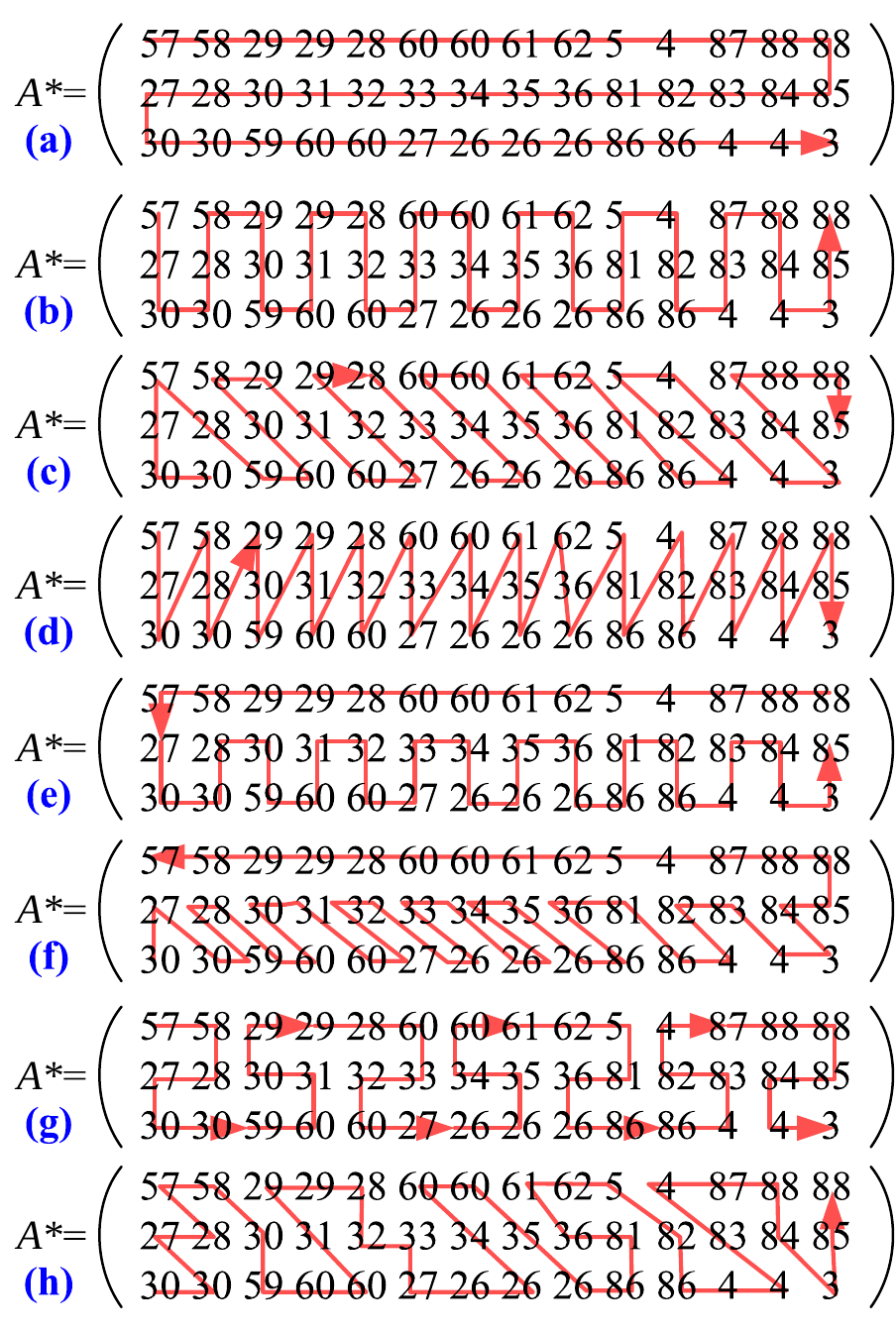}\\
\caption{\label{fig:TB-from-basic-11} {\small The lines in (a), (b), (c) and (d) are the basic rules for producing TB-paws; others are examples for showing there exist many $1$-line TB-paws.}}
\end{figure}

\textbf{Example 2.} By Fig.\ref{fig:TB-from-basic-11}, we have the following TB-paws:
$${
\begin{split}
T^{(a)}_b=&575829292860606162548788882\\
&72830313233343536818283848\\
&53030596060272626268686443
\end{split}}$$

$${
\begin{split}
T^{(b)}_b=&572730302858293059603129283\\
&2602733606034262635616236\\
&268681548286483878884438588
\end{split}}$$

$${
\begin{split}
T^{(c)}_b=&303027572859603058293160273\\
&229283326263460603526863861\\
&6281864825483438487888885
\end{split}}$$

$${
\begin{split}
T^{(d)}_b=&57273058283029305929316028\\
&32606033276034266135266236265\\
&818648386878348884488853
\end{split}}$$

$${
\begin{split}
T^{(e)}_b=&888887456261606028292958572\\
&73030283059603132602733342626\\
&35362686818286483844385
\end{split}}$$

$${
\begin{split}
T^{(f)}_b=&3027305928306060313227263334\\
&26263536863681864828343848588\\
&888745626160602829295857
\end{split}}$$

$${
\begin{split}
T^{(g)}_b=&5758282730305960313029292860\\
&33326027262635346061625\\
&8136268686483824878888858443
\end{split}}$$

$${
\begin{split}
T^{(h)}_b=&30302728575830596060312929283\\
&2332726263460603526868136\\
&6162582864483487888438588
\end{split}}$$

Clearly, $T^{(i)}_b\neq T^{(j)}_b$ for $i,j\in \{a,b,c,d,e,f,g,h\}$ and $i\neq j$. In Fig.\ref{fig:TB-from-basic-4-curves}, we give each basic $1$-line O-$k$ and its reciprocal $1$-line O-$k$-r and inverse $1$-line O-$k$-i with $k\in [1,4]$. Moreover, there are the following efficient algorithms for writing TB-paws from $1$-line O-$k$-r and inverse $1$-line O-$k$-i ($k\in [1,4]$) introduced above.

\begin{figure}[h]
\centering
\includegraphics[height=6cm]{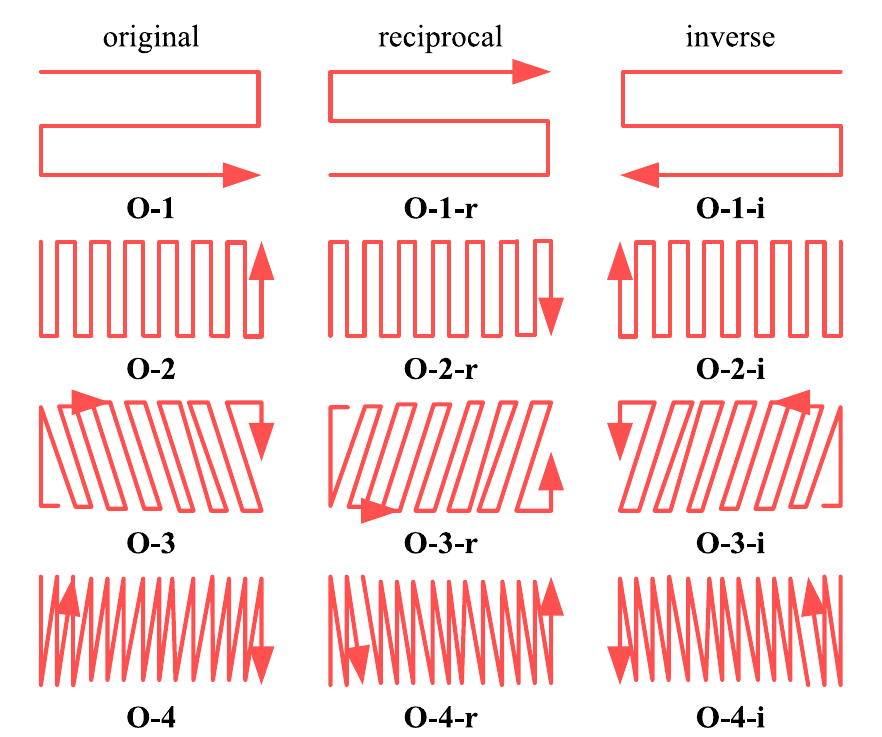}\\
\caption{\label{fig:TB-from-basic-4-curves} {\small Four basic $1$-line O-$k$ with $k\in [1,4]$ used in general matrices, such as  Fig.\ref{fig:TB-from-basic-11} (a), (b), (c) and (d).}}
\end{figure}

\vskip 0.2cm

\textbf{ALGORITHM-I ($1$-line-O-$1$)}

\textbf{Input:} $A_{vev}(G)=(X~W~Y)^{-1}_{3\times q}$

\textbf{Output:} TB-paws: $1$-line O-$1$ TB-paw $T^{O-1}_b$; $1$-line O-$1$-$r$ TB-paw $T^{O-1-r}_b$; and $1$-line O-$1$-$i$ TB-paw $T^{O-1-i}_b$ as follows:
$${
\begin{split}
T^{O-1}_b=&x_1x_2\cdots x_qe_qe_{q-1}\cdots e_2e_1y_1 y_2\cdots y_q\\
T^{O-1-r}_b=&y_1 y_2\cdots y_qe_qe_{q-1}\cdots e_2e_1x_1x_2\cdots x_q\\
T^{O-1-i}_b=&x_qx_{q-1}\cdots x_2x_1 e_1 e_2\cdots e_qy_qy_{q-1}\cdots y_2y_1
\end{split}}$$

\vskip 0.2cm

\textbf{ALGORITHM-II ($1$-line-O-$2$)}

\textbf{Input:} $A_{vev}(G)=(X~W~Y)^{-1}_{3\times q}$

\textbf{Output:} TB-paws: $1$-line O-$2$ TB-paw $T^{O-2}_b$; $1$-line O-$2$-$r$ TB-paw $T^{O-2-r}_b$; and $1$-line O-$2$-$i$ TB-paw $T^{O-2-i}_b$ as follows:
$${
\begin{split}
T^{O-2}_b=&x_1e_1y_1y_2e_2x_2x_3e_3y_3y_4\dots x_{q-1}x_qe_qy_q~(\textrm{odd }q)\\
T^{O-2}_b=&x_1e_1y_1y_2e_2x_2x_3e_3y_3y_4\dots y_{q-1}y_qe_qx_q~(\textrm{even }q)\\
T^{O-2-r}_b=&y_1e_1x_1x_2e_2y_2y_3\dots y_{q-1}y_qe_qx_q~(\textrm{odd }q)\\
T^{O-2-r}_b=&y_1e_1x_1x_2e_2y_2y_3\dots x_{q-1}x_qe_qy_q~(\textrm{even }q)\\
T^{O-2-i}_b=&x_qe_qy_qy_{q-1}e_{q-1}x_{q-1}x_{q-2}\dots x_1e_1y_1~(\textrm{odd }q)\\
T^{O-2-i}_b=&x_qe_qy_qy_{q-1}e_{q-1}x_{q-1}x_{q-2}\dots y_1e_1x_1~(\textrm{even }q)
\end{split}}$$

\vskip 0.2cm

\textbf{ALGORITHM-III ($1$-line-O-$3$)}

\textbf{Input:} $A_{vev}(G)=(X~W~Y)^{-1}_{3\times q}$

\textbf{Output:} TB-paws: $1$-line O-$3$ TB-paw $T^{O-3}_b$; $1$-line O-$3$-$r$ TB-paw $T^{O-3-r}_b$; and $1$-line O-$3$-$i$ TB-paw $T^{O-3-i}_b$ as follows:
$${
\begin{split}
T^{O-3}_b=&y_2y_1e_1x_1e_2y_3y_4e_3x_2x_3e_4y_5y_6\dots \\
&y_qe_{q-1}x_{q-2}x_{q-1}x_qe_q\\
T^{O-3-r}_b=&x_2x_1e_1y_1e_2x_3x_4e_3y_2y_3e_4x_5x_6\dots \\
&x_qe_{q-1}y_{q-2}y_{q-1}y_qe_q\\
T^{O-3-i}_b=&y_{q-1}y_qe_qx_qe_{q-1}y_{q-2}y_{q-3}e_{q-2}x_{q-1}x_{q-2}\dots \\
&y_1e_2x_2x_1e_1
\end{split}}$$

\vskip 0.2cm

\textbf{ALGORITHM-IV ($1$-line-O-$4$)}

\textbf{Input:} $A_{vev}(G)=(X~W~Y)^{-1}_{3\times q}$

\textbf{Output:} TB-paws: $1$-line O-$4$ TB-paw $T^{O-4}_b$; $1$-line O-$4$-$r$ TB-paw $T^{O-4-r}_b$; and $1$-line O-$4$-$i$ TB-paw $T^{O-4-i}_b$ as follows:
$${
\begin{split}
T^{O-4}_b=&x_1e_1y_1x_2e_2y_2\dots y_{q-1}x_{q}e_qy_q\\
T^{O-4-r}_b=&y_1e_1x_1y_2e_2x_2\dots x_{q-1}y_{q}e_qx_q \\
T^{O-4-i}_b=&x_q e_q y_{q}x_{q-1}e_{q-1}y_{q-1}x_{q-2}\dots x_2e_2y_2x_1e_1 y_1
\end{split}}$$

\subsection{Basic ways for producing TB-paws from adjacency e-value and ve-value matrices}

Let $f:V(G)=\{1,2,\dots, p\}\rightarrow [a,b]$ be a labelling of a $(p,q)$-graph $G$, and the induced edge labelling $f(ij)=F(f(i),f(j))$ for each edge $ij\in E(G)$.

Each \emph{adjacency e-value matrix} $A(G,f)=(f(a_{ij}))_{p\times p}$ is a symmetric matrix along its main diagonal, here, $f(a_{ij})=f(ij)$ for each edge $ij\in E(G)$. An \emph{adjacency ve-value matrix} $V_A(G,f)=(f(b_{ij}))_{(p+1)\times (p+1)}$ is defined as: $f(b_{ij})=f(a_{ij})\in A(G,f)$ if $i\neq 1$ and $j\neq 1$, and $f(b_{1j})=f(j)$ for $j\geq 2$ and $f(b_{i1})=f(i)$ for $i\geq 2$ (see two examples $A^{(1)}(H_{4043})$ and $A^{(2)}(H_{4043})$ shown in Fig.\ref{fig:ren-ve-value-matrix}).

\begin{figure}[h]
\centering
\includegraphics[height=6.4cm]{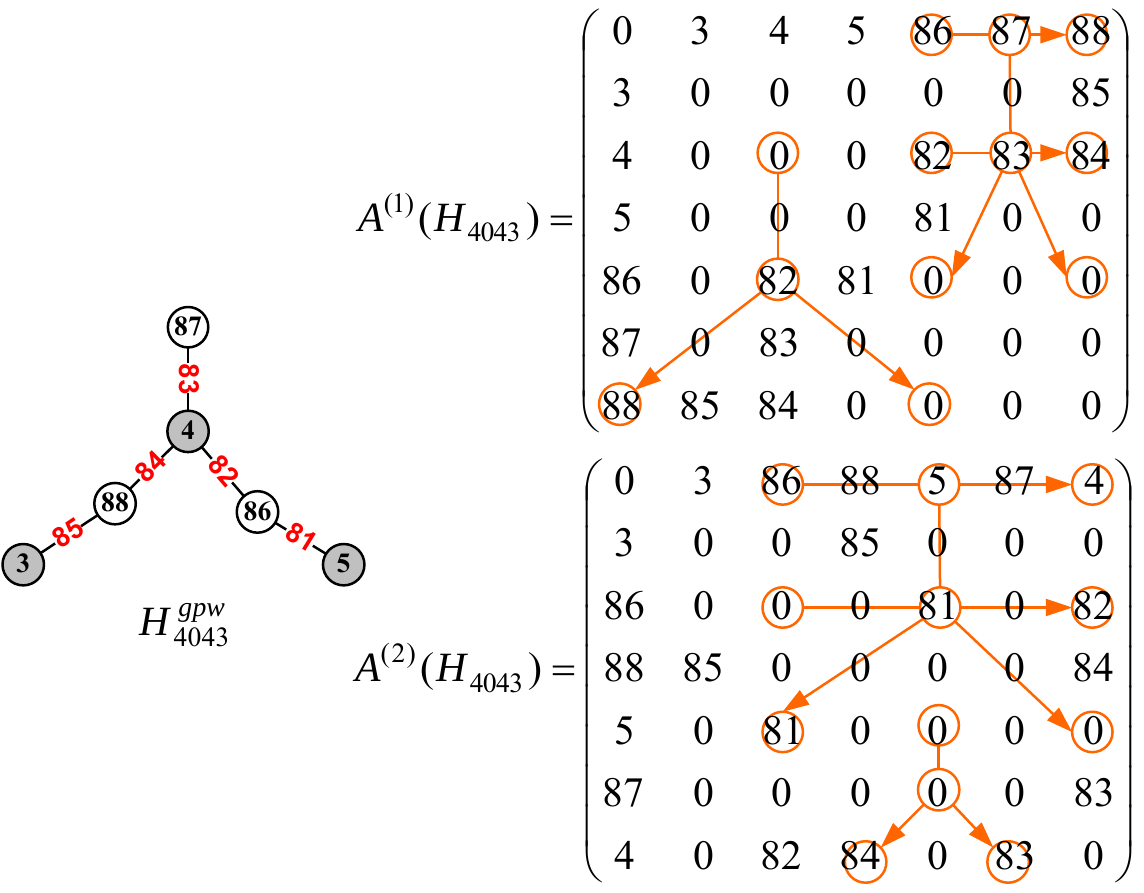}\\
\caption{\label{fig:ren-ve-value-matrix} {\small A Hanzi $H_{4043}$ with its Hanzi-gpw $H^{gpw}_{4043}$ and two adjacency ve-value matrices $A^{(1)}(H_{4043})$ and $A^{(2)}(H_{4043})$.}}
\end{figure}

Two matrices $A^{(1)}(H_{4043})$ and $A^{(2)}(H_{4043})$ shown in Fig.\ref{fig:ren-ve-value-matrix} are two \emph{ve-value matrices}. We present four standard $1$-line Vo-$t$ with $t\in [1,4]$ in Fig.\ref{fig:TB-from-basic-4-curves-00} for making TB-paws from adjacency e-value and ve-value matrices.

\begin{figure}[h]
\centering
\includegraphics[height=8cm]{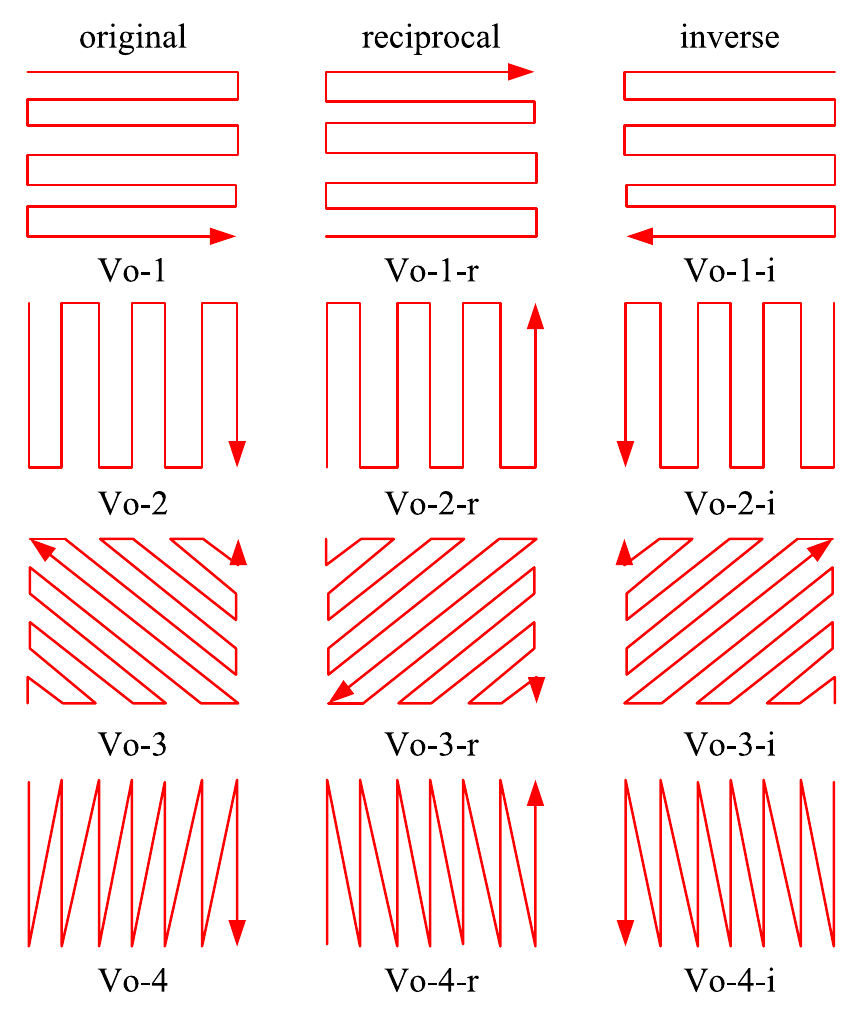}\\
\caption{\label{fig:TB-from-basic-4-curves-00} {\small Four standard $1$-line Vo-$t$ with $t\in [1,4]$ used in general matrices, such as the adjacency e-value matrices and the adjacency ve-value matrices.}}
\end{figure}

It is easy to write efficient algorithms for $1$-line Vo-$t$ with $t\in [1,4]$, we omit them here. Thereby, we show examples for using these four standard $1$-line Vo-$t$ with $t\in [1,4]$. Thereby, we have
$${
\begin{split}
T^{\textrm{Vo-}1}_{b4043}=&03458687888500000340008\\
&2838400810005860828100\\
&00000830878885840000=T^{\textrm{Vo-}2}_{b4043}
\end{split}}$$

$${
\begin{split}
T^{\textrm{Vo-}3}_{b4043}=&888785840865083000820430\\
&081000000000300810000820\\
&45083084086878588
\end{split}}$$

$${
\begin{split}
T^{\textrm{Vo-}4}_{b4043}=&0345868788300000854000828\\
&3845000810086082810008708\\
&300008885840000
\end{split}}$$

\begin{figure}[h]
\centering
\includegraphics[height=7cm]{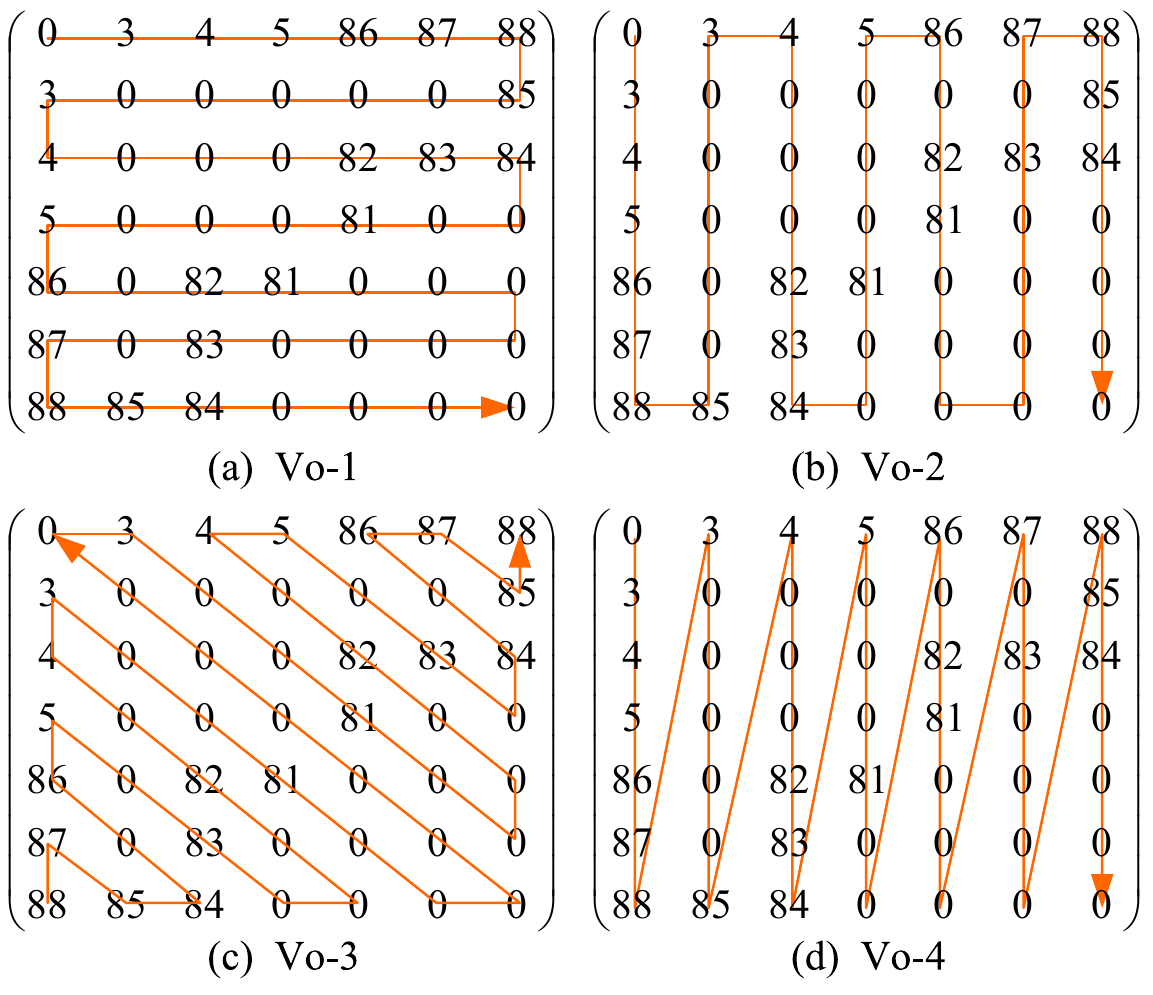}\\
\caption{\label{fig:TB-from-basic-4-curves-00-lizi} {\small Examples for using four standard $1$-line Vo-$t$ with $t\in [1,4]$ shown in Fig.\ref{fig:TB-from-basic-4-curves-00}.}}
\end{figure}

The authors in \cite{Mu-Yao-ITNEC2019} introduce a way of producing TB-paws by combination of Hanzi-graphs and their various matrices. There are two Hanzi-graphs $T_{4043}$ and $T_{4476}$ under the adjacent ve-value matrix $A^{(1)}(H_{4043})$ shown in Fig.\ref{fig:ren-ve-value-matrix}, we can write two TB-paws

$D_{4476}=86878882838487830830$\\
and $D_{4043}=00820888200$. Another group of two Hanzi-graphs $T_{4043}$ and $T_{4476}$ under the adjacent ve-value matrix $A^{(2)}(H_{4043})$ shown in Fig.\ref{fig:ren-ve-value-matrix} distributes us two TB-paws

$D'_{4476}=86885874008108250810818100$\\
and $D'_{4043}=0084083$.

The space of TB-paws made by the adjacent e-value matrices and the adjacent ve-value matrices is larger than that of Topsnut-matrices. Suppose that $A=(a_{ij})_{n\times n}$ is a popular matrix, so it has $n^2$ elements $a_{ij}$ with $i,j\in [1,n]$. We can obtain $(n^2)!$ TB-paws with $n^2$ bytes from $A=(a_{ij})_{n\times n}$. However, a $(p,q)$-graph $G$ has many its own adjacent e-value and adjacent ve-value matrices, in which two adjacent e-value matrices (or adjacent ve-value matrices) $A^{(i)},A^{(j)}$ are similar to each other $A^{(i)}\sim A^{(j)}$, that is, $A^{(i)}=B A^{(j)}B^{-1}$ by a non-singular matrix $B$ in linear algebra. See an example shown in Fig.\ref{fig:ren-ve-value-matrix}, where $A^{(1)}(H_{4043})$ is similar with $A^{(2)}(H_{4043})$, so there exists a non-singular matrix $P$ with its inverse $P^{-1}$ holding $A^{(1)}(H_{4043})=PA^{(2)}(H_{4043})P^{-1}$.

\subsection{Writing stroke orders of Hanzis in Hanzi-gpws}
We can abide the writing stroke order of a Hanzi to yield TB-paws. By Fig.\ref{fig:TB-haoren} and Fig.\ref{fig:directed-tian-00}, we have the following TB-paws
$${
\begin{split}
T_b(H^{gpw}_{4476})=&62362635612733603228263460\\
&312930595828302757
\end{split}}$$
and
$$T_b(H^{gpw}_{4043})=87834848885348286815.$$
Clearly, this way is natural for making TB-paws from Hanzi-gpws.

\subsection{Compound TB-paws from Hanzi-keys vs Hanzi-keys}

Compound TB-paws are similar with compound functions of calculus. Hanzi-couplets of Hanzi-keys vs Hanzi-keys can provide complex Hanzi-gpws made by Hanzi-graphs and various graph labellings. Here, we only show several Hanzi-couplets of Hanzi-keys vs Hanzi-keys in Fig.\ref{fig:group-keys-in-keys}, and omit detail process for producing Hanzi-gpws.

\begin{figure}[h]
\centering
\includegraphics[width=8.2cm]{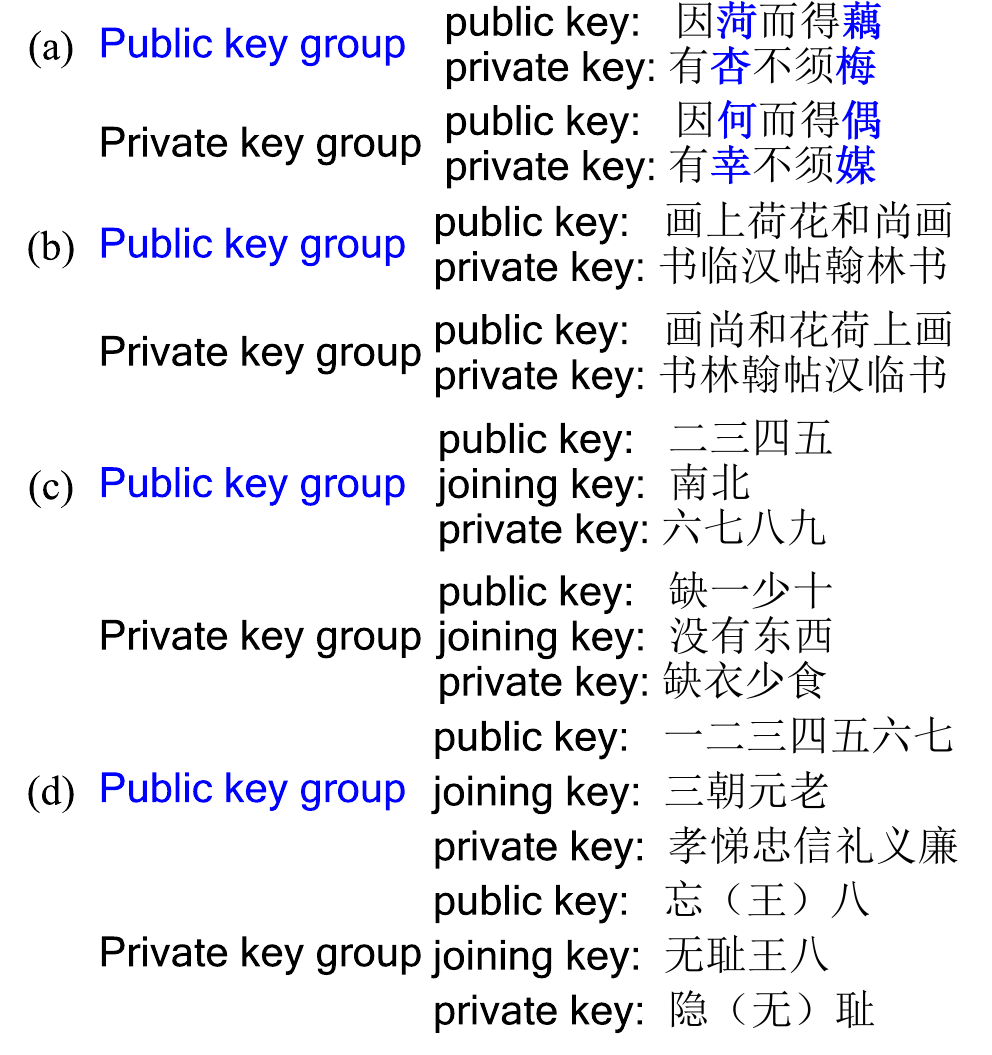}\\
\caption{\label{fig:group-keys-in-keys} {\small Four groups of keys in keys.}}
\end{figure}

\subsection{TB-paws from matrices with elements of Hanzi-GB2312-80 or Chinese code}

\begin{defn}\label{defn:Hanzi-GB2312-80-matrix}
$^*$ A \emph{Hanzi-GB2312-80 matrix} $A_{han}(H)$ of a Hanzi-sentence $H=\langle H_i \rangle ^m_{i=1}$ made by $m$ Hanzis $H_1$, $H_2$, $\dots$, $H_m$ is defined as
\begin{equation}\label{eqa:a-formula}
\centering
{
\begin{split}
A_{han}(H)&= \left(
\begin{array}{ccccc}
a_{1} & a_{2} & \cdots & a_{m}\\
b_{1} & b_{2} & \cdots & b_{m}\\
c_{1} & c_{2} & \cdots & c_{m}\\
d_{1} & d_{2} & \cdots & d_{m}
\end{array}
\right)_{4\times m}\\
&=(A~B~C~D)^{-1}_{4\times m}
\end{split}}
\end{equation}\\
where
\begin{equation}\label{eqa:three-vectors}
{
\begin{split}
&A=(a_1 ~ a_2 ~ \cdots ~a_m), B=(b_1 ~ b_2 ~ \cdots ~b_m)\\
&C=(c_1 ~ c_2 ~ \cdots ~c_m), D=(d_1 ~ d_2 ~ \cdots ~d_m)
\end{split}}
\end{equation}
where each Hanzi $H_i$ has its own Hanzi-code $a_ib_ic_id_i$ defined in ``GB2312-80 Encoding of Chinese characters'' in  \cite{GB2312-80}.\qqed
\end{defn}

It is easy to define a Hanzi-code matrix with elements defined by Chinese code defined in \cite{GB2312-80}, we omit it here. Since there are efficient algorithms for four standard $1$-line Vo-$t$ with $t\in [1,4]$ in Fig.\ref{fig:TB-from-basic-4-curves-00} for making TB-paws from adjacency e-value and ve-value matrices. We, as an example, have a Hanzi-GB2312-80 matrix

\begin{equation}\label{eqa:Hanzi-sentence-GB2312-80}
\centering
{
\begin{split}
A_{han}(G^*)&= \left(
\begin{array}{ccccccccc}
4 & 4 & 2 & 2 & 5 & 4 & 4 & 4 & 3\\
0 & 0 & 6 & 5 & 2 & 4 & 7 & 4 & 8\\
4 & 4 & 3 & 1 & 8 & 7 & 3 & 1 & 2\\
3 & 3 & 5 & 1 & 2 & 6 & 4 & 1 & 9
\end{array}
\right)
\end{split}}
\end{equation}
according to a Hanzi-sentence $G^*=H_{4043}$ $H_{4043}$ $H_{2635}$ $H_{2511}$ $H_{5282}$ $H_{4476}$ $H_{4734}$ $H_{4411}H_{3829}$ shown in Fig.\ref{fig:rrhg-GB2312-80} (a). Moreover, the matrix $A_{han}(G^*)$ distributes us some TB-paws as follows:
$$D_1(G^*)=442254443847425600443187312914621533$$
$$D_2(G^*)=404334042635115252826744473411443829$$
$$D_3(G^*)=343540403121642586472254319174442843$$
and
$$D_4(G^*)=404340432635251152824476473444113829$$
by four standard $1$-line Vo-$t$ with $t\in [1,4]$. This Hanzi-sentence $G^*$ induces another matrix as follows:
{\small
\begin{equation}\label{eqa:Hanzi-sentence-code}
\centering
{
\begin{split}
A_{han}(G^*)&= \left(
\begin{array}{ccccccccc}
4 & 4 & 5 & 5 & 5 & 5 & 4 & 5 & 5\\
E & E & 9 & 1 & 2 & 9 & E & 9 & E\\
B & B & 7 & 6 & 1 & 2 & 0 & 1 & 7\\
A & A & D & C & 9 & 9 & B & A & 3
\end{array}
\right)
\end{split}}
\end{equation}
}by Chinese code of Chinese dictionary, and it gives us four TB-paws
{\small
$$C_1(G^*)=445555455E9E9219EEBB76120173AB99CDAA$$
$$C_2(G^*)=4EBAABE4597DC615521992954E0BA1955E73$$
$$C_3(G^*)=ABADBE4E7C96945119B225590A31E5497E55$$
and
$$C_4(G^*)=4EBA4EBA597D516C521959294E0B591A5E73$$
}by four standard $1$-line Vo-$t$ with $t\in [1,4]$.

\subsection{Systems of linear equations of Hanzis}

Let $X=(x_1~x_2~x_3~x_4)^{-1}$ and $Y=(y_1~y_2~y_3~y_4)^{-1}$ be two vectors with $x_i,y_i\in [0,9]$, and let $A_{han}(H)=(A~B~C~D)^{-1}_{4\times 4}$. So, we have a system of linear equations
\begin{equation}\label{eqa:system-linear-equations-hanzi}
Y=A_{han}(H)X.
\end{equation}
The system (\ref{eqa:system-linear-equations-hanzi}) can help us to find unknown \emph{private key} $Y$ from known \emph{public key} $X$, or other applications.

\subsection{Hanzi-gpws with variable labellings}

\begin{figure}[h]
\centering
\includegraphics[width=8cm]{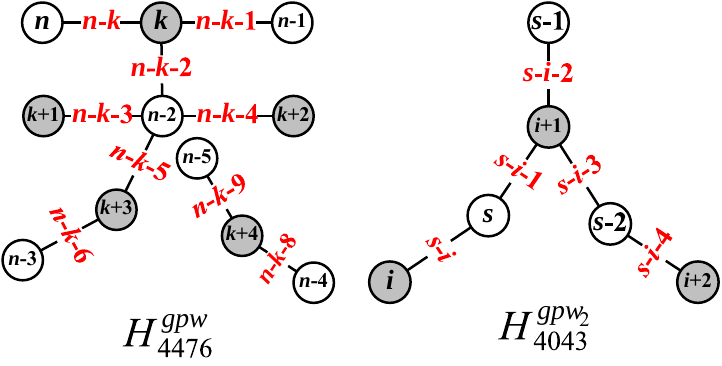}\\
\caption{\label{fig:variable-tian-ren} {\small Two Hanzi-gpws having variable labellings.}}
\end{figure}

We show an example of Hanzi-gpws with variable labellings in Fig.\ref{fig:variable-tian-ren}. By the writing stroke order of Hanzis, we have
$${
\begin{split}
&T_b(H_{4476};n,k)=n(n-k)k(n-k-1)(n-1)(k+1)\\
&(n-k-3)(n-2)(n-k-4)(k+2)k(n-k-2)\\
&(n-2)(n-k-5)(k+3)(n-k-6)(n-3)(n-5)\\
&(n-k-9)(k+4)(n-k-8)(n-4)
\end{split}}$$
$${
\begin{split}
&T_b(H_{4043};s,i)=(s-1)(s-i-2)(i+1)(s-i-1)\\
&s(s-i)i(i+1)(s-i-3)(s-2)(s-i-4)(i+2).
\end{split}}$$

We can apply Hanzi-gpws with variable labellings to build up large scale of Abelian groups, also graph groups introduced in \cite{Yao-Zhang-Sun-Mu-Sun-Wang-Wang-Ma-Su-Yang-Yang-Zhang-2018arXiv}, for encrypting dynamic networks.

\subsection{Hanzis in xOy-plane}

In Fig.\ref{fig:2-hanzi-plane}, we express two Hanzi-graphs $T_{4535}$ and $T_{8630}$ into the popular xOy-plane, such that each vertex has a coordinate $(x,y)$ with non-negative integers $x,y$. So, we can write an edge with two ends $(x,y)$ and $(u,v)$ as $(x,y|u,v)$ (or $(u,v|x,y)$), for instance, $uv=(1,0|2,1)$ (or $uv=(2,1|1,0)$) and $st=(3,0|3,2)$ (or $st=(3,2|3,0)$) shown in Fig.\ref{fig:2-hanzi-plane} (a). We refer to the \emph{graphic expressions} of two Hanzis $T_{4535}$ and $T_{8630}$ as \emph{analytic Hanzis}.

\begin{figure}[h]
\centering
\includegraphics[height=7cm]{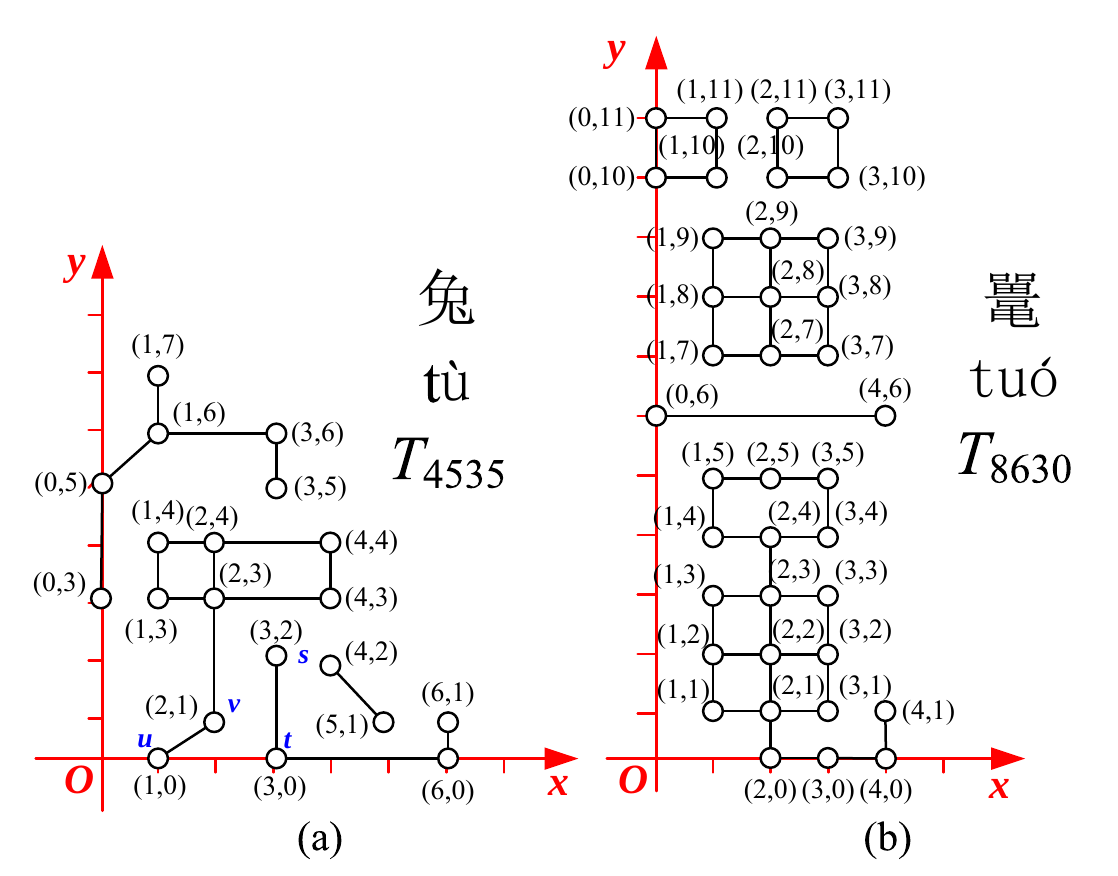}\\
\caption{\label{fig:2-hanzi-plane} {\small (a) A Hanzi-graph $T_{4535}$ is expressed in xOy-plane; (b) a Hanzi-graph $T_{8630}$ is expressed in xOy-plane.}}
\end{figure}

By Fig.\ref{fig:1-hanzi-xOy-plane-TB} and the $1$-line O-$k$ with $k\in [1,4]$ based on (\ref{eqa:k-lines-TB-paw}), we can write the following TB-paws
$${
\begin{split}
T^{(1)}_b(T_{4585})=&02121201111110108765432\\
&11222111122100020,
\end{split}}$$

$${
\begin{split}
T^{(2)}_b(T_{4585})=&0211222212123111140\\
&111522106111070020810,
\end{split}}$$

$${
\begin{split}
T^{(3)}_b(T_{4585})=&22121022111131212422\\
&10501116002071110108,
\end{split}}$$
and
$${
\begin{split}
T^{(4)}_b(T_{4585})=&021121222212311014\\
&1111522116101070010820
\end{split}}$$

\begin{figure}[h]
\centering
\includegraphics[height=5cm]{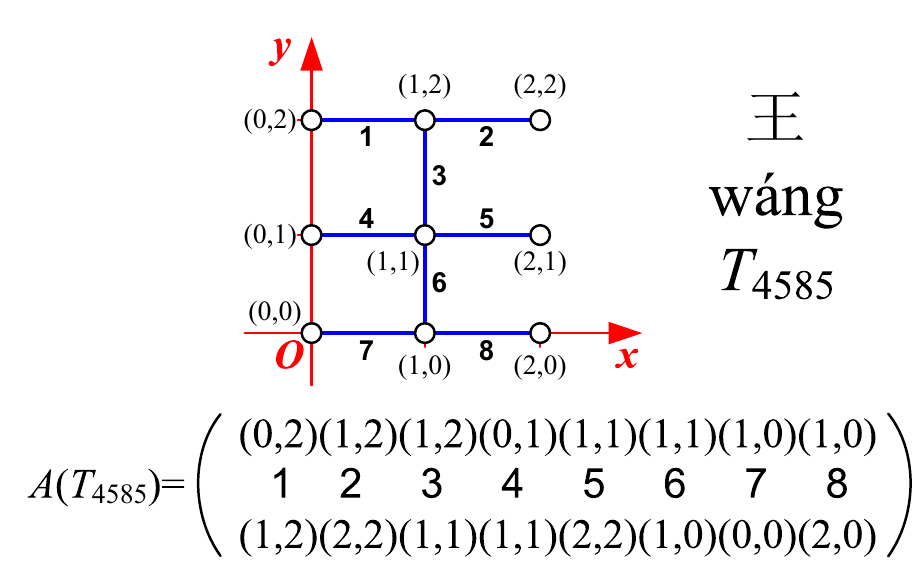}\\
\caption{\label{fig:1-hanzi-xOy-plane-TB} {\small A Hanzi-graph $T_{4585}$ in xOy-plane and its analytic Hanzi-matrix.}}
\end{figure}

\section{Self-similar Hanzi-networks}

Self-similarity is common phenomena between a part of a complex system and the whole of the system. The similarity between the fine structure or property of different parts can reflect the basic characteristics of the whole. In other word, the invariance under geometric or non-linear transformation: the similar properties in different magnification multiples, including geometry. The mathematical expression of self-similarity is defined by
\begin{equation}\label{eqa:self-similarity}
\theta(\lambda r)=\lambda\alpha \theta(r), \textrm{ or } \theta(r)\sim r\alpha,
\end{equation} where $\lambda $ is called \emph{scaling factor}, and $\alpha$ is called \emph{scaling exponent} (fractal dimension) and describes the spatial properties of the structure. The function $\theta(r)$ is a measure of the occupancy number, quantity and other properties of area, volume, mass, etc (Wikipedia).

\subsection{An example of self-similar Hanzi-networks}

The previous four steps of a self-similar Hanzi-network $N_{4043}(t)$ are shown in Fig.\ref{fig:self-similar-00} and Fig.\ref{fig:self-similar-00aa}.

\begin{figure}[h]
\centering
\includegraphics[height=6cm]{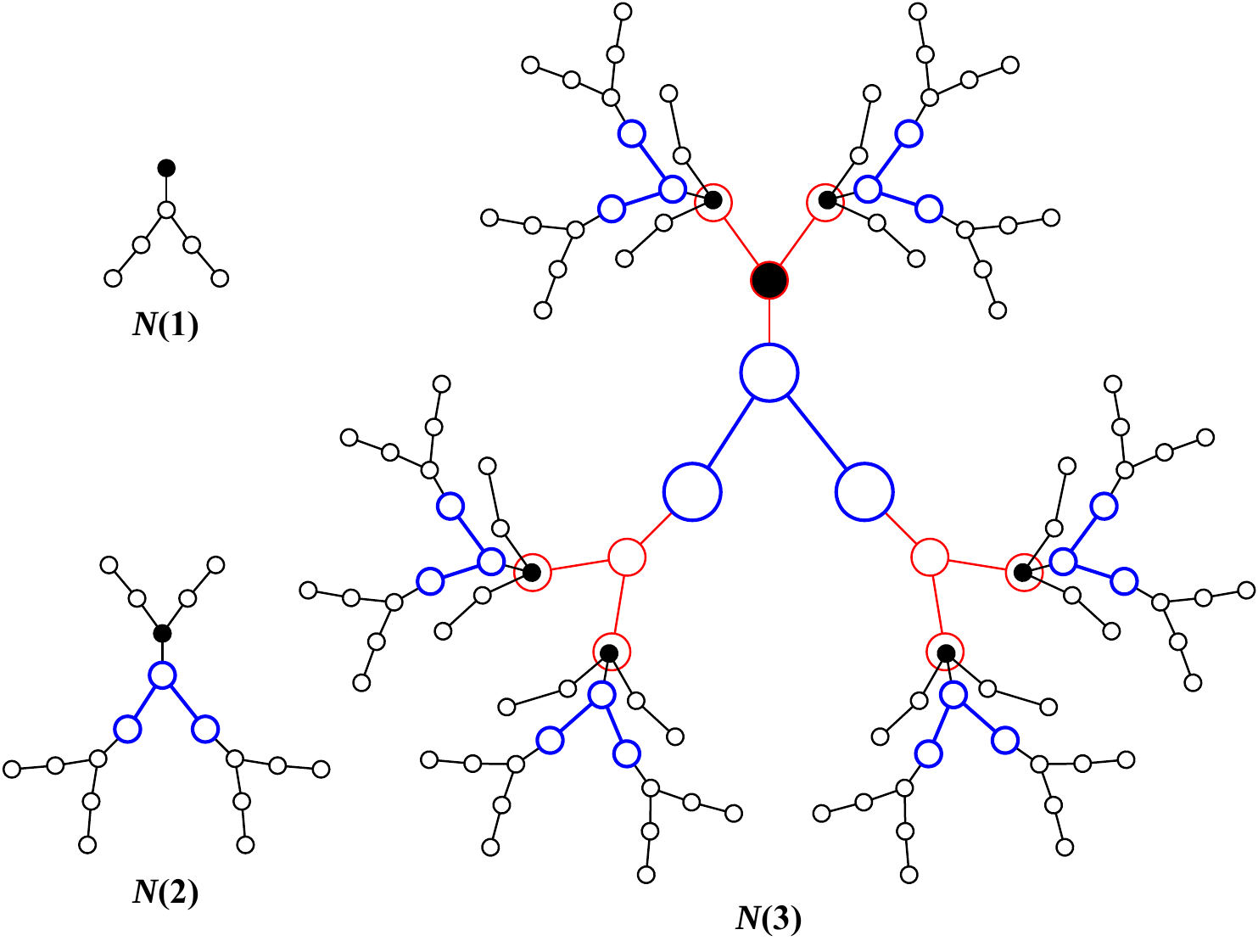}\\
\caption{\label{fig:self-similar-00} {\small A self-similar Hanzi-network $N_{4043}(t)$ with first three steps $t=1,2,3$, where $N(1)=T_{4043}$.}}
\end{figure}

\begin{figure}[h]
\centering
\includegraphics[height=7.8cm]{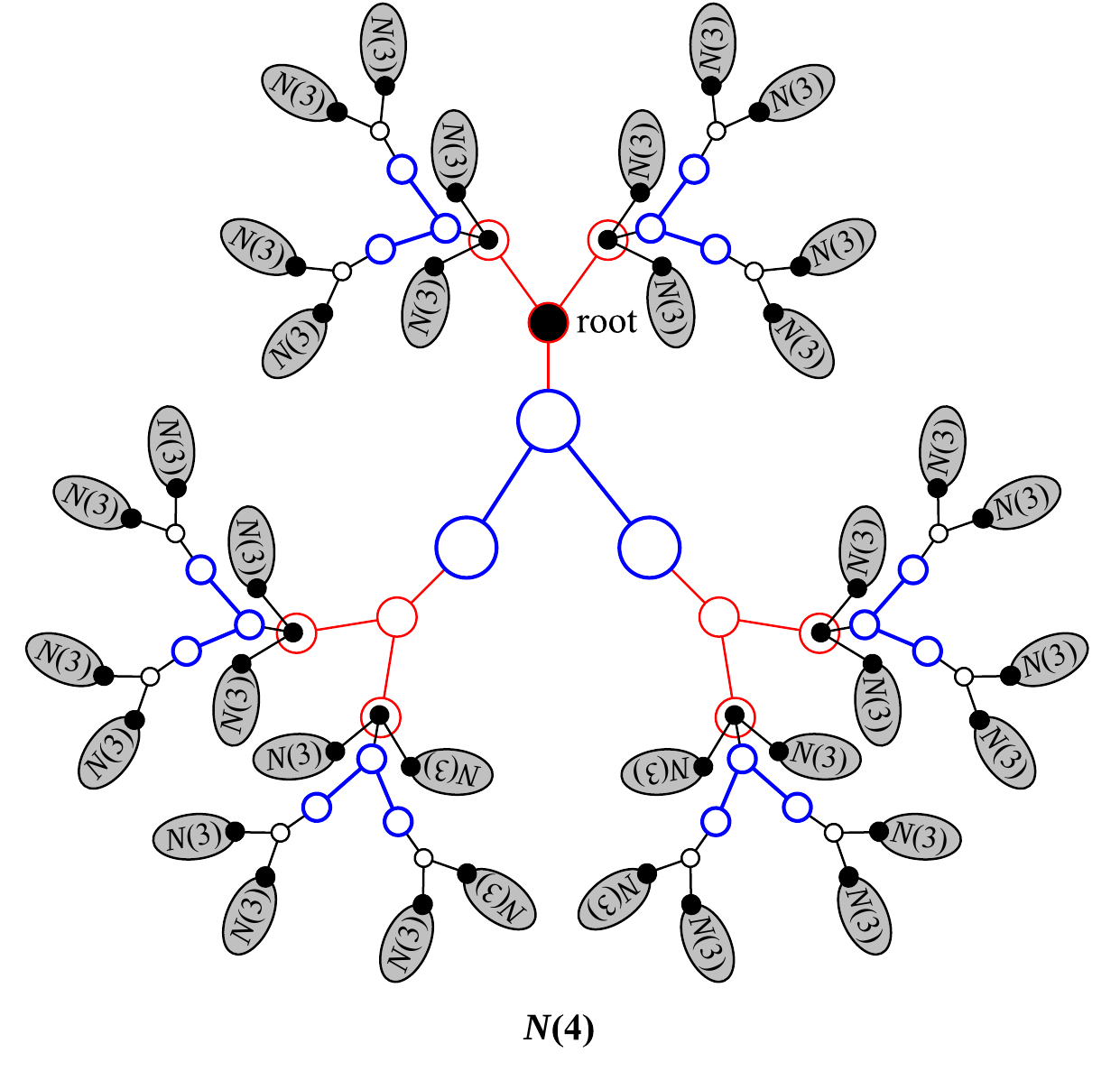}\\
\caption{\label{fig:self-similar-00aa} {\small The self-similar Hanzi-network $N_{4043}(t)$ shown in Fig.\ref{fig:self-similar-00} at the forth step $t=4$.}}
\end{figure}

\subsection{Self-similar tree-like Hanzi-graphs}

In mathematics, a self-similar object is exactly or approximately similar to a part of itself (i.e. the whole has the same shape as one or more of the parts). Many objects in the real world, such as coastlines, are statistically self-similar: parts of them show the same statistical properties at many scales (Ref. \cite{Mandelbrot-Benoit-B-1967}). Some of self-similar Hanzi-graphs are shown in Fig.\ref{fig:self-similar-11}, Fig.\ref{fig:self-similar-algorithm-A}, Fig.\ref{fig:self-similar-00}, Fig.\ref{fig:self-similar-00aa} and Fig.\ref{fig:self-similar-wang}, we call them \emph{tree-like self-similar Hanzi-graphs}, since there are no cycles in them. Another reason is that many tree-like Hanzi-graphs admit many graph labellings for making Hanzi-gpws easily.

\begin{figure}[h]
\centering
\includegraphics[width=8.2cm]{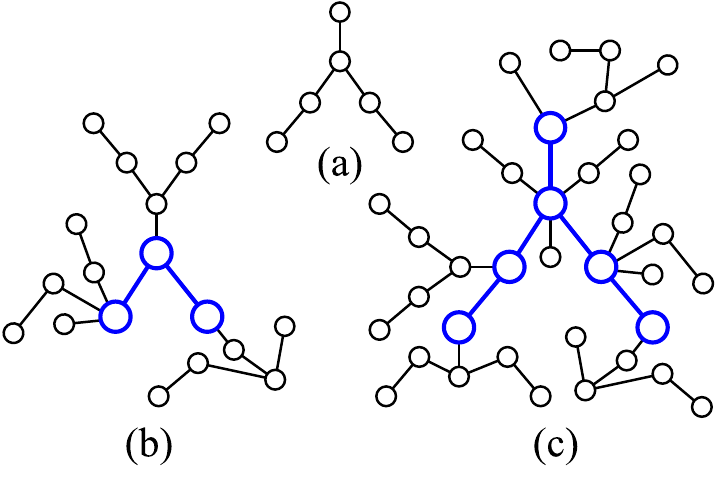}\\
\caption{\label{fig:self-similar-11} {\small (a) Hanzi-graph $T_{4043}$; (b) and (c) are two non-uniformly self-similar Hanzi-graphs.}}
\end{figure}

\begin{figure}[h]
\centering
\includegraphics[width=8.2cm]{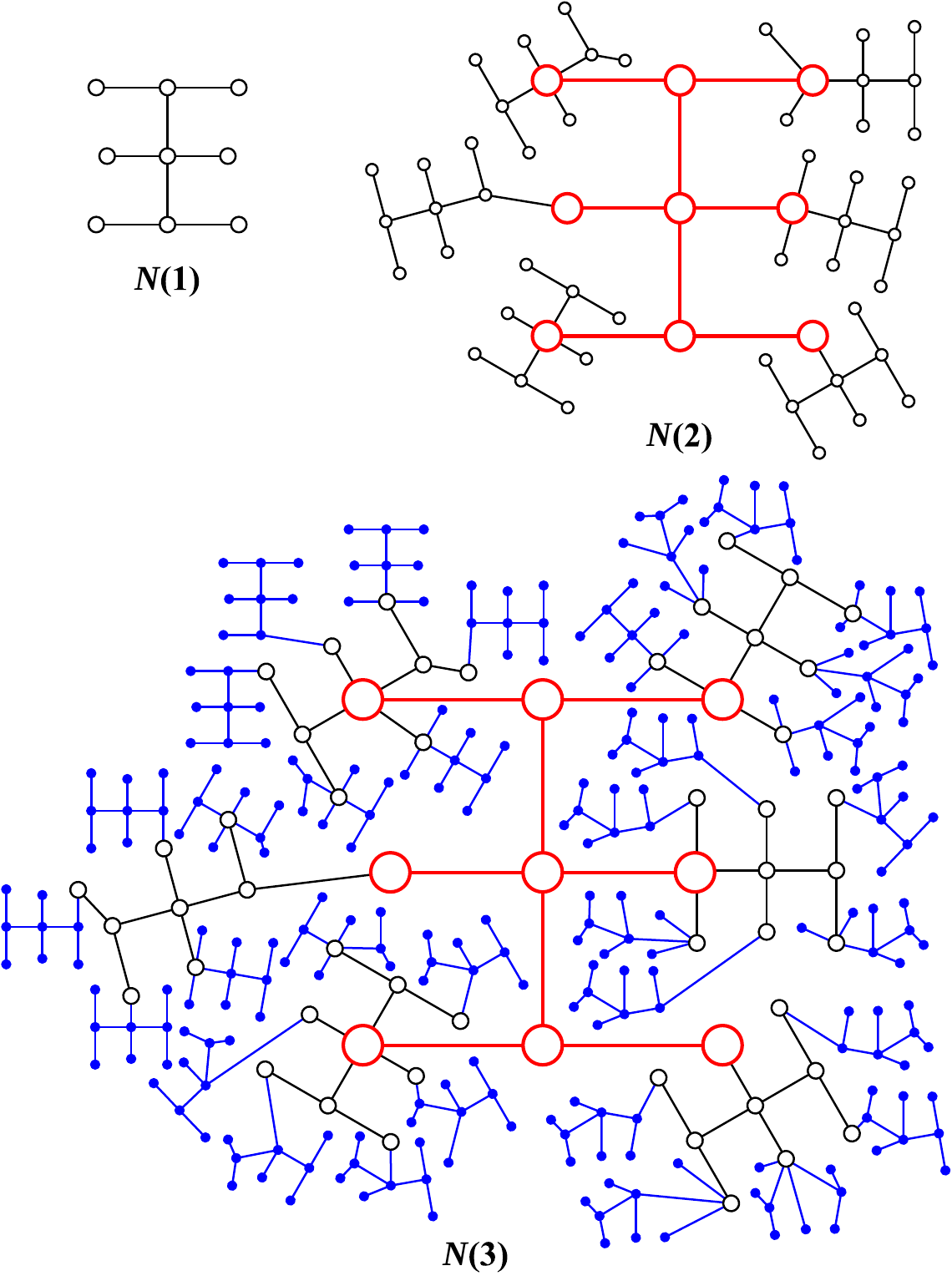}\\
\caption{\label{fig:self-similar-wang} {\small A self-similar Hanzi-network $N_{4585}(t)$ based on a Hanzi-graph $T_{4585}$ for $t=1,2,3$.}}
\end{figure}

We present the following constructive leaf-algorithms for building up particular self-similar tree-like networks. Let $T_0$ be a tree on $n\geq 3$ vertices and let $L(T_0)=\{u_i:i\in[1,m]\}$ be the set of leaves of $T_0$. We refer a vertex $u_0\in V(T_0)$ to be the \emph{root} of $T_0$, and write $L(T_0)\setminus \{u_0\}=\{u_j:j\in [1,m]\}$, where $m=|L(T_0)\setminus \{u_0\}|$, and assume that each leaf $u_i\in L(T_0)\setminus \{u_0\}$ is adjacent to $v_i\in V(T_0)\setminus L(T_0)$ with $i\in [1,m]$.

\vskip 0.2cm

\subsubsection{\textbf{Leaf-algorithm-A}} There are the copies $T_{0,i}$ of $T_0$ with $u_{0,i}\in V(T_{0,i})$ to be the image of the root $u_0$ with $i\in [1,m]$. Deleting each leaf $u_i\in L(T_0)\setminus \{u_0\}$, and then coinciding the \textrm{root vertex} $u_{0,i}$ of the tree $T_{0,i}$ with $v_i$ into one vertex for $i\in [1,m]$, the resultant tree is denoted as $T_1=\odot_A\langle T_0,\{T_{0,i}\}^m_1\rangle$ and called a \emph{uniformly $1$-rank self-similar tree} with root $u_0$. Go on in the way, we have \emph{uniformly $T_0$-leaf $t$-rank self-similar trees} $T_t=\odot_A \langle T_{0},\{T_{t-1,i}\}^m_1\rangle$ with the root $u_0$ and $t\geq 1$. Moreover, we called $T_t$ as a \emph{uniformly $T_0$-leaf self-similar Hanzi-network} with the root at time step $t$ if $T_0$ is a Hanzi-graph. (see Fig.\ref{fig:self-similar-algorithm-A} (a), which is a Hanzi-graph obtained from a Hanzi $H_{4043}$)

\vskip 0.4cm

Obviously, every uniformly $k$-rank self-similar tree $T_k=\odot_A \langle T_{0},\{T_{k-1,i}\}^m_1\rangle$ is similar with $T_0$ as regarding each $T_{k-1,i}$ as a ``leaf''. If the root $u_0$ is a leaf of $T_0$, then the uniformly $k$-rank self-similar trees $T_k=\odot_A \langle T_{0},\{T_{k-1,i}\}^m_1\rangle$ have some good properties.

The vertex number $v(T_t)$ and edge number $e(T_t)$ of each uniformly $T_0$-leaf $t$-rank self-similar tree $T_t=\odot_A \langle T_{0},\{T_{t-1,i}\}^m_1\rangle$ can be computed by the following way:
\begin{equation}\label{eqa:Leaf-algorithm-A}
\left\{
{
\begin{split}
v(T_t)&=v(T_0)m^t+[v(T_0)-2m]\sum^{t-1}_{k=0}m^k\\
e(T_t)&=v(T_t)-1
\end{split}}
\right.
\end{equation}
where $m=|L(T_0)\setminus \{u_0\}|$.

\begin{figure}[h]
\centering
\includegraphics[height=12cm]{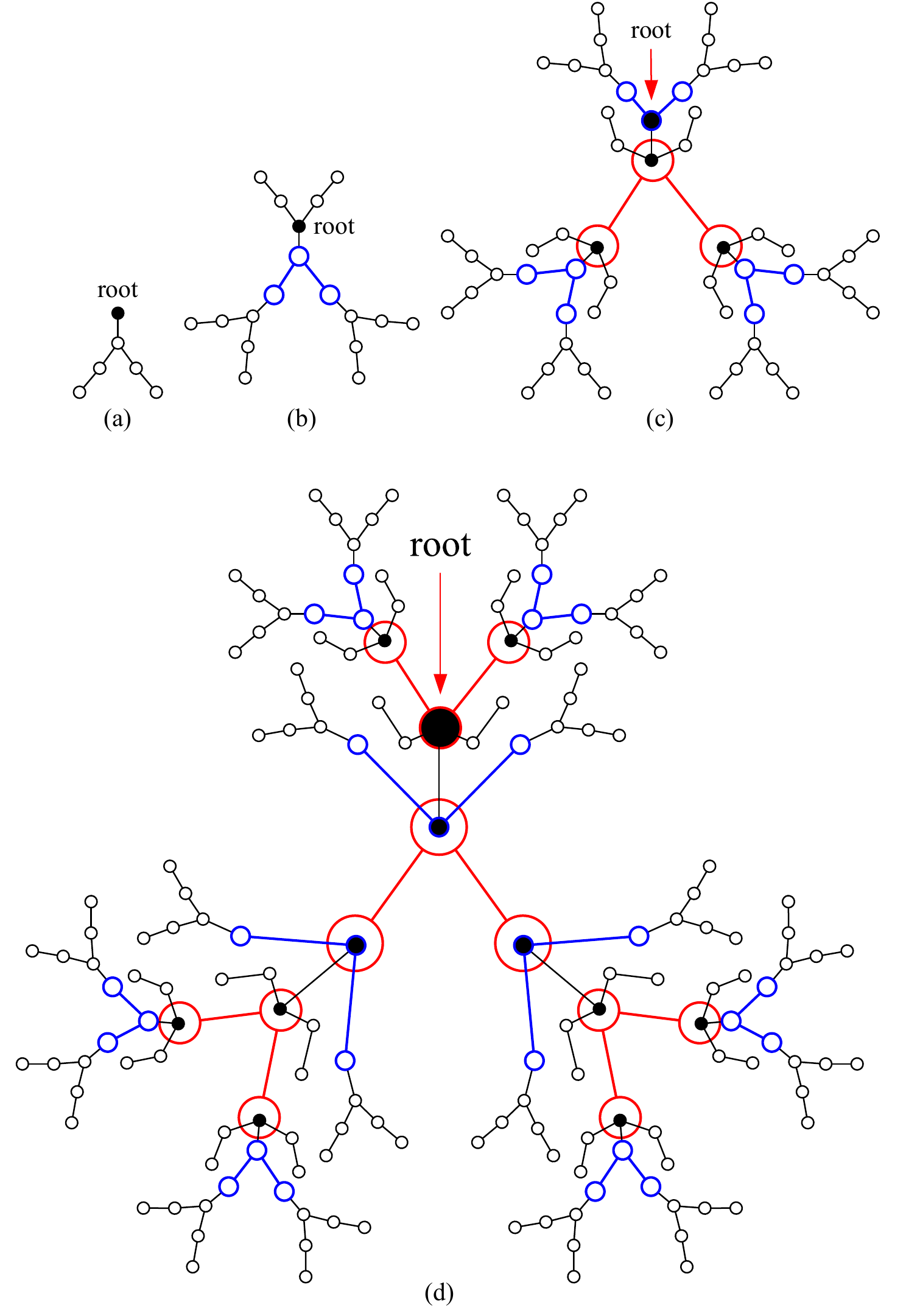}\\
\caption{\label{fig:self-similar-algorithm-A} {\small According to $T_0=T_{4043}$ and Leaf-algorithm-A: (a) A Hanzi-graph $T_{0}$; (b) a uniformly self-similar Hanzi-graph $T_1=\odot_A\langle T_0,\{T_{0,i}\}^3_1\rangle$ with the root; (c) $T_2=\odot_A\langle T_{0},\{T_{2,i}\}^3_1\rangle$ with the root; (d) $T_3=\odot_A\langle T_{0},\{T_{3,i}\}^3_1\rangle$ with the root.}}
\end{figure}

\vskip 0.4cm

\subsubsection{\textbf{Leaf-algorithm-B}} We take $n$ copies $H_{0,1}$, $H_{0,2}$, $\dots $, $H_{0,n}$ of a tree $H_0$, where $n=|L(H_0)|$ is the number of leaves of $H_0$, and do: (1) delete each leaf $x_i$ from $H_0$, where $x_i$ is adjacent with $y_i$ such that the edge $x_iy_i\in E(H_0)$, clearly, $y_i$ may be adjacent two or more leaves; then (2) coincide some vertex $x_{0,i}\in V(H_{0,i})$ with the vertex $y_i$ into one vertex with $i\in [1,n]$. The resultant tree is denoted as $H_1=\odot_B \langle H_0,\{H_{0,i}\}^n_1\rangle $. Proceeding in this way, we get trees $H_j=\odot_B\langle H_0, \{H_{j-1,i}\}^n_1\rangle$ for $j\geq 2$, where $H_{j-1,1}$, $H_{j-1,2}$, $\dots $, $H_{j-1,n}$ are the copies of $H_{j-1}$, and deleting leaves $x_i$ from $H_0$ and coincide an arbitrary vertex $x_{j-1,i}\in V(H_{j-1,i})$ with the vertex $y_i$ of $H_0$ into one for $i\in [1,n]$. It refers to each tree $H_k=\odot_B\langle H_0, \{H_{k-1,i}\}^n_1\rangle$ as an \emph{$H_0$-leaf $k$-rank self-similar tree} without root with $k\geq 1$. We name $H_k$ as an \emph{$H_0$-leaf self-similar Hanzi-network} at time step $k$ if $H_0$ is a Hanzi-graph.

The vertex number $v(H_t)$ and edge number $e(H_t)$ of each $H_0$-leaf $k$-rank self-similar tree $H_t=\odot_B \langle H_{0},\{H_{t-1,i}\}^n_1\rangle$ can be computed in the way:
\begin{equation}\label{eqa:Leaf-algorithm-B}
\left\{
{
\begin{split}
v(H_t)&=v(H_0)n^t+[v(H_0)-2n]\sum^{t-1}_{k=0}n^k\\
e(H_t)&=v(H_t)-1.
\end{split}}
\right.
\end{equation}

\vskip 0.4cm

\subsubsection{\textbf{Leaf-algorithm-C}} Let disjoint trees $G_{0,1}$, $G_{0,2}$, $\dots $, $G_{0,m(0)}$ be the copies of a tree $G_0$, where $m(0)=|L(G_0)|$ is the number of leaves of $G_0$. We delete each leaf $x_i$ from $G_0$ and coincide some vertex $x_{0,i}$ of $G_{0,i}$ with the vertex $y_i$ into one for $i\in [1,m(0)]$, where the edge $x_iy_i\in E(G_0)$. The resultant tree is denoted as $G_1=\odot _C\langle G_0,\{G_{0,i}\}^{m(0)}_1\rangle $. Proceeding in this way, each tree $G_k=\odot_C\langle G_{k-1}, \{G_{k-1,i}\}^{m(k-1)}_1\rangle$ with $k\geq 2$ is obtained by removing each leaf $w_i$ of $G_{k-1}$, and then coincide some vertex $z_{k-1,i}$ of $G_{k-1,i}$ being a copy of $G_{k-1}$ with the vertex $w'_i$ of $G_{k-1}$ into one vertex with $i\in [1,m(k-1)]$, where the leaf $w_i$ is adjacent with $w'_i$ in $G_{k-1}$, and $m(k-1)=|L(G_{k-1})|$ is the number of leaves of $G_{k-1}$. We refer to each tree $G_k$ as a \emph{leaf-$k$-rank self-similar tree} with $k\geq 1$, and call $G_k=\odot_C\langle G_{k-1}, \{G_{k-1,i}\}^{m(k-1)}_1\rangle$ a \emph{leaf-$k$-rank self-similar Hanzi-network} at time step $k\geq 1$ if $G_0$ is a Hanzi-graph.

\vskip 0.4cm

Let $(a_1,a_2,\dots ,a_{k})$ be a combinator selected from integer numbers $0,1,\dots ,s-1$ with $s>1$, so we have the number ${s\choose k}$ of different combinators $(a_1,a_2,\dots ,a_{k})$ in total, and put them into a set $F_k$. Each leaf-$k$-rank self-similar tree $G_t$ has its vertex number $v(G_t)$ and edge number $e(G_t)$ as follows:
\begin{equation}\label{eqa:Leaf-algorithm-B}
\left\{
{
\begin{split}
v(G_t)&=v(G_0)\prod^{s-1}_{k=0}[1+m(k)]\\
&-2\sum^{s-1}_{k=1}\sum^{{s\choose k}}_{(a_1,a_2,\dots ,a_{k})\in F_k}m(a_1)m(a_2)\cdots m(a_k)\\
e(G_t)&=v(G_t)-1.
\end{split}}
\right.
\end{equation}
Moreover, if the vertex $z_{k-1,i}$ of $G_{k-1,i}$ is not a leaf of $G_{k-1,i}$, then each $G_{k-1}$ has $m(k)=[m(0)]^{2^k}$ leaves in total.

\vskip 0.2cm

It is noticeable, some $H_k$ contains $T_k$ in Leaf-algorithm-B. Moreover, each trees $T_k$, $H_k$ and $G_k$ are similar to $T_0$ while we see each subgraphs $T_{k-1}$, $H_{k-1}$ and $G_{k-1}$ of $T_k$, $H_k$ and $G_k$ as `leaves' in the above three algorithms. Common phenomena are that a local part and the whole are similar to each other, and a local shape of a local part and the whole are similar to each other too.

\subsection{Self-similar graphs}

What is a definition of a self-similar graph? We apply the vertex-coincident operation and the vertex-split operation of graphs in Definition \ref{defn:split-operation-combinatoric} to show some types of self-similar graphs.
\begin{defn}\label{defn:vertex-edge-split-self-similar-graphs}
$^*$ Let $G_0$ be a $(p_0,q_0)$-graph. If a graph $G_1$ can be vertex-split (resp. edge-split) into proper subgraphs $H_{1,j}$ with $j\in [1,m_1]$ by the vertex-split (resp. edge-split) operation defined in Definition \ref{defn:split-operation-combinatoric}, such that $H_{1,j}\cong G_0$ for each $j\in [1,m_1]$ and $E(G_1)=\bigcup ^{m_1}_{j=1}E(H_{1,j})$, we say $G_1$ to be a \emph{vertex-split (resp. edge-split) $m_1$-scaling $G_0$-similar graph}, denoted as $G_1=\odot^{m_1}_{j} H_{1,j}$ (resp. $G_1=\ominus^{m_1}_{j} H_{1,j}$). Furthermore, if we can do a series of vertex-split operations to a graph $G_i$  to obtain subgraphs $H_{i,j}$ with $j\in [1,m_i]$ such that $H_{i,j}\cong G_{i-1}$ and $E(G_i)=\bigcup ^{m_i}_{j=1}E(H_{i,j})$ with $i\in [1,n]$, we call $G_i$ a \emph{vertex-split (resp. edge-split) $(m_j)^i_1$-scaling $G_0$-similar graph}, denoted as $G_i=\odot^{m_i}_{j} H_{i,j}$ (resp. $G_i=\ominus^{m_i}_{j} H_{i,j}$), with $i\in [1,n]$ for $n\geq 1$.\qqed
\end{defn}

We, in Definition \ref{defn:vertex-edge-split-self-similar-graphs}, can see: (1) $E(H_{i,k})\cap E(H_{i,j})$ may be not empty in general; (2) do only vertex-split operations, or do only edge-split operation, no mixed. Thereby, we can give the following self-similar graphs:

\begin{defn}\label{defn:vertex-edge-split-mixed}
$^*$ Let $G_0$ be a $(p_0,q_0)$-graph. If a graph $G_1$ can be vertex-split and edge-split into proper subgraphs $H_{1,j}$ with $j\in [1,m_1]$ by two vertex-split and edge-split operations defined in Definition \ref{defn:split-operation-combinatoric}, such that $H_{1,j}\cong G_0$ for each $j\in [1,m_1]$ and $E(G_1)=\bigcup ^{m_1}_{j=1}E(H_{1,j})$, we say $G_1$ to be a \emph{mixed-split $m_1$-scaling self-similar graph with respect to $G_0$}, denoted as $G_1=(\ominus\cup \odot)^{m_1}_{j} H_{1,j}$. Furthermore, if we can do a series of vertex-split and edge-split operations to each graph $G_i$ ($i\in [1,n]$ for $n\geq 1$) to obtain subgraphs $H_{i,j}$ with $j\in [1,m_i]$ such that $H_{i,j}\cong G_{i-1}$ and $E(G_i)=\bigcup ^{m_i}_{j=1}E(H_{i,j})$, we call $\{G_i\}$ a \emph{mixed-split $\prod^i_1 m_j$-scaling self-similar sequence with respect to $G_0$}, denoted as $G_i=(\ominus\cup \odot)^{m_i}_{j} H_{i,j}$.\qqed
\end{defn}

In Definition \ref{defn:vertex-edge-split-mixed}, we can see $\theta(m_1\cdot G_0)=m_1\cdot \theta(G_0)$ by the mathematical expression of self-similarity defined in the equation (\ref{eqa:self-similarity}), where $m_1$ is the \emph{scaling factor}. Moreover, we show another type of self-similar graphs in Definition \ref{defn:edge-contracting-self-similar} below.

\begin{defn}\label{defn:edge-contracting-self-similar}
$^*$ Let $G$ be a $(p,q)$-graph. If a graph $H$ contains proper subgraphs $H_{i}$ with $n\in [1,n]$, and each $H_{i}$ is a mixed-split $m_i$-scaling self-similar graph with respect to $G$, we do an edge-contracting operation defined in Definition \ref{defn:split-operation-combinatoric} to all edges of $H_{i}$ such that each $H_{i}$ contracts a graph having a unique vertex and no loops, the resultant graph is denoted as $J$. If $J$ does not contains any proper subgraph being a mixed-split $r$-scaling self-similar graph with respect to $G$, and $J\cong G$, we call $H$ a \emph{complete self-similar graph with respect to $G$}.\qqed
\end{defn}

In general, we may meet so-called pan-self-similar graphs defined as follows:

\begin{defn}\label{defn:multiple-element-self-similar}
$^*$ Let $F(k)$ be a graph set of graphs $G^1_{0},G^2_{0},\dots ,G^k_{0}$ with $k\geq 2$. If a graph $H$ can be vertex-split (edge-split) into proper subgraphs $H_{j}$ with $j\in [1,m]$ by the vertex-split (edge-split) operation defined in Definition \ref{defn:split-operation-combinatoric}, such that each $H_{j}$ holds $H_{j}\cong G^t_0$ for some $G^t_0\in F$ and $E(H)=\bigcup ^{m}_{j=1}E(H_{j})$, we say $H$ to be a \emph{vertex-split (edge-split) $m$-scaling pan-self-similar graph with respect to the set $F(k)$}, denoted as $H=F\odot^{m_i}_{j} H_{i,j}$.\qqed
\end{defn}

For example, assume that $G_0$ defined in Definition \ref{defn:multiple-element-self-similar} is a vertex-split (edge-split) $m$-scaling pan-self-similar graph with respect to a graph set $F(k)$, then the mixed-split $(m_j)^i_1$-scaling self-similar sequence $\{G_i\}$ with respect to $G_0$ is a pan-self-similar graph with respect to the set $F(k)$. We introduce the following ways for constructing self-similar networks:

\vskip 0.2cm

\textbf{$\bullet$ Self-similar vertex-coincided algorithm.}

\emph{Step 1.} Let $N(0)$ be a network with a unique \emph{active vertex} $u_0$, and let $n_v(0)$ be the number of vertices of the network $N(0)$, so we have the vertex set $V(N(0))=\{v_i:~i\in[1,n_v(0)]\}$. We take $n_v(0)$ copies $N^i(0)$ of $N(0)$ with $i\in[1,n_v(0)]$, and the vertex $u^i_0$ of the $i$th copy $N^i(0)$ is the image of the active vertex $u_0$ of $N(0)$. And then we coincide the $i$th image $u^i_0$ with the vertex $v_i$ of the network $N(0)$ into one for $i\in[1,n_v(0)]$. The resultant network is denoted as $N(1)=N(0)\overrightarrow{\odot}N(0)$. Here, we define the active vertex $u_0$ of $N(0)$ as the active vertex of $N(1)$.

\vskip 0.2cm

\emph{Step 2.} The following vertex-coincided algorithms are independent to each other:

\vskip 0.2cm

\emph{Step 2.1.} \textbf{Vertex-coincided algorithm-I.}

We coincide the unique active vertex $u^s_{k}$ of the $s$th copy $N^s(k)$ of $N(k)$ with the vertex $v^s_{k-1}$ of the $N(k-1)$ into one vertex for $s\in [1,n_v(k-1)]$ with $k\geq 1$, where $n_v(k-1)=|V(N(k-1))|$ is the number of vertices of $N(k-1)$, the resultant network is denoted as $N(k+1)=N(k)\overrightarrow{\odot }N(k-1)$, and called a \emph{uniformly vertex-split $n_v(j)^k_0$-scaling $N(0)$-similar network} according to Definition \ref{defn:vertex-edge-split-self-similar-graphs}, and let the active vertex $u_{k-1}$ of $N(k-1)$ to be the unique active vertex of $N(k+1)$. Then, $n_v(j)^k_0$ are called \emph{scaling factors}.

\begin{figure}[h]
\centering
\includegraphics[height=7.8cm]{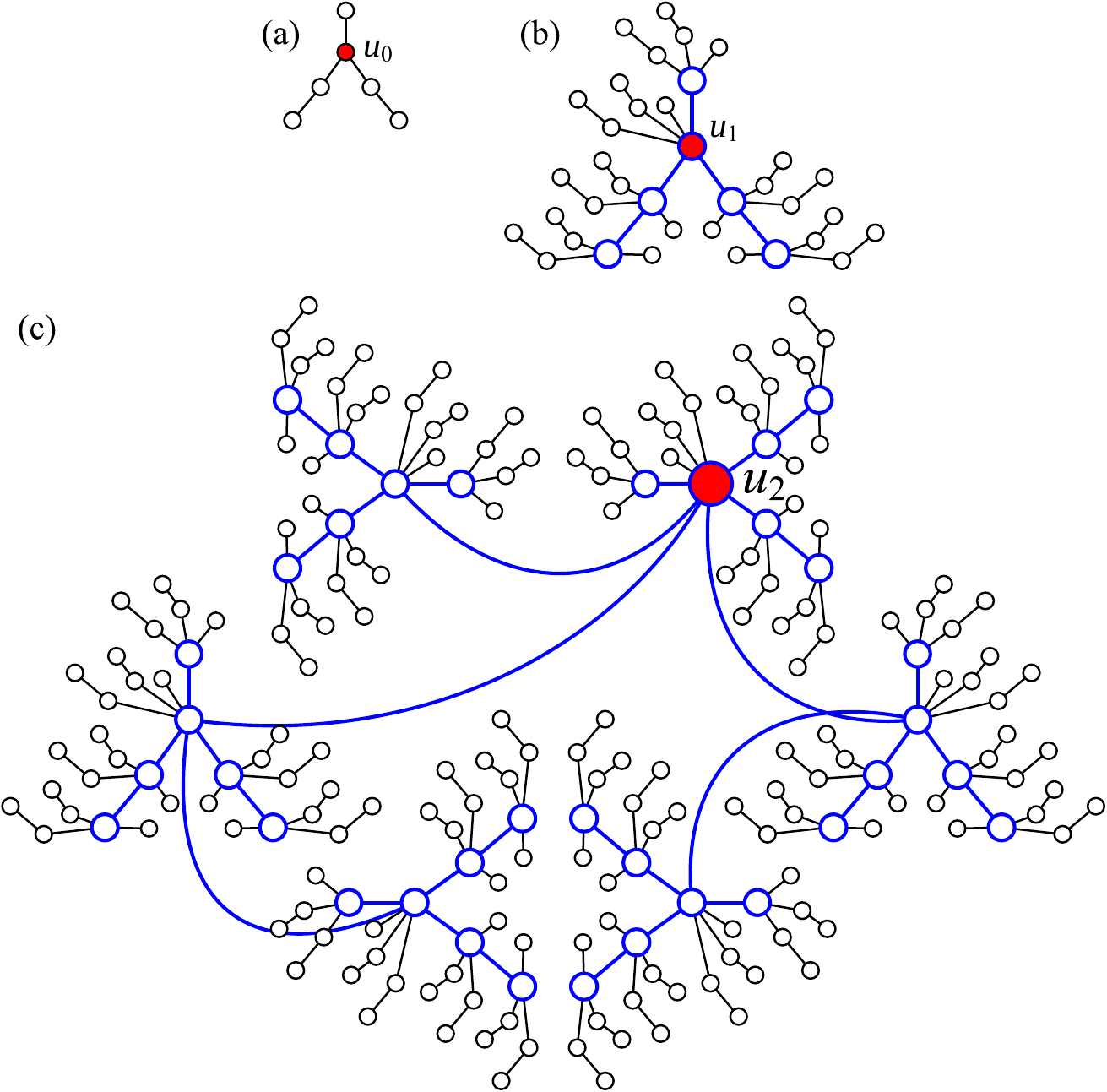}\\
\caption{\label{fig:vertex-coincided-algorithm-b} {\small For illustrating the vertex-coincided algorithm-I.}}
\end{figure}

Fig.\ref{fig:vertex-coincided-algorithm-b} is for illustrating the vertex-coincided algorithm-I, where (a) is $N(0)=T_{4043}$, (b) is $N(1)=N(0)\overrightarrow{\odot }N(0)$, and (c) is $N(2)=N(1)\overrightarrow{\odot }N(0)$.

\vskip 0.2cm

\emph{Step 2.2.} \textbf{Vertex-coincided algorithm-II.}

Coinciding the unique active vertex $u^i_{k-1}$ of the $i$th copy $N^i(k-1)$ of $N(k-1)$ with the vertex $v^i_{k}$ of the network $N(k)$ into one vertex for $i\in [1,n_v(k)]$ and $k\geq 1$, where $n_v(k)=|V(N(k))|$ is the number of vertices of $N(k)$, will produce a new network $N(k+1)$, we write it as $N(k+1)=N(k-1)\overrightarrow{\odot }N(k)$, and call it a \emph{uniformly vertex-split $n_v(j)^k_0$-scaling $N(0)$-similar network} according to Definition \ref{defn:vertex-edge-split-self-similar-graphs}. Here, we define the unique active vertex $u_k$ of $N(k)$ as the \emph{ unique active vertex} of $N(k+1)$, rewrite it as $u_{k+1}$. Then, $n_v(j)^k_0$ are called \emph{scaling factors}. (see examples shown in Fig.\ref{fig:vertex-coincided-algorithm-b})

\begin{figure}[h]
\centering
\includegraphics[height=6cm]{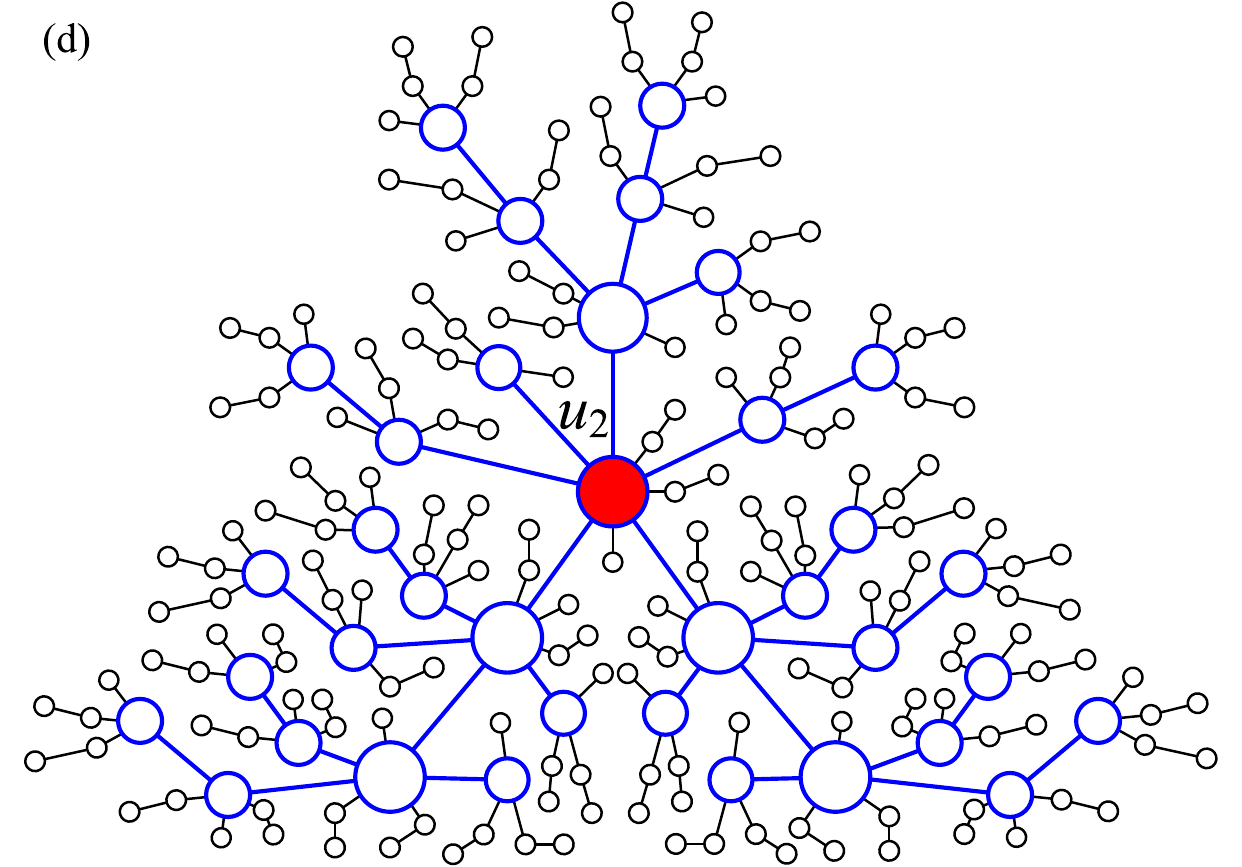}\\
\caption{\label{fig:vertex-coincided-algorithm-c} {\small For illustrating the vertex-coincided algorithm-II.}}
\end{figure}

Fig.\ref{fig:vertex-coincided-algorithm-c} is for illustrating the vertex-coincided algorithm-II, so we can see (d)$=N(2)=N(0)\overrightarrow{\odot}N(1)$ by (a)$=N(0)$ and (b)$=N(1)$ shown in Fig.\ref{fig:vertex-coincided-algorithm-b}.

\vskip 0.2cm

\emph{Step 2.3.} \textbf{Vertex-coincided algorithm-III.}

Let $N^*(k)$ be a duplicated network of $N(k)$. Coinciding the unique active vertex $u^j_{k}$ of the $s$th copy $N^j(k)$ of $N(k)$ with the vertex $v^s_{k}$ of the $N(k)$ into one vertex for $j\in [1,n_v(k)]$ with $k\geq 1$ produces a new network $N(k+1)=N(k)\overrightarrow{\odot}N(k)$, called a \emph{uniformly vertex-split $n_v(k)$-scaling $N(0)$-similar network} according to Definition \ref{defn:vertex-edge-split-self-similar-graphs}, and the unique active vertex of $N(k)$ is defined as the unique active vertex of $N(k+1)$. Then, $n_v(k)$ is the \emph{scaling factor}.

\vskip 0.2cm

\emph{Step 2.4.} \textbf{Vertex-coincided algorithm-IV.} Let $N'(0)$, $N'(1)$, $ \dots $, $N'(m)$ be defined well by Step 1 and $N'(s)\overrightarrow{\odot }N'(t)$ be defined by: We take the copies $N'_j(s)$ of $N'(s)$ having unique active vertex $u_{s,0}$ with $j\in [1,|V(N'(t))|]$, each copy $N'_j(s)$ has its own unique active vertex $u^j_{s,0}$. We coincide the active vertex $u^j_{s,0}$ of $N'_j(s)$ with the vertex $w_j\in V(N'(s))=\{w_j:~j\in [1,|V(N'(t))|]\}$ into one vertex, the resultant network is just $N'(s)\overrightarrow{\odot }N'(t)$.

Then we can make self-similar network
\begin{equation}\label{eqa:similar-network-make-11}
{
\begin{split}
N'(m+1)&=\overrightarrow{\odot }^k_{i=1}N'(a_i)\\
&=N'(a_1)\overrightarrow{\odot }N'(a_2)\overrightarrow{\odot }\cdots \overrightarrow{\odot }N'(a_k)
\end{split}}
\end{equation} with $k\geq 2$, where pairwise distinct $a_1,a_2,\dots ,a_k$ are selected from the set $[0,m]$. As the following form
\begin{equation}\label{eqa:similar-network-make-33}
N'(t+1)=N'(t-2)\overrightarrow{\odot}N'(t-1)\overrightarrow{\odot}N'(t)
\end{equation}
for each $t\geq 2$ holds true, we say $N'(t+1)$ to be a \emph{Fibonacci $N'(0)$-similar network} on time step $t\geq 0$.

\vskip 0.4cm

\textbf{$\bullet$ Self-similar edge-coincided algorithm.}

\vskip 0.2cm

\emph{Step (1)} Let $M(0)$ be a network with a unique \emph{active edge} $x_0w_0$, and let $n_e(0)$ be the number of vertices of the network $M(0)$, so we have the edge set $E(M(0))=\{x_iy_i:~i\in[1,n_e(0)]\}$. We take $n_e(0)$ copies $M^i(0)$ of $M(0)$ with $i\in[1,n_e(0)]$, and the edge $x^i_0w^i_0$ of the $i$th copy $M^i(0)$ is the image of the active edge $x_0w_0$ of $M(0)$. Next, the $i$th image $x^i_0w^i_0$ is identified with the edge $x_iy_i$ of the network $M(0)$ into one edge for each $i\in[1,n_e(0)]$. The resultant network is denoted as $M(1)=M(0)\overrightarrow{\ominus}M(0)$. Here, we define the active edge $x_0w_0$ of $M(0)$ as the active edge of $M(1)$, denoted as $x_1w_1$.

\vskip 0.2cm

\emph{Step (2)} The following edge-coincided algorithms are independent to each other:

\vskip 0.2cm

\emph{Step (2.1)} \textbf{Edge-coincided algorithm-I.}

Overlapping the unique active edge $x^i_{t-1}w^i_{t-1}$ of the $i$th copy $M^i(t-1)$ of $M(t-1)$ with the edge $x^i_{k}y^i_{k}$ of the network $M(t)$ into one edge for $i\in [1,n_e(t)]$ and $k\geq 1$, where $n_e(t)=|E(M(t))|$ is the number of edges of $M(t)$, gives us a new network $M(t+1)=M(t-1)\overrightarrow{\ominus}M(t)$, called as a \emph{uniformly edge-split $n_e(j)^t_0$-scaling $M(0)$-similar network} by Definition \ref{defn:vertex-edge-split-self-similar-graphs}. Here, we define the unique active edge $x_kw_k$ of $M(t)$ as the \emph{unique active edge} of $M(t+1)$, rewrite it as $x_{k+1}w_{k+1}$, $n_e(j)^t_0$ are called \emph{edge-scaling factors} (see example shown in Fig.\ref{fig:edge-coincided-algorithm-b} and Fig.\ref{fig:edge-coincided-algorithm-c}).

\begin{figure}[h]
\centering
\includegraphics[height=5cm]{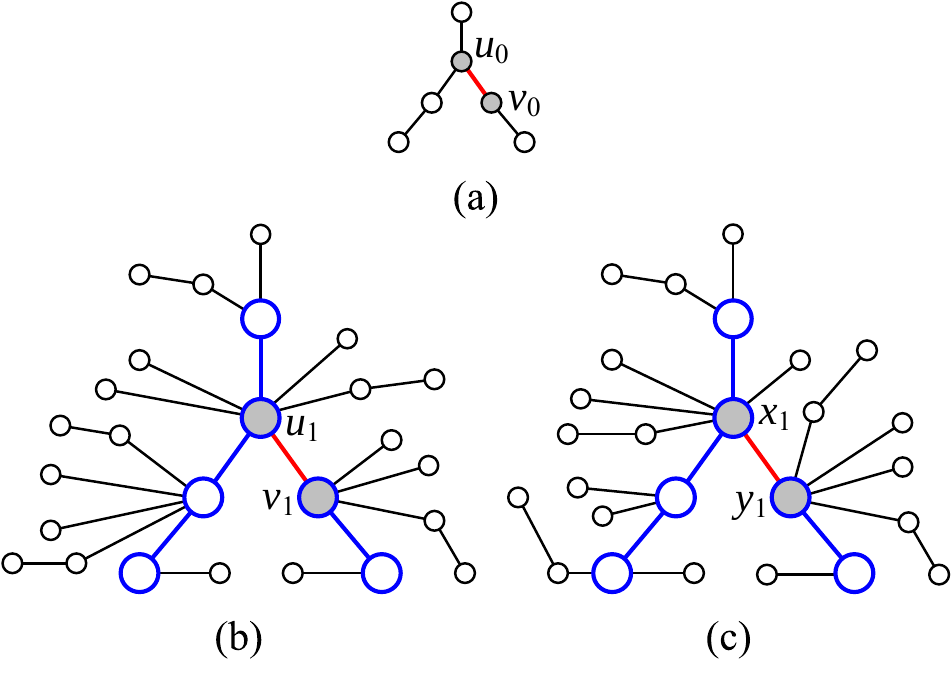}\\
\caption{\label{fig:edge-coincided-algorithm-b} {\small Self-similar Hanzi-graphs on the edge-coincided algorithm-I.}}
\end{figure}

Fig.\ref{fig:edge-coincided-algorithm-b} is for illustrating the edge-coincided algorithm-I: (a) $M(0)$ was made by a Hanzi-graph $T_{4043}$ with a unique active edge $u_0v_0$; (b) $M(1)=M(0)\overrightarrow{\ominus} M(0)$ with a unique active edge $u_1v_1$; (c) another $M'(1)=M(0)\overrightarrow{\ominus} M(0)$ with a unique active edge $x_1y_1$. Clearly, $M(1)\not \cong M'(1)$.

\begin{figure}[h]
\centering
\includegraphics[width=8cm]{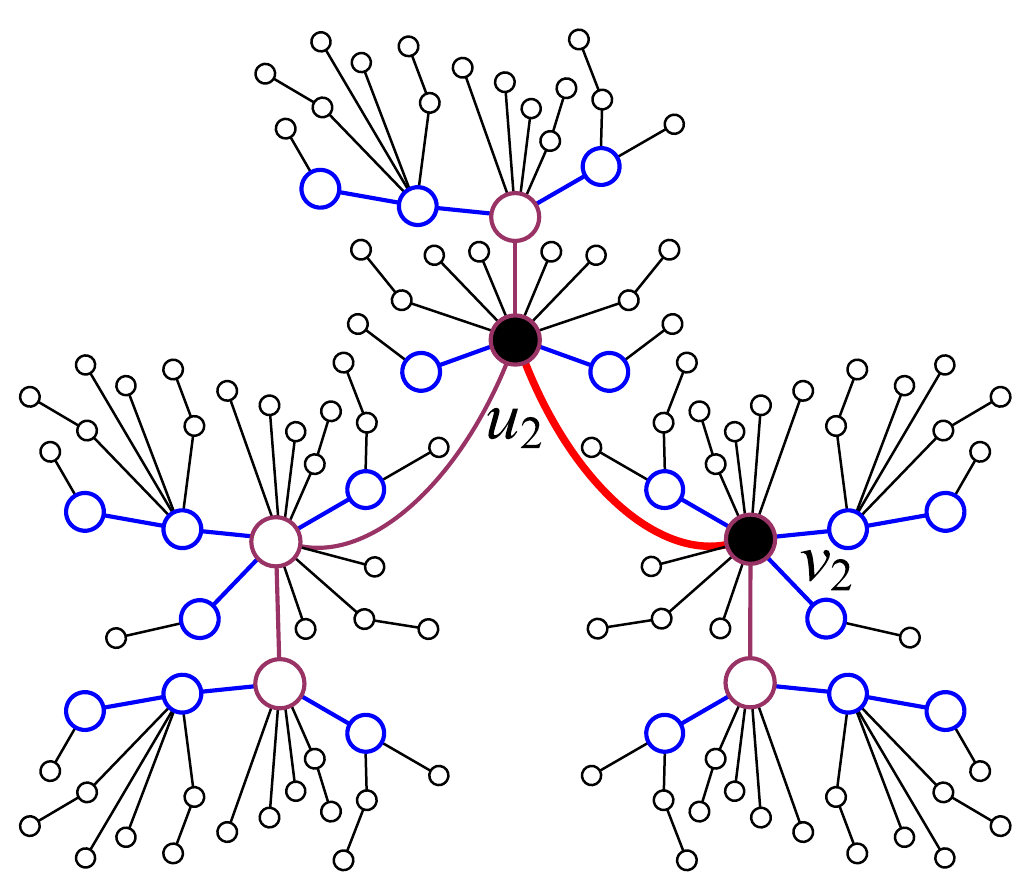}\\
\caption{\label{fig:edge-coincided-algorithm-c} {\small A self-similar Hanzi-graph with its unique \emph{active edge} $u_2v_2$.}}
\end{figure}

The self-similar Hanzi-graph shown in Fig.\ref{fig:edge-coincided-algorithm-c} is $M(2)=M(1)\overrightarrow{\ominus} M(0)$, where $M(0)$=(a), $M(1)$=(b) shown in Fig.\ref{fig:edge-coincided-algorithm-b}.

\vskip 0.2cm

\emph{Step (2.2)} \textbf{Edge-coincided algorithm-II.}

We overlap the unique active edge $x^s_{t}w^s_{t}$ of the $s$th copy $M^s(t)$ of $M(t)$ with the edge $x^s_{t-1}y^s_{t-1}$ of the $M(t-1)$ into one edge for $s\in [1,n_e(t-1)]$ with $k\geq 1$, where $n_e(t-1)=|E(M(t-1))|$ is the edge number of $M(t-1)$, so we get a new network denoted as $M(t+1)=M(t)\overrightarrow{\ominus}M(t-1)$, called a \emph{uniformly edge-split $n_e(j)^t_0$-scaling $M(0)$-similar network} under Definition \ref{defn:vertex-edge-split-self-similar-graphs}, and let the unique active edge $x_{t-1}w_{t-1}$ of $M(t-1)$ to be the unique active edge of $M(t+1)$, denoted as $x_{t+1}w_{t+1}$. We say $n_e(j)^k_0$ to be \emph{edge-scaling factors} (see examples shown in Fig.\ref{fig:edge-coincided-algorithm-b} and Fig.\ref{fig:edge-coincided-algorithm-d}).

\begin{figure}[h]
\centering
\includegraphics[width=8cm]{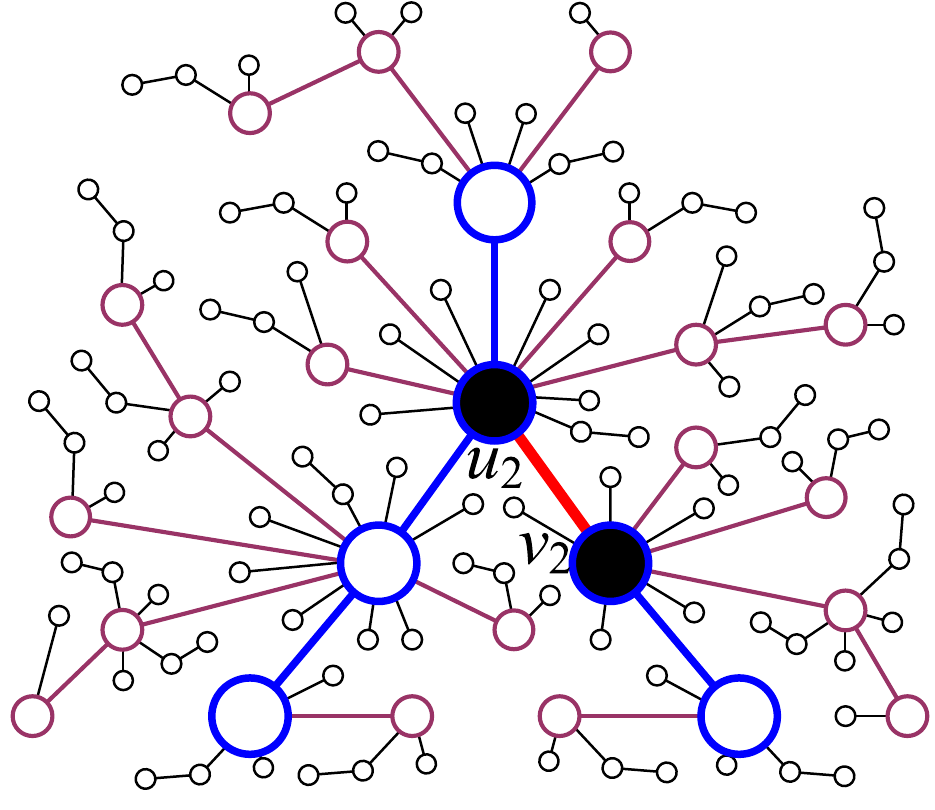}\\
\caption{\label{fig:edge-coincided-algorithm-d} {\small A self-similar Hanzi-graph made by the edge-coincided algorithm-II, where $u_2v_2$ is its unique \emph{active edge}.}}
\end{figure}

The self-similar Hanzi-graph shown in Fig.\ref{fig:edge-coincided-algorithm-d} is $M(2)=M(0)\overrightarrow{\ominus} M(1)$, where $M(0)$=(a), $M(1)$=(b) shown in Fig.\ref{fig:edge-coincided-algorithm-b}.

\vskip 0.2cm

\emph{Step (2.3)} \textbf{Edge-coincided algorithm-III.}

Take a duplicated network of $M(t)$, and denote as $M^*(t)$. We overlap the unique active edge $x^s_{t}w^s_{t}$ of the $s$th copy $M^j(t)$ of $M(t)$ with the edge $x^s_{t}y^s_{t}$ of the $M(t)$ into one edge for $s\in [1,n_e(t)]$ with $k\geq 1$, such that the resultant network $M(t+1)=M(t)\overrightarrow{\ominus}M(t)$, called a \emph{uniformly edge-split $n_e(t)$-scaling $M(0)$-similar network} by Definition \ref{defn:vertex-edge-split-self-similar-graphs}, and the unique active edge of $M(t+1)$ is defined as the unique active edge $x_{t}w_{t}$ of $M(t)$. Then, $n_e(t)$ is \emph{edge-scaling factor}.

It is noticeable, self-similar networks $M(k+1)=M(k)\overrightarrow{\ominus}M(k-1)$ made by the edge-coincided algorithm-I, the edge-coincided algorithm-II and the edge-coincided algorithm-III are not unique. So, we claim that self-similar networks $M(k+1)=M(k)\overrightarrow{\ominus}M(k-1)$ are very complex greater than self-similar networks made by three vertex-coincided algorithm-$x$ with $x=$I,II.

\vskip 0.2cm

\emph{Step (2.4)} \textbf{Edge-coincided algorithm-IV.} Let $M'(0)$, $M'(1)$, $ \dots $, $M'(n)$ be defined well by Step 1 and the operation $M'(s)\overrightarrow{\ominus }M'(t)$ defined as: There are the copies $M'_j(s)$ of $M'(s)$ having unique active edge $x_{s,0}u_{s,0}$ with $j\in [1,|E(M'(t))|]$, each copy $M'_j(s)$ has its own unique active edge $x^j_{s,0}u^j_{s,0}$. Now, we overlap the active edge $x^j_{s,0}u^j_{s,0}$ of $M'_j(s)$ with the edge $z_jw_j\in E(M'(s))=\{z_jw_j:~j\in [1,|E(M'(t))|]\}$ into one edge, the resultant network is denoted by $M'(s)\overrightarrow{\ominus }M'(t)$. A self-similar network is produced as
\begin{equation}\label{eqa:similar-network-edge-make-11}
{
\begin{split}
M'(m+1)&=\overrightarrow{\ominus }^k_{i=1}M'(b_i)\\
&=M'(b_1)\overrightarrow{\ominus }M'(b_2)\overrightarrow{\ominus }\cdots \overrightarrow{\ominus }M'(b_k)
\end{split}}
\end{equation} with $k\geq 2$, where $b_1,b_2,\dots ,b_k$ are selected from the set $[0,n]$ and are pairwise distinct. If every self-similar network
\begin{equation}\label{eqa:similar-network-edge-make-33}
M(t+1)=M(t-2)\overrightarrow{\ominus }M(t-1)\overrightarrow{\ominus }M(t)
\end{equation}
for each time step $t\geq 2$ holds true, then $M(t+1)$ is called an \emph{edge-Fibonacci $M'(0)$-similar network} on time step $t\geq 0$, under the edge-coincided algorithms defined above.

\vskip 0.2cm

\textbf{$\bullet$ Self-similar mixed-coincided algorithm.} We allow: (i) A new self-similar network is constructed by the vertex-coincided algorithms and the edge-coincided algorithms defined above, see (\ref{eqa:Mixed-coincided-algorithm-11}) and (\ref{eqa:Mixed-coincided-algorithm-22}); (ii) a new self-similar network is made by coinciding the unique active vertex/edge of a self-similar network with part of vertices/edges of another self-similar network in probabilistic way. Here, we list the following cases
\begin{equation}\label{eqa:Mixed-coincided-algorithm-11}
N(k+1)=\left\{
\begin{array}{l}
N(k)\overrightarrow{\odot }N(k-1);\\
N(k-1)\overrightarrow{\odot }N(k);\\
N(k)\overrightarrow{\ominus }N(k-1);\\
N(k-1)\overrightarrow{\ominus }N(k);
\end{array}
\right.
\end{equation}
and
\begin{equation}\label{eqa:Mixed-coincided-algorithm-22}
N(k+2)=\left\{
\begin{array}{l}
N(k+1)\overrightarrow{\odot }N(k)\overrightarrow{\ominus }N(k-1);\\
N(k+1)\overrightarrow{\ominus }N(k)\overrightarrow{\odot }N(k-1);\\
N(k-1)\overrightarrow{\ominus }N(k)\overrightarrow{\odot }N(k+1);\\
N(k-1)\overrightarrow{\odot }N(k)\overrightarrow{\ominus }N(k+1),
\end{array}
\right.
\end{equation}
and omit writing detail expressions of such probabilistic self-similar networks. It may be interesting to uncover properties of self-similar networks made by (\ref{eqa:Mixed-coincided-algorithm-22}) and the vertex-coincided algorithms and the edge-coincided algorithms.

\subsection{Computation of parameters of self-similar networks made by the vertex-coincided algorithms and the edge-coincided algorithms}

\subsubsection{\textbf{Numbers of vertices and edges}} We compute the number $n_v(t)$ of vertices of a self-similar network $N(t)$, and the number $n_e(t)$ of edges of $N(t)$, where $N(t+1)=N(t)\overrightarrow{\odot }N(t-1)$, or $N(t+1)=N(t-1)\overrightarrow{\odot }N(t)$ with $t\geq 1$. Notice that $N(1)=N(0)\overrightarrow{\odot }N(0)$, and let $a_0=n_v(0)$ and $b_0=n_e(0)$. Thereby, $$n_v(1)=n_v(0)[n_v(0)-1]+n_v(0)=n^2_v(0)=a^2_0.$$

\begin{thm}\label{thm:vertex-coincided-algorithm-numbers}
For $N(t+1)=N(t)\overrightarrow{\odot }N(t-1)$ with $N(1)=N(0)\overrightarrow{\odot }N(0)=N(0)\overrightarrow{\odot }N(0)$ and time step $t\geq 1$, we have
\begin{equation}\label{eqa:vertex-coincided-algorithm-11}
\left\{
\begin{array}{l}
\displaystyle n_v(t+1)=n_v(t-1)n_v(t)=a_0^{r(t+1)};\\
\displaystyle n_e(t+1)=b_0\sum^{r(t)-1}_{k=0}a^k_0
\end{array}
\right.
\end{equation}
where $r(0)=1$, $r(1)=2$, $r(2)=3$, $r(3)=5$, $\dots$, $r(t+1)=r(t-1)+r(t)$, namely, $\{r(t)\}$ is a generalized Fibonacci sequence with the initial values $r(0)=1$ and $r(1)=2$.
\end{thm}
\begin{proof}
The proof of the formula (\ref{eqa:vertex-coincided-algorithm-11}) is by the mathematical induction on time step $t$. Notice that $n_v(0)=a^{r(0)}_0$, $n_v(1)=a^{r(2)}_0$, and
\begin{equation}\label{eqa:vertex-coincided algorithm-22}
n_e(1)=n_v(0)n_e(0)+n_e(0)=b_0(a_0+1).
\end{equation}

\emph{Step-1.} $t=1$. We get
$${
\begin{split}
n_v(2)&=n_v(0)[n_v(1)-1]+n_v(0)=n_v(0)n_v(1)\\
&=a_0^{r(0)+r(1)}=a_0^{r(2)}
\end{split}}
$$ since $N(2)=N(1)\overrightarrow{\odot }N(0)$. Next, we compute
$${
\begin{split}
n_e(2)&=n_v(0)n_e(1)+n_e(0)=a_0b_0(a_0+1)+b_0\\
&=b_0(a^2_0+a_0+1)=b_0\sum^{r(2)-1}_{j=0}a^j_0.
\end{split}}$$

\emph{Step-2.} Assume that the formula (\ref{eqa:vertex-coincided-algorithm-11}) holds true for $t=k$. By
$N(k+1)=N(k)\overrightarrow{\odot }N(k-1)$, we have
$${
\begin{split}
n_v(k+1)&=n_v(k-1)[n_v(k)-1]+n_v(k-1)\\
&=n_v(k-1)n_v(k)=a_0^{r(k-1)+r(k)}\\
&=a_0^{r(k+1)},
\end{split}}$$
and
$${
\begin{split}
n_e(k+1)&=n_v(k-1)n_e(k)+n_e(k-1)\\
&=a_0^{r(k-1)}n_e(k)+n_e(k-1)\\
&=a_0^{r(k-1)}b_0\sum^{r(k)-1}_{k=0}a^k_0+b_0\sum^{r(k-1)-1}_{k=0}a^k_0\\
&=b_0\sum^{r(k+1)-1}_{k=0}a^k_0.
\end{split}}$$
Thereby, we claim that the formula (\ref{eqa:vertex-coincided-algorithm-11}) holds true for $t\geq 1$ by the hypothesis of induction.
\end{proof}

\vskip 0.2cm

The number of vertices of an edge-split $n_e(t)$-scaling $M(0)$-similar network $M(t)$ is denoted as $m_v(t)$, the number of edges of the network $M(t)$ is denoted as $m_e(t)$. Let $m_v(0)=c_0$ and $m_e(0)=d_0$. We show a group of formulas to compute $m_v(t)$ and $m_e(t)$ as follows.
\begin{thm}\label{thm:edge-coincided-algorithm-numbers}
For $N(t+1)=N(t)\overrightarrow{\ominus}N(t-1)$ with $M(1)=M(0)\overrightarrow{\ominus }M(0)$ and $t\geq 1$, we have
\begin{equation}\label{eqa:edge-coincided-algorithm-11}
\left\{
\begin{array}{l}
\displaystyle m_v(t+1)=2+(c_0-2)\sum^{r(t)-1}_{j=0}a^j_0;\\
\displaystyle m_e(t+1)=m_e(t-1)m_e(t)=d_0^{r(t+1)}
\end{array}
\right.
\end{equation}
where $r(0)=1$, $r(1)=2$, $r(2)=3$, $r(3)=5$, $\dots$, $r(t+1)=r(t-1)+r(t)$, and $\{r(t)\}$ is called a generalized Fibonacci sequence with the initial values $r(0)=1$ and $r(1)=2$.
\end{thm}
\begin{proof} By induction on time step $t$. For $M(1)=M(0)\overrightarrow{\ominus }M(0)$, we compute:

$m_v(1)=d_0[c_0-2]+c_0=c_0(d_0+1)-2d_0=2+[c_0-2]\sum^{r(1)-1}_{k=0}a^k_0$ and

$m_e(1)=d_0[d_0-1]+d_0=d^2_0=d^{r(1)}_0$.

As $M(2)=M(1)\overrightarrow{\ominus }M(0)$, we have

$${
\begin{split}
m_v(2)&=d_0[m_v(1)-2]+c_0\\
&=d_0[c_0(d_0+1)-2d_0-2]+c_0\\
&=c_0(d^2_0+d_0+1)-2d_0(d_0+1)-2+2\\
&=2+[c_0-2]\sum^2_{k=0}d^k_0\\
&=2+[c_0-2]\sum^{r(2)-1}_{k=0}a^k_0
\end{split}}$$ and
$$m_e(2)=d_0[m_e(1)-1]+d_0=d^3_0=d^{r(0)+r(1)}_0=d^{r(2)}_0.$$

Assume the formula (\ref{eqa:edge-coincided-algorithm-11}) is true for $t=k$. For $N(k+1)=N(k)\overrightarrow{\ominus}N(k-1)$, then we come to compute
$${
\begin{split}
m_v(k+1)&=m_e(k-1)[m_v(k)-2]+m_v(k-1)\\
&=d_0^{r(k-1)}(c_0-2)\sum^{r(k)-1}_{j=0}a^j_0\\
&\quad +2+(c_0-2)\sum^{r(k-1)-1}_{j=0}a^j_0\\
&=2+(c_0-2)\sum^{r(k)-1}_{j=0}a^j_0
\end{split}}
$$
and
$${
\begin{split}
m_e(k+1)&=m_e(k-1)[m_e(k)-1]+m_e(k-1)\\
&=d^{r(k-1)+r(k)}_0\\
&=d^{r(k+1)}_0.
\end{split}}
$$
Thereby, the formula (\ref{eqa:edge-coincided-algorithm-11}) holds true by the hypothesis of induction.
\end{proof}

Since the edge-coincided algorithm $N(t+1)=N(t)\overrightarrow{\ominus}N(t-1)$ with $M(1)=M(0)\overrightarrow{\ominus }M(0)$ produces a set of distinct self-similar networks more than two, so we claim
\begin{equation}\label{eqa:not-unique}
N(t)\overrightarrow{\ominus}N(t-1)\neq N(t-1)\overrightarrow{\ominus}N(t).
\end{equation}

\vskip 0.2cm

It is noticeable in Theorem \ref{thm:vertex-coincided-algorithm-numbers}, the number $n_v(t)=a_0^{r(t)}$ of vertices of a self-similar network $N(t)$ and the number $m_e(t)=d_0^{r(t)}$ of edges of a self-similar network $M(t)$ for time step $t\geq 1$ obey \emph{Fibonacci's phenomena}.

\vskip 0.2cm

\subsubsection{\textbf{Uniqueness on the operation} ``$\odot$''} We show the following formula (\ref{eqa:Uniqueness-diandot-operation})
\begin{lem}\label{thm:vertex-coincided-exchanged}
By the vertex-coincided algorithm under the operation $\overrightarrow{\odot} $, we have
\begin{equation}\label{eqa:Uniqueness-diandot-operation}
N(k-1)\overrightarrow{\odot}N(k)=N(k)\overrightarrow{\odot}N(k-1)
\end{equation}
\end{lem}
\begin{proof} By induction, we have
$$N(k)=N(k-1)\overrightarrow{\odot}N(k-2)=N(k-2)\overrightarrow{\odot}N(k-1),$$
then
$${
\begin{split}
N(k)\overrightarrow{\odot}N(k-1)&=N(k-2)\overrightarrow{\odot}N(k-1)\overrightarrow{\odot}N(k-1)\\
&=N(k-1)\overrightarrow{\odot}N(k-2)\overrightarrow{\odot}N(k-1)\\
&=N(k-1)\overrightarrow{\odot}N(k),
\end{split}}
$$ so the formula (\ref{eqa:Uniqueness-diandot-operation}) is true, as we have desired.
\end{proof}

By Lemma \ref{thm:vertex-coincided-exchanged}, we have

\begin{thm}\label{thm:vertex-coincided-semi-group}
Let $\{N(0),\overrightarrow{\odot} \}$ be a set containing self-similar graphs made by the vertex-coincided algorithm under the operation $\overrightarrow{\odot} $ with respect to an initial graph $N(0)$, then $\{N(0),\overrightarrow{\odot} \}$ is a \emph{graphic semi-group}, and moreover $\{N(0),\overrightarrow{\odot} \}$ obeys the \emph{commutative law}, namely, $\{N(0),\overrightarrow{\odot} \}$ is a \emph{commutative graphic semi-group}.
\end{thm}
\begin{proof} Let $G,H,J$ are self-similar graphs with respect to a graph $N_0$. Clearly, $G\overrightarrow{\odot} H\in \{N(0),\overrightarrow{\odot} \}$, that is the closure. In the left form in the equation (\ref{eqa:dot-semi-group}) means that each vertex $w$ of $J$ is vertex-planted with a $G\overrightarrow{\odot} H$, which is equivalent to each vertex $w$ of $J$ is vertex-planted with a $H$ and then each vertex $y$ of $H\overrightarrow{\odot}J$ is vertex-planted with a $G$.
\begin{equation}\label{eqa:dot-semi-group}
(G\overrightarrow{\odot} H)\overrightarrow{\odot}J=G\overrightarrow{\odot} (H\overrightarrow{\odot}J).
\end{equation}
Thereby, we claim that the equation (\ref{eqa:dot-semi-group}) holds true. On the other hands, it is easy to see $G\overrightarrow{\odot} H=H\overrightarrow{\odot} G$, so
\begin{equation}\label{eqa:dot-semi-group-commotative}
{
\begin{split}
(G\overrightarrow{\odot} H)\overrightarrow{\odot}J&=(H\overrightarrow{\odot} G)\overrightarrow{\odot}J\\
&=H\overrightarrow{\odot} (G\overrightarrow{\odot}J)=H\overrightarrow{\odot} (J\overrightarrow{\odot}G)\\
&=(H\overrightarrow{\odot} J)\overrightarrow{\odot}G=(J\overrightarrow{\odot} H)\overrightarrow{\odot}G.
\end{split}}
\end{equation}
It has shown that the semi-group $\{N(0),\overrightarrow{\odot} \}$ obeys the commutative law.
\end{proof}

In general, we have a commutative graphic semi-group $\{S,\overrightarrow{\odot} \}$, where $S$ is the set of connected and simple graphs. Unfortunately, this commutative graphic semi-group $\{S,\overrightarrow{\odot} \}$ has infinite connected and simple graphs.

\subsubsection{\textbf{Scale-freeness of self-similar networks made by the operations} ``$\odot$'' and ``$\ominus$''} Self-similar networks have \emph{growth} and \emph{preferential attachment}, these are two important criterions for networks having scale-free structure. Let $P(k)$ be the \emph{probability distribution} of a vertex joined with $k$ vertices inn a network $N(t)$. Barab\'{a}si and Albert \cite{Barabasi-Albert1999} have shown
\begin{equation}\label{eqa:Barabasi-Albert1999}
P(k)=\textrm{Pr}(x=k)\propto k^{-\alpha}
\end{equation}
to be one of standard topological properties of a scale-free network, where $\alpha$ falls in the range $2<\alpha<3$. The criterion (\ref{eqa:Barabasi-Albert1999}) has been shown in many real networks that have scale-free behavior. Another approach proposed by Newman in \cite{M-E-J-Newman-SIAM-2003} and by Dorogovtsev \emph{et al.} \cite{Dorogovtsev-Goltsev-Mendes-2002} for verifying the scale-freeness of a network $N(t)$ is given as
\begin{equation}\label{eqa:Newman-Dorogovtsev}
P_{cum}(k)=\sum_{k'>k}\frac{N(k', t)}{n_v(t)}\sim k^{1-\gamma},
\end{equation}
where $N(k', t)$ represents the number of vertices which the degree greater than $k$ at time step $t$, and $2 < \gamma=1 + \frac{\ln 3}{\ln 2}< 3$. Dorogovtsev \emph{et al.} called $P_{cum}(k)$ as the \emph{cumulative degree distribution} of the network $N(t)$, so we obtain exact (analytical) and precise (numerical) answers
for main structural and topological characteristics of scale-free networks.

However, it is not a slight job to verify the scale-freeness of a self-similar network obtained by the vertex/edge-coincided algorithms though computation of $P(k)$ and $P_{cum}(k)$.
\begin{defn}\label{defn:any-one-vertex-edge-coincided}
$^*$ (1) There are a vertex-split $n_v(k)$-scaling $N(0)$-similar network $N(k)$ and another vertex-split $n_v(s)$-scaling $N(0)$-similar network $N(s)$ with its vertex set $V(N(s))=\{u_j:j\in[1,n_v(s)]\}$, where $n_v(s)$ is the number of vertices of $N(s)$. We make $n_v(s)$ copies $N_j(k)$ of $N(k)$ with $j\in[1,n_v(s)]$, and coincide the vertex $u_j$ of $N(s)$ with some vertex (probabilistic selected) of the $j$th copy $N_j(k)$ into one vertex for each $j\in[1,n_v(s)]$. The resultant network is denoted as $N(k)\overrightarrow{\odot }N(s)$. Clearly, $N(k)\overrightarrow{\odot }N(s)$ is a vertex-split $n_v(k)n_v(s)$-scaling $N(0)$-similar network.

(2) We use an edge-split $n_e(a)$-scaling $M(0)$-similar network $M(a)$ and another edge-split $n_e(b)$-scaling $M(0)$-similar network $M(b)$ with its edge set $E(M(b))=\{x_iy_i:i\in[1,n_e(b)]\}$, here $n_e(b)$ is the number of edges of $M(b)$, to make $M(a)\overrightarrow{\ominus }M(b)$ by making $n_e(b)$ copies $M_i(a)$ of $M(a)$ with $i\in[1,n_e(b)]$, and overlapping the edge $x_iy_i\in E(M(b))$ with some edge (probabilistic selected) of the $i$th copy $M_i(a)$ into one edge for each $i\in[1,n_e(b)]$.\qqed
\end{defn}

Definition \ref{defn:any-one-vertex-edge-coincided} can help us to construct self-similar networks having scale-free behavior, although this definition will produce deterministic scale-free networks. For example, suppose that $N(0)$ is a scale-free tree, also, a BA-model:

(1) We coincide the vertex $u_j$ of $N(s)$ with the vertex $v$ of the $j$th copy $N_j(k)$, if the degree difference $|deg_{N(s)}(u_j)-deg_{N_j(k)}(v)|$ is the smallest, into one vertex for each $j\in[1,n_v(s)]$. There resultant self-similar network $N(k)\overrightarrow{\odot }N(s)$ is a tree self-similarly with $N(0)$ and holds scale-free behavior.

(2) We coincide the vertex $u_j$ of $N(s)$ with the vertex $v$ of the $j$th copy $N_j(k)$ by $\textrm{deg}_{N(s)}(u_j)+\textrm{deg}_{N_j(k)}(v)$ obeys $\prod(k_i)=k_i/\sum_j k_j$ into one vertex for each $j\in[1,n_v(s)]$. See a scale-free self-similar tree shown in Fig.\ref{fig:1-scale-free-tree}.

\begin{figure}[h]
\centering
\includegraphics[width=8.2cm]{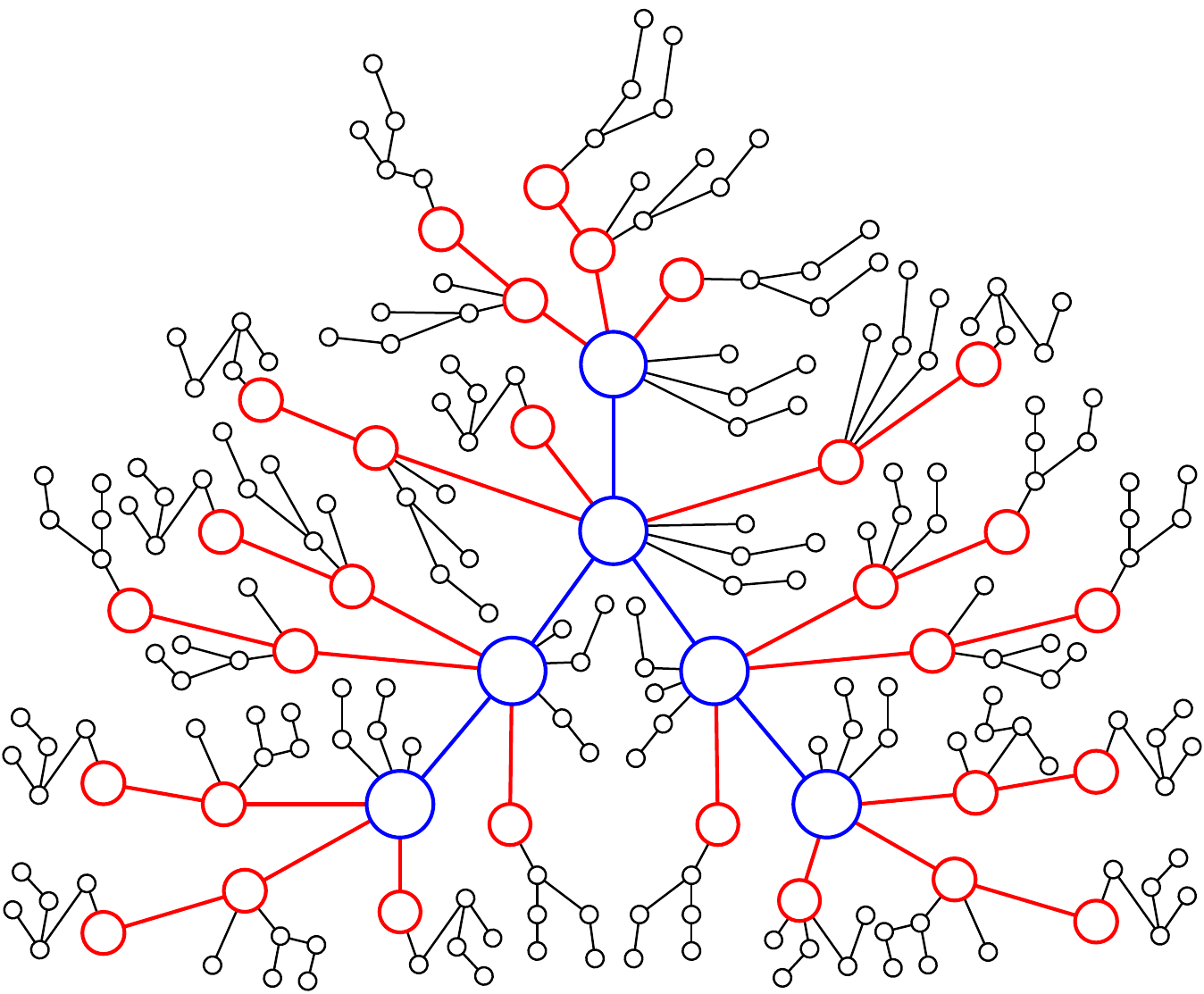}\\
\caption{\label{fig:1-scale-free-tree} {\small A scale-free self-similar Hanzi-tree made by two self-similar Hanzi-trees shown in Fig.\ref{fig:vertex-coincided-algorithm-b} (a)$=T_{4043}$ and (b).}}
\end{figure}

We will apply a technique introduced in \cite{Ma-Wang-Yao-2019}, since it can produce Fibonacci series trees obeying power-law $k^{-\gamma}$ with $\gamma>3$. Let $F(0)$ be a Hanzi-network. We present an algorithm for constructing \emph{Fibonacci self-similar Hanzi-networks} having scale-free features by the following rules:

\vskip 0.4cm

\begin{asparaenum}[Rule-1. ]
\item The vertex-planting operation: The sentence ``vertex-planting disjoint graphs $H_1,H_2,\dots, H_s$ on a vertex $u$ of a given graph $T$'' means that we coincide a vertex $v_i$ of $H_i$ with the vertex $u$ into one vertex for $i\in [1,s]$, such that the resultant graph $H_{u}=\{H_i\}^s_1\overrightarrow{\odot }T$ holds $V(H_{u})=V(T)\bigcup \big\{\bigcup ^s_{i=1}[ V(H_i)\setminus \{v_i\}]\big\}$ and $E(H_{u})=E(T)\bigcup \big[\bigcup ^s_{i=1} E(H_i)\big]$. Notice that $[V(H_i)\setminus \{v_i\}]\cap [V(H_j)\setminus \{v_j\}]=\emptyset$ for $i\neq j$.
\item The edge-planting operation: The sentence ``edge-planting disjoint graphs $H_1,H_2,\dots, H_s$ on a vertex $u$ of a given graph $T$'' means that we join a vertex $v_i$ of $H_i$ with the vertex $u$ by an new edge $uv_i$ for $i\in [1,s]$, such that the resultant graph $H_{vu}=\{H_i\}^s_1\overrightarrow{\ominus}T$ holds $V(H_{vu})=V(T)\bigcup \big[\bigcup ^s_{i=1} V(H_i)\big]$ and $E(H_{vu})=E(T)\bigcup \{uv_i:i\in [1,s]\}\bigcup \big[\bigcup ^s_{i=1} E(H_i)\big]$.
\end{asparaenum}

\vskip 0.4cm

Let $\{F_n\}$ be Fibonacci sequence with $F_1=1$, $F_2=1$ and $F_{k+1}=F_{k-1}+F_{k}$ for $k\geq 2$. Fibonacci self-similar Hanzi-networks $N(t)$ with time step $t\geq 0$ are obtained by the following FIBONACCI-VERTEX algorithm. Since we will use the famous Fibonacci sequence $\{F_n\}$, for distinguishing, we write ``$\overrightarrow{\odot}$'' by ``$\overrightarrow{\odot_F}$'', and ``$\overrightarrow{\ominus}$'' by ``$\overrightarrow{\ominus\odot_F}$'' in the following process of making Fibonacci self-similar Hanzi-networks.

\vskip 0.2cm

\textbf{FIBONACCI-VERTEX algorithm based on the vertex-planting operation:}

\emph{Step 0.} Let $C_0$ be the set of copies $H_i$ of an initial Hanzi-network $N(0)$ with a \emph{unique active vertex} $u_1$, such that each copy $H_i$ has a \emph{unique active vertex} $u^i_1$, which is the image of $u_1$. We define a level-label function to label the vertices of the initial Hanzi-network $N(0)$ with $f(v_{0,a})=0a$ for $v_{0,a}\in V(N(0))$ with $a\in [1,n_v(0)]$, where $n_v(0)$ is the number of vertices of $N(0)$.

\emph{Step 1.} A Hanzi-network $N(1)=N(0)\overrightarrow{\odot_F}N(0)$ is obtained by doing the vertex-planting a graph $H_i\in C_0$ on each vertex $x$ of the initial Hanzi-network $N(0)$, such that the unique active vertex $u^i_1$ of $H_i$ is coincided with $x$ into one vertex. So, we can divide the vertex set of $N(1)$ as $V(N(1))=V_0\bigcup V_1$ and $V_0\cap V_1=\emptyset$ with $V_0=V(N(0))$, and furthermore we label each vertex $x_{1,i}$ with $f(x_{1,i})=1i$ for $x_{1,i}\in V_1$.

\emph{Step 2.} A Hanzi-network $N(2)=N(0)\overrightarrow{\odot_F}N(1)$ is made in the following way: doing the vertex-planting a graph $H_i\in C_0$ on each vertex $x$ of the Hanzi-network $N(1)$ by coinciding the unique active vertex $u^i_1$ of $H_i$ with $x$ into one vertex. And then, $V(N(2))=V(N(1))\bigcup V_2$, where $V(N(1))=V_0\bigcup V_1$, we label each vertex $x_{2,i}$ with $f(x_{2,i})=2i$ for $x_{2,i}\in V_2$.

A Hanzi-network $N(3)=N(0)\overrightarrow{\odot_F}N(2)$ is constructed by doing the vertex-planting disjoint graph $H_1,H_2,\dots H_{F_{3-i}}\in C_0$ on each vertex $x\in V_i$ by coinciding the unique active vertex $u^a_0$ of $H_a$ with $x$ into one vertex, where $a\in[1,3-i]$ and $i\in [0,2]$. Thereby, $V(N(3))=V(N(2))\bigcup V_3$, where $V(N(2))=V_0\bigcup V_1\bigcup V_2$, we label each vertex $x_{3,i}$ with $f(x_{3,i})=3i$ for $x_{3,i}\in V_3$.

\emph{Step 3.} Each Hanzi-network $N(k+1)=N(0)\overrightarrow{\odot_F}N(k)$ is obtained by implementing the vertex-planting disjoint graph $H_1,H_2,\dots H_{F_{k-i}}\in C_0$ on each vertex $x\in V_i$ by coinciding the unique active vertex $u^a_0$ of $H_a$ with $x$ into one vertex, where $a\in[1,k-i]$ and $i\in [0,k-1]$. Thereby, $V(N(k+1))=V(N(k))\bigcup V_{k+1}$, where $V(N(k))=\bigcup ^k_{j=0} V_j$, we label each vertex $x_{3,i}\in V_{k+1}$ with $f(x_{k+1,i})=(k+1)i$.

\vskip 0.2cm

In Fig.\ref{fig:Fibonacci-self-similar}, we give a vertex $02$ that is planted with the initial network shown in Fig.\ref{fig:Fibonacci-self-similar} (d). By FIBONACCI-VERTEX algorithm, we have: (a) the vertex $02$ is \emph{vertex-planted} with $F_1=1$ initial network in $N(1)$; (b) the vertex $02$ is vertex-planted with $F_2=2$ initial network in $N(2)$; (c) the vertex $02$ is vertex-planted with $F_3=2$ initial networks in $N(3)$; and we claim that the vertex $02$ is vertex-planted with $F_k$ initial networks in $N(k)$.

\begin{figure}[h]
\centering
\includegraphics[height=8.8cm]{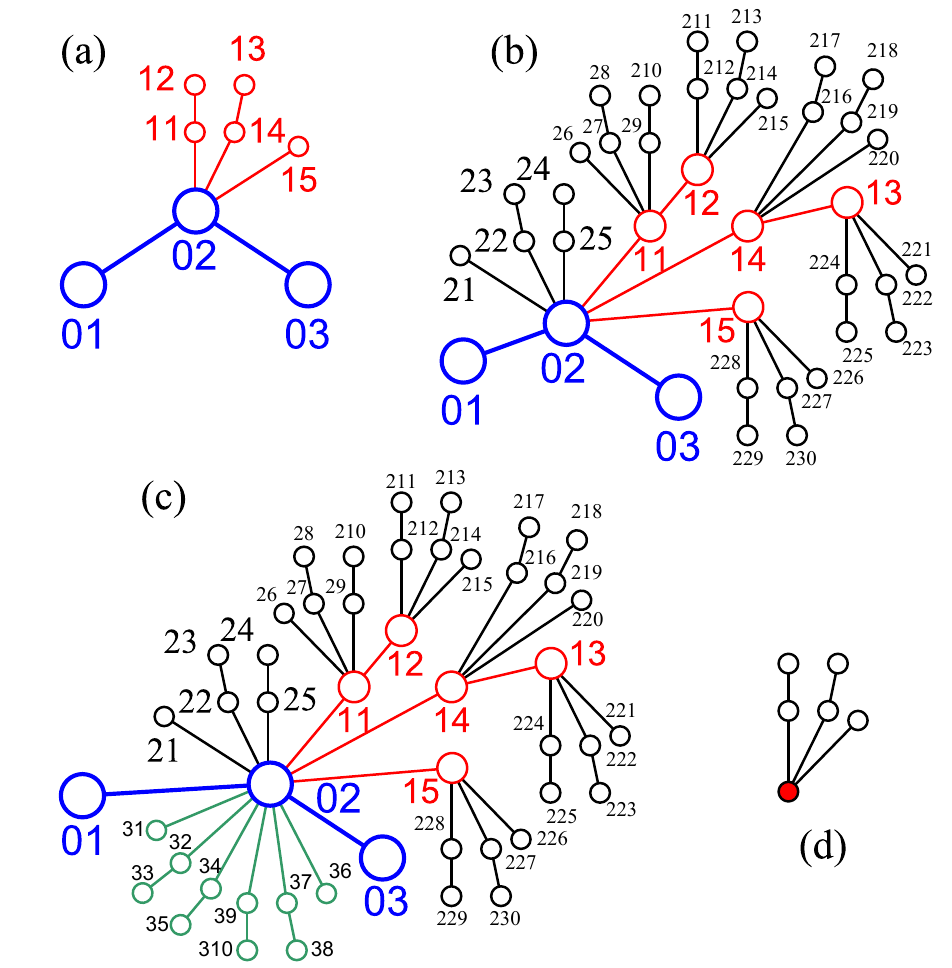}\\
\caption{\label{fig:Fibonacci-self-similar} {\small (a) $N(1)$; (b) $N(2)$; (c) $N(3)$; (d) $N(0)$.}}
\end{figure}

Fibonacci self-similar Hanzi-networks obtained by FIBONACCI-EDGE algorithm based on the edge-planting operation are denoted as $N'(t)$ with time step $t\geq 0$. Here we omit the definition of \emph{FIBONACCI-EDGE algorithm} since it is very similar with \emph{FIBONACCI-VERTEX algorithm}.

We will show some properties of Fibonacci self-similar Hanzi-networks $N(t)$ made by FIBONACCI-VERTEX algorithm. Let $n_v(t)$ and $n_e(t)$ be the numbers of vertices and edges of Fibonacci self-similar Hanzi-networks $N(t)$ with time step $t\geq 0$, respectively.

$V(N(0))=\{u_i:~i\in [1,n_v(0)]\}$, we rewrite $n_v(0)$ and $n_e(0)$ by $\alpha_0$ and $\beta_0$ respectively, and each degree $deg_{N(0)}(u_i)$ is simply written as $d_{0,i}$ with $i\in [1,n_v(0)]$.
\begin{lem}\label{thm:Fibonacci-self-similar Hanzi-networks-vv}
Let $N(t)$ with $t\geq 1$ be Fibonacci self-similar Hanzi-networks defined by FIBONACCI-VERTEX algorithm based on the vertex-planting operation, then the number $n_v(t)$ of vertices and the number $n_e(t)$ of edges of $N(t)$ are computed by
\begin{equation}\label{eqa:Fibonacci-self-similar-Hanzi-networks-ve}
\left\{
\begin{array}{ll}
\displaystyle n_v(t)=n_v(t-1)+(\alpha_0-1)\sum^{t-1}_{j=0}F_{t-j}\cdot |V_j|;\\
\displaystyle n_e(t)=n_e(t-1)+\beta_0\sum^{t-1}_{j=0}F_{t-j}\cdot |V_j|.
\end{array}
\right.
\end{equation}
\end{lem}
\begin{proof} The proof is by induction on time step $t$. For $t=1$, we have

$n_v(1)=n_v(0)+F_{1}\cdot |V_0|\cdot [n_v(0)-1]=n_v(0)+(\alpha_0-1)F_{1}\cdot |V_0|$, and

$n_e(1)=\beta_0+|V_0|\cdot F_1\beta_0$.

Since $V(N(2))=\bigcup^2_{i=0}V_i$, we get

$n_v(2)=n_v(1)+F_{1}\cdot |V_1|\cdot (\alpha_0-1)+F_{2}\cdot |V_0|(\alpha_0-1)$, and

$n_e(2)=n_e(1)+|V_1|\cdot F_1\beta_0+|V_0|\cdot F_2\beta_0=\beta_0+|V_0|\beta_0+|V_1|\beta_0+|V_0|\beta_0$.

Suppose that the formula (\ref{eqa:Fibonacci-self-similar-Hanzi-networks-ve}) are true for $t=k$.

As $t=k+1$, since $N(k+1)=N(0)\overrightarrow{\odot_F}N(k)$ is obtained in Step 3 of FIBONACCI-VERTEX algorithm based on the vertex-planting operation. We have $V(N(k+1))=\bigcup ^{k+1}_{j=0} V_j$, where $V(N(k))=\bigcup ^{k}_{j=0} V_j$, and the vertices of $V_{k+1}$ are newly added into $N(k)$. Thereby, each vertex $x\in V_i$ with $i\in[0,k]$ is planted with disjoint graph $H_1,H_2,\dots H_{F_{k-i}}\in C_0$, thus we have $F_{k-i}\cdot (\alpha_0-1)|V_i|$ new vertices added into $N(k)$, and furthermore we have $|V_{k+1}|=(\alpha_0-1)\sum^{k}_{j=0}F_{k-j}\cdot |V_j|$ new vertices added into $N(k)$ in total.

We come to consider the set $E_{k+1}$ of edges newly added into $N(k)$. Notice that $H_i\cong N(0)$ for each $H_i\in C_0$, each vertex $x\in V_i$ is related with $F_{k-i}\cdot \beta_0$, so the set $V_i$ distributes us  $F_{k-i}\cdot \beta_0\cdot |V_i|$ new edges, finally,
$$|E_{k+1}|=\beta_0\sum^{k}_{j=0}F_{k-j}\cdot |V_j|.$$ Because of $n_e(k+1)=n_e(k)+|E_{k+1}|$, we have shown the formula (\ref{eqa:Fibonacci-self-similar-Hanzi-networks-ve}).
\end{proof}

\begin{lem}\label{thm:Fibonacci-self-similar Hanzi-networks-ee}
Let $N'(t)$ with $t\geq 1$ be Fibonacci self-similar Hanzi-networks defined by FIBONACCI-EDGE algorithm based on the edhe-planting operation, then the number $n'_v(t)$ of vertices and the number $n'_e(t)$ of edges of $N'(t)$ can be computed as
\begin{equation}\label{eqa:Fibonacci-self-similar-Hanzi-networks-ee-ve}
\left\{
\begin{array}{ll}
\displaystyle n'_v(t)=n'_v(t-1)+\alpha_0\sum^{t-1}_{j=0}F_{t-j}\cdot |V_j|;\\
\displaystyle n'_e(t)=n'_e(t-1)+(\beta_0+1)\sum^{t-1}_{j=0}F_{t-j}\cdot |V_j|.
\end{array}
\right.
\end{equation}
\end{lem}

A vertex $02$ shown in Fig.\ref{fig:Fibonacci-self-similar-edge} is planted with disjoint initial networks $N'(0)$ shown in Fig.\ref{fig:Fibonacci-self-similar-edge} (d). One, according to FIBONACCI-EDGE algorithm, can see: (a) the vertex $02$ is \emph{edge-planted} with $F_1=1$ initial network in $N'(1)=N'(0)\overrightarrow{\ominus\odot_F}N'(0)$; (b) the vertex $02$ is edge-planted with $F_2=2$ initial network in $N'(2)=N'(0)\overrightarrow{\ominus\odot_F}N'(1)$; (c) the vertex $02$ is edge-planted with $F_3=2$ initial networks in $N'(3)=N'(0)\overrightarrow{\ominus\odot_F}N'(2)$. Thereby, the vertex $02$ is edge-planted with $F_k$ initial networks in $N'(k)$. In general, we have $N'(k+1)=N'(0)\overrightarrow{\ominus\odot_F} N'(k)$.

\begin{figure}[h]
\centering
\includegraphics[height=10cm]{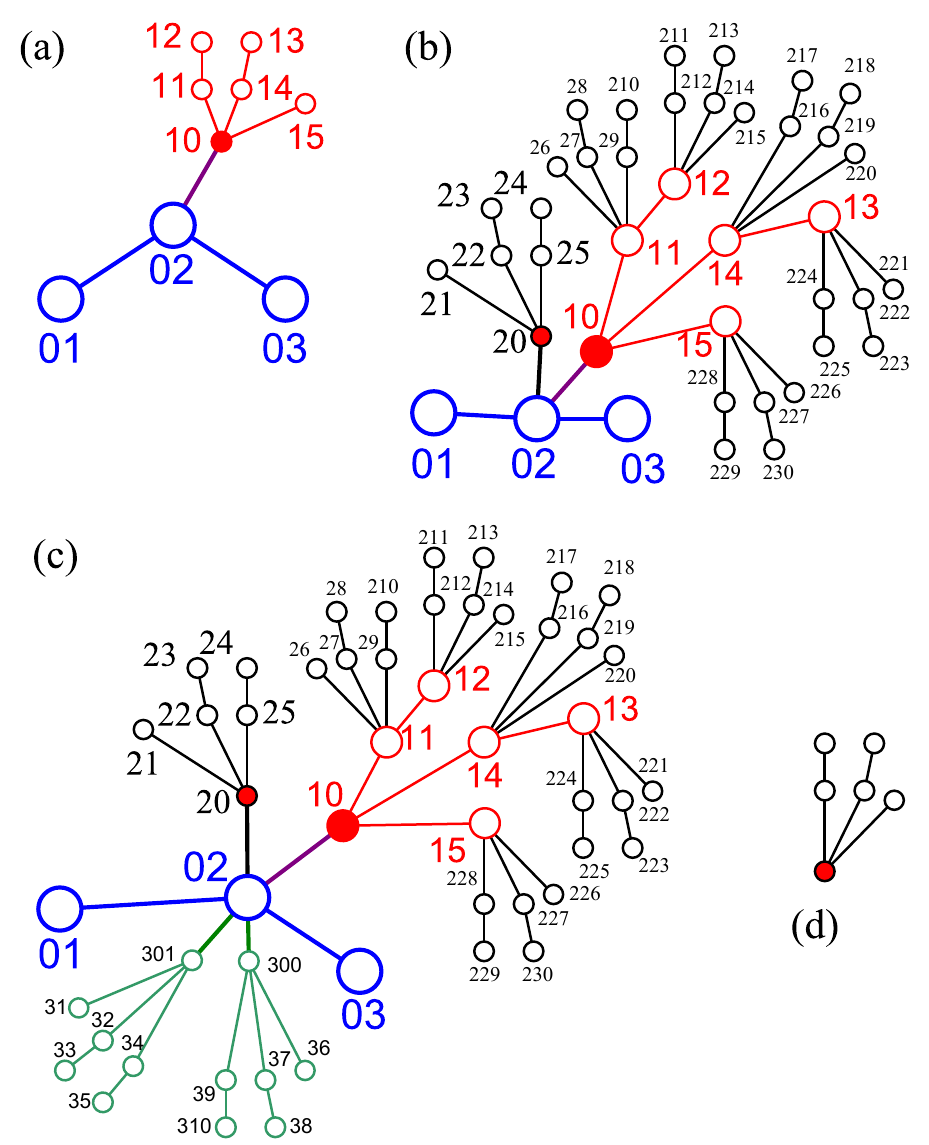}\\
\caption{\label{fig:Fibonacci-self-similar-edge} {\small (a) $N'(1)$; (b) $N'(2)$; (c) $N'(3)$; (d) $N'(0)$.}}
\end{figure}

Combining the formula (\ref{eqa:Fibonacci-self-similar-Hanzi-networks-ve}) and the formula (\ref{eqa:Fibonacci-self-similar-Hanzi-networks-ee-ve}) together produces the following formula
\begin{equation}\label{eqa:Fibonacci-self-similar-Hanzi-networks-ve-00}
\left\{
\begin{array}{ll}
\displaystyle n_v(t)=(\alpha_0-1)M(t,k);\\
\displaystyle n_e(t)=\beta_0M(t,k).
\end{array}
\right.
\end{equation}
and
\begin{equation}\label{eqa:Fibonacci-self-similar-Hanzi-networks-ee-ve-00}
\left\{
\begin{array}{ll}
\displaystyle n'_v(t)=\alpha_0M(t,k);\\
\displaystyle n'_e(t)=(\beta_0+1)M(t,k).
\end{array}
\right.
\end{equation}
where $M(t,k)=\sum^{t-1}_{k=1}\sum^{k-1}_{j=0}F_{k-j}\cdot |V_j|$. Moreover, the formula (\ref{eqa:Fibonacci-self-similar-Hanzi-networks-ve}) and the formula (\ref{eqa:Fibonacci-self-similar-Hanzi-networks-ee-ve}) give us
\begin{equation}\label{eqa:0000}
\left\{
\begin{array}{ll}
\displaystyle n'_v(t)-n_v(t)=n'_v(t-1)-n_v(t-1)+F\bullet V;\\
\displaystyle n'_e(t)-n_e(t)=n'_e(t-1)-n_e(t-1)+F\bullet V.
\end{array}
\right.
\end{equation}
where the vector dot-product
$${
\begin{split}
F\bullet V&=(F_{1}~F_{2}~\cdots ~F_{t})\bullet (|V_{t-1}|~|V_{t-2}|~\cdots ~|V_0|)^{-1}\\
&=\sum^{t-1}_{j=0}F_{t-j}\cdot |V_j|,
\end{split}}
$$ it is not related with the numbers of vertices and edges of the initial network $N(0)$ at all.

The scale-freeness of two types of Fibonacci self-similar Hanzi-networks $N(t)$ and $N'(t)$ (concluded the network models introduced in\cite{Ma-Wang-Yao-2019}) can be verified by the formula (\ref{eqa:Newman-Dorogovtsev}) (Ref. \cite{M-E-J-Newman-SIAM-2003, Dorogovtsev-Goltsev-Mendes-2002}), which is a powerful tool in investigation of scale-free networks. Other topological properties of Fibonacci self-similar Hanzi-networks $N(t)$ and $N'(t)$ are: Degree distribution, Cumulative degree distribution, Correlation coefficient, Clustering coefficient, Diameter, Average distance, Small-world, and so on.

\subsection{Encrypting Hanzi-networks}

Encrypting a network is new in our knowledge. Here, our `encrypting network' differs greatly from `network encryption' of network science.

\subsubsection{Strategy of encrypting networks} First of all, we think of encrypting larger scale of networks with computer, not manually. A computer encryption is cryptographic hash function that is used in `Blockchain' invented by Satoshi Nakamoto in 2008 to serve as the public transaction ledger of the cryptocurrency bitcoin.

We assign a private key for each vertex (node) of a network $N(t)$, and distribute an authentication of each edge (link) of $N(t)$ at time step $t$, such that an end node $u$ of an edge $uv$ can communicate with another end node $v$ of the edge $uv$ though the edge $uv$. We write this encryption system by $E_n(N(t))$ at time step $t$. Clearly, $E_n(N(t_1))\neq E_n(N(t_2))$ if two time steps $t_1\neq t_2$.

Second, we have known the following cipher codes:

1. The authors, in \cite{Sun-Zhang-Zhao-Yao-2017, Yao-Sun-Zhao-Li-Yan-2017, Yao-Mu-Sun-Zhang-Wang-Su-Ma-IAEAC-2018, Yao-Sun-Zhang-Mu-Sun-Wang-Su-Zhang-Yang-Yang-2018arXiv, Yao-Zhang-Sun-Mu-Sun-Wang-Wang-Ma-Su-Yang-Yang-Zhang-2018arXiv}, have introduced Topsnut-gpws with variable labellings related with large scale of \emph{Abelian groups}, also \emph{graphic groups}, for exploring the encryption of dynamic networks.

2. Another way is to use commutative graphic semi-groups $\{S,\overrightarrow{\odot} \}$ to implement encryption of self-similar Hanzi-networks, such that the resultant networks are `\emph{super-networks}'' or ``\emph{hypergraphs}''. We divide $S$ into two subsets $S(V)$ and $S(E)$, and use elements of $S(V)$ to replace the vertices of a self-similar Hanzi-network $N(t)$, in other word, we have a mapping $\theta: V(N(t))\rightarrow S(V)$, and each edge $uv\in E(N(t))$ is substituted by $\theta(uv)=\theta(u)\overrightarrow{\odot} \theta(v)\in S(E)$, and then join a vertex of $\theta(u)$ with some vertex of $\theta(uv)$ by an edge, and join a vertex of $\theta(uv)$ with some vertex of $\theta(v)$ by another edge; the resultant network is denoted as $N_{super}(t)=\langle N(t); \{S,\overrightarrow{\odot} \}\rangle$. Clearly, there are two or more super-networks $N_{super}(t)$, and each one has more vertices and edges for opposing attachers.

3. We propose using Hanzis (not Hanzi-graphs and Hanzi-gpws) to encrypting a network. For example, a group of Hanzis $H_{4043}$, $H_{4043}$, $H_{2635}$, $H_{2511}$, $H_{5282}$, $H_{4476}$, $H_{4734}$, $H_{4411}$, $H_{3829}$, we have used here, distributes us $9!=362,880$ sentences, so they can be candidate cipher codes for encrypting smaller networks. Chinese couplets are naturally public keys and private keys mentioned in Section II.

Since networks, very often, contain thousand and thousand of vertices and edges, we think that ``graphic groups'' made by Topsnut-gpws is better than commutative graphic semi-groups $\{S,\overrightarrow{\odot} \}$ and Hanzis in terms of memory, convenience and large quantity.

However, we are facing a difficult problem: How to construct quickly large scale of Topsnut-gpws with thousand and thousand of vertices? As known, the number of elements in a graphic group made by a Topsnut-gpw $G$ is equal to the number of vertices of graphical configuration of $G$ or the range size of the labelling $f$ of $G$. In other words, we will make a (compound) Topsnut-gpw $G(t)$ such that the number of vertices of $G(t)$ equals the number of nodes in a network $N(t)$ to be encrypted at time step $t$.

\subsubsection{Every-zero graphic groups, every-zero string groups} By means of graceful labelling, we restate graphic group definition and properties as follows (Ref. \cite{Sun-Zhang-Zhao-Yao-2017, Yao-Sun-Zhao-Li-Yan-2017, Yao-Mu-Sun-Zhang-Wang-Su-Ma-IAEAC-2018, Yao-Sun-Zhang-Mu-Sun-Wang-Su-Zhang-Yang-Yang-2018arXiv, Yao-Zhang-Sun-Mu-Sun-Wang-Wang-Ma-Su-Yang-Yang-Zhang-2018arXiv}):

Let a Topsnut-gpw $G$ be made by a $(p,q)$-graph $G$ with a graceful labelling $f: V(G)\rightarrow [0,q]$. We set each Topsnut-gpw $G_i$ such that (i) $G_i$ is a copy of $G$, namely, $G_i\cong G$, so $V(G_i)=V(G)$; and (ii) $G_i$ admits a labelling $f_i: V(G_i)\rightarrow [0,q]$ defined by $f_i(x)=f(x)+i-1~(\bmod~q)$ with $i\in [1,q]$. Then the set $F_f(G)=\{G_1,G_2,\dots, G_q\}$ and an additive operation $\oplus$ form a \emph{graphic group} $\{F_f(G),\oplus\}$, also, called an \emph{Ablian additive group}, where the additive operation $G_i\oplus G_j$ is defined as follows:
\begin{equation}\label{eqa:basic-operation}
[f_i(x)+f_j(x)-f_k(x)]~(\bmod~q)=f_{\lambda}(x)
\end{equation}where $\lambda=i+j-k~(\bmod~q)$ for $x\in V(G)$, and $G_i,G_j\in F_f(G)$, and a fixed element $G_k\in F_f(G)$. Notice that $\{F_f(G),\oplus\}$ holds
\begin{asparaenum}[(1) ]
\item \emph{Zero}. $G_k$ is the \emph{zero} such that $G_i\oplus G_k=G_i$.

\item \emph{Closure and uniqueness}. If $G_i\oplus G_j=G_s$, $G_i\oplus G_j=G_t$, then $s=t$. And $G_i\oplus G_j\in \{F_f(G),\oplus\}$.

\item \emph{Inverse}. Each $G_i$ has its own \emph{inverse} $G_{i^{-1}}$ such that $G_i\oplus G_{i^{-1}}=G_k$ determined by $f_i(x)+f_{i^{-1}}(x)=2f_{k}(x)$ for $x\in V(G)$.

\item \emph{Associative law}. $G_i\oplus [G_j\oplus G_s]=[G_i\oplus G_j]\oplus G_s$.
\end{asparaenum}

Moreover, this graphic group $\{F_f(G),\oplus\}$ has the following properties:

\begin{asparaenum}[GP-1. ]
\item \emph{Every-zero}. Each element of $\{F_f(G),\oplus\}$ can be as ``zero'' by the additive operation $\oplus $ defined in (\ref{eqa:basic-operation}), also, we call $\{F_f(G),\oplus\}$ an \emph{every-zero graphic group}.
\item \emph{Commutative}. $G_i\oplus G_j=G_j\oplus G_i$, so the every-zero graphic group $\{F_f(G),\oplus\}$ is \emph{commutative}.
\item If $G_s=G_i\oplus G_j$, then $G_{s^{-1}}=G_{i^{-1}}\oplus G_{j^{-1}}=G_{j^{-1}}\oplus G_{i^{-1}}$.
\item If $q$ is even, then there exists an element $G_{i_0}$ to be itself inverse, that is, $G_{i_0}\oplus G_{i_0}$ equals to the zero.
\item Each $G_i$ has its own \emph{inverse} $G_{i+2j}$ based on the zero $G_{i+j}$, since $f_i(x)+f_{i+2j}(x)=2f_{i+j}(x)$. Thereby, each $G_i$ has \emph{different inverses} based on \emph{different zeros}.
\end{asparaenum}

As examples, a tree $G$ is a $(13,12)$-graph with an \emph{odd-even separable 6C-labelling} $f$ (Ref. \cite{Yao-Sun-Zhang-Mu-Sun-Wang-Su-Zhang-Yang-Yang-2018arXiv}). By this $(13,12)$-graph $G$, we construct an \emph{every-zero graphic group} $G_{roup}=\{F^{odd}_f(G)\cup F^{even}_f(G),\oplus\}$ in the following: In Fig.\ref{fig:6C-group-odd}), the Topsnut-gpw set $F^{odd}_f(G)=\{G_{1},G_{3},\dots ,G_{25}\}$ holds $G_{2i-1}\cong G$ and admits a labelling $f_{2i-1}$ with $i\in [1,13]$, and each $f_{2i-1}$ is defined by $f_{2i-1}(x)=f(x)+2(i-1)~(\bmod~25)$ for $x\in V(G)$ and $f_{2i-1}(xy)=f(xy)+2(i-1)~(\bmod~24)$ for $xy\in E(G)$. In Fig.\ref{fig:6C-group-even}, we have another Topsnut-gpw set $F^{even}_f(G)=\{G_{2},G_{4},\dots ,G_{26}\}$ has each element $G_{2i}$ holding $G_{2i}\cong G$ and admitting a labelling $f_{2i}$ with $i\in [1,13]$, where each $f_{2i}$ is defined by $f_{2i}(u)=f(u)+(2i-1)~(\bmod~25)$ for $u\in V(G)$ and $f_{2i}(uv)=f(uv)+(2i-1)~(\bmod~24)$ for $uv\in E(G)$. Thereby, we have
\begin{equation}\label{eqa:c3xxxxx}
[f_a(x)+f_b(x)-f_c(x)]~(\bmod~25)=f_{\lambda}(x)
\end{equation}
where $\lambda=a+b-c~(\bmod~25)$ for $x\in V(G)$, and $a,b,c\in [1,26]$.

Furthermore, each element $G_{2i-1}$ is a Topsnut-gpw and distributes us a TB-paw $D_{2i-1}$ as
$${
\begin{split}
D_{2i-1}=&f_{2i-1}(x_1)f_{2i-1}(x_1x_2)f_{2i-1}(x_2)f_{2i-1}(x_2x_3)\\
&f_{2i-1}(x_3)f_{2i-1}(x_3x_4)f_{2i-1}(x_4)f_{2i-1}(x_3x_5)\\
&f_{2i-1}(x_5)f_{2i-1}(x_3x_6)f_{2i-1}(x_6)f_{2i-1}(x_3x_7)\\
&f_{2i-1}(x_7)f_{2i-1}(x_7x_3)f_{2i-1}(x_3)f_{2i-1}(x_7x_8)\\
&f_{2i-1}(x_8)f_{2i-1}(x_8x_7)f_{2i-1}(x_7)f_{2i-1}(x_8x_{11})\\
&f_{2i-1}(x_{11})f_{2i-1}(x_{11}x_{12})f_{2i-1}(x_{12})\\
&f_{2i-1}(x_{11}x_{13})f_{2i-1}(x_{13})
\end{split}}
$$ with $i\in [1,13]$. Similarly, the string set $\{D_{2i}:~i\in [1,13]\}$ is defined by
$${
\begin{split}
D_{2i}=&f_{2i}(x_1)f_{2i}(x_1x_2)f_{2i}(x_2)f_{2i}(x_2x_3)f_{2i}(x_3)f_{2i}(x_3x_4)\\
&f_{2i}(x_4)f_{2i}(x_3x_5)f_{2i}(x_5)f_{2i}(x_3x_6)f_{2i}(x_6)f_{2i}(x_3x_7)\\
&f_{2i}(x_7)f_{2i}(x_7x_3)f_{2i}(x_3)f_{2i}(x_7x_8)f_{2i}(x_8)f_{2i}(x_8x_7)\\
&f_{2i}(x_7)f_{2i}(x_8x_{11})f_{2i}(x_{11})f_{2i}(x_{11}x_{12})\\
&f_{2i}(x_{12})f_{2i}(x_{11}x_{13})f_{2i}(x_{13}).
\end{split}}
$$ It is not hard to verify that $\{D_{2i-1}:~i\in [1,13]\}\cup \{D_{2i}:~i\in [1,13]\}$ forms an \emph{every-zero string group} under the operation $\oplus $.

On the other hands, we say an element $G_k\in F^{odd}_f(G)\cup F^{even}_f(G)$ to be the \emph{zero}, this particular group is denoted as $G^k_{roup}=\{F^{odd}_f(G)\cup F^{even}_f(G),\oplus_k\}$ with $k\in [1,26]$. Observe the every-zero graphic group $G_{roup}=\{F^{odd}_f(G)\cup F^{even}_f(G),\oplus\}$ again, we can discover that
\begin{equation}\label{eqa:c3xxxxx}
[_{a'}(uv)+f_{b'}(uv)-f_{c'}(uv)]~(\bmod~26)=f_{a'+b'-c'}(uv)
\end{equation}
where $\lambda'=a+b-c~(\bmod~26)$ for $uv \in E(G)$, and $a',b',c'\in [1,26]$. So, $G_{roup}=\{F^{odd}_f(G)\cup F^{even}_f(G),\oplus\}$ is an \emph{edge-every-zero graphic group}.

\begin{figure}[h]
\centering
\includegraphics[width=8.2cm]{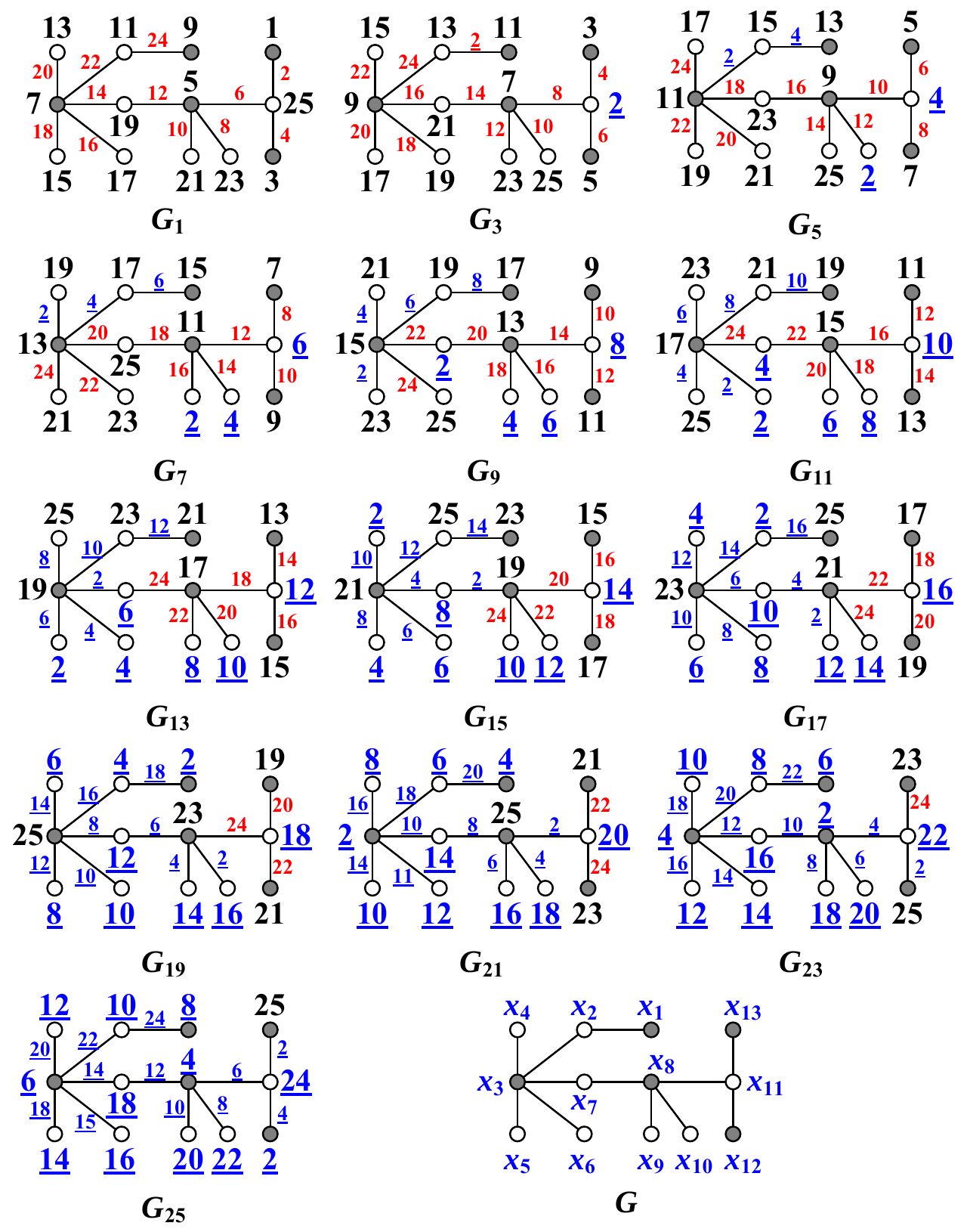}\\
\caption{\label{fig:6C-group-odd} {\small A Topsnut-gpw set $F^{odd}_f(G)$.}}
\end{figure}

\begin{figure}[h]
\centering
\includegraphics[width=8.2cm]{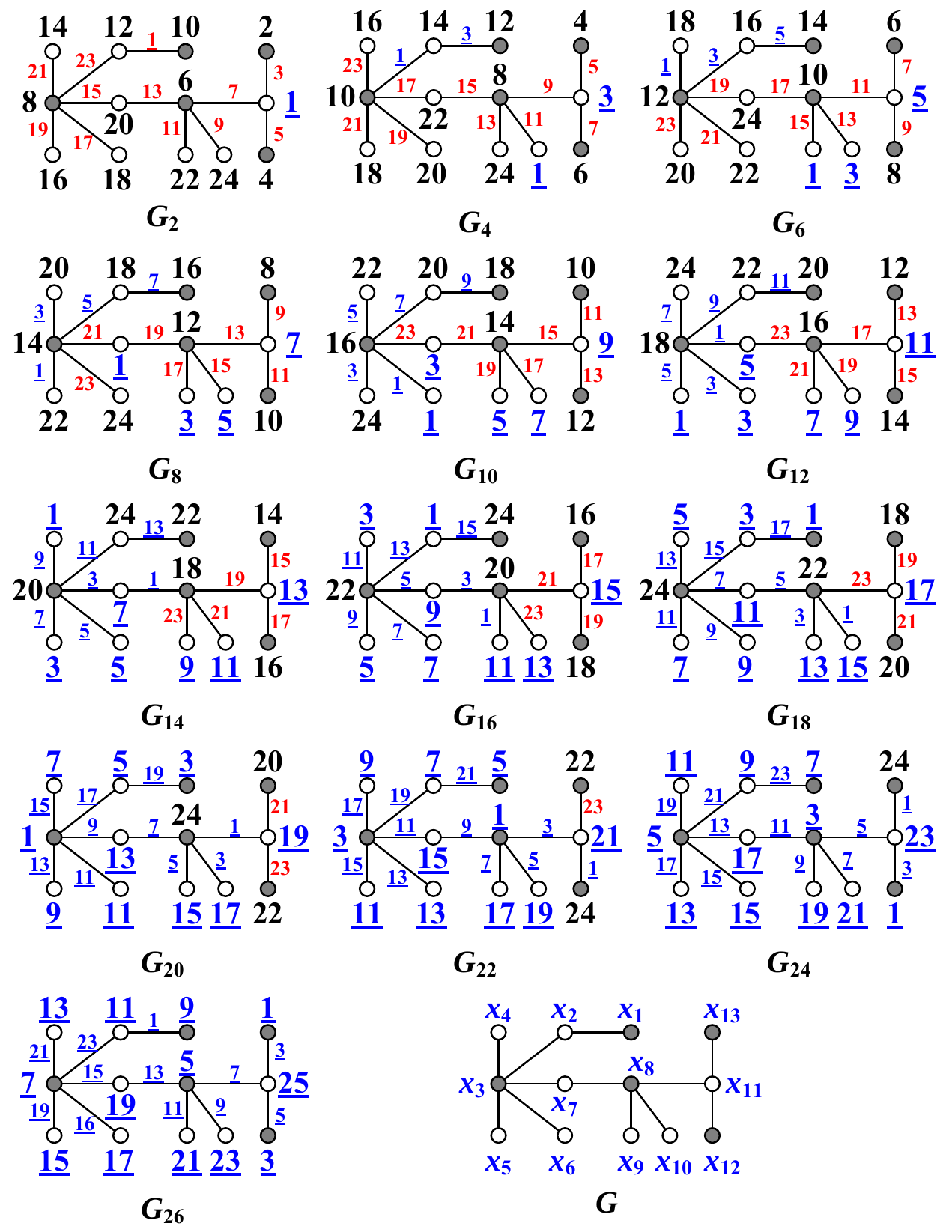}\\
\caption{\label{fig:6C-group-even} {\small Another Topsnut-gpw set $F^{odd}_f(G)$ differing from one shown in Fig.\ref{fig:6C-group-odd}.}}
\end{figure}

Obviously, an every-zero graphic group $\{F_f(G),\oplus\}$, very often, is made by a simple graph $G$ (may be disconnected) admitting a (flawed) graph labelling $f$ (Ref. \cite{Gallian2018}, and see for examples shown in Fig.\ref{fig:tian-group-formulae}).

\subsubsection{Encrypting Hanzi-networks by every-zero graphic/string groups}

Let $N(t)$ be a Hanzi-network, and let $G_{roup}=\{F_f(G),\oplus\}$ be an every-zero graphic group or an edge-every-zero graphic group. We use a mapping $F$ to assign each edge $uv\in E(N(t))$ and its two ends $u,v$ with three Topsnut-gpws $G_i,G_j,G_{i+i-k}\in G_{roup}$, such that $F(u)=G_i$, $F(v)=G_j$, $F(uv)=G_{i+i-k}$ holding $G_i\oplus_{k_{uv}}G_j=G_{i+i-k}~(\bmod~q)$, also, $F(u)\oplus_{k_{uv}}F(v)=F(uv)$, where $k_{uv}$ is an index of an element $G_{k_{uv}}\in G_{roup}$. Thereby, we have
\begin{equation}\label{eqa:c3xxxxx}
[f_{i}(x)+f_{j}(x)-f_{k_{uv}}(x)]~(\bmod~q)=f_{\mu}(x)
\end{equation}where $\mu=i+j-k_{uv}~(\bmod~q)$ for each vertex $x\in V(G)$. For distinguishing, we restrict $F(u)\oplus_{k_{uv}}F(v)=F(uv)\neq F(uw)=F(u)\oplus_{k_{uw}}F(w)$ by selecting suitable indices $k_{uv}$ and $k_{uw}$.

\begin{thm}\label{thm:encrypting-Hanzi-network-11}
If $N(t)$ is a tree-like Hanzi-network, and an every-zero graphic group $\{F_f(G),\oplus\}$ has $q$ Topsnut-gpws such that $q$ is not less than the number of neighbors of any vertex of $N(t)$, then we can have a mapping $F$ to encrypt $N(t)$ with $\{F_f(G),\oplus\}$ such that $F(uv)\neq F(uw)$, that is,
\begin{equation}\label{eqa:graphic-labelling-distinguishing}
F(u)\oplus_{k_{uv}}F(v)\neq F(u)\oplus_{k_{uw}}F(w)
\end{equation} for any pair of adjacent edges $uv$ and $uw$ of $N(t)$.
\end{thm}

\begin{thm}\label{thm:encrypting-Hanzi-network-22}
If a Hanzi-network $N(t)$ has its cardinality not greater than that of an every-zero graphic group $\{F_f(G),\oplus\}$, then we have a labelling $F$ to encrypt $N(t)$ with $\{F_f(G),\oplus\}$ such that
\begin{equation}\label{eqa:graphic-labelling-distinguishing-1}
F(uv)=F(u)\oplus_{k}F(v)\neq F(u)\oplus_{k}F(w)=F(uw)
\end{equation} for any pair of edges $uv$ and $uw$ of $N(t)$, where $k$ is a fixed index of some Topsnut-gpw $G_k\in F_f(G)$.
\end{thm}

\begin{figure}[h]
\centering
\includegraphics[width=8.2cm]{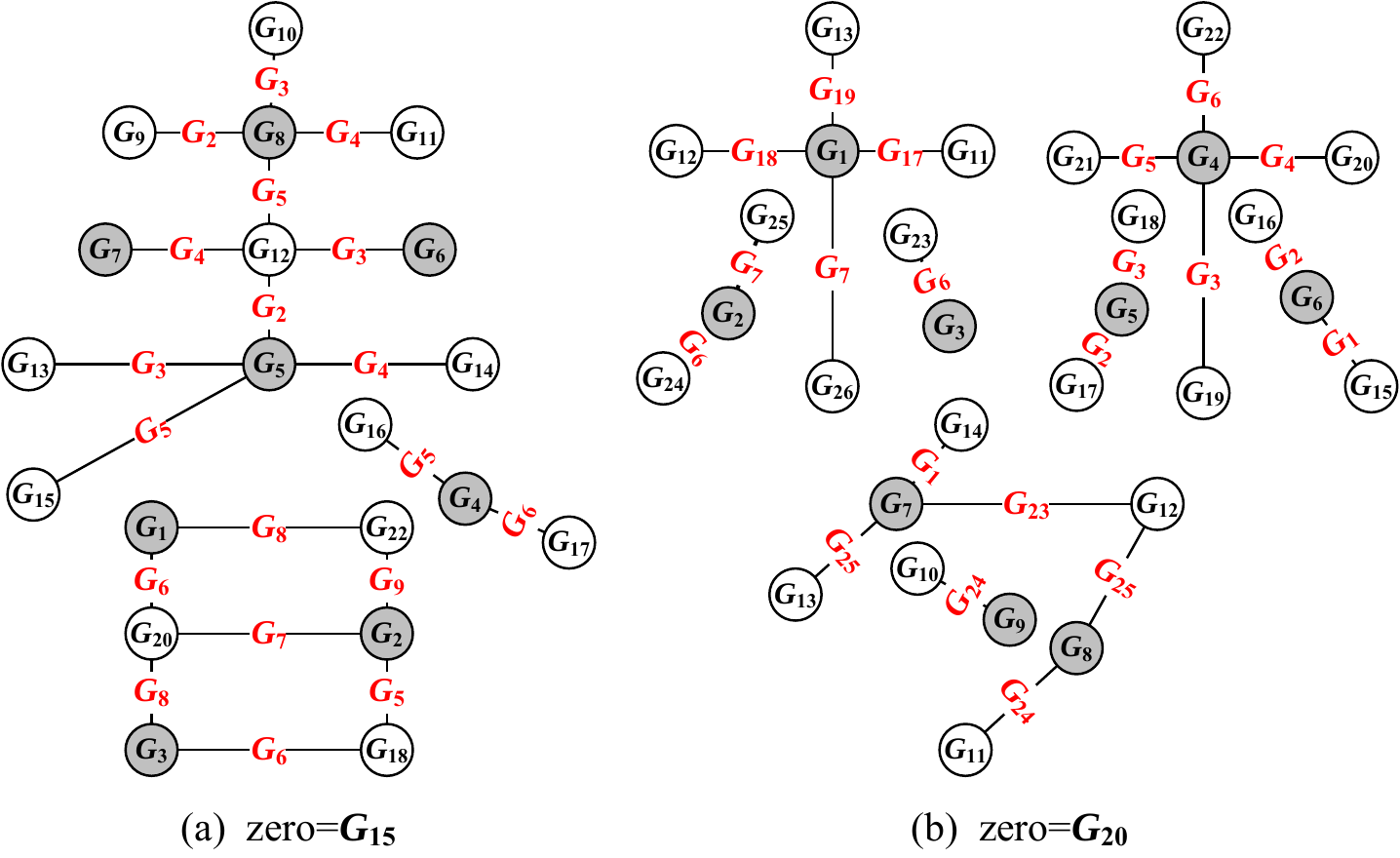}\\
\caption{\label{fig:6C-group-chunmeng} {\small (a) A Hanzi-graph $T_{2026}$ is labelled with an every-zero graphic group $G^{15}_{roup}=\{F^{odd}_f(G)\cup F^{even}_f(G),\oplus_k\}$ under the zero $G_{15}$; (b) a Hanzi-graph $T_{3546}$ is labelled with an every-zero graphic group $G^{20}_{roup}=\{F^{odd}_f(G)\cup F^{even}_f(G),\oplus_k\}$ under the zero $G_{20}$.}}
\end{figure}

By observing Fig.\ref{fig:6C-group-chunmeng} carefully, the indices of Topsnut-gpw labellings of vertices of the Hanzi-graph $T_{2026}$ forms a flawed graceful labelling. In Fig.\ref{fig:6C-group-chunmeng-edge}, the edges of two Hanzi-graphs $T_{2026}$ and $T_{3546}$ are labelled with two edge-every-zero graphic groups under two edge-zeros $G_{15}$ and $G_{20}$, respectively. However, (a)=(c), (b)$\neq $(d) after comparing Fig.\ref{fig:6C-group-chunmeng} and Fig.\ref{fig:6C-group-chunmeng-edge}. Clearly, the above examples inspire us to do more works on the encryption of Hanzi-networks by graphic groups.

\begin{figure}[h]
\centering
\includegraphics[width=8.2cm]{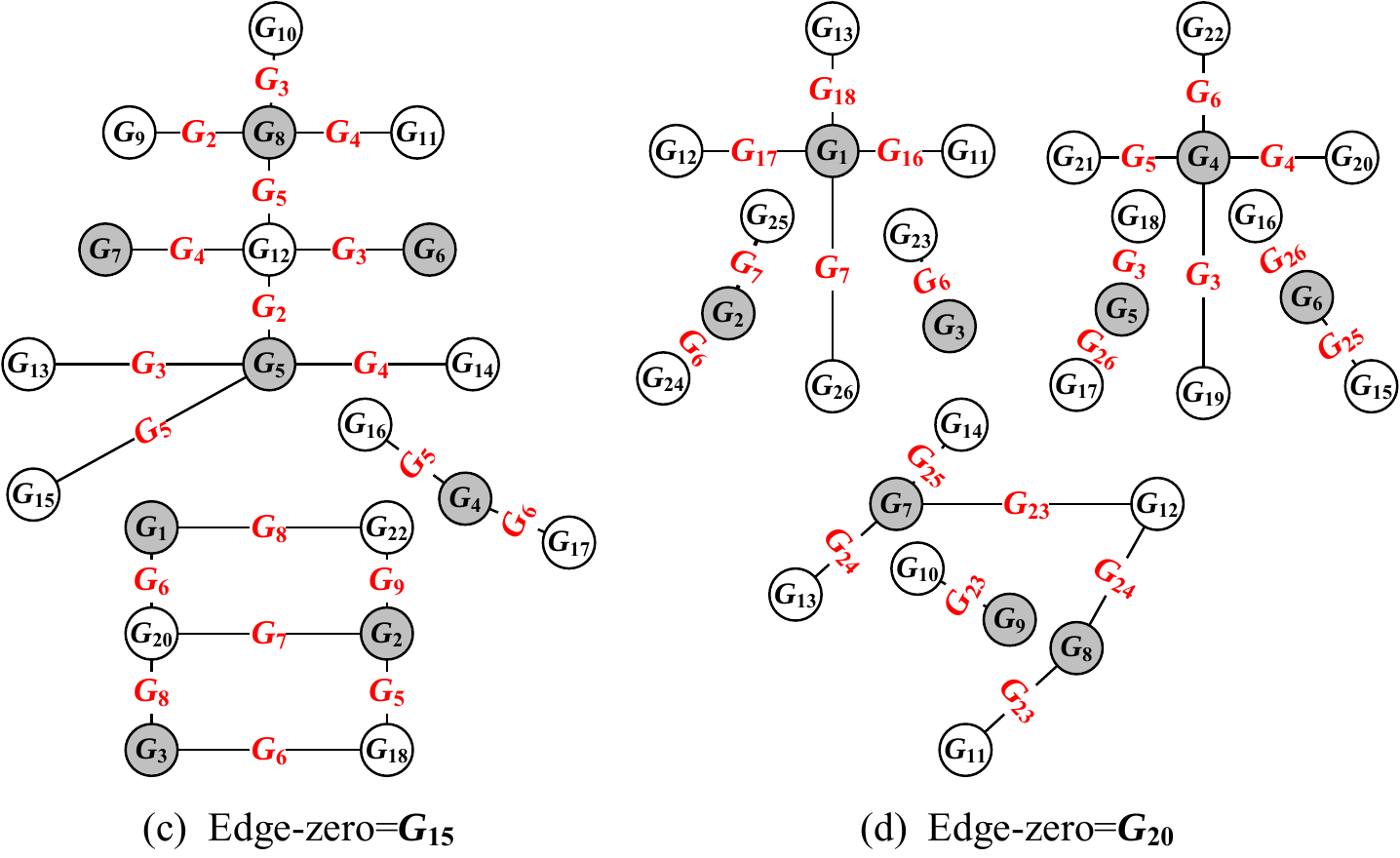}\\
\caption{\label{fig:6C-group-chunmeng-edge} {\small (c) A Hanzi-graph $T_{2026}$ is labelled with an edge-every-zero graphic group $G^{15}_{roup}=\{F^{odd}_f(G)\cup F^{even}_f(G),\oplus_k\}$ under the edge-zero $G_{15}$; (d) a Hanzi-graph $T_{3546}$ is labelled with an edge-every-zero graphic group $G^{20}_{roup}=\{F^{odd}_f(G)\cup F^{even}_f(G),\oplus_k\}$ under the edge-zero $G_{20}$.}}
\end{figure}

\subsubsection{Encrypting networks constructed by communities (blocks, subgraphs)}

Suppose that $N(t)$ is a $k$-partite network, that is, $V(N(t))=\bigcup ^k_{i=1}V_i$, each $V_i$ induces a subgraph $T_i$ such that $V(T_i)=V_i$, let $E_{i,j}$ be the set of edges in which each edge has one end in $T_i$ and another end in $T_j$. Thereby, we have the edge set $E(N(t))=\big(\bigcup^k_{i=1}E(T_i)\big)\bigcup \big(\bigcup_{i\neq j}E_{i,j}\big)$. Notice that each $E_{i,j}$ can induce a subgraph $H_{i,j}$ with $E(H_{i,j})=E_{i,j}$, so $N(t)=\big(\bigcup^k_{i=1}T_i\big)\bigcup \big(\bigcup_{i\neq j}H_{i,j}\big)$, and each $H_{i,j}$ joins $G_i$ and $G_j$ together. In general, each subgraph $T_i$ is a \emph{community} in $N(t)$.

We take an every-zero graphic group $\{F_f(G),\oplus\}$ with $F_f(G)=\{G_i\}^n_1$ and $G$ is a $(p,q)$-graph, such that, $n\geq k$ and $n\geq \max\{|V(T_i)|:i\in[1,k]\}$ to encrypt $N(t)$ in the way: each subgraph $T_i$ is encrypted with an every-zero graphic group $\{F_f(G),\oplus_{k_{i}}\}$ under the zero $G_{k_{i}}$, and each subgraph $H_{i,j}$ is encrypted with an every-zero graphic group $\{F_f(G),\oplus_{k_{i,j}}\}$ under the zero $G_{k_{i,j}}$, where $G_{k_{i,j}}\neq G_{k_{i}}$ and $G_{k_{i,j}}\neq G_{k_{j}}$ since $n\geq k$. See an illustration shown in Fig.\ref{fig:community}.

\begin{figure}[h]
\centering
\includegraphics[width=8cm]{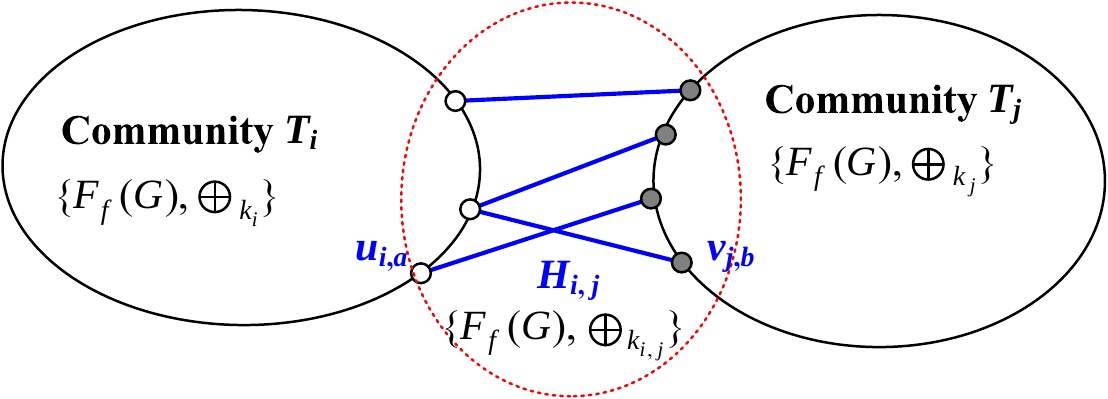}\\
\caption{\label{fig:community} {\small A scheme for encrypting communities.}}
\end{figure}

We have ${n \choose k}$ groups of $G_{i_1},G_{i_2},\dots, G_{i_k}$ as zeros for encrypting communities $T_1,T_2,\dots ,T_k$, where $i_1,i_2,\dots, i_k$ is a permutation of $k$ number of $1,2,\dots ,n $. Thereby, $N(t)$ can be encrypted by at least ${n \choose k}k!$ methods.

Suppose that the graph $G$ in an every-zero graphic group $\{F_f(G),\oplus\}$ admits $m$ labellings $f^{(l)}$ with $l\in [1,m]$. If these labellings $f^{(1)},f^{(2)}, \dots, f^{(m)}$ are equivalent to each other, that is, there exists a one-one mapping $\theta_{a,b}$ for any pair of $f^{(a)}$ and $f^{(b)}$ such that
\begin{equation}\label{eqa:equivalent-mapping}
f^{(a)}(x)=\theta_{a,b}(f^{(b)}(x)),~ f^{(b)}(x)=\theta^{-1}_{a,b}(f^{(a)}(x))
\end{equation} for $x\in V(G)$. Then we get $m$ every-zero graphic groups $\{F_{f^{(l)}}(G),\oplus\}$ with $l\in [1,m]$. Each community $T_i$ of $N(t)$ can be encrypted by $\{F_{f^{(l)}}(G),\oplus_{l,m_i}\}$ with $l\in [1,m]$ under the zero $G_{l,m_i}$ admitting the labelling $f^{(l)}$. By the equivalent transformation (\ref{eqa:equivalent-mapping}), we are easy to build up an equivalent transformation between two every-zero graphic groups $\{F_{f^{(a)}}(G),\oplus_{a,m_i}\}$ and $\{F_{f^{(b)}}(G),\oplus_{b,m_i}\}$, namely, we can change quickly Topsnut-gpws for the network $N(t)$.

If $G$ admits $c(l)$ labellings like $f^{(l)}$, then by the above deduction, $N(t)$ can be encrypted by approaches with number
\begin{equation}\label{eqa:encrypt-methods-number}
N_{encrypt}(N(t))\geq m{n \choose k}k!\sum ^k_{l=1}c(l)
\end{equation} for supporting us to encrypt the network $N(t)$. On the other hand, there are many $(p,q)$-graphs $G$ admitting labellings $f^{(l)}$ with $l\in [1,m]$, which means we have many every-zero graphic groups $\{F_{f^{(l)}}(G),\oplus\}$ for encrypting the network $N(t)$. There are more caterpillars with $p$ vertices, and they admit at least over 25 equivalent labellings (see Table-2 shown in Appendix B).

\subsubsection{Definitions of graphic/string group colorings} Using every-zero graphic groups and edge-every-zero graphic groups to encrypting networks are similar with graph colorings.

\begin{defn} \label{defn:graphic-group-coloring}
$^*$ Let $H$ be a $(p,q)$-graph, and $\{F_f(G),\oplus\}$ be an every-zero graphic group containing Topsnut-gpws $G_i$ with $i\in [1,M_G]$. Suppose that $H$ admits a mapping $\theta: S\rightarrow F_f(G)$ with $S\subseteq V(H)\cup E(H)$. Write $\theta(S)=\{\theta(w):w\in S\}$, $\theta(N_{ei}(u))=\{\theta(v):v\in N_{ei}(u)\}$ and $\theta(N'_{ei}(u))=\{\theta(uv):v\in N_{ei}(u)\}$. There are the following so-called \emph{graphic group constraints}:
\begin{asparaenum}[Gg-1. ]
\item \label{prob:only-vertex} $S=V(H)$.
\item \label{prob:only-edge} $S=E(H)$.
\item \label{prob:total} $S=V(H)\cup E(H)$.
\item \label{prob:vertex-not} $\theta(u)\neq \theta(v)$ for $v\in N_{ei}(u)$.
\item \label{prob:edge-not} $\theta(uv)\neq \theta(uw)$ for $v,w\in N_{ei}(u)$.
\item \label{prob:vertex-distinghuishing} $\theta(N_{ei}(x))\neq \theta(N_{ei}(y))$ for any pair of vertices $x,y\in V(H)$.
\item \label{prob:adjacent-vertex-distinghuishing} $\theta(N_{ei}(u))\neq \theta(N_{ei}(v))$ for each edge $uv\in E(H)$.
\item \label{prob:vertex-edge-distinghuishing} $\theta(N'_{ei}(x))\neq \theta(N'_{ei}(y))$ for any pair of vertices $x,y\in V(H)$.
\item \label{prob:adjacent-vertex-edge-distinghuishing} $\theta(N'_{ei}(u))\neq \theta(N'_{ei}(v))$ for each edge $uv\in E(H)$.
\item \label{prob:adjacent-vertex-equitable} $\big ||\theta(N_{ei}(u))|-|\theta(N_{ei}(v))|\big |\leq 1$ for each edge $uv\in E(H)$.
\item \label{prob:adjacent-edge-equitable} $\big ||\theta(N'_{ei}(u))|-|\theta(N'_{ei}(v))|\big |\leq 1$ for each edge $uv\in E(H)$.
\item \label{prob:adjacent-total-distinghuishing} $\{\theta(u)\}\cup \theta(N_{ei}(u))\cup \theta(N'_{ei}(u))\neq \{\theta(v)\}\cup \theta(N_{ei}(v))\cup \theta(N'_{ei}(v))$ for each edge $uv\in E(H)$.
\item \label{prob:induced-edge} There exists an index $k_{uv}$ such that $\theta(uv)=\theta(u)\oplus _{k_{uv}}\theta(v)\in F_f(G)$ for each edge $uv\in E(H)$.
\item \label{prob:induced-vertex} Each vertex $w$ is assigned with a set $\{\theta(ww_i)\oplus _{k_{ij}}\theta(ww_j): w_i,w_j\in N_{ei}(w)\}$.
\end{asparaenum}

Then we call $\theta$ as: (1) a \emph{proper gg-coloring} if Gg-\ref{prob:only-vertex} and Gg-\ref{prob:vertex-not} hold true; (2) a \emph{proper edge-gg-coloring} if Gg-\ref{prob:only-edge} and Gg-\ref{prob:edge-not} hold true; (3) a \emph{proper total gg-coloring} if Gg-\ref{prob:total}, Gg-\ref{prob:vertex-not} and Gg-\ref{prob:edge-not} hold true; (4) a \emph{vertex distinguishing proper gg-coloring} if Gg-\ref{prob:only-vertex}, Gg-\ref{prob:vertex-not} and Gg-\ref{prob:vertex-distinghuishing} hold true; (5) an \emph{adjacent vertex distinguishing proper gg-coloring} if Gg-\ref{prob:only-vertex}, Gg-\ref{prob:vertex-not} and Gg-\ref{prob:adjacent-vertex-distinghuishing} hold true; (6) an \emph{edge distinguishing proper gg-coloring} if Gg-\ref{prob:only-edge}, Gg-\ref{prob:edge-not} and Gg-\ref{prob:vertex-edge-distinghuishing} hold true; (7) an \emph{adjacent edge distinguishing proper gg-coloring} if Gg-\ref{prob:only-edge}, Gg-\ref{prob:edge-not} and Gg-\ref{prob:adjacent-vertex-edge-distinghuishing} hold true; (8) an \emph{equitable adjacent-v proper gg-coloring} if Gg-\ref{prob:only-vertex}, Gg-\ref{prob:vertex-not} and Gg-\ref{prob:adjacent-vertex-equitable} hold true; (9) an \emph{equitable adjacent-e proper gg-coloring} if Gg-\ref{prob:only-edge}, Gg-\ref{prob:edge-not} and Gg-\ref{prob:adjacent-edge-equitable} hold true; (10) an \emph{adjacent total distinguishing proper gg-coloring} if Gg-\ref{prob:total}, Gg-\ref{prob:vertex-not}, Gg-\ref{prob:edge-not} and Gg-\ref{prob:adjacent-total-distinghuishing} hold true; (11) a \emph{v-induced total proper gg-coloring} if Gg-\ref{prob:only-vertex}, Gg-\ref{prob:vertex-not} and Gg-\ref{prob:induced-edge} hold true; (12) an \emph{induced e-proper v-set gg-coloring} if Gg-\ref{prob:only-edge}, Gg-\ref{prob:edge-not} and Gg-\ref{prob:induced-vertex} hold true.\qqed
\end{defn}

Based on Definition \ref{defn:graphic-group-coloring}, we have new parameters:
\begin{asparaenum}[New-1. ]
\item $\chi _{gg}(H)$ is the minimum number of $M_G$ Topsnut-gpws $G_i$ in some every-zero graphic group $\{F_f(G),\oplus\}$ for which $H$ admits a \emph{proper gg-coloring}. Bruce Reed in 1998 conjectured that $\chi(H)\leq \lceil \frac{1}{2}[\Delta(H)+1+K(H)]\rceil$ (Ref. \cite{Bondy-2008}), where $\chi _{gg}(H)=\chi(H)$.
\item $\chi' _{gg}(H)$ is the minimum number of $M_G$ Topsnut-gpws $G_i$ in some every-zero graphic group $\{F_f(G),\oplus\}$ for which $H$ admits a \emph{proper edge-gg-coloring}. We have $\Delta(H)\leq \chi'(H)\leq \Delta(H)+1$ (Vadim G. Vizing, 1964; Guppta, 1966 \cite{Bondy-2008}), where $\chi' _{gg}(H)=\chi'(H)$.

\item $\chi'' _{gg}(H)$ is the minimum number of $M_G$ Topsnut-gpws $G_i$ in some every-zero graphic group $\{F_f(G),\oplus\}$ for which $H$ admits a \emph{proper total gg-coloring}. It was conjectured $\Delta(H)+1\leq \chi''(H)\leq \Delta(H)+2$ (Behzad, 1965; Vadim G. Vizing, 1964 \cite{Bondy-2008}), where $\chi'' _{gg}(H)=\chi''(H)$.

\item $\chi _{ggs}(H)$ is the minimum number of $M_G$ Topsnut-gpws $G_i$ in some every-zero graphic group $\{F_f(G),\oplus\}$ for which $H$ admits a \emph{vertex distinguishing proper gg-coloring}.

\item $\chi _{ggas}(H)$ is the minimum number of $M_G$ Topsnut-gpws $G_i$ in some every-zero graphic group $\{F_f(G),\oplus\}$ for which $H$ admits an \emph{adjacent vertex distinguishing proper gg-coloring}.

\item $\chi' _{ggs}(H)$ is the minimum number of $M_G$ Topsnut-gpws $G_i$ in some every-zero graphic group $\{F_f(G),\oplus\}$ for which $H$ admits an \emph{edge distinguishing proper gg-coloring}.

\item $\chi' _{ggas}(H)$ is the minimum number of $M_G$ Topsnut-gpws $G_i$ in some every-zero graphic group $\{F_f(G),\oplus\}$ for which $H$ admits an \emph{adjacent edge distinguishing proper gg-coloring}. We have a conjecture $\chi' _{as}(H)\leq \Delta(H)+2$ by Zhang Zhongfu, Liu Linzhong, Wang Jianfang, 2002 \cite{Zhang-Liu-Wang-2002-strong}, where $\chi' _{ggas}(H)=\chi' _{as}(H)$.
\item $\chi _{ggeq}(H)$ is the minimum number of $M_G$ Topsnut-gpws $G_i$ in some every-zero graphic group $\{F_f(G),\oplus\}$ for which $H$ admits an \emph{equitable adjacent-v proper gg-coloring}.
\item $\chi' _{ggeq}(H)$ is the minimum number of $M_G$ Topsnut-gpws $G_i$ in some every-zero graphic group $\{F_f(G),\oplus\}$ for which $H$ admits an \emph{equitable adjacent-e proper gg-coloring}.
\item $\chi'' _{ggas}(H)$ is the minimum number of $M_G$ Topsnut-gpws $G_i$ in some every-zero graphic group $\{F_f(G),\oplus\}$ for which $H$ admits an \emph{adjacent total distinguishing proper gg-coloring}.
\end{asparaenum}

\vskip 0.4cm

\subsubsection{Transformation between six every-zero graphic groups}

Let $T_{4476}=G,O,M,L,E,C$ in the following discussion. Based on Hanzi-graph $T_{4476}$ and some graph labellings, we have the following six every-zero graphic groups $\{F_{f_X}(X),\oplus\}$ with $X=G,O,M,L,E,C$:

\begin{asparaenum}[Group-1. ]
\item $\{F_{f_G}(G),\oplus\}$ with \emph{flawed set-ordered graceful labellings} $f^{(i)}_G$ with $i\in [1,11]$, where $G=G_1$ shown in Fig.\ref{fig:tian-group-formulae} (a).
\item $\{F_{f_O}(O),\oplus\}$ with\emph{ flawed set-ordered odd-graceful labellings} $f^{(j)}_O$ with $j\in [1,21]$, where $O=O_1$ shown in Fig.\ref{fig:tian-group-formulae} (b).
\item $\{F_{f_M}(M),\oplus\}$ with \emph{flawed set-ordered edge-magic total labellings} $f^{(k)}_M$ with $k\in [1,21]$, where $M=M_1$ shown in Fig.\ref{fig:tian-group-formulae} (c).
\item $\{F_{f_L}(L),\oplus\}$ with \emph{flawed set-ordered odd-even separable edge-magic total labellings} $f^{(i)}_L$ with $i\in [1,21]$, where $L=L_1$ shown in Fig.\ref{fig:tian-group-formulae} (d).
\item $\{F_{f_E}(E),\oplus\}$ with \emph{flawed set-ordered odd-elegant labellings} $f^{(j)}_E$ with $j\in [1,11]$, where $E=E_1$ shown in Fig.\ref{fig:tian-group-formulae} (e).
\item $\{F_{f_C}(C),\oplus\}$ with \emph{flawed set-ordered odd-elegant labellings} $f^{(k)}_C$ with $k\in [1,20]$, where $C=C_1$ shown in Fig.\ref{fig:tian-group-formulae} (f).
\end{asparaenum}

\begin{figure}[h]
\centering
\includegraphics[width=7.2cm]{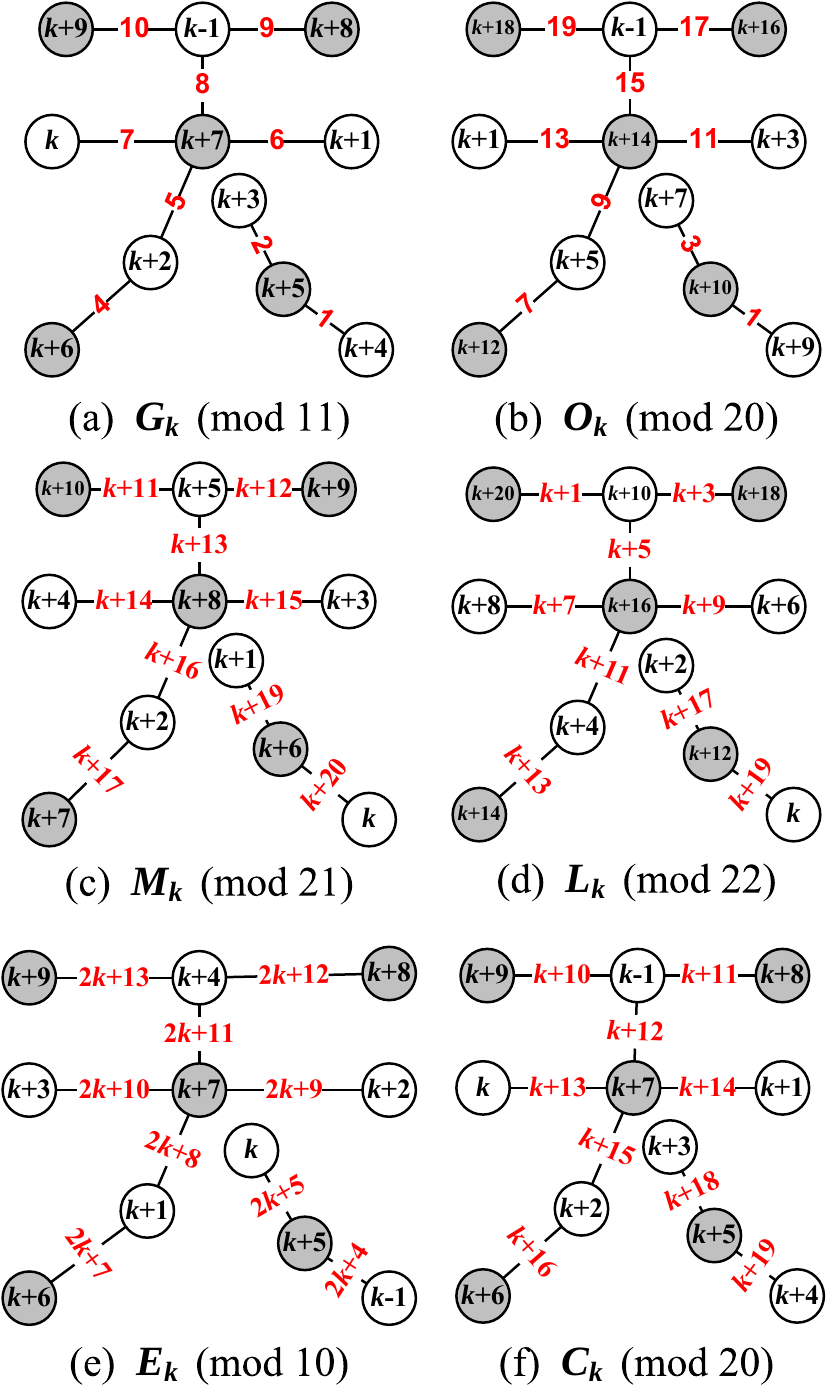}\\
\caption{\label{fig:tian-group-formulae} {\small Six every-zero graphic groups.}}
\end{figure}

\begin{figure}[h]
\centering
\includegraphics[width=8cm]{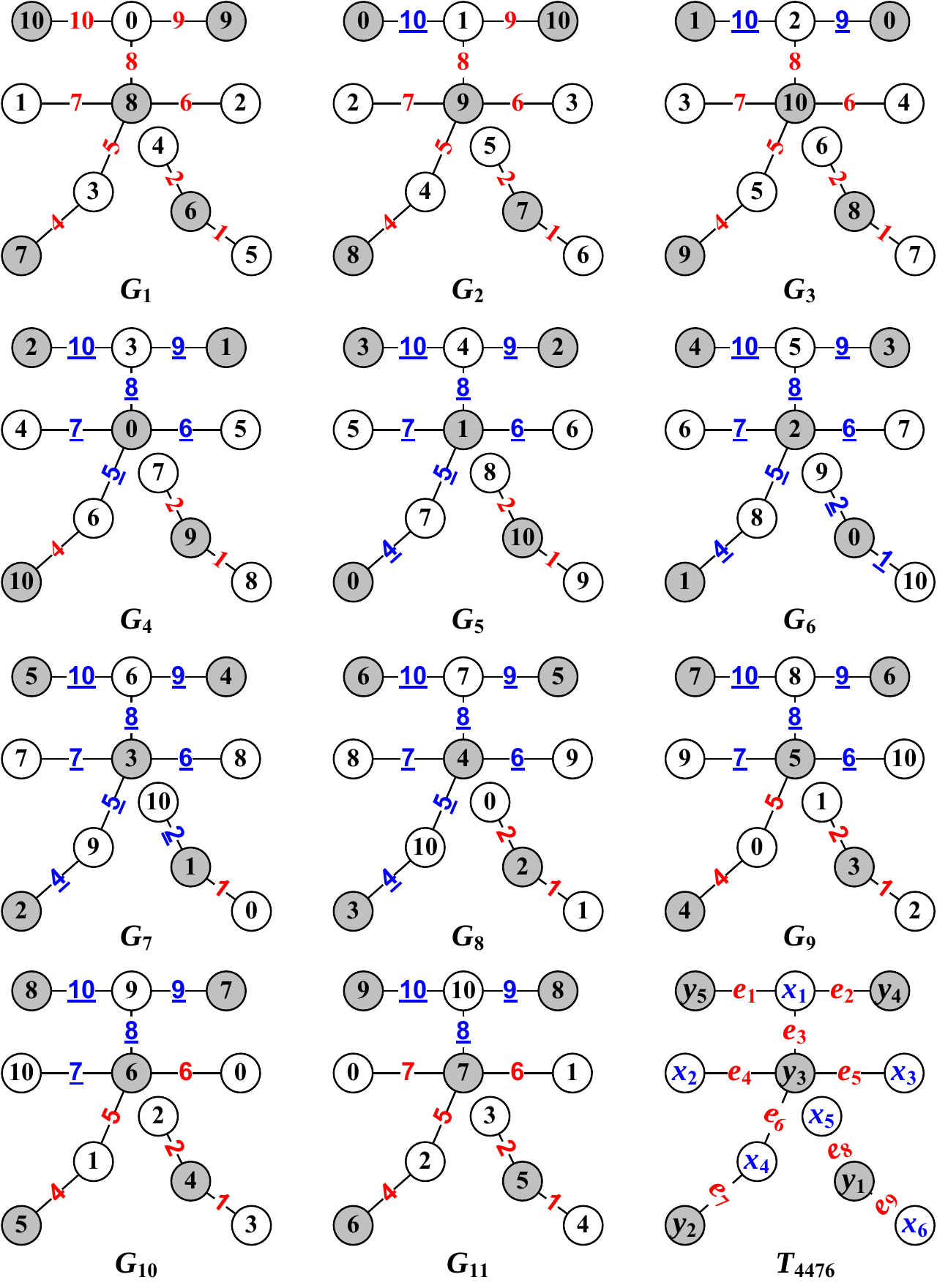}\\
\caption{\label{fig:tian-group} {\small An every-zero graphic group $\{F_{f_G}(G),\oplus\}$ made by Hanzi-graph $T_{4476}$ admitting a set-ordered graceful labelling.}}
\end{figure}

Let $V(T_{4476})=X\cup Y$ be the set of vertices of a Hanzi-graph $T_{4476}$ shown in Fig.\ref{fig:tian-group}, where \begin{equation}\label{eqa:bipartite-sets}
{
\begin{split}
&X=\{x_1,x_2,x_3,x_4,x_5,x_6\},~Y=\{y_1,y_2,y_3,y_4,y_5\};\\
&E(T_{4476})=\{e_i:~i\in [1,9]\}.
\end{split}}
\end{equation} where $e_1=x_1y_5$, $e_2=x_1y_4$, $e_3=x_1y_3$, $e_4=y_3x_2$, $e_5=y_3x_3$, $e_6=y_3x_4$, $e_7=y_2x_4$, $e_8=y_1x_5$, $e_9=y_1x_6$.

The property of ``set-ordered'' tells us: $f^{(j)}_G(x_i)<f^{(j)}_G(x_{i+1})$, $f^{(j)}_G(x_6)<f^{(j)}_G(y_1)$, $f^{(j)}_G(y_j)<f^{(j)}_G(y_{j+1})$ for $T_{4476}=G=G_1$ shown in Fig.\ref{fig:tian-group}, and $j\in [1,11]$. Also, we have $f^{(k)}_X(x_i)<f^{(k)}_X(x_{i+1})$, $f^{(k)}_X(x_6)<f^{(k)}_X(y_1)$, $f^{(k)}_X(y_j)<f^{(k)}_X(y_{j+1})$ for $X=O_1,M_1,L_1,E_1,C_1$ shown in Fig.\ref{fig:tian-group-formulae}.

Thereby, $G=G_1$ admits its own flawed set-ordered graceful labelling as: $f^{(1)}_G(x_i)=i-1$ for $i\in [1,6]$, and $f^{(1)}_G(y_j)=6+j-1$ for $i\in [1,5]$, and the induced edge labels $f^{(1)}_O(e_{s})=11-s$ for $s\in [1,7]$, as well as $f^{(1)}_O(e_{s})=10-s$ for $s=8,9$. We can write $f^{(k)}_G(x_i)=k+i-2$ with $i\in [1,6]$ and $k\in [1,11]$, $f^{(k)}_G(y_j)=k+i-2$ with $j\in [1,5]$ and $k\in [1,11]$. Hence, $f^{(k)}_G(w)=f^{(1)}_G(w)+k-1$ for $w\in V(G)$ and $k\in [1,11]$.

We have the following equivalent relationships:
\begin{asparaenum}[(Rel-1) ]
\item $\{F_{f_G}(G),\oplus\}$ and $\{F_{f_O}(O),\oplus\}$ under $(\bmod~21)$: $f^{(1)}_O(x_i)=2f^{(1)}_G(x_i)$ for $x_i\in X$ and $f^{(1)}_O(y_i)=2f^{(1)}_G(y_i)-1$ for $y_i\in Y$, and the induced edge labels $f^{(1)}_O(e_{s})=2f^{(1)}_G(e_{s})-1$ for $s\in [1,9]$ (see Fig.\ref{fig:tian-group-head}). Thereby, we get $f^{(k)}_O(z)=f^{(1)}_O(z)+k-1$ for $z\in V(G)$, $f^{(k)}_O(e_j)=|f^{(k)}_O(y_l)-f^{(k)}_O(x_i)|$ for $e_j=x_iy_l\in E(G)$ and $k\in [1,21]$, and $G_j$ corresponds with $O_{2j-1}$ for $i\in [1,11]$.

\item $\{F_{f_G}(G),\oplus\}$ and $\{F_{f_M}(M),\oplus\}$ under $(\bmod~21)$: $f^{(1)}_M(x_i)=f^{(1)}_G(x_{6-i+1})+1$ for $i\in [1,6]$ and $f^{(1)}_M(y_i)=f^{(1)}_G(y_i)+1$ for $y_i\in Y$, and $f^{(1)}_M(e_j)=f^{(1)}_G(e_{6-j+1})+11$ for $i\in [1,9]$ (see Fig.\ref{fig:tian-group-head}). Hence, $f^{(k)}_M(z)=f^{(1)}_M(z)+k-1$ for $z\in V(G)$, $f^{(k)}_M(e_j)=f^{(1)}_M(e_j)+k-1$ for $e_j\in E(G)$ and $k\in [1,21]$.

\item $\{F_{f_G}(G),\oplus\}$ and $\{F_{f_L}(L),\oplus\}$ under $(\bmod~21)$: $f^{(1)}_L(x_i)=2f^{(1)}_G(x_{6-i+1})+1$ for $i\in [1,6]$ and $f^{(1)}_L(y_i)=2f^{(1)}_G(y_i)+1$ for $y_i\in Y$, and $f^{(1)}_L(e_j)=2f^{(1)}_G(e_{6-j+1})$ for $i\in [1,9]$ (see Fig.\ref{fig:tian-group-head}). In general, $f^{(k)}_L(z)=f^{(1)}_L(z)+k-1$ for $z\in V(G)$, $f^{(k)}_L(e_i)=f^{(1)}_L(e_i)+k-1$ for $e_i\in E(G)$ and $k\in [1,21]$.

\item $\{F_{f_G}(G),\oplus\}$ and $\{F_{f_E}(E),\oplus\}$ under $(\bmod~10)$: $f^{(1)}_E(x_i)=f^{(1)}_G(x_{6-i+1})$ for $i\in [1,6]$ and $f^{(1)}_E(y_i)=f^{(1)}_G(y_i)$ for $y_i\in Y$, and $f^{(1)}_E(e_j)=f^{(1)}_G(x_{6-i+1})+f^{(1)}_G(y_i)~(\bmod~10)$ for $i\in [1,9]$ (see Fig.\ref{fig:tian-group-head}). We obtain $f^{(k)}_E(z)=f^{(1)}_E(z)+k-1$ for $z\in V(G)$, $f^{(k)}_E(e_j)=f^{(k)}_E(x_i)+f^{(k)}_E(y_l)~(\bmod~10)$ for $e_j=x_iy_l\in E(G)$ and $k\in [1,11]$.

\item $\{F_{f_G}(G),\oplus\}$ and $\{F_{f_C}(C),\oplus\}$ under $(\bmod~20)$: $f^{(1)}_C(w)=f^{(1)}_G(w)+1$ for $w\in V(G)$, and $f^{(1)}_C(e_i)=f^{(1)}_G(e_{9-i+1})+9$ for $i\in [1,9]$ (see Fig.\ref{fig:tian-group-head}). The above transformation enables us to claim $f^{(k)}_E(z)=f^{(1)}_E(z)+k-1$ for $z\in V(G)$, $f^{(k)}_C(e_i)=f^{(1)}_C(e_i)+k$ for $e_i\in E(G)$ and $k\in [1,20]$.
\end{asparaenum}

\begin{figure}[h]
\centering
\includegraphics[width=8.2cm]{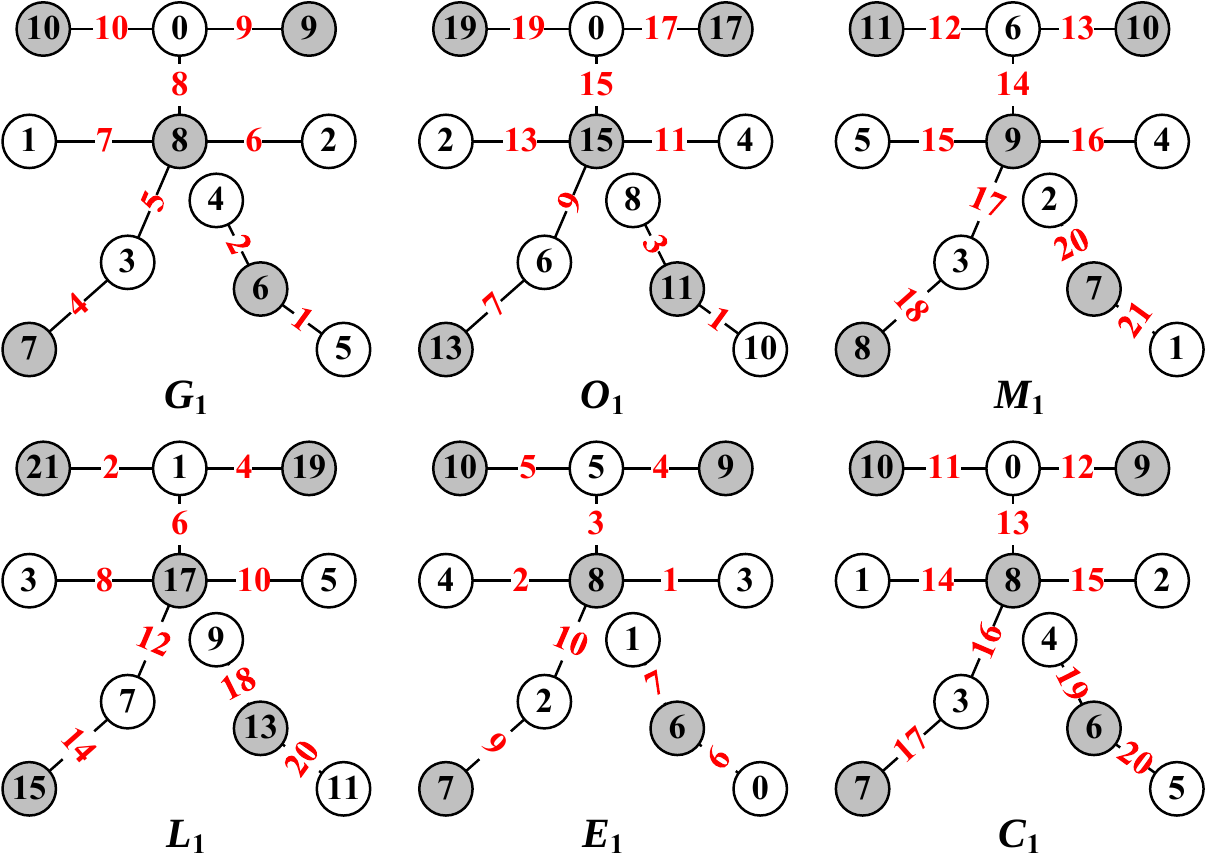}\\
\caption{\label{fig:tian-group-head} {\small The first Topsnut-gpw of six every-zero graphic groups shown in Fig.\ref{fig:tian-group-formulae}.}}
\end{figure}

Hanzi-graph $T_{4476}$, in fact, admits a flawed set-ordered $0$-rotatable system of (odd-)graceful labellings according to Definition \ref{defn:mf-graceful-mf-odd-graceful}, see Fig.\ref{fig:tian-0-rotatable}, namely, each vertex of Hanzi-graph $T_{4476}$ can grow new vertices and edges (see examples shown in Fig.\ref{fig:self-growing-11} and Fig.\ref{fig:self-growing-22}). Thereby, we have at least eight every-zero graphic groups $\{F_{f_G}(G),\oplus\}$ based on $T_{4476}$.

\begin{figure}[h]
\centering
\includegraphics[width=8cm]{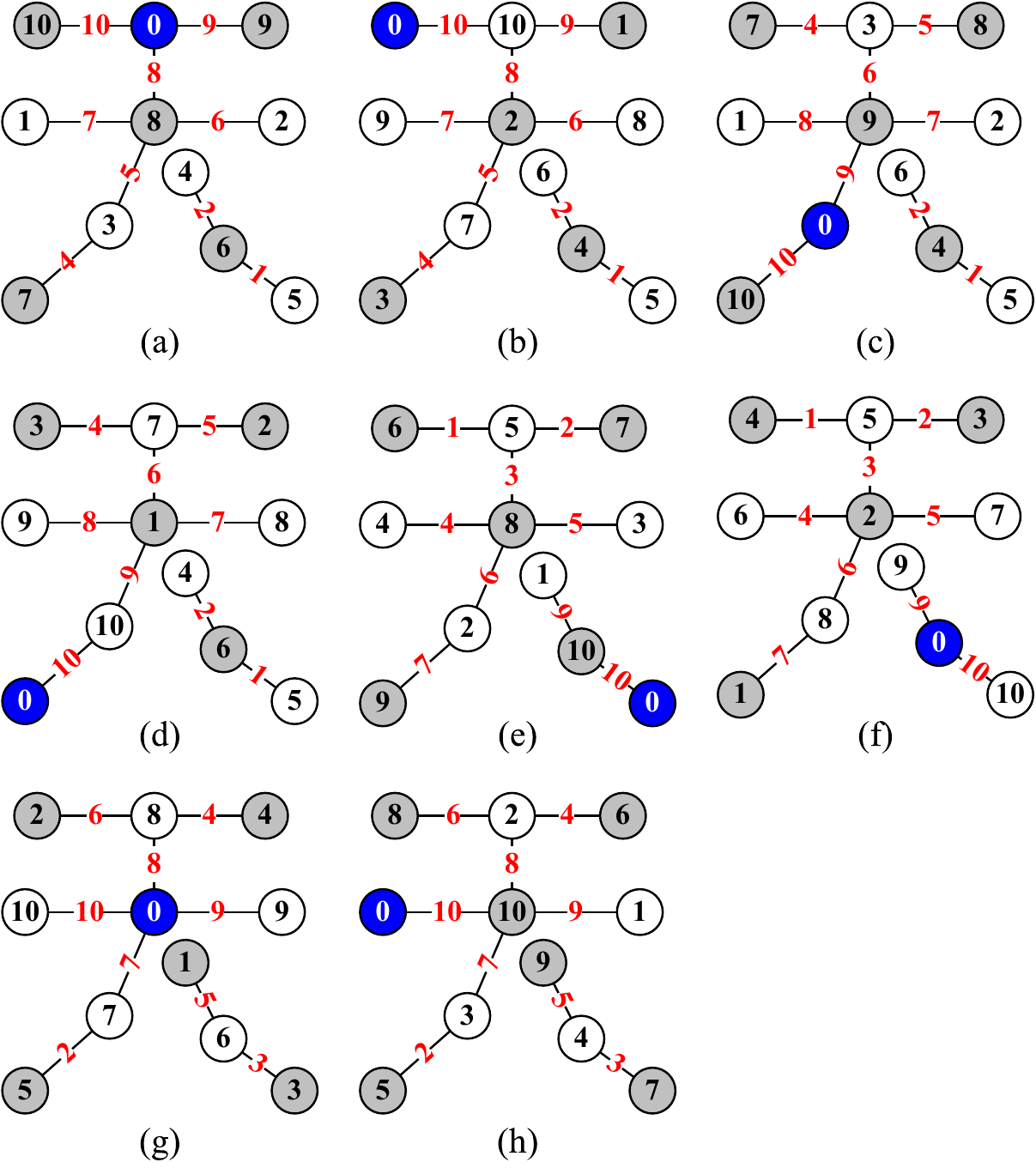}\\
\caption{\label{fig:tian-0-rotatable} {\small Hanzi-graph $T_{4476}$ admits a set-ordered $0$-rotatable system of (odd-)graceful labellings.}}
\end{figure}

By the writing stroke order of Hanzis, a Topsnut-gpw shown in Fig.\ref{fig:tian-0-rotatable} (a) enables us to obtain a TB-paw
\begin{equation}\label{eqa:t-4476-tb-paw}
{
\begin{split}
D(G_1)=101009917862088534742615
\end{split}}
\end{equation}
with $24$ bytes. Furthermore, a Hanzi-network $T_{4476}$ shown in Fig.\ref{fig:tian-encrypted-6-groups} (a), which was encrypted by an every-zero graphic group $\{F_{f_G}(G),\oplus\}$, can induce a TB-paw
\begin{equation}\label{eqa:t-4476-tb-paw}
{
\begin{split}
D(T_{4476})=&D(G_{11})D(G_4)D(G_1)D(G_3)D(G_{10})\\
&D(G_{2})D(G_3)D(G_9)D(G_4)D(G_{3})\\
&D(G_{1})D(G_2)D(G_9)D(G_5)D(G_{4})\\
&D(G_4)D(G_{8})D(G_{5})D(G_4)D(G_7)\\
&D(G_5)D(G_{6})
\end{split}}
\end{equation} with at least $500$ bytes ($500 \times 8=4000$ bits). Moreover, this TB-paw $D(T_{4476})$ has its own \emph{strong-rank} $H(D(T_{4476}))$ computed by
\begin{equation}\label{eqa:strong-rank}
{
\begin{split}
H(D(T_{4476}))&=L(D(T_{4476}))\cdot \log_2 |X|\\
&=500\cdot \log_2~10\\
&\approx 1661 ~(\textrm{bits}),
\end{split}}
\end{equation} here $X=[0,9]$ in (\ref{eqa:strong-rank}). If $X=[0,9]\cup \{a,b,\dots ,z\}\cup \{A,B,\dots ,Z\}$, then $\log_2~|X|=\log_2 62\approx 5.9542$ bits. A string $x$ has its own strong-rank $H(x)$ can be computed by the strong-rank formula $H(x)=L(x)\cdot \log_2 |X|$ (``Is your password secure? How to design a password that others can't guess and they can remember easily?'', WWW), where $L(x)$ is the number of  bytes of the string $x$ written by the letters of $X$.

We have known: Each permutation $G_{a_1}G_{a_2}\dots G_{a_{11}}$ of $G_1G_2 \dots G_{11}$ in the every-zero graphic group $\{F_{f_G}(G),\oplus\}$ can be used to label the vertices of the Hanzi-network $T_{4476}$, and each permutation $G_{a_1}G_{a_2}\dots G_{a_{11}}$ has $11$ zeros to encrypt the Hanzi-network $T_{4476}$, so we have at least $11\cdot (11)!~(=439,084,800> 2^{28})$ different approaches to encrypt the Hanzi-network $T_{4476}$. Moreover, this Hanzi-network $T_{4476}$ induces a matrix $A_{(a)}$ as follows:

\begin{equation}\label{eqa:tian-network-matrix}
\centering
A_{(a)}= \left(
\begin{array}{l}
G_1 ~ ~G_1 ~~ G_1 ~ G_3~ G_9~ G_9~ G_4~ G_5~ G_7\\
G_4 ~ ~G_3 ~ ~G_2 ~ G_2~ G_4~ G_5~ G_4~ G_4~ G_5\\
G_{11} ~ G_{10} ~ G_9 ~ G_9~ G_3~ G_4~ G_8~ G_7~ G_6
\end{array}
\right)
\end{equation}
\\So, we have $27!~(>2^{93})$ permutations obtained from the matrix $A_{(a)}$ shown in (\ref{eqa:tian-network-matrix}), and each permutation distributes us a TB-paw with at least 648 bytes.

\begin{figure}[h]
\centering
\includegraphics[width=8.2cm]{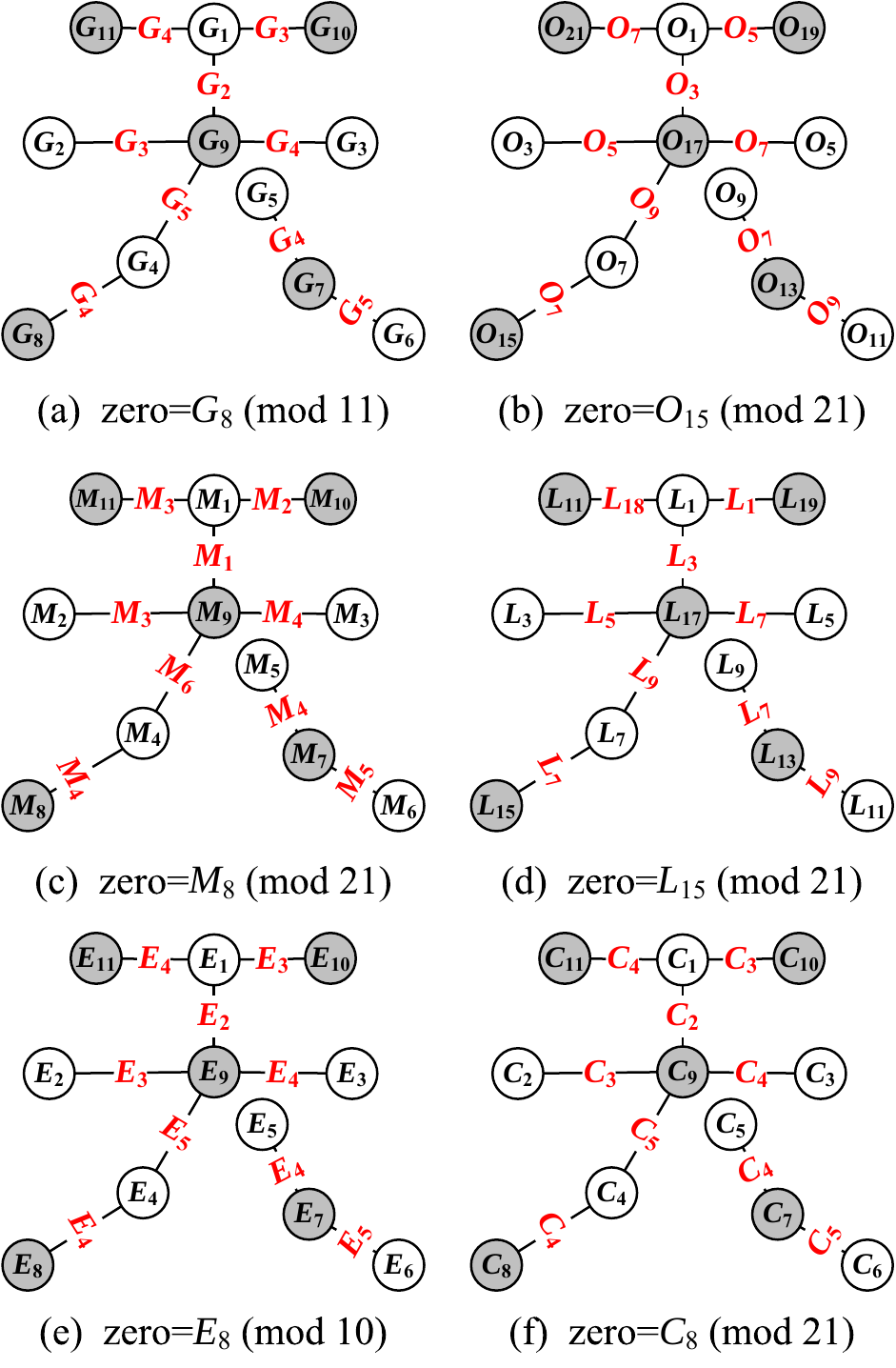}\\
\caption{\label{fig:tian-encrypted-6-groups} {\small A Hanzi-network $T_{4476}$ is encrypted by six every-zero graphic groups shown in Fig.\ref{fig:tian-group-formulae}, respectively.}}
\end{figure}

\begin{thm}\label{thm:equivalent-graphic-groups}
If a $(p,q)$-graph $G$ admits two mutually equivalent labellings $f_a$ and $f_b$, then two every-zero graphic groups $\{F_{f_a}(G),\oplus\}$ and $\{F_{f_b}(G),\oplus\}$ induced by these two labellings of $G$ are equivalent to each other.
\end{thm}

\subsubsection{Encrypting Hanzi-networks with many restrictions}

Notice that there are many gg-colorings introduced in Definition \ref{defn:graphic-group-coloring}, they are unlike popular graph colorings. First of all, we prove a lemma as follows:
\begin{lem}\label{thm:graphic-groups-to-trees}
Suppose that $T$ is a tree, and $\{F_{f}(G),\oplus\}$ is an every-zero graphic group made by a  $(p,q)$-graph $G$ admitting a labelling $f$. If the number $|F_{f}(G)|$ is greater than the maximum degree $\Delta(T)$ of $T$, then $T$ admits a proper total gg-coloring $\theta: V(T)\cup E(T)\rightarrow F_f(G)$.
\end{lem}
\begin{proof} By induction on the number of vertices of $T$. We take a leaf $u$ of $T$, and $v$ is adjacent with $u$ in $T$, such that the number of neighbors of the neighbor set $N_{ei}(v)$ is the least in $T$. Notice that $|F_{f}(G)|\geq \Delta(T)+1$ by the hypothesis of the lemma.

If $\Delta(T-u)<\Delta(T)$, then $T$ is a star, we are done. Otherwise, $\Delta(T-u)=\Delta(T)$, so the tree $T-u$ admits a proper total gg-coloring $\theta': V(T-u)\cup E(T-u)\rightarrow F_f(G)$. Clearly, the neighbor set $N'_{ei}(v)$ in the tree $T-u$ holds $|N'_{ei}(v)|<\Delta(T)$ true. On the other hand, $|F_{f}(G)|-1\geq \Delta(T)>|N'_{ei}(v)|$, there exists a Topsnut-gpw $G_k$ holding $G_k\not \in \{\theta'(v)\}\cup \{\theta'(w):~w\in N'_{ei}(v)\}$. We define a  proper total gg-coloring of $T$ as: $\theta(z)=\theta'(z)$ for $z\in V(T-u)\cup E(T-u)$, $\theta(u)=G_k$, and $\theta(uv)=\theta(u)\oplus \theta(v)$. Notice that if $\theta(uv)=\theta(vw)$ for some $w\in N'_{ei}(v)$, thus, $\theta(u)\oplus \theta(v)=\theta(w)\oplus \theta(v)$, which shows $\theta(w)=\theta(u)$, a contradiction. The proof is complete by the principle of induction.
\end{proof}

Lemma \ref{thm:graphic-groups-to-trees} can help us to get a  result as follows:

\begin{thm}\label{thm:Hanzi-self-similar-networks-trees}
If each Hanzi-network  $N(t)$ is a tree at time step $t$, for a proper total gg-coloring $\theta: V(N(t))\cup E(N(t))\rightarrow F_f(G(t))$ based on an every-zero graphic group $\{F_{f}(G(t)),\oplus\}$, write $C'_{\theta}(x,t)= \{\theta(x),\theta(xy):y\in N_{ei}(x)\}$ for each $x\in V(N(t))$, and $I_{\theta}(x)=\{k(xy):y\in N_{ei}(x)\}$ for $\theta(xy)=\theta(x)\oplus_{k(xy)} \theta(y)$ for $y\in N_{ei}(x)$, called an \emph{index set} of a vertex $x$.

$(1)$ There are infinite every-zero graphic groups $\{F_{f}(G(t)),\oplus\}$ with $|F_{f}(G(t))|\geq 1+\Delta(N(t))$ at time step $t$, such that $N(t)$ admits a proper total gg-coloring $\theta: V(N(t))\cup E(N(t))\rightarrow F_f(G(t))$ for each every-zero graphic group.

$(2)$ $N(t)$ admits a proper total gg-coloring $\theta: V(N(t))\cup E(N(t))\rightarrow F_f(G(t))$, and holds $C'_{\theta}(u,t)\neq C'_{\theta}(v,t)$ for each edge $uv\in E(N(t))$ if each every-zero graphic group $\{F_{f}(G(t)),\oplus\}$ with $|F_{f}(G(t))|\geq 2+\Delta(N(t))$ at time step $t$.

$(3)$ $N(t)$ admits a proper total gg-coloring $\theta: V(N(t))\cup E(N(t))\rightarrow F_f(G(t))$ with $I_{\theta}(u)\neq I_{\theta}(v)$ for each edge $uv\in E(N(t))$ if each every-zero graphic group $\{F_{f}(G(t)),\oplus\}$ with $|F_{f}(G(t))|\geq 1+\Delta(N(t))$ at time step $t$.
\end{thm}
\begin{proof} (1) In fact, we take $N(t)$ to be a caterpillar  with $\Delta(N(t))\geq 1+\Delta(N(t))$ at each time step $t$, since any caterpillar admits over $35$ graph labellings.

(2) By induction on the numbers of vertices of three-like networks. Let $P=x_1x_2\dots x_{m-1}x_m$ be a longest path in $N(t)$, and $d_{N(t)}(x_{m-1})\leq d_{N(t)}(x_{2})\leq \Delta(N(t))$. Clearly, the tree $N(t)-x_m$ obtained from the deletion of a leaf $x_m$ of $N(t)$ holds $d_{N(t)-x_m}(x_{m-1})=d'\leq  \Delta(N(t))-1$. Thereby, $N(t)-x_m$ admits a proper total gg-coloring $\theta': V(N(t)-x_m)\cup E(N(t)-x_m)\rightarrow F_f(G(t))$ such that $\theta'(uv)=\theta'(u)\oplus_s \theta'(v)$, and $\{\theta'(u),\theta'(uy):y\in N_{ei}(u)\}=C'_{\theta'}(u,t)\neq C'_{\theta'}(v,t)=\{\theta'(v),\theta'(vw):w\in N_{ei}(v)\}$ for each edge $uv\in E(N(t)-x_m)$ according to the hypothesis of induction. By $d'\leq  \Delta(N(t))-1$, we have a Topsnut-gpw $G_i\not \in C'_{\theta'}(x_{m-1},t)=\{\theta'(x_{m-1}),\theta'(x_{m-1}w):w\in N_{ei}(x_{m-1})\}\}$. Notice that there is a vertex $z\in N_{ei}(x_{m-1})$ such that $d_{N(t)}(z)\geq 2$ and $C'_{\theta'}(z,t)\neq C'_{\theta'}(x_{m-1},t)$.

We define a  proper total  gg-coloring $\theta$ of the three-like network $N(t)$ as: $\theta(x)=\theta'(x)$ for $x\in V(N(t)-x_m)\cup E(N(t)-x_m) \subset V(N(t))\cup E(N(t))$, and $\theta(x_{m-1}x_{m})=G_i$, and select a Topsnut-gpw $G_j$ to be $G_j\theta(x_{m})$ holding $\theta(x_{m-1}x_{m})=\theta(x_{m-1})\oplus_s \theta(x_{m})$. If $C'_{\theta'}(z,t)=C'_{\theta}(z,t)=C'_{\theta}(x_{m-1},t)=C'_{\theta'}(x_{m-1},t)\cup \{G_i\}$, we get a contradiction if $d_{N(t)}(x_{m-1})=\Delta(N(t))$, since $|C'_{\theta}(x_{m-1},t)=C'_{\theta'}(x_{m-1},t)\cup \{G_i\}|=\Delta(N(t))+2>|C'_{\theta'}(z,t)|$; so we discuss the case $d_{N(t)}(x_{m-1})=\Delta(N(t))-1$, then there are two distinct Topsnut-gpws $G_i,G'_i\not \in C'_{\theta'}(x_{m-1},t)$, we choose $G'_i$ to be $\theta(x_{m-1}x_{m})=G'_i$. Thus, the  proper total  gg-coloring $\theta$ holds $C'_{\theta}(u,t)\neq C'_{\theta}(v,t)$ for each edge $uv\in E(N(t))$ by the principle of induction.

(3) We choose a leaf $u_0$ of  the three-like network  $N(t)$ if $u_0$ is on a longest path of $N(t)$, and $u_0\in N_{ei}(x)$, as well as $y\in N_{ei}(x)$ with $d_{N(t)}(y)\geq 2$. Assume that the subnetwork $N(t)-u_0$ admits a proper total gg-coloring $\phi': V(N(t)-u_0)\cup E(N(t)-u_0)\rightarrow F_f(G(t))$ with $I_{\phi'}(u)\neq I_{\phi'}(v)$ for each edge $uv\in E(N(t)-u_0)$ by induction. There are $k_s\not \in I_{\phi'}(y)$ and $G_j\not \in \{\phi'(xx_i): x_i\in N_{ei}(x)\}$ since $|I_{\phi'}(y)|\leq \Delta(N(t))$ and $d_{N(t)-u_0}(x)\leq \Delta(N(t))-1$. We define a proper total  gg-coloring $\phi$  of $N(t)$ by setting $\phi(z)=\phi'(z)$ for $z\in (V(N(t))\cup E(N(t)))\setminus \{u_0\}$, and $\phi(xu_0)=G_j$ holding $\phi(xu_0)=\phi(x)\oplus _{k_s}\phi(u_0)$ for $\phi(u_0)=G_l\in F_f(G(t))$. Clearly, $I_{\phi}(y)\neq I_{\phi}(x)$, we are done according to the principle of induction.
\end{proof}

\begin{rem}\label{thm:CCCCCC}
We  ask for other types of gg-colorings as follows:

\vskip 0.4cm

\begin{asparaenum}[LC-1.  ]
\item  \label{prob:total-graceful-gg-coloring} If an every-zero graphic group $\{F_{f}(G(t)),\oplus\}$ holds that $f$ is a graceful labelling of the $(p,q)$-graph $G$, and $N(t)$ is a \emph{tree-like network} and admits a proper total gg-coloring $\theta: V(N(t))\cup E(N(t))\rightarrow F_f(G(t))$ at time step $t$, such that the \emph{vertex-index set} $I_{\theta}(V(N(t)))=\{i:\theta(y)=G_i \in F_f(G(t))\}$ satisfy $\{|i-j|:~\theta(x)=G_i,\theta(y)=G_j,~xy\in E(N(t))\}=[1,|V(N(t))|-1]$. Then we say the gg-coloring $\theta$ matches with $f$, call it a \emph{totally graceful gg-coloring}. Find such tree-like network $N(t)$ and its total graceful gg-colorings.
\item  \label{prob:perfectly-total-graceful-gg-coloring}  Suppose that $\theta$ is a totally graceful gg-coloring of a tree-like network $N(t)$, so we have the  \emph{edge-index set} $I_{\theta}(E(N(t)))=\{k:\theta(u)\oplus \theta(v)=\theta(uv)=G_k, ~uv \in E(N(t))\}$, and call $\theta$ a \emph{perfectly total graceful gg-coloring} of $N(t)$ if $I_{\theta}(E(N(t)))=[1,|V(N(t))|-1]$. Determine tree-like networks $N(t)$ admitting perfectly  total graceful gg-colorings.
\item  \label{prob:perfectly-total-odd-graceful-gg-coloring} If $f$ is an odd-graceful labelling in an every-zero graphic group $\{F_{f}(G(t)),\oplus\}$, we replace ``graceful'' by ``odd-graceful'' and substitute ``$[1,|V(N(t))|-1]$'' by ``$[1,2|V(N(t))|-3]$'' in LC-\ref{prob:total-graceful-gg-coloring} and LC-\ref{prob:perfectly-total-graceful-gg-coloring}, so we want to look for tree-like networks $N(t)$ admitting \emph{totally odd-graceful gg-colorings} and \emph{perfectly total odd-graceful gg-colorings}.
\end{asparaenum}

\vskip 0.4cm

The above LC-\ref{prob:total-graceful-gg-coloring}, LC-\ref{prob:perfectly-total-graceful-gg-coloring} and LC-\ref{prob:perfectly-total-odd-graceful-gg-coloring} can be handled if $N(t)$ and $G(t)$ are both caterpillars of $p(t)$  vertices at each time step $t$.
\end{rem}

Definition \ref{defn:graphic-group-coloring} can be generalized into dynamic networks:
\begin{defn} \label{defn:gg-coloring-dynamic-networks}
$^*$ Let $N(t)$ be a network  admitting a mapping $\theta_t: V(N(t))\rightarrow F_f(G)$ with  $\theta_t(u)\neq \theta_t(v)$ for  each edge $uv\in E(N(t))$, where $\{F_f(G),\oplus\}$ is an every-zero graphic group containing Topsnut-gpws $G_i$ with $i\in [1,M_G]$. If each edge $uv\in E(N(t))$ corresponds an index $k(t,uv)$ such that $\theta_t(uv)=\theta_t(u)\oplus _{k(t,uv)}\theta_t(v)\in F_f(G)$, we call $\theta_t$ a v-induced total dynamic gg-coloring at time step $t$.\qqed
\end{defn}

\begin{rem}\label{thm:gg-colorings}
(1) A v-induced total proper gg-coloring $\theta$ defined in Definition \ref{defn:gg-coloring-dynamic-networks} may be not a proper total gg-coloring, since such cases $\theta(uv)=\theta(u)$ or $\theta(uv)=\theta(uw)$ may happen. Thereby, we can add $\theta(uv)\neq \theta(u)$ or $\theta(uv)\neq \theta(uw)$ to guarantee $\theta$ is really  a proper total gg-coloring.

(2) A  v-induced total dynamic gg-coloring $\theta_t$ of a network $N(t)$  defined in Definition \ref{defn:gg-coloring-dynamic-networks} can be used as a dynamic encryption of $N(t)$. For example, we let $k(t,uv)=k(t)$ for  each edge $uv\in E(N(t))$, and $k(t_1)\neq k(t_2)$ for two distinct  time steps $t_1,t_2$, as well as $\theta_{t_1}\neq \theta_{t_2}$, in other word, an encrypted network $N(t_1)$ is not equal to another encrypted network $N(t_1)$. It is easy to find a connected graph $G$ having thousand and thousand vertices and admitting a set-ordered graceful labelling (since such labelling can deduce many other labellings).
\end{rem}

\subsection{Labelling self-similar Hanzi-graphs}

It has been conjectured that every tree is graceful (Graceful Tree Conjecture, Alexander Rosa, 1966) and odd-graceful (R.B. Gnanajothi, 1991) in \cite{Gallian2018}. As far as we know, these two conjectures are open now. It is noticeable, no polynomial algorithm is for labelling trees to be graceful, or odd-graceful, except a few number of particular trees. So, we can say it is not easy to labelling self-similar Hanzi-graphs having large scales.

\begin{thm} \label{thm:two-graphs-link-identify}
\cite{Mu-Sun-Zhang-Yao-2019} Let $G_1,G_2$ be two disjoint bipartite graphs admits set-ordered graceful labellings. Then there exist vertices $u\in
V(G_1)$ and $v\in V(G_2)$ such that the graph obtained by joining
$u$ with $v$ with an edge or by identifying $u$ with $v$ into one
vertex is a set-ordered graceful labelling.
\end{thm}

Let $T$ be  a $(n,m)$-graph and let $G$ be a $(p,q)$-graph. We define a  \emph{near-symmetric graph} $H=\langle T \circ G\rangle$ such that $H$ contains $T$ and $n$ edge-disjoint copies $G_i$ of $G$ with $|E(H-E(T))|=nq$, $|H|\leq np$ and $i\in [1,n]$.

\begin{thm} \label{thm:two-graphs-link-identify}
\cite{Mu-Sun-Zhang-Yao-2019} Let $T$ be a set-ordered graceful Hanzi-graph and $H=\langle T \circ G\rangle$ is graceful.
\end{thm}

\begin{thm} \label{thm:11111}
\cite{Mu-Sun-Zhang-Yao-2019} Suppose that a Hanzi $N_{zi}(0)$ admits a set-ordered graceful labelling. Let $N_{zi}(0)$ execute the Leaf-algorithm  to get a large self-similar network $N_{zi}(n)$ with $n\in integer$, then each $N_{zi}(n)$ admits a  set-ordered graceful labelling and an edge-magic total labelling.
\end{thm}

\section{Directed Hanzi-gpws}

Directed Hanzi-gpw is a new topic in our memory. We will present terminology and notation, definitions about directed  Hanzi-gpws. We call a one-way directed edge as an \emph{arc} $\overrightarrow{uv}$, where $v$ with a narrow is the \emph{head} of the arc $\overrightarrow{uv}$ and $u$ is the \emph{tail} of the arc $\overrightarrow{uv}$. A $(p,q)$-graph having some arcs and the remainder being proper edges is called a \emph{half-directed $(p,q)$-graph}, denoted as $G^{\rightarrow}$; a $(p,q)$-graph having only arcs is called a \emph{directed $(p,q)$-graph} or a \emph{$(p,q)$-digraph}, and denoted as $\overrightarrow{G}$. Particularly, the set of all arcs of a digraph $\overrightarrow{G}$ is denoted as $A(\overrightarrow{G})$.

\begin{figure}[h]
\centering
\includegraphics[height=4cm]{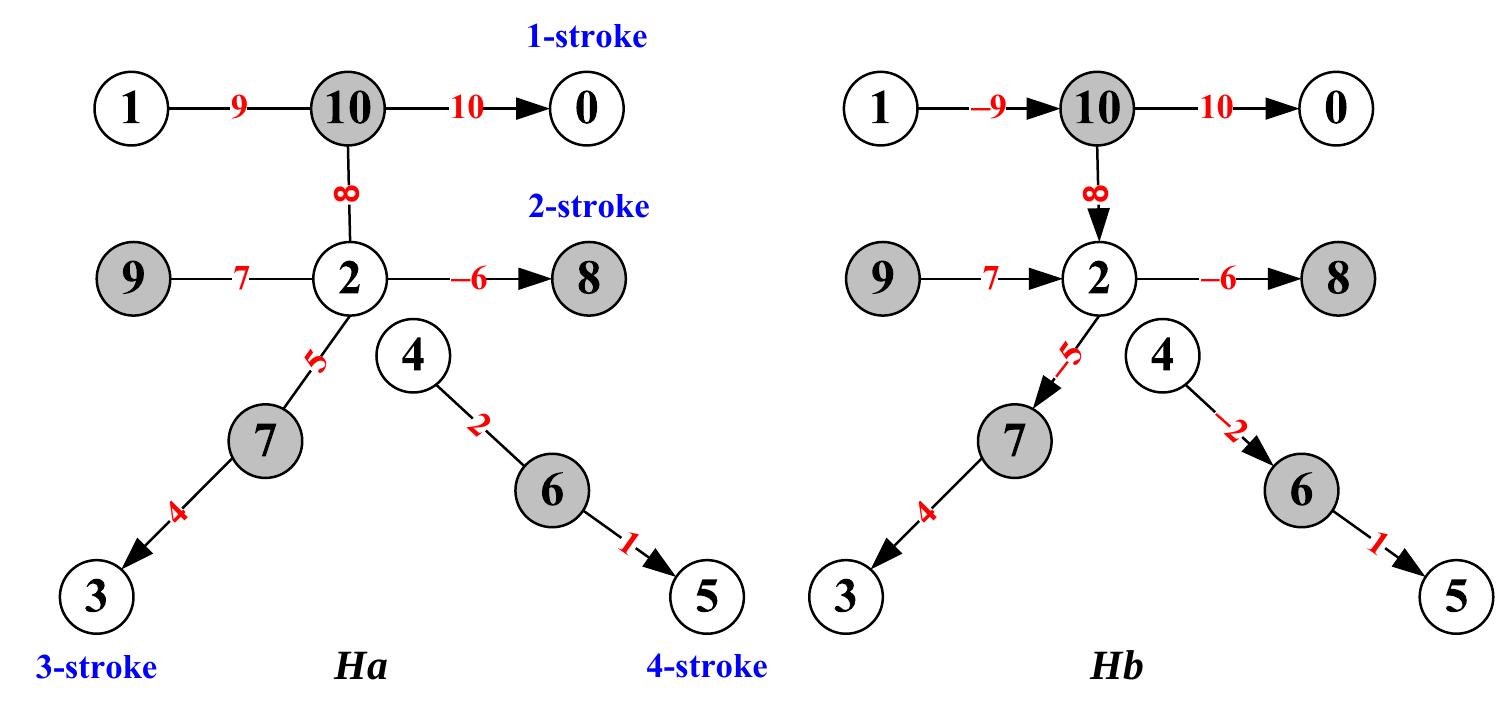}\\
\caption{\label{fig:directed-tian-00} {\small $Ha$ (left) is a half-directed Hanzi-gpw made by a Hanzi $H_{4476}$, and $Hb$ (right) is a directed Hanzi-gpw made by a Hanzi $H_{4476}$.}}
\end{figure}

Here, replacing arcs by proper edges in a directed graph, or a half-directed graph produces a graph, we call this graph as the \emph{underlying graph} of the directed graph, or the half-directed graph.

\begin{defn}\label{defn:half-directed-graceful}
$^*$ Let $f:V(G^{\rightarrow}) \rightarrow [0,q]$ (resp. $[0,2q-1]$) be a labelling of a half-directed $(p,q)$-graph $G^{\rightarrow}$ with its connected underlying graph, we define $f(xy)=|f(x)-f(y)|$ for proper edge $xy$ of $G^{\rightarrow}$ and $f(\overrightarrow{uv})=f(u)-f(v)$ for each arc $\overrightarrow{uv}$ of $G^{\rightarrow}$. If $\{f(xy),f(\overrightarrow{uv}):~uv, \overrightarrow{uv}\in E(G)\cup A(G)\}=[1,q]$ (resp. $[1,2q-1]^o$), we call $f$ a \emph{half-directed graceful labelling} (resp. \emph{half-directed odd-graceful labelling}) of $G^{\rightarrow}$.\qqed
\end{defn}

\begin{defn}\label{defn:directed-graceful}
$^*$ Let $g:V(\overrightarrow{G}) \rightarrow [0,q]$ (resp. $[0,2q-1]$) be a labelling of a $(p,q)$-digraph $\overrightarrow{G}$ with its connected underlying graph, and define $g(\overrightarrow{uv})=g(u)-g(v)$ for each arc $\overrightarrow{uv}$ of $\overrightarrow{G}$. If $\{|g(\overrightarrow{uv})|:~\overrightarrow{uv}\in A(\overrightarrow{G})\}=[1,q]$ (resp. $[1,2q-1]^o$), we call $g$ a \emph{directed graceful labelling} (resp. \emph{directed odd-graceful labelling}) of the digraph $\overrightarrow{G}$. Moreover, if $g(\overrightarrow{uv})>0$ (resp. $g(\overrightarrow{uv})<0$) for each arc $\overrightarrow{uv}\in A(\overrightarrow{G})$, we say $g$ to be uniformly.\qqed
\end{defn}

\begin{equation}\label{eqa:directed-1}
\centering
A(G^{\rightarrow})={\Big \downarrow} \left(
\begin{array}{l}
6~6~7~7\quad 2\quad 9~10~10~10\\
1~2~4~5~-6~7~~8~~9~~10\\
5~4~3~2\quad 8\quad 2~~2~~1~~0
\end{array}
\right)
\end{equation}
\begin{equation}\label{eqa:directed-2}
\centering
A(\overrightarrow{G})={\Big \downarrow} \left(
\begin{array}{l}
6~~~~4~~7~~~2~\quad 2~\quad 9~ ~10~~1~~~10\\
1~-2~4~-5~-6~~7~~~8~-9~10\\
5~~~~6~~3~~~7~\quad 8~\quad 2~~ ~2~~10~~~0
\end{array}
\right)
\end{equation}

A half-directed $(11,9)$-graph $G^{\rightarrow}$ of a Hanzi-gpw $Ha$ shown in Fig.\ref{fig:directed-tian-00} has its own half-directed Hanzi-matrix $A(G^{\rightarrow})$ shown in (\ref{eqa:directed-1}), and a $(11,9)$-digraph $\overrightarrow{G}$ of a Hanzi-gpw $Hb$ shown in Fig.\ref{fig:directed-tian-00} has its own directed Hanzi-matrix $A(\overrightarrow{G})$ shown in (\ref{eqa:directed-2}). The ``down-arrow $\downarrow$'' means that the elements in first line minus the elements in third line respectively, and the results are in second line in $A(\overrightarrow{G})$. We can derive TB-paws from half-directed Hanzi-matrices and directed Hanzi-matrices by replacing minus sign ``$-$'' with a letter ``$x$'' (or other letters, or signs). For example, the half-directed Hanzi-matrix $A(G^{\rightarrow})$ shown in (\ref{eqa:directed-1}) distributes us a TB-paw
$$D(A(G^{\rightarrow}))=66772910101010987x65421543282.$$

There is no equivalent relation between uniformly directed graceful labellings and set-ordered graceful labellings. In Fig.\ref{fig:uniformly-dir-graceful}, a directed tree $\overrightarrow{T}$ admits a uniformly directed graceful labelling, but its underlying tree $T$ does not admit a set-ordered graceful labelling. Thereby, directed labellings of directed graphs are complex than that of undirected graphs.

\begin{figure}[h]
\centering
\includegraphics[height=4cm]{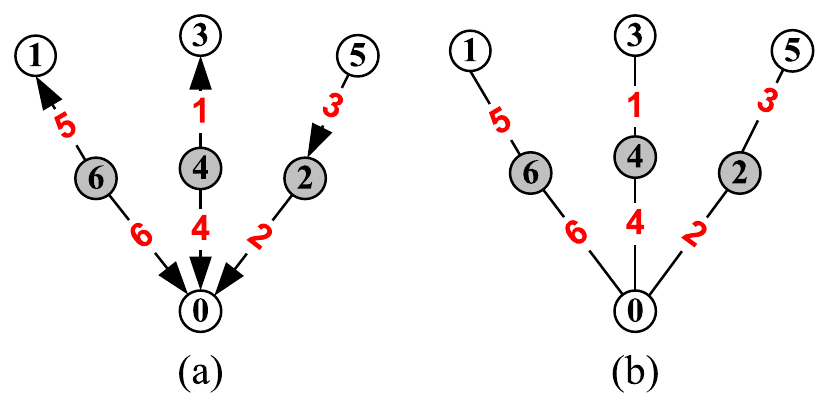}\\
\caption{\label{fig:uniformly-dir-graceful} {\small (a) A directed tree admits a uniformly directed graceful labelling; (b) a tree admits a graceful labelling.}}
\end{figure}

\begin{prop}\label{thm:GTC-conjecture}
Let $\overrightarrow{O}(T)$ be a set of directed trees having the same underlying tree $T$. If GTC-conjecture (resp. OGTC-conjecture) holds true, then $\overrightarrow{O}(T)$ contains at least two directed trees admitting uniformly directed (odd-)graceful labellings, and each directed tree of $\overrightarrow{O}(T)$ admits a directed (odd-)graceful labelling.
\end{prop}

Since there are $2^{p-1}$ directed trees in $\overrightarrow{O}(T)$ with the same underlying tree $T$ of $p$ vertices, so we can make $2^{p-1}\cdot (3p-3)!$ TB-paws by directed Hanzi-matrices on $\overrightarrow{O}(T)$.

If the underlying graph of a half-directed $(p,q)$-graph $G^{\rightarrow}$, or a $(p,q)$-digraph $\overrightarrow{G}$ is disconnected, we can add a set $E^*$ of arcs or edges to join some vertices of $G^{\rightarrow}$ or $\overrightarrow{G}$ such that the resultant half-directed $G^{\rightarrow}+E^*$ or the resultant digraph $\overrightarrow{G}+E^*$ has a connected underlying graph, we say $G^{\rightarrow}+E^*$ to be a \emph{connected half-directed graph}, or $\overrightarrow{G}+E^*$ a \emph{connected digraph}, not based on directed paths. We define a flawed half-directed graceful labelling and a flawed directed graceful labelling on half-directed graphs and digraphs in the following Definition \ref{defn:flawed-half-directed-graceful} and Definition \ref{defn:flawed-directed-graceful}, respectively.

\begin{defn}\label{defn:flawed-half-directed-graceful}
$^*$ Suppose that a half-directed $(p,q)$-graph $G^{\rightarrow}$ has its disconnected underlying graph, and $G^{\rightarrow}+E^*$ is a connected half-directed $(p,q+q')$-graph, where $q'=|E^*|$. Let $f:V(G^{\rightarrow}+E^*) \rightarrow [0,q+q']$ (resp. $[0,2(q+q')-1]$) be a half-directed graceful labelling (resp. a half-directed odd-graceful labelling) $f$ of $G^{\rightarrow}+E^*$, then $f$ is called a \emph{flawed half-directed graceful labelling} (resp. \emph{flawed half-directed odd-graceful labelling}) of the half-directed $(p,q)$-graph $G^{\rightarrow}$.\qqed
\end{defn}

\begin{defn}\label{defn:flawed-directed-graceful}
$^*$ Suppose that the underlying graph of a $(p,q)$-digraph $\overrightarrow{G}$ is disconnected, and $\overrightarrow{G}+E^*$ is a connected directed $(p,q+q')$-graph, where $q'=|E^*|$. Let $f:V(\overrightarrow{G}+E^*) \rightarrow [0,q+q']$ (resp. $[0,2(q+q')-1]$) be a directed graceful labelling (resp. a directed odd-graceful labelling) $f$ of $\overrightarrow{G}+E^*$, then $f$ is called a \emph{flawed directed graceful labelling} (resp. \emph{flawed directed odd-graceful labelling}) of the $(p,q)$-digraph $\overrightarrow{G}$.\qqed
\end{defn}

We divide the arc label set $f(A(\overrightarrow{G}))$ into $f^+(A(\overrightarrow{G}))=\{f(\overrightarrow{uv})>0:~\overrightarrow{uv}\in A(\overrightarrow{G})\}$ and $f^-(A(\overrightarrow{G}))=\{f(\overrightarrow{uv})<0:~\overrightarrow{uv}\in A(\overrightarrow{G})\}$, namely, $f(A(\overrightarrow{G}))=f^+(A(\overrightarrow{G}))\cup f^-(A(\overrightarrow{G}))$. We say $f$ to be a \emph{uniformly (flawed) directed (odd-)graceful labelling} if one of $f^+(A(\overrightarrow{G}))=\emptyset $ and $f^-(A(\overrightarrow{G}))=\emptyset $ holds true. Determining the existence of directed (odd-)graceful labellings of digraphs are very difficult, even for directed trees (Ref. \cite{Yao-Yao-Cheng-2012}).

It is not hard to define other \emph{flawed directed labellings} or \emph{flawed half-directed labellings} of digraphs and half-directed graphs by the ways used in Definition \ref{defn:flawed-half-directed-graceful} and Definition \ref{defn:flawed-directed-graceful}. For example, we modify the 6c-labelling in \cite{Yao-Sun-Zhang-Mu-Sun-Wang-Su-Zhang-Yang-Yang-2018arXiv} into a flawed directed 6C-labelling:

\begin{defn}\label{defn:directed-6C-labelling}
$^*$ A total labelling $f:V(\overrightarrow{G})\cup A(\overrightarrow{G})\rightarrow [1,p+q]$ for a bipartite and connected $(p,q)$-digraph $\overrightarrow{G}$ is a bijection and holds:

(i) (e-magic) $f(\overrightarrow{uv})+|f(u)-f(v)|=k$, a constant;

(ii) (ee-difference) each arc $\overrightarrow{uv}$ matches with another arc $\overrightarrow{xy}$ holding $f(\overrightarrow{uv})=|f(x)-f(y)|$ (or $f(\overrightarrow{uv})=2(p+q)-|f(x)-f(y)|$);

(iii) (ee-balanced) let $s(\overrightarrow{uv})=|f(u)-f(v)|-f(\overrightarrow{uv})$ for $\overrightarrow{uv}\in A(\overrightarrow{G})$, then there exists a constant $k'$ such that each arc $\overrightarrow{uv}$ matches with another arc $\overrightarrow{u'v'}$ holding $s(\overrightarrow{uv})+s(\overrightarrow{u'v'})=k'$ (or $2(p+q)+s(\overrightarrow{uv})+s(\overrightarrow{u'v'})=k'$) true;

(iv) (EV-ordered) $f_{\min}(V(\overrightarrow{G}))>f_{\max}(A(\overrightarrow{G}))$ (or $\max f(V(\overrightarrow{G}))<\min f(A(\overrightarrow{G}))$, or $f(V(\overrightarrow{G}))\subseteq f(A(\overrightarrow{G}))$, or $f(A(\overrightarrow{G}))\subseteq f(V(\overrightarrow{G}))$, or $f(V(\overrightarrow{G}))$ is an odd-set and $f(A(\overrightarrow{G}))$ is an even-set);

(v) (ve-matching) there exists a constant $k''$ such thateach arc $\overrightarrow{uv}$ matches with one vertex $w$ such that $f(\overrightarrow{uv})+f(w)=k''$, and each vertex $z$ matches with one arc $\overrightarrow{xy}$ such that $f(z)+f(\overrightarrow{xy})=k''$, except the \emph{singularity} $f(x_0)=\lfloor \frac{p+q+1}{2}\rfloor $;

(vi) (set-ordered) $\max f(X)<\min f(Y)$ (or $\min f(X)>\max f(Y)$) for the bipartition $(X,Y)$ of $V(\overrightarrow{G})$.

We call $f$ a \emph{directed 6C-labelling}.\qqed
\end{defn}

\begin{figure}[h]
\centering
\includegraphics[height=6cm]{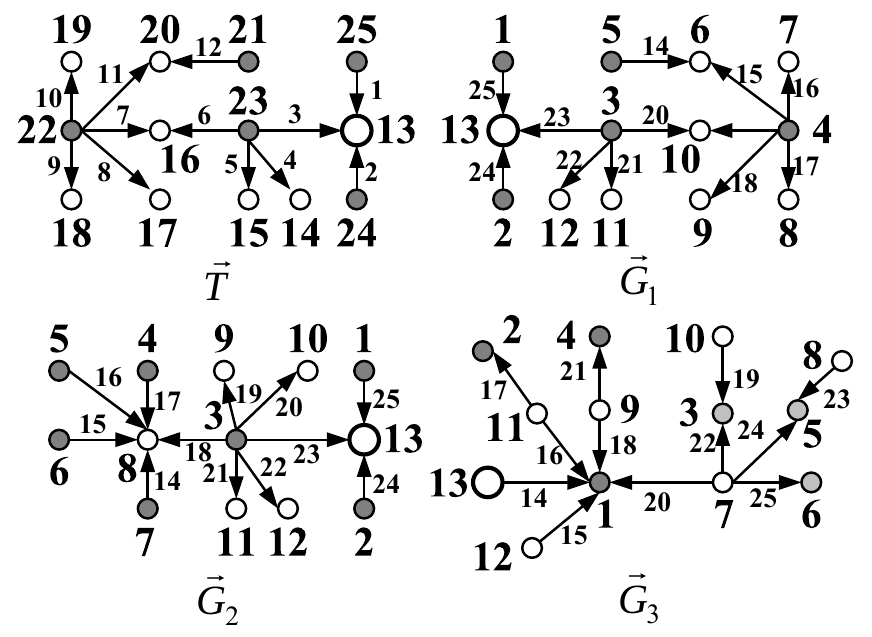}\\
\caption{\label{fig:directed-6c-labelling} {\small Four bipartite and connected digraphs admit directed 6C-labellings.}}
\end{figure}

Four digraphs shown in Fig.\ref{fig:directed-6c-labelling} admit their own directed 6C-labellings.

We can remove a set of arcs from the bipartite and connected $(p,q)$-digraph $\overrightarrow{G}$ admitting a directed 6C-labelling $f$ such that the resultant digraph $\overrightarrow{H}$ is disconnected. Then we say that the digraph $\overrightarrow{H}$ admits a \emph{flawed directed 6C-labelling} $f$ if $\overrightarrow{H}$ holds each of the restrict conditions (i)-(vi) of Definition \ref{defn:directed-6C-labelling}.

Proposition \ref{thm:obtained-caterpillar-directed-graceful} and Proposition \ref{thm:obtained-lobster-directed-graceful-11} are two corollaries of caterpillars admitting (odd-)graceful labellings.

\begin{prop}\label{thm:obtained-caterpillar-directed-graceful}
A digraph $\overrightarrow{T}$ obtained by orienting all of $q$ edges of a caterpillar admits a uniformly directed (odd-)graceful labelling $f$ such that $\{f(\overrightarrow{uv}):~\overrightarrow{uv}\in A(\overrightarrow{T})\}=[1,q]$ (or $[1,2q-1]^o$).
\end{prop}

\begin{prop}\label{thm:obtained-lobster-directed-graceful-11}
A digraph $\overrightarrow{T}$ obtained by orienting all of $q$ edges of a lobster admits a uniformly directed odd-graceful labelling $f$ such that $\{f(\overrightarrow{uv}):~\overrightarrow{uv}\in A(\overrightarrow{T})\}=[1,2q-1]^o$.
\end{prop}

\begin{thm}\label{thm:set-ordered-vs-uniformly-directed-graceful}
There is a way to orient all edges of a connected graph $G$ admitting a set-ordered graceful labelling $f$ such that it results in a connected digraph $\overrightarrow{G}$ admitting a uniformly directed graceful labelling $g$ holding $f(E(G))=g(A(\overrightarrow{G}))$ true.
\end{thm}

\section{Problems in researching Hanzi-graphs and Hanzi-gpws}

There are a number of unanswered questions generated from our researching Hanzi-graphs and Hanzi-gpws for deeply study. We, in the first subsection, propose these questions, some of them may be challenged and excited; and discuss some questions in the second subsection.
\subsection{Problems for further research}
\begin{asparaenum}[\textrm{Problem}-1. ]
\item \textbf{Find} a Hanzi $H$ such that $D(H)=s(T^{CC})-s(S^{CC})$ is the largest in all Hanzis.
\item \textbf{Find} a Hanzi $H_{abcd}$ which is a noun and can yields the largest numbers of word groups and idioms with other Hanzis. For example, ``$H_{4043}$ (= man)'' can construct many word groups and idioms shown in Fig.\ref{fig:man-word-group} and Fig.\ref{fig:man-chengyu}. As known, a Hanzi $H_{5551}$ has over 4400 word groups matching with it, and the number of  Hanzis matching with $H_{5551}$ is over 4000 too.
\item \textbf{Construct} various Hanzi idiom-graphs (see Fig.\ref{fig:2-idioms}), and study them.
\item Given a group of $m$ Hanzis, how many paragraphs made by this group of Hanzis are there? And how many such paragraphs are meaningful in Chinese? Some numbers $n!$ are given in Table-1 in Appendix A.
\item \textbf{Conjecture 1}: $T=\bigcup ^m_{i=1}T_i$ is a forest having trees $T_1,T_2,\dots ,T_m$, then $T$ admits a flawed graceful labelling.
\item \textbf{Conjecture 2}: $T=\bigcup ^m_{i=1}T_i$ is a forest having trees $T_1,T_2,\dots ,T_m$, then $T$ admits a flawed odd-graceful labelling.
\item Can we have mathematical approaches to design Chinese ``Duilian'' couplets? See examples shown in Fig.\ref{fig:formula-couplets}. How many Chinese sentences having real meanings can be produced by a mathematical formula like one of (a) and (b) shown in Fig.\ref{fig:formula-couplets}? Clearly, it is not easy to answer this question since we do not know how many such mathematical formulae are there first of all. One may use an approach called ``'graph network' introduced in \cite{Battaglia-27-authors-arXiv1806-01261v2} to research this question by Deep Learning, or Strong Deep Learning. On the other hands, a language formula ``'$x,x,x,xxyz$' shown in Fig.\ref{fig:formula-couplets} (a) can be considered as a public key, so there are many private keys matching with it.

\begin{figure}[h]
\centering
\includegraphics[height=4.4cm]{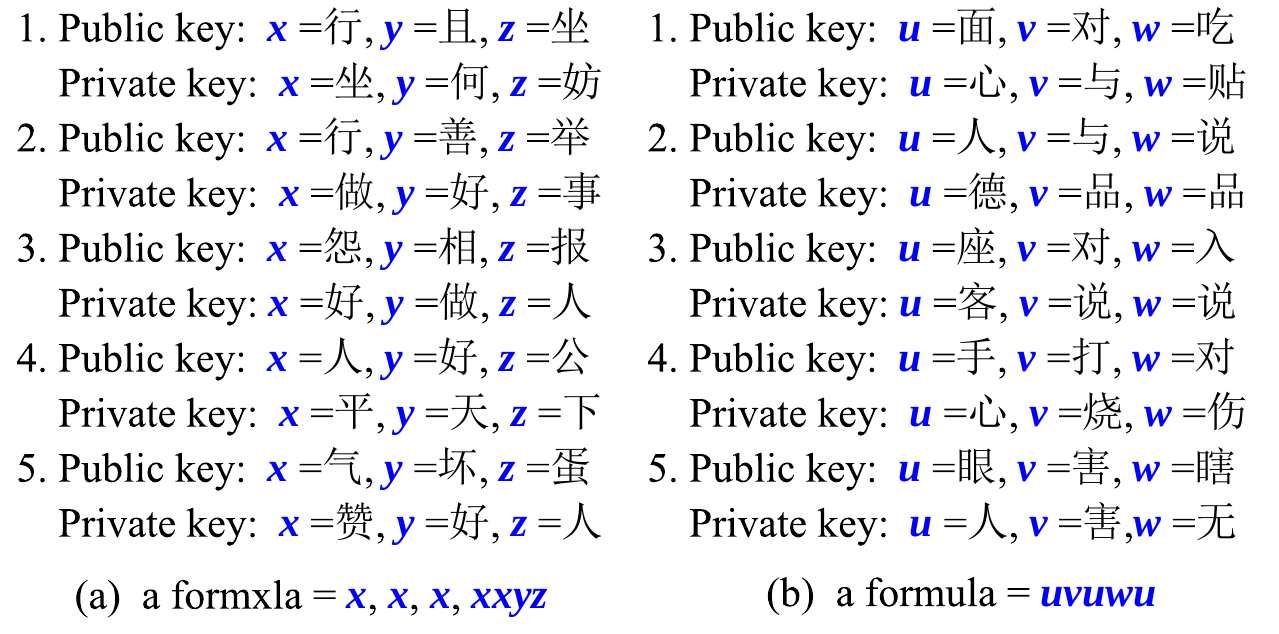}\\
\caption{\label{fig:formula-couplets} {\small Two groups of formula couplets.}}
\end{figure}

\item Suppose that a graph $G$ can be split into some Hanzi-graphs $H_1,H_2,\dots ,H_n$ by the v-split operation, we write $n=vs(G)$. \textbf{Find} $\min \{vs(G)\}$ over all v-splitting results to $G$.
\item In Fig.\ref{fig:2-hanzi-plane}, we put two Hanzis $H_{4535}$ and $H_{8630}$ into xOy-plane, and get their graphic expressions, called \emph{analytic Hanzis}. \textbf{Find} possible applications of analytic Hanzis in graphical passwords or other areas. For example, we can express some strokes of Hanzis by smooth functions having derivatives, it may be interesting in researching Chinese handwriting, ancient Hanzis and `Jia gu wen', and so on.
\item \textbf{Do} both non-isomorphic Hanzi-gpws introduce the same TB-paw?
\item \textbf{Use} Hanzis or Hanzi-gpws to encrypt a dynamic network.
\item Build up hz-$\varepsilon$-graphs for some $\varepsilon$-labelling based on Hanzis from GB2312-80 \cite{GB2312-80}, and then show properties of these hz-$\varepsilon$-graphs.
\item In \cite{YAO-SUN-WANG-SU-XU2018arXiv}, the authors defined: ``Let $\eta$-labeling be a given graph labelling, and let a connected graph $G$ admit an $\eta$-labeling. If every connected proper subgraph of $G$ also admits a labelling to be $\eta$-labeling, then we call $G$ a \emph{perfect $\eta$-labeling graph}.'' Caterpillars are perfect $\eta$-labeling graphs if these $\eta$-labelings are listed in Theorem \ref{defn:group-flawed-labellings}, and each lobster is a perfect (odd-)graceful labeling graph. Conversely, we ask for: \emph{If every connected proper subgraph of a connected graph $G$ admits an $\eta$-labelling, then does $G$ admit this $\eta$-labelling too}?
\item Suppose that each tree $H_i$ of $p$ vertices has $k$ leaves and admits an $\varepsilon$-labelling $f_i$ with $i=1,2$. If any leaf $x'$ of $H_i$ corresponds a leaf $x''$ of $H_{3-i}$ such that $H_i-x'$ is isomorphic to $H_{3-i}-x''$ holding $f_i(y)=f_{3-i}(y)$ for $y\in V(H_{3-i}-x'')=V(H_i-x')$ with $i=1,2$, then can we claim that $H_1\cong H_2$ with $f_1(y)=f_2(y)$ for $y\in V(H_1)=V(H_2)$?

\item How many Hanzi-graphs is a simple graph decomposed  by the operations defined in Definition \ref{defn:split-operation-combinatoric}?
\item How many Hanzi-graphs can a simple graph be decomposed  only by the vertex-split operation defined in Definition \ref{defn:split-operation-combinatoric}?
\item What is a simple graph that can be decomposed into different Hanzi-graphs by the vertex-split operation defined in Definition \ref{defn:split-operation-combinatoric}? Or we consider this question only on trees, see some examples shown in Fig.\ref{fig:split-stars-Hanzi}.

\begin{figure}[h]
\centering
\includegraphics[height=4.8cm]{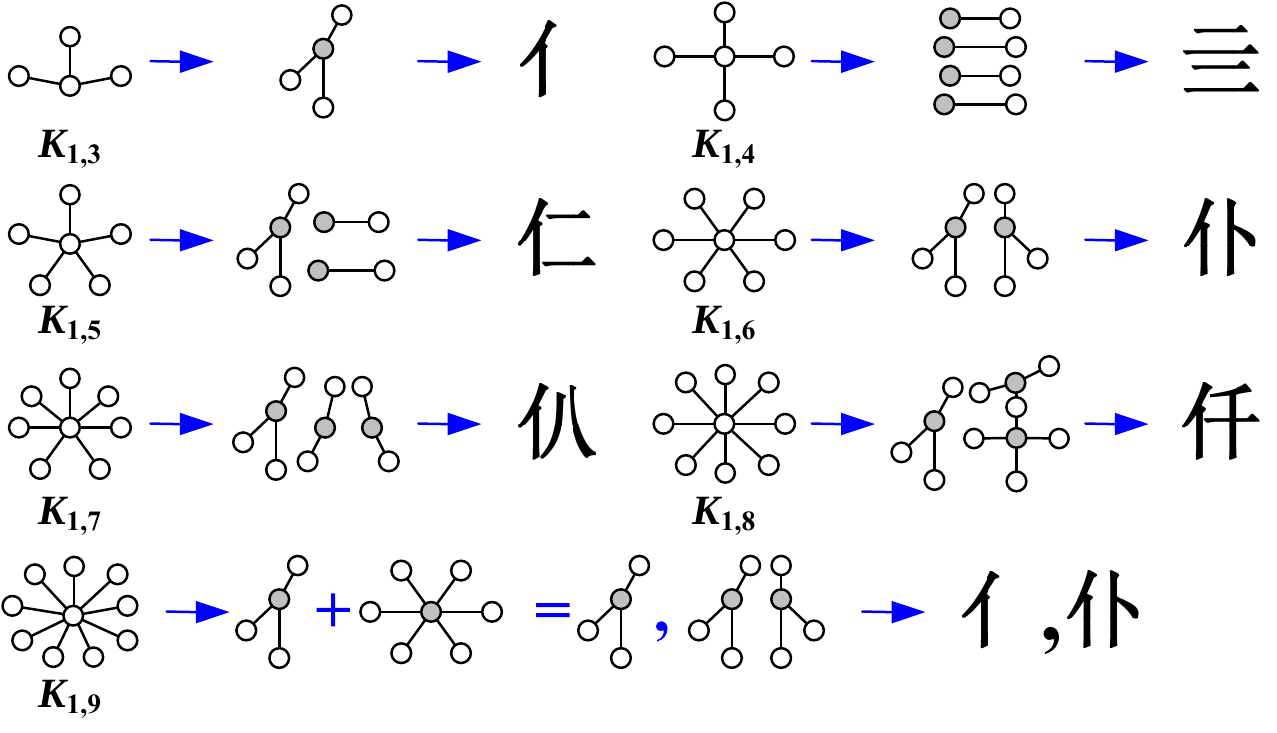}\\
\caption{\label{fig:split-stars-Hanzi} {\small Stars $K_{1,j}$ with $j\in [3,9]$ can be decomposed into Hanzi-graphs only by the vertex-split operation defined in Definition \ref{defn:split-operation-combinatoric}.}}
\end{figure}
\item Since any graph can be decomposed into Hanzi-graphs, \textbf{determine} the smallest number of such Hanzi-graphs from the decomposition of this graph (Ref. \cite{Yao-Sun-Zhang-Mu-Wang-Xu-2018}). We may provide a public key (graph) $G$ to ask for private keys made by some groups of Hanzi-graphs from the decomposition of $G$, and these groups of Hanzi-graphs produce meaningful sentences or paragraphs in Chinese.

\item In Fig.\ref{fig:directed-rrhgztxtp-digraph}, labelling a digraph $\overrightarrow{H}$ with a \emph{Hanzi labelling} $f$ deduces a directed Hanzi-gpw, where $f:V(\overrightarrow{H})\rightarrow H_F$, and $H_F$ is a set of Hanzis defined in \cite{GB2312-80}. For any connected digraph $\overrightarrow{L}$, can we label its vertices with Hanzis such that $\overrightarrow{L}$ can be decomposed into several arc-disjoint digraphs $\overrightarrow{L}_1,\overrightarrow{L}_2,\dots ,\overrightarrow{L}_m$, and each digraph $\overrightarrow{L}_j$ has at least one arc and can express a \emph{complete meaning paragraph} in Chinese (see examples shown in Fig.\ref{fig:directed-rrhgztxtp-digraph}).

\begin{figure}[h]
\centering
\includegraphics[height=9cm]{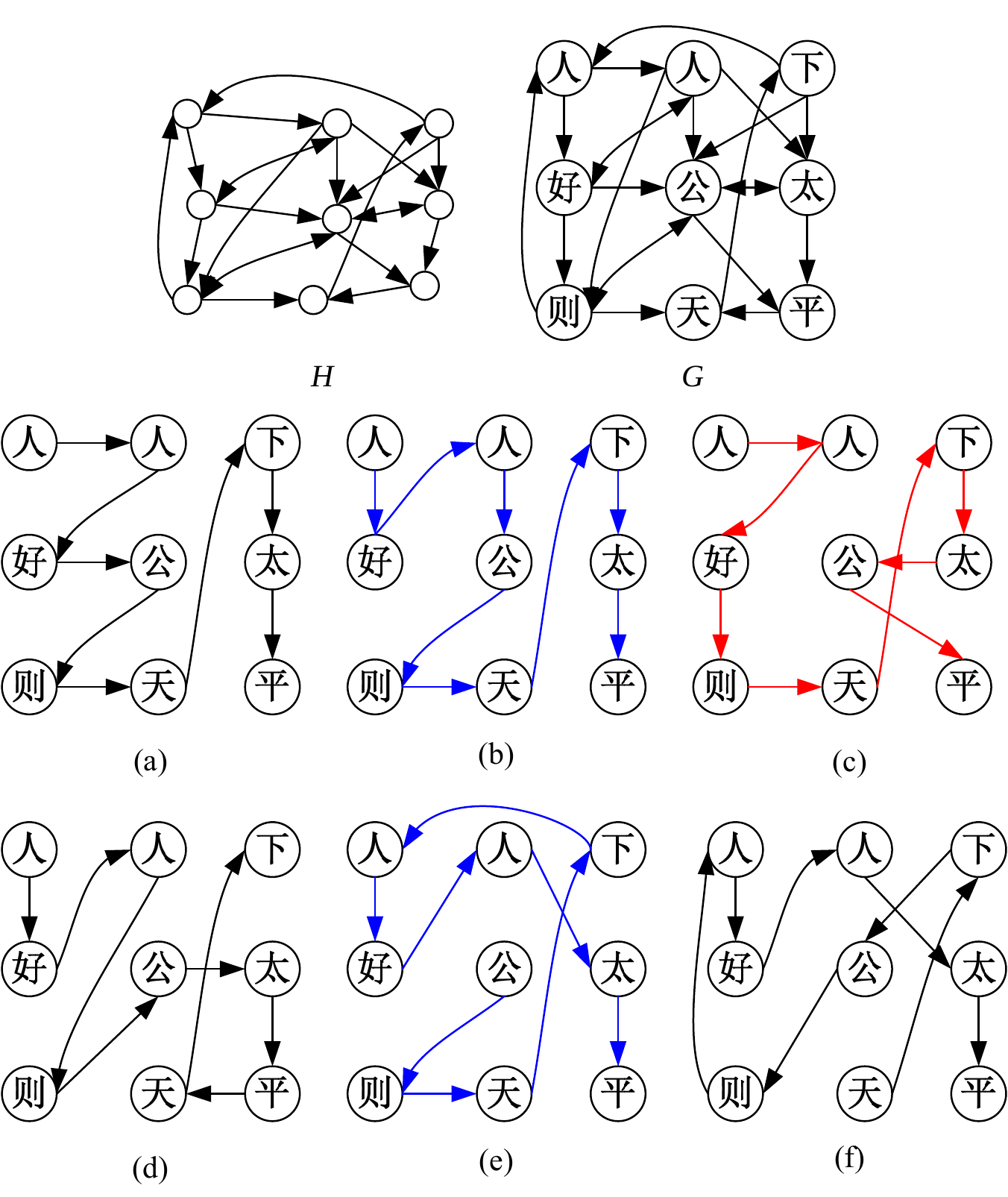}\\
\caption{\label{fig:directed-rrhgztxtp-digraph} {\small A digraph $G$ labellined with Hanzis is made by six directed paths shown in (a)-(f).}}
\end{figure}

\item In Definition \ref{defn:generalization-flawed-labellings-00}, it is easy to see $m-1=\min\{|E_j|:j\in [1,n]\}$, we hope to determine $\max \{|E_j|:j\in [1,n]\}$ for each disconnected graph $G=\bigcup^m_{i=1}G_i$ having disjoint components $G_1,G_2,\dots, G_m$.
\item \label{prob:anti-problem} How to replace each triangular vertex (see Fig.\ref{fig:anti-problem}) of a \emph{public key} $T$ being a labelled graph with a total labelling $f$ by a connected and labelled graph to obtain a \emph{private key} $H$ (Hanzi-graph) with a set-ordered graceful labelling $g$, such that the disconnected graph $H-E(T)$ admits a flawed set-ordered graceful labelling $h$ holding $f(V(T))\subset g(V(T))=h(V(H-E(T)))$ true? We have the following questions:

\quad Problem-\ref{prob:anti-problem}-1. Fig.\ref{fig:anti-problem} shows two public keys $T_1$ and $T_2$ to ask for solutions, and Fig.\ref{fig:anti-problem-solutions} shows two private keys $H_1$ and $H_2$ to answer the questions proposed above.

\quad Problem-\ref{prob:anti-problem}-2. In general, the number of connected private keys corresponding a public key $T$ may be two or more. In other word, a public key (Topsnut-gpw) $T$ with triangular vertices and labelled well may have private keys (Topsnut-gpws) $H_1,H_2,\dots ,H_m$ with $m\geq 2$, such that $H_i$ is not isomorphic to $H_j$ if $i\neq j$.

\begin{figure}[h]
\centering
\includegraphics[height=2.6cm]{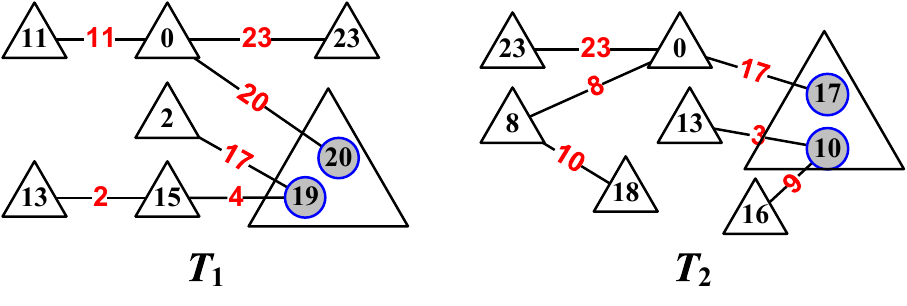}\\
\caption{\label{fig:anti-problem} {\small Two public key Topsnut-gpw $T_1$ and $T_2$ for finding private key Topsnut-gpws.}}
\end{figure}

\begin{figure}[h]
\centering
\includegraphics[height=6.8cm]{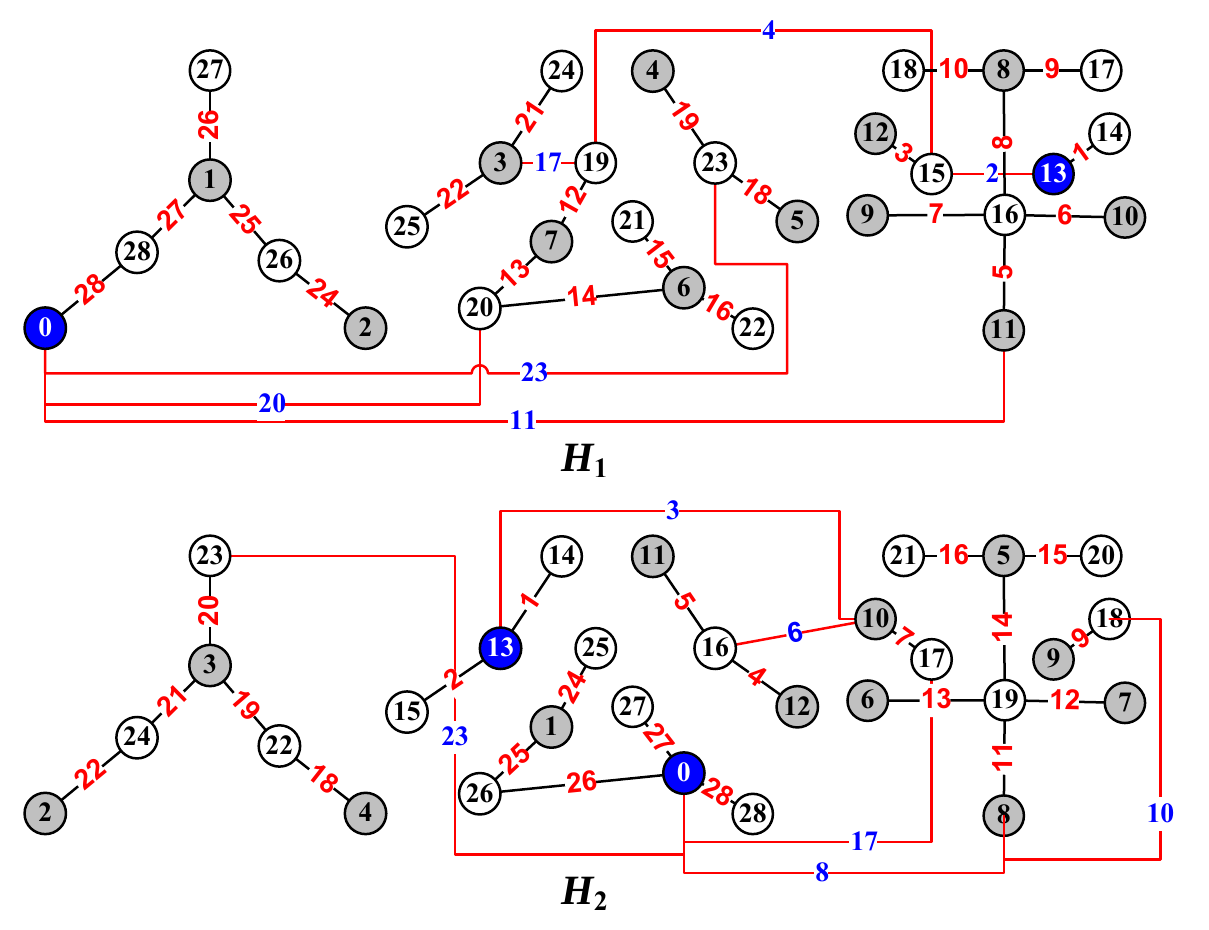}\\
\caption{\label{fig:anti-problem-solutions} {\small Two private key Topsnut-gpws $H_1$ and $H_2$ as two solutions for the question based on two public key Topsnut-gpws $T_1$ and $T_2$ shown in Fig.\ref{fig:anti-problem}.}}
\end{figure}

\item \textbf{Determine} graphs admitting $0$-rotatable set-ordered system of (odd-)graceful labellings. If every tree admits $0$-rotatable system of graceful labellings, then it sets down a famous conjecture ``\emph{Every tree admits a graceful labelling}'' proposed first by Rosa in 1967 (Ref. \cite{A-Rosa-1966}), and most graph labeling approaches trace their origin to a labelling (also, $\beta$-valuation) introduced by Rosa (Ref. \cite{Gallian2018}). However, determining graphs admitting $0$-rotatable set-ordered system of (odd-)graceful labellings is far more difficult to finding a graceful labelling for a graph, even for trees.

\item \textbf{Find} graph labellings of \emph{Fibonacci self-similar Hanzi-networks} obtained by FIBONACCI-EDGE algorithm or FIBONACCI-VERTEX algorithm.
\item \label{prob:matching-pairs} Given a disconnected graph $H=\bigcup^m_{j=1}K^j_{2}$ having disjoint subgraphs $K^1_{2},K^2_{2},\dots ,K^m_{2}$ holding $K^j_{2}\cong K_2$ for $j\in[1,m]$. Assume that $H$ admits a total labelling $h:V(H)\cup E(H)\rightarrow [0,M]$, \textbf{find} another disconnected graph $G=\bigcup^{m+1}_{i=1}G_{i}$ having disjoint subgraphs $G_{1},G_{2},\dots ,G_{m+1}$ and admitting a flawed graceful labelling $g$, such that the graph $H\odot G$, obtained by coinciding the vertices of two graphs $H$ and $G$ labelled with the same labels into one vertex, is connected, and admits a graceful labelling $F$ deduced by $h$ and $g$, write $F=\langle h\odot g \rangle$, and holds $V(H\odot G)=V(G)$ and $E(H\odot G)=E(H)\cup E(G)$. See such a disconnected graph $H$ (as a public key) shown in Fig.\ref{fig:matching-problem-0} and another disconnected graph $G$ (as a private key) shown in Fig.\ref{fig:matching-problem-1}. It is not hard to see that there are two or more private keys like $G$ shown in Fig.\ref{fig:matching-problem-1} to be the matchings with the public key $H$. In fact, there are two or more such matching pairs $\langle H,G\rangle$.
\item \label{prob:group-encryption} For a $(p,q)$-graph $H$ in Definition \ref{defn:graphic-group-coloring}, \textbf{determine} parameters $\chi _{gg}(H)$, $\chi' _{gg}(H)$, $\chi'' _{gg}(H)$, $\chi _{ggs}(H)$, $\chi _{ggas}(H)$, $\chi' _{ggs}(H)$, $\chi' _{ggas}(H)$, $\chi _{ggeq}(H)$, $\chi' _{ggeq}(H)$ and $\chi'' _{ggas}(H)$.

\item In Fig.\ref{fig:pseudo-Hanzi-pianpang}, we show several \emph{pseudo Hanzis} that are not allowed to use in China. Clearly, these pseudo Hanzis can be decomposed into the pianpangs in \cite{GB2312-80}, so each graph can be decomposed into Hanzi-graphs and so-called \emph{pseudo Hanzi-graphs}. We cite a short introduction cited from WWW. Baidu Encyclopedia as: ``Xu Bing, a professor at the Central Academy of Fine Arts, has personally designed and printed thousands of `\emph{new Hanzis}' to explore the essence and mode of thinking of Chinese culture in terms of graphs and symbolism, which has become a classic in the history of Chinese contemporary art.'' Here, `\emph{new Hanzis}' are pseudo Hanzis. Some pseudo Hanzis shown in Fig.\ref{fig:pseudo-xubing} and Fig.\ref{fig:pseudo-xubing-11} were made by English characters of Bing Xu. In Fig.\ref{fig:pseudo-xubing-11}, ``\emph{Nothing gold can stay so dawn goes down to day leaf so Eden Sanr to grief an hour, then leaf subsides to leafs a flower, but only so hardest hue to hold her early natures, first green is gold her nothing gold can stay Robert frost}. Calligraphy by Xu Bing.''

\begin{figure}[h]
\centering
\includegraphics[width=8cm]{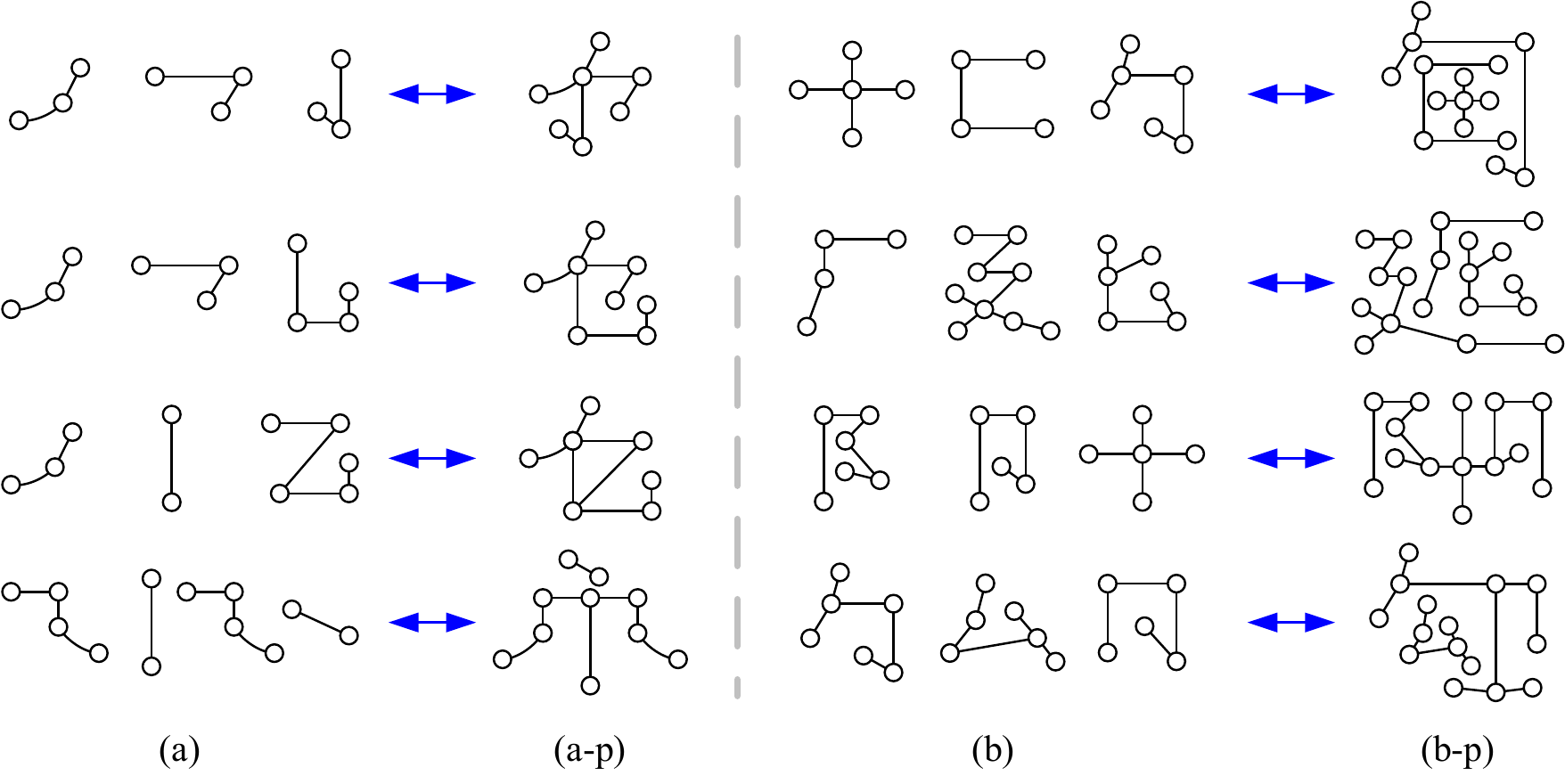}\\
\caption{\label{fig:pseudo-Hanzi-pianpang} {\small (a) On-stroke pianpangs; (a-p) pseudo Hanzis made by (a); (b) two-stroke pianpangs; (b-p) pseudo Hanzis made by (b).}}
\end{figure}
\begin{figure}[h]
\centering
\includegraphics[width=8.2cm]{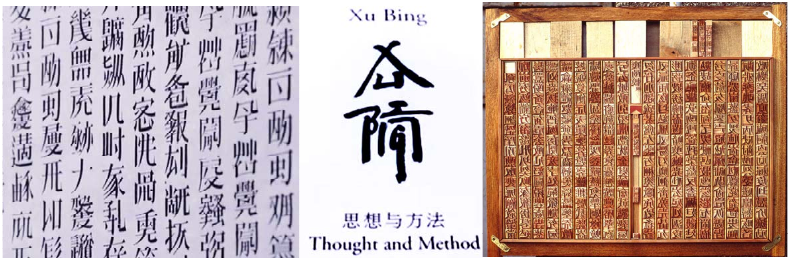}\\
\caption{\label{fig:pseudo-xubing} {\small Prof. Xu's pseudo Hanzis.}}
\end{figure}

\begin{figure}[h]
\centering
\includegraphics[width=8.2cm]{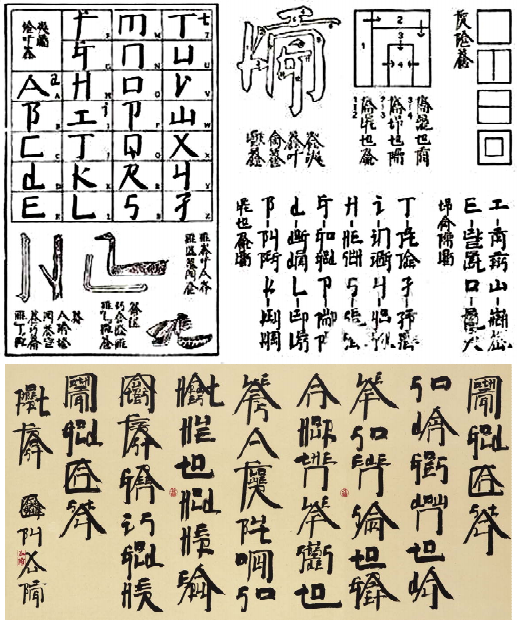}\\
\caption{\label{fig:pseudo-xubing-11} {\small Top is a connection of English letters and Hanzis made by Prof.  Xu Bing; Lower is a calligraphy in Prof.  Xu's pseudo Hanzis  written by Prof. Xu Bing.}}
\end{figure}

\item We show some other type of connected parameters for further researching graph connectivity:
\begin{defn}\label{defn:H-connected}
Let $H$ be a proper subgraph of a graph $G$. We say $G$ to be $H$-connected if $G-V(H)$ is disconnected. Particularly, we have new parameters:

(1) $\kappa_{cycle}(G)=\min \{m:~\textrm{$G$ is $C_m$-connected}\}$, where each $C_m$ is a cycle of $m$ vertices in $G$.

(2) $\kappa_{path}(G)=\min \{n:~\textrm{$G$ is $P_n$-connected}\}$, where each $P_n$ is a path of $n$ vertices in  $G$.

(3) $\kappa_{tree}(G)=\min \{|V(T)|:~\textrm{$G$ is $T$-connected}\}$, where each $T$ is a tree of $G$.

(4) $\kappa_{indep}(G)=\min \{|V(H)|:~\textrm{$G$ is $H$-connected}\}$, where each $H$ has no edge in $G$.

(5) $\kappa_{edge}(G)=\min \{|E(H)|:~\textrm{$G$ is $H$-connected}\}$.
\end{defn}

\item \textbf{Labelling Hanzi-graphs with Chinese phrases and idioms.} We label the vertices of a $(p,q)$-graph $G$ with  Chinese phrases and idioms, and obtain a labelling $\phi$ such that  each edge $uv\in E(G)$ is labelled as $\phi(uv)=\phi(u)\cap \phi(v)$, where $\phi(u)\cap \phi(v)$ means the common Hanzis.  So, we called $\phi$ \emph{Chinese phrase/idiom labelling} (or \emph{Hanzi-labelling}), and get a new graphic password, called as \emph{a password made by topological structure plus Chinese phrases/idioms}. In Fig.\ref{fig:idiom-labelling}, a graph $G$ admits an \emph{idiom labelling} $\phi:V(G)\rightarrow H_{idiom}$, where $H_{idiom}$ is the set of all Chinese idioms. Another example is shown in Fig.\ref{fig:6-idiom-graph}. It is not hard to see that Fig.\ref{fig:idiom-labelling} and Fig.\ref{fig:6-idiom-graph} show two different ways for Hanzi-idiom labellings on graphs.

\begin{figure}[h]
\centering
\includegraphics[width=8.4cm]{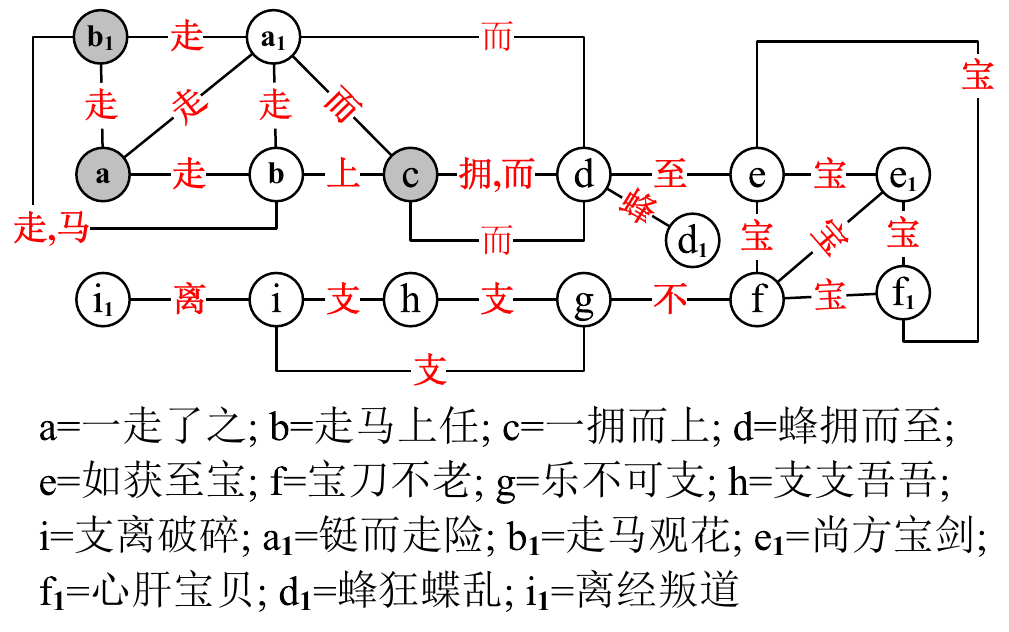}\\
\caption{\label{fig:idiom-labelling} {\small A  graph admitting a \emph{Hanzi-idiom labelling}.}}
\end{figure}
\begin{figure}[h]
\centering
\includegraphics[width=8.4cm]{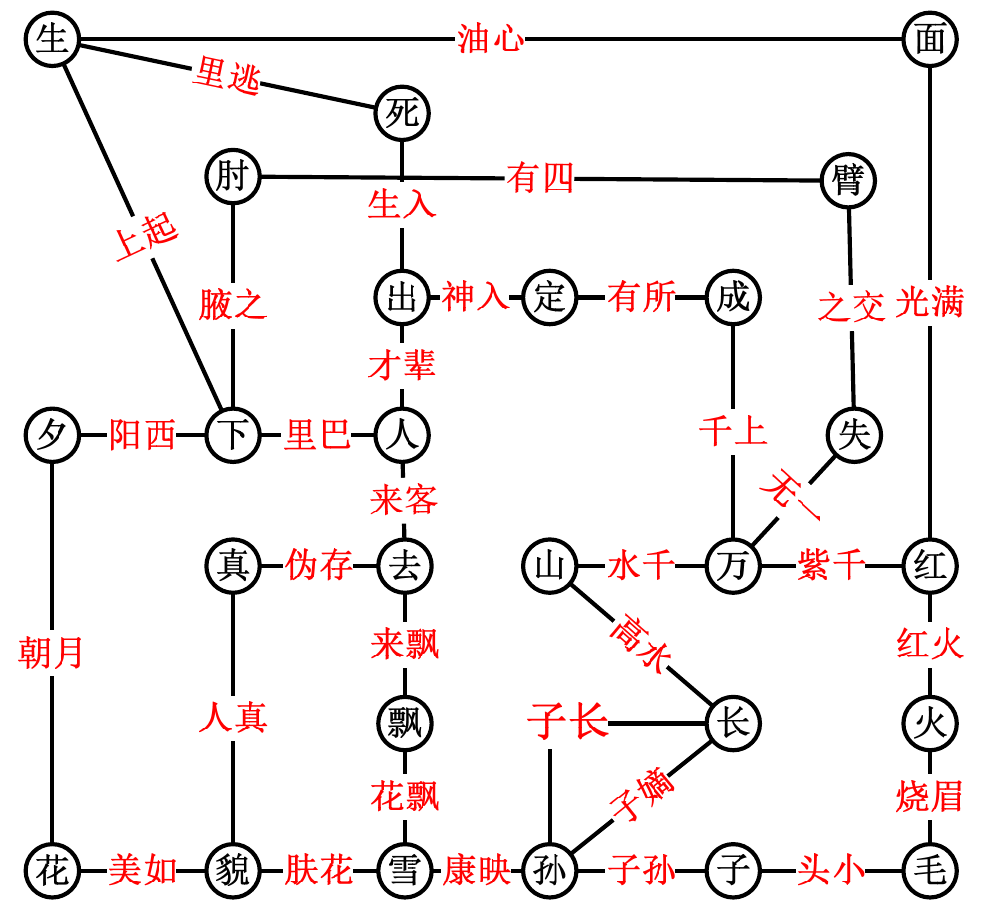}\\
\caption{\label{fig:6-idiom-graph} {\small Another  graph admitting a  \emph{Hanzi-idiom labelling}.}}
\end{figure}
\item A group of topological structures in Fig.\ref{fig:1-group-more-flawed-graceful}(a) can be considered as a public key, so we want to find some private keys shown in  Fig.\ref{fig:1-group-more-flawed-graceful}(b). Hence, for a given group of topological structures, \textbf{how many} paragraphs written in Chinese are there?
\end{asparaenum}

\subsection{Discussion}

\begin{figure}[h]
\centering
\includegraphics[height=2cm]{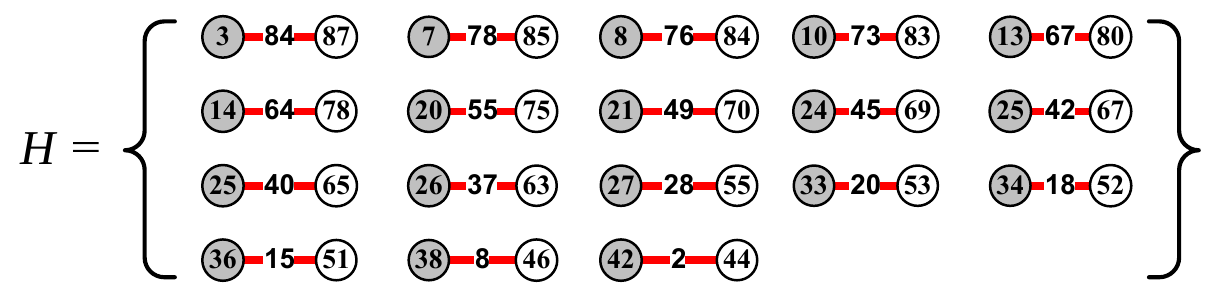}\\
\caption{\label{fig:matching-problem-0} {\small A disconnected graph $H$ admitting a total labelling.}}
\end{figure}

\begin{figure}[h]
\centering
\includegraphics[height=2.8cm]{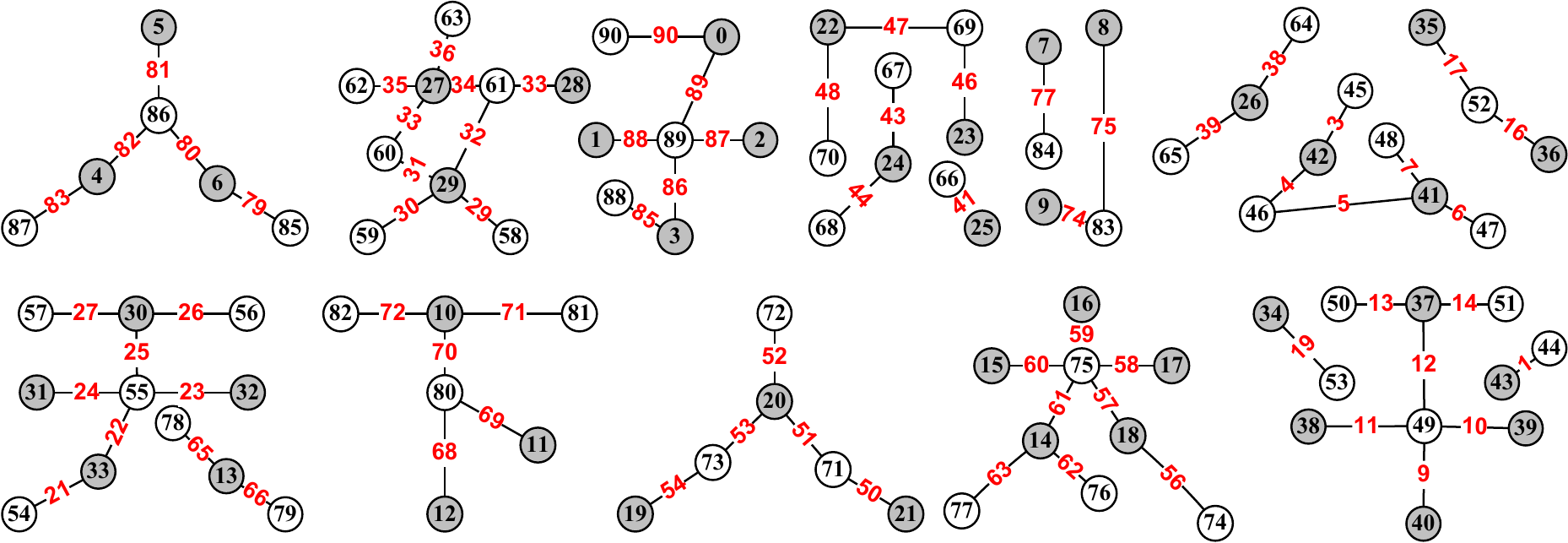}\\
\caption{\label{fig:matching-problem-1} {\small A disconnected graph $G$ admitting a flawed graceful labelling.}}
\end{figure}

(1) Problem-\ref{prob:matching-pairs} can be planted into the graph labellings defined or listed here and others defined in \cite{Gallian2018}. We show some solutions for Problem Problem-\ref{prob:matching-pairs} below.

\begin{thm}\label{thm:G-matchs-K2s}
Given a disconnected graph $H=\bigcup^m_{j=1}K^j_{2}$ having disjoint subgraphs $K^1_{2},K^2_{2},\dots ,K^m_{2}$ holding $K^j_{2}\cong K_2$ for $j\in[1,m]$. Assume that $H$ admits a total labelling $h:V(H)\cup E(H)\rightarrow [0,M]$, such that $h(xy)=|h(x)-h(y)|>1$ for each edge $xy\in E(H)$. We rewrite the edge set $E(H)=\{x_1y_1,x_2y_2,\dots,x_my_m\}$ with $E(K^i_{2})=\{x_iy_i\}$ and $h(x_iy_i)=h(y_i)-h(x_i)>1$ for $j\in[1,m]$, as well as $h(x_j)<h(x_{j+1})$ and $h(y_j)<h(y_{j+1})$ with $j\in [1,m-1]$.

If $1<h(x_jy_j)<h(x_{j+1}y_{j+1})$ with $j\in [1,m-1]$, then we can find another disconnected graph $G=\bigcup^{m+1}_{i=1}G_{i}$ having disjoint subgraphs $G_{1},G_{2},\dots ,G_{m+1}$ such that $G$ admits a flawed graceful labelling $g$. Furthermore, we have a connected graph $H\odot G$ obtained by coinciding the vertices of two graphs $H$ and $G$ labelled with the same labels into one, such that $H\odot G$ holds $V(H\odot G)=V(G)$ and $E(H\odot G)=E(H)\cup E(G)$, and admits a graceful labelling $f=\langle h\odot g \rangle$ deduced by two labellings $h$ and $g$.
\end{thm}
\begin{proof} Let $a_i=h(x_iy_i)$ in the following process. We use stars $K^i_{1,n_i}$ with $i\in [1,n]$ to complete the proof, and let $K^i_{1,n_i}$ has its own vertex set and edge set as

$V(K^i_{1,n_i})=\{u_{i,1},v_{i,1},v_{i,2},\dots ,v_{i,n_i}\}$ and

$E(K^i_{1,n_i})=\{u_{i,1}v_{i,j}:~j\in [1,n_i]\}$, $i\in [1,n]$.

\vskip 0.2cm

\emph{Step 1.} In $K^1_2=x_1y_1$ of $H$, we have $h(x_1y_1)=h(y_1)-h(x_1)>1$ and $h(x_2)<h(y_1)$.

We join $u_{2,1}$ of $K^2_{1,n_2}$ with $y_1$, and join $u_{1,1}$ of $K^1_{1,n_1}$ with $x_1$, thus, we get a tree $H_1=K^2_{1,n_2}+x_1y_1+K^1_{1,n_1}$. We define a labelling $F_1$ as:

(i) label $F_1(u_{2,1})=0$, $F_1(v_{2,j})=h(y_1)-j$ with $j\in [1,n_2]$, such that $F_1(u_{2,1}v_{2,n_2})=h(y_1)-n_2=a_1+1$, also $n_2=h(y_1)-a_1-1$.

(ii) label $F_1(u_{1,1})=1$, $F_1(v_{1,i})=i+1$ with $i\in [1,h(x_1)-2]$, $F_1(v_{1,i})=i+1$ with $i\in [h(x_1),n_1]$, such that $F_1(u_{1,1}v_{2,n_1})=F_1(v_{2,n_1})-1=a_1-1$, also $n_1+1=a_1$, $n_1=a_1-1$.

So, $F_1(E(K^1_{1,n_1}))=[1,a_1-1]$, $F_1(E(K^2_{1,n_2}))=[a_1+1,h(y_1)-1]$, notice $F_1(u_{2,1}y_1)=h(y_1)$, $F_1(u_{1,1}x_1)=h(x_1)-1$ and $h(x_1y_1)=a_1$, we get $F_1(E(H_1))=[1,h(y_1)]$.

If $h(x_2)>h(y_1)$, we add new vertices $v'_{2,k}$ ($k\in [1,m_1]$) to join each of them with $u_{2,1}$, and label $F_1(v'_{2,k})=h(y_1)+k$, such that $F_1(v'_{2,m_1})=h(x_2)+a_2-1$, the resultant tree is denoted as $H'_1=H_1+\{u_{2,1}v'_{2,k}:k\in [1,m_1]\}$ with a graceful labelling $F_1$ ($F_1(E(H'_1))=[1,h(x_2)+a_2-1]$). (See an example shown in Fig.\ref{fig:G-matchs-K2s-1})

\emph{Step 2.} Consider the second graph $K^2_2=x_2y_2$ of $H$, $a_2=h(x_2y_2)=h(y_2)-h(x_2)$, we join $y_2$ with $x_2$ of $H_1$ (or $H'_1$ if $h(x_3)>h(y_2)$), next we join $u_{3,1}$ of $K^3_{1,n_3}$ with $y_2$ for forming the resultant tree $H_2=K^3_{1,n_3}+H_1$ (or $H_2=K^3_{1,n_3}+H'_1$ if $h(x_3)>h(y_2)$). We define a labelling $F_2$ for $H_2$ in the following way:

(1) $F_2(w)=F_1(w)+1$ for $w\in V(K^1_{1,n_1})$;

(2) $F_2(z)=F_1(z)+1$ for $z\in V(K^1_{2,n_1})$;

(3) $F_2(\alpha )=F_1(\alpha)$ for $\alpha\in \{x_1,x_2,y_1,y_2\}$;

(4) $F_2(u_{3,1})=0$, $F_2(v_{3,j})=h(y_2)-j$ with $j\in [1,n_3]$, such that $F_2(u_{3,1}v_{3,n_3})=h(y_2)-n_3=a_2+1$, also $n_2=h(y_2)-a_2-1$.

Thereby, $H_2=H_1+K^3_{1,n_3}$ (or $H'_2=H_2+\{u_{2,1}v'_{2,k}:k\in [1,m_2]\}$) admits a graceful labelling $F_2$, $H'_2\setminus E(H_2)$ admits a flawed graceful labelling. (See an example shown in Fig.\ref{fig:G-matchs-K2s-2})

\emph{Step 3.} Iteration. We use the method introduced in Step 1 and Step 2, until to a graceful tree $H\odot G$, where $H=\bigcup^m_{j=1}K^j_{2}$ admits a total labelling $h$ and $G=\bigcup^{m+1}_{i=1}G_{i}$ admits a flawed graceful labelling $g$.
\end{proof}
\begin{figure}[h]
\centering
\includegraphics[height=4.2cm]{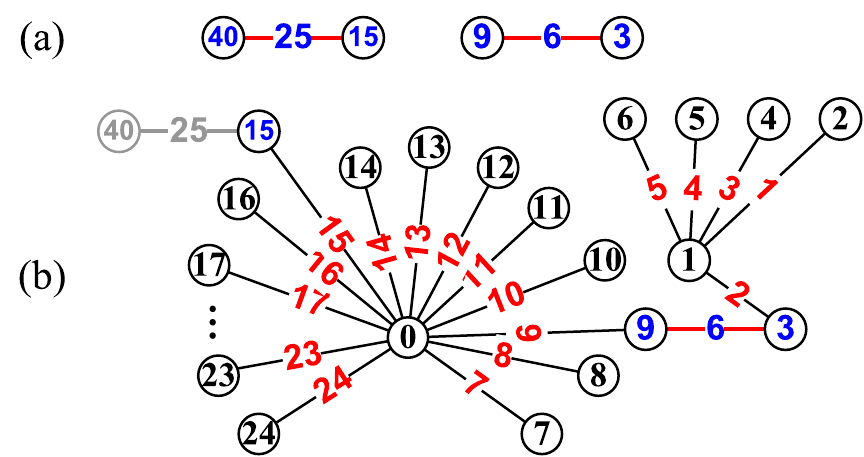}\\
\caption{\label{fig:G-matchs-K2s-1} {\small (a) $H=K^1_{2}\cup K^2_{2}$; (b) $H'_1$.}}
\end{figure}

\begin{figure}[h]
\centering
\includegraphics[height=10cm]{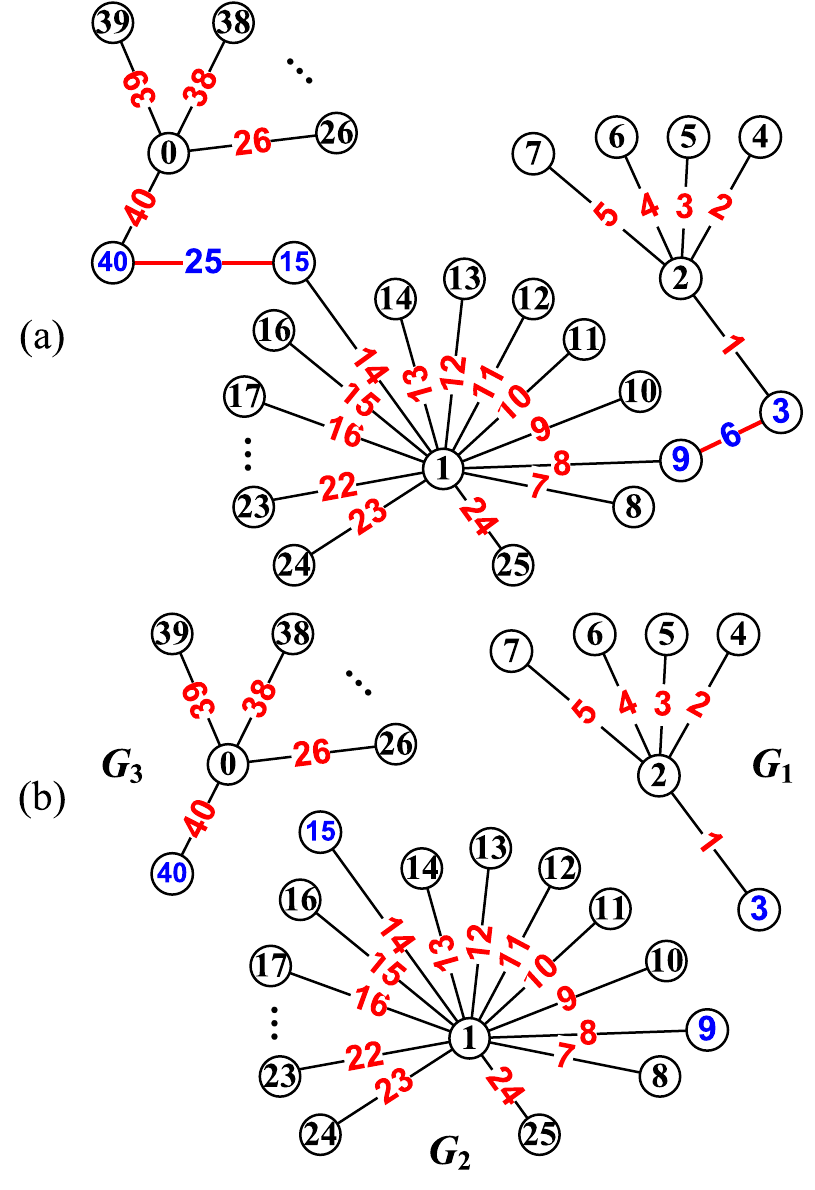}\\
\caption{\label{fig:G-matchs-K2s-2} {\small (a) $H_2$; (b) $G=\bigcup^3_{i=1}G_i$.}}
\end{figure}

\begin{rem}\label{thm:CCCCCC}
(1) In the proof of Theorem \ref{thm:G-matchs-K2s}, we use stars to provide a solution. In general, assume that each tree $T_k$ of $p_k$ vertices admits a set-ordered graceful labelling $f_k$ such that $\max f_k(A_k)<\min f_k(B_k)$, where $V(T_k)=A_k\cup B_k$, $A_k=\{u_{k,1},u_{k,2},\dots ,u_{k,s_k}\}$ and
$$B_k=\{w_{k,1},w_{k,2},\dots ,w_{k,t_k}\},~p_k=|V(T_k)|=s_k+t_k,$$
as well as $f_k(u_{k,i}w_{k,j})=f_k(w_{k,j})-f_k(u_{k,i})$ for each edge $u_{k,i}w_{k,j}\in E(T_k)$. So, we can replace stars $G_k$ by some non-stars $T_k$. (See an example shown in Fig.\ref{fig:G-matchs-K2s-no-star})

\begin{figure}[h]
\centering
\includegraphics[height=5.4cm]{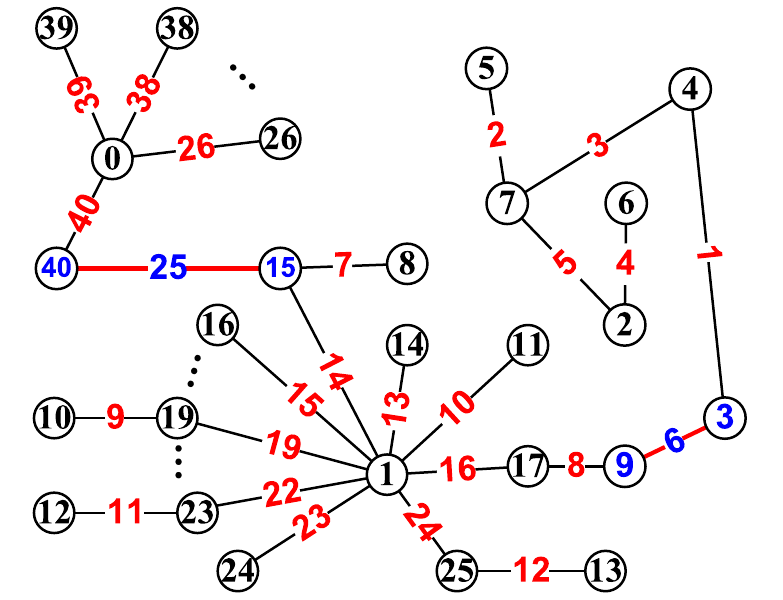}\\
\caption{\label{fig:G-matchs-K2s-no-star} {\small Some $G_k$ are not stars.}}
\end{figure}

(2) Find conditions for the tree $H\odot G$ in Theorem \ref{thm:G-matchs-K2s} to admit a set-ordered graceful labelling.

(3) The result of Theorem \ref{thm:G-matchs-K2s} is about graceful labellings, however, we can replace graceful labelling by odd-graceful labelling, elegant labelling, odd-elegant labelling, edge-magic total labelling, and so on.

\begin{figure}[h]
\centering
\includegraphics[height=5.4cm]{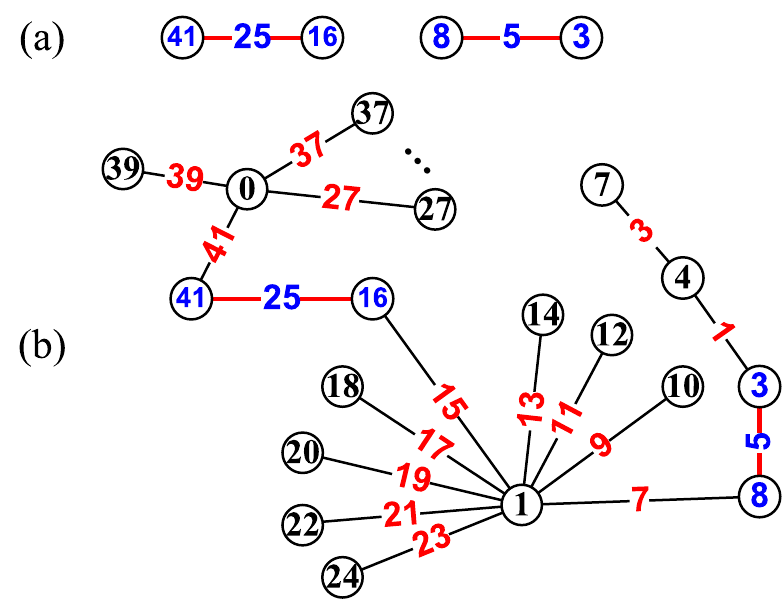}\\
\caption{\label{fig:G-matchs-K2s-odd-graceful} {\small $H\odot G$ admits an odd-graceful labelling, $G=(H\odot G)-E(H)$ admits a flawed odd-graceful labelling.}}
\end{figure}
\end{rem}

(2)  Outline Hanzis (hollow font) are popular in Chinese cultural. Some Hanzis can be obtained by cycles $C_n$ with some labellings and Euler graphs (see Fig.\ref{fig:hollow-ren}). We can split some Euler-gpws into hollowed Hanzi-gpws. Two Euler-gpws with v-set e-proper graceful labellings shown in  Fig.\ref{fig:hollow-ren}(c) and Fig.\ref{fig:hollow-wang}(c) can be split into two hollowed Hanzi-gpws  shown in  Fig.\ref{fig:hollow-ren}(b) and Fig.\ref{fig:hollow-wang}(b), conversely, we can coincide some vertices of a hollowed Hanzi-gpw into one for producing Euler-gpws with v-set e-proper graceful labellings. \textbf{Determine} v-set e-proper $\varepsilon$-labellings for Euler graphs with $\varepsilon \in \{$graceful, odd-graceful, elegant, odd-elegant, edge-magic total, $k$-graceful, 6C, etc.$\}$ (Ref. \cite{Gallian2018, Yao-Sun-Zhang-Mu-Sun-Wang-Su-Zhang-Yang-Yang-2018arXiv, Yao-Zhang-Sun-Mu-Sun-Wang-Wang-Ma-Su-Yang-Yang-Zhang-2018arXiv}).  As known, the cycles $C_{4m+1}$ and $C_{4m+2}$ are not graceful. Rosa  \cite{Gallian2018} showed that the $n$-cycle $C_{n}$ is graceful if and only if $n\equiv 0$ or 3 ($\bmod~4$).

\begin{figure}[h]
\centering
\includegraphics[width=8cm]{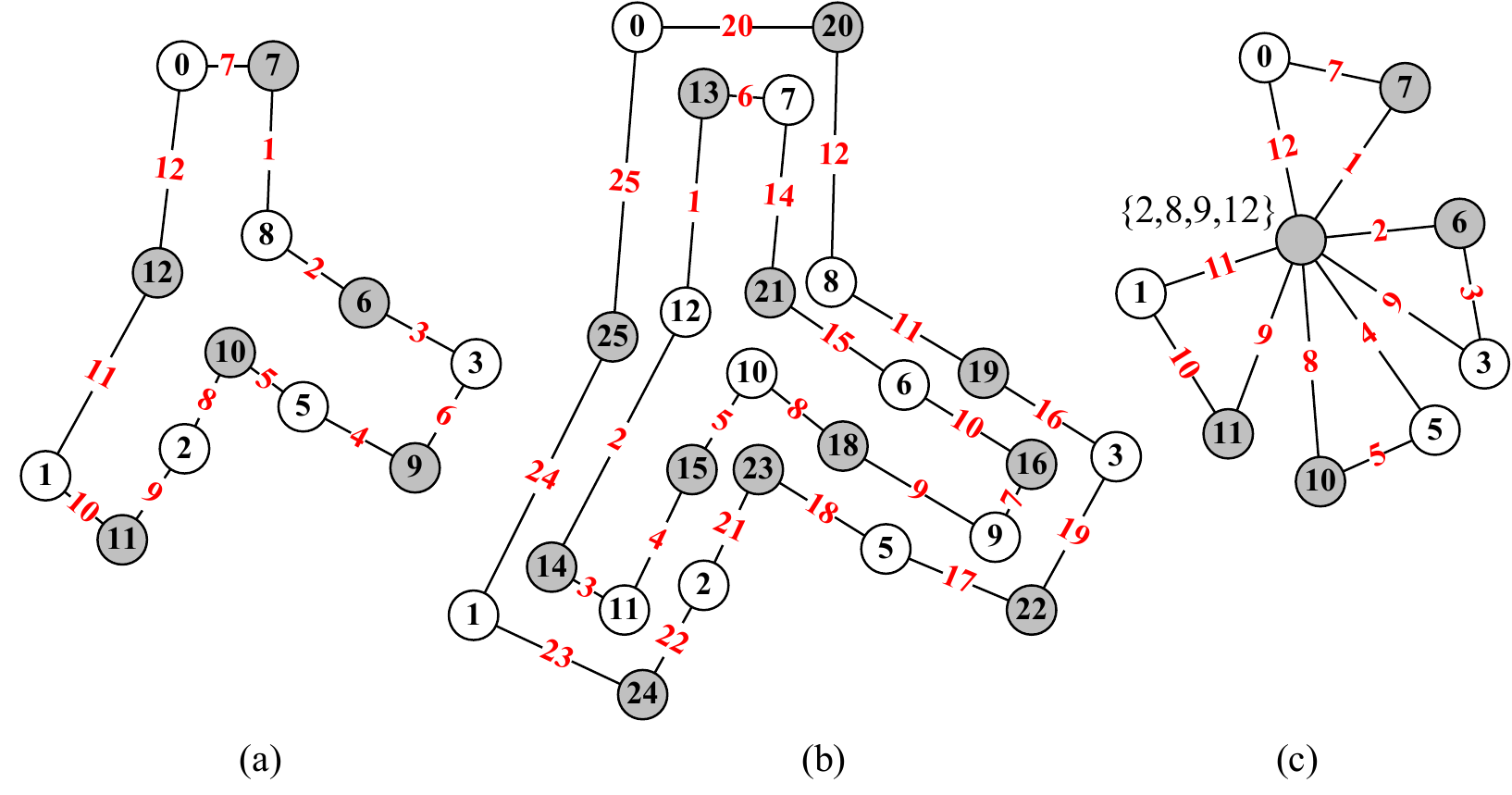}\\
\caption{\label{fig:hollow-ren} {\small (a) A Hanzi-gpw $H^{gpw}_{4043}$; (b) a hollowed Hanzi-gpw with a  flawed set-ordered graceful labelling based on Hanzi-graph $T_{4043}$; (c) an Euler-gpw with a  v-set e-proper graceful labelling.}}
\end{figure}
\begin{figure}[h]
\centering
\includegraphics[width=8cm]{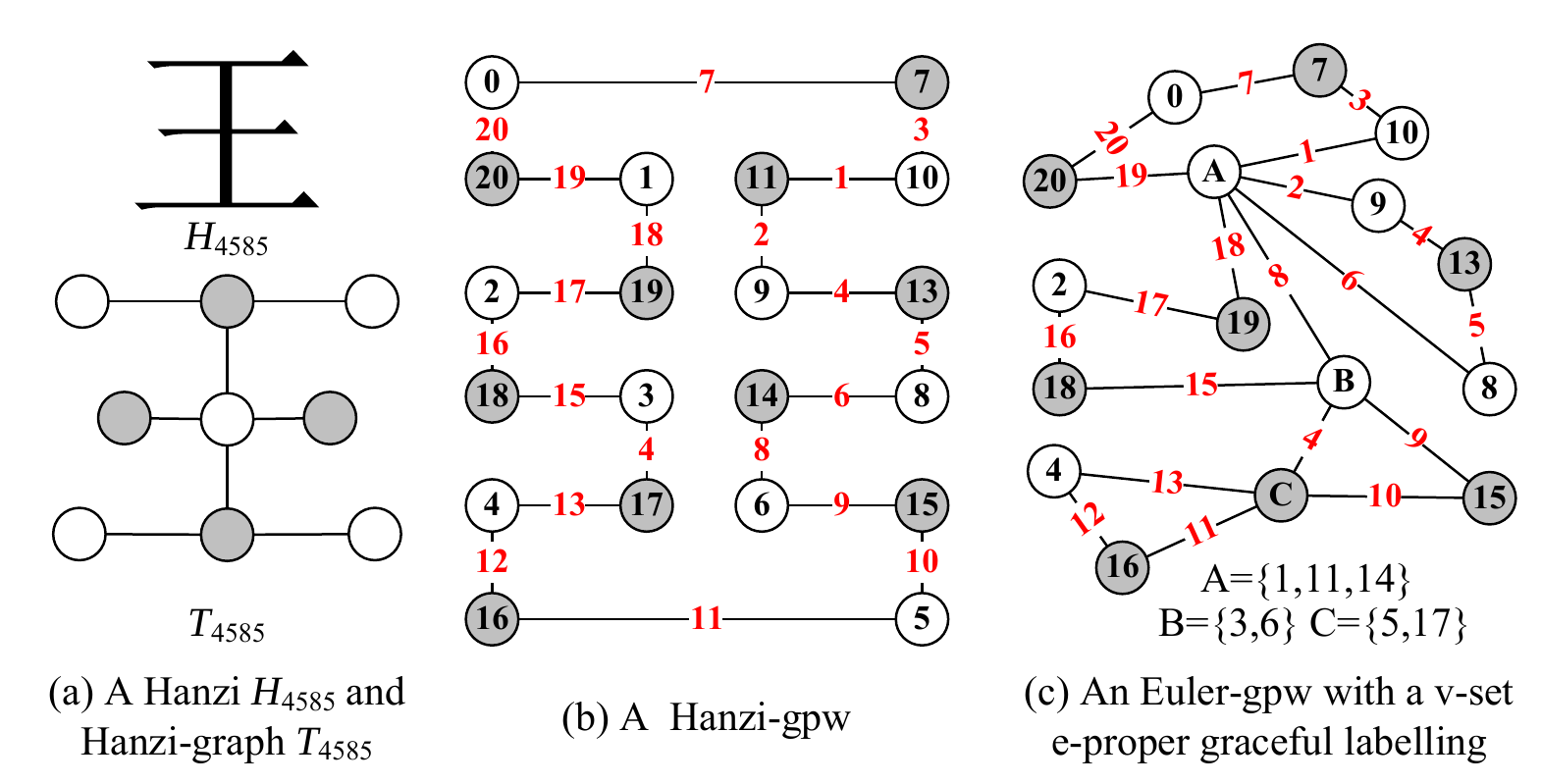}\\
\caption{\label{fig:hollow-wang} {\small (a) A Hanzi $H_{4585}$ and its Hanzi-graph $H_{4585}$; (b) a hollowed Hanzi-gpw with a graceful labelling based on Hanzi $H_{4585}$; (c) an Euler-gpw with a  v-set e-proper graceful labelling.}}
\end{figure}

In \cite{Wang-Wang-Yao2019-Euler-Split}, the authors answer the above question, they show:

\begin{thm}\label{thm:Euler-v-set-e-proper-graceful-labelling}
\emph{\cite{Wang-Wang-Yao2019-Euler-Split}} Every connected Euler's graph of $n$ edges with $n\equiv 0,3$ ($\bmod~4$) admits a v-set e-proper graceful labelling.
\end{thm}

\begin{thm}\label{thm:Euler-v-set-e-proper-harmonious-labelling}
\emph{\cite{Wang-Wang-Yao2019-Euler-Split}} Every connected Euler's graph with odd $n$ edges admits a v-set e-proper harmonious labelling.
\end{thm}

\begin{defn}\label{defn:4C-k-labelling}
$^*$ A $(p,q)$-graph $G$ admits a vertex labelling $f:V(G)\rightarrow \{0,1,\dots ,q\}$, and an edge labelling $g:E(G)\rightarrow \{1,2,\dots \}$, such that

(1) (e-magic) each edge $xy\in E(G)$ holds $g(xy)+|f(x)-f(y)|=k$, a constant;

(2) (ee-balanced) let $s(uv)=|f(u)-f(v)|-g(uv)$ for an edge, and each edge $uv$ matches another edge $xy$ holding $s(uv)+s(xy)=k'$ true;

(3) (EV-ordered)  $\max f(V(G))< \min g(E(G))$ (or $\min f(V(G))> \max g(E(G))$);

(4) (set-ordered) $\max f(X)<\min f(Y)$ for the bipartition $(X,Y)$ of $V(G)$.

We call the labelling pair $(f,g)$ as a \emph{$(4C,k,k')$-labelling}.\qqed
\end{defn}

\begin{thm}\label{thm:3-equivalent-labellings}
\emph{\cite{Wang-Wang-Yao2019-Euler-Split}} If an Euler's graph has $4m$ edges for $m\geq 2$, then it admits a v-set e-proper graceful labelling, a v-set e-proper odd-graceful labelling, a v-set e-proper edge-magic total labelling, and a v-set e-proper $(4C,k,k')$-labelling.
\end{thm}
\begin{proof} Since each Euler's graph of $n$ edges can be v-divided into a cycle of $n$ edges, so we prove this theorem only to cycles.

When $n=4m$, we make a labelling $f$ for the cycle $C_{4m}=x_1x_2\cdots x_{4m}x_1$ as: $f(x_1)=0$, $f(x_2)=2m$; $f(x_{2j+1})=2m-j-1$ with $j=1,2,\dots ,2m-2$; and $f(x_{4m-2i+2})=4m-i+1$ with $i=1,2,\dots ,m$. Compute the edge labels: $f(x_1x_{4m})=4m$, $f(x_1x_{2})=2m$; $\{f(x_{4m-i}x_{4m-i-1})=f(x_{4m-i})-f(x_{4m-i-1}):i=0,1,2,\dots ,m\}=\{4m-1,4m-2,\dots, 2m+1\}$, $\{f(x_{2i}x_{2i+1})=f(x_{2i})-f(x_{2i+1}):i=1,2,\dots ,m\}=\{1,2,\dots, 2m-1\}$. Therefore, $f$ is a graceful labelling, since $f(E(C_{4m}))=\{1,2,\dots, 4m\}$. Moreover, we have $V(C_{4m})=V_{odd}\cup V_{even}$, where $V_{odd}=\{x_1,x_3,\dots, x_{4m-1}\}$ and $V_{even}=\{x_2,x_4,\dots, x_{4m}\}$. Clearly, $f(x_1)<f(x_{4k-1})<f(x_2)<f(x_{2s})<f(x_{4m})$, we write this case by $\max f(V_{even})<\min f(V_{odd})$, and call $f$ a \emph{set-ordered graceful labelling} of the cycle $C_{4m}$ By Definition \ref{defn:proper-bipartite-labelling-ongraphs}.

Next, we set a labelling $h_1$ of the cycle $C_{4m}$ in the way: $h_1(x_{4k-1})=2f(x_{4k-1})$ for $x_{4k-1}\in V_{odd}$, and  $h_1(x_{2i})=2f(x_{2i})-1$ for $x_{2i}\in V_{even}$. We compte edge labels $h_1(x_{1}x_{4m})=2f(x_{4m})-1-2f(x_1)=8m-1$, $h_1(x_{1}x_{2})=2f(x_{2})-1-2f(x_1)=4m-1$; $\{h_1(x_{4m-i}x_{4m-i-1})=2f(x_{4m-i}x_{4m-i-1})-1:i=0,1,2,\dots ,m\}$; $\{h_1(x_{2j}x_{2j+1})=2f(x_{2j}x_{2j+1})-1:j=1,2,\dots ,m\}$. Thereby, we get $h_1(E(C_{4m}))=\{1,3,5,\dots, 8m-1\}$. Hence, we prove that $h_1$ is a \emph{set-ordered odd-graceful labelling} of  the cycle $C_{4m}$  By Definition \ref{defn:proper-bipartite-labelling-ongraphs}.

For $h_2$ (edge-magic total labelling): $h_2(x_{4k-1})=f(x_{4k-1})$ for $x_{4k-1}\in V_{odd}$, $h_2(x_{2i})=f(x_{4m-2i})$ for $x_{2i}\in V_{even}$; and $h_2(uv)=f(uv)$ with $uv\in E(C_{4m})$.  For the edge $x_1x_{4m}$ of $C_{4m}$, we have

$h_2(x_1)+h_2(x_1x_{4m})+h_2(x_{4m})=f(x_1)+f(x_1x_{4m})+f(x_{2})=0+4m+2m=6m$.

Consider the edge $x_1x_{4m}$ of $C_{4m}$, we get

$h_2(x_1)+h_2(x_1x_{2})+h_2(x_{2})=f(x_1)+f(x_1x_{2})+f(x_{4m})=0+2m+4m=6m$.

Notice that $f(x_{4m-2i})+f(x_{2i})=6m$. For edges $x_{4m-2j+1}x_{4m-2j}$ with $j=1,2,\dots ,2m-1$, we compute

$h_2(x_{4m-2j+1})+h_2(x_{4m-2j+1}x_{4m-2j})+h_2(x_{4m-2j})=f(x_{4m-2j+1})+f(x_{4m-2j+1}x_{4m-2j})+f(x_{2j})
=f(x_{4m-2j+1})+f(x_{4m-2j+1}x_{4m-2j})+6m-f(x_{4m-2i})=f(x_{4m-2j+1})+f(x_{4m-2j})-f(x_{4m-2j+1})+6m-f(x_{4m-2i})=6m$.

For edges $x_{4m-2j}x_{4m-2j-1}$ with $j=1,2,\dots ,2m-1$, we obtain

$h_2(x_{4m-2j})+h_2(x_{4m-2j}x_{4m-2j-1})+h_2(x_{4m-2j-1})=f(x_{2j})+f(x_{4m-2j}x_{4m-2j-1})+f(x_{4m-2j-1})
=6m-f(x_{4m-2i})+f(x_{4m-2j})-f(x_{4m-2j-1})+f(x_{4m-2j-1})=6m$.

Now, we have found a constant $6m$ such that  $h_2(u)+h_2(uv)+h_2(v)=6m$ for each edge $uv\in E(C_{4m})$. Therefore, $h_2$ is really a \emph{set-ordered edge-magic total labelling} of  the cycle $C_{4m}$ by Definition \ref{defn:connections-several-labellings}.

Finally, a labelling pair $(h_3,g)$ for the cycle $C_{4m}$ is defined as: $h_3(x)=f(x)$ for $x\in V(C_{4m})$, and $g(xy)=8m+t-f(xy)$ for each edge  $xy\in E(C_{4m})$ with an constant $t\geq 1$. We verify the restrictions in Definition \ref{defn:4C-k-labelling} to the labelling pair $(h_3,g)$.

(1) (e-magic) each edge $xy$ holds $g(xy)+|h_3(x)-h_3(y)|=8m+t-f(xy)+|f(x)-f(y)|=8m+t$;

(2) (ee-balanced) since $s(uv)=|h_3(u)-h_3(v)|-g(uv)=|f(u)-f(v)|-[8m+t-f(xy)]=-8m-t$, so each edge $uv$ matches another edge $xy$ holding $s(uv)+s(xy)=-16m-2t$ true;

(3) (EV-ordered)  $\max h_3(V(G))=4m< 4m+t=\min g(E(G))$;

(4) (set-ordered) $\max h_3(X)<\min h_3(Y)$ for the bipartition $(X,Y)$ of $V(G)$, since $C_{4m}$ has its own bipartition $(X,Y)$, and $f$ is set-ordered.

So, we claim that $h_3$ is really a \emph{$(4C,8m+t,-16m-2t)$-labelling} of  the cycle $C_{4m}$.

The proof of the theorem is complete.
\end{proof}

\begin{cor}\label{thm:edge-deleted-graphs}
\emph{\cite{Wang-Wang-Yao2019-Euler-Split}} If an Euler's graph has $n$ edges holding $n\equiv 1,2~(\bmod~4)$, then it admits a v-set e-proper $\varepsilon$-labelling, where $\varepsilon$-labelling  $\in \{$odd-graceful labelling, edge-magic total labelling, $(4C,k,k')$-labelling$\}$.
\end{cor}

\begin{cor}\label{thm:general-graphs}
\emph{\cite{Wang-Wang-Yao2019-Euler-Split}} Any non-Euler $(p,q)$-graph $G$ of $n$ edges corresponds an Euler's graph $H=G+E^*$  admitting a v-set e-proper $\varepsilon$-labelling $f$ for $\varepsilon$-labelling $\in \{$odd-graceful labelling, edge-magic total labelling, $(4C,k,k')$-labelling$\}$, such that $G$ admits a v-set e-proper labelling $g$ induced by $f$ and $g(E(G))=\{1,2,\dots ,|E^*|+q\}\setminus \{f(xy):xy\in E^*\}$.
\end{cor}

\begin{thm}\label{thm:more-Euler-graphs}
\emph{\cite{Wang-Wang-Yao2019-Euler-Split}} There are connected and disjoint Euler's graphs $H_1,H_2,\dots,H_m$, such that another connected Euler's graph $G$ is obtained by coinciding $H_i$ with some $H_j$ for $i\neq j$. Then $G$ admits a v-set e-proper $\varepsilon$-labelling for $\varepsilon$-labelling  $\in \{$odd-graceful labelling, edge-magic total labelling, $(4C,k,k')$-labelling$\}$ if $\sum^m_{i=1} |E(H_i)|\equiv 0~(\bmod~4)$.
\end{thm}

\begin{thm}\label{thm:v-pseudo-e-proper-graceful}
A connected Euler's graph $G$ of $n$ edges admits a graceful labelling $f$ if and only if a cycle $C_n$ obtained by dividing some vertices of $G$ admits this labelling $f$ as a its v-pseudo e-proper graceful labelling.
\end{thm}

By the above investigation, we present an obvious result as follows:

\begin{thm}\label{thm:pseudo-v-set-e-proper-graceful}
Every connected graph admits a v-pseudo e-proper graceful labelling.
\end{thm}

If an Euler's graph has odd edges, then it admits a v-set e-proper harmonious labelling. For odd $m$, each complete graph $K_{2m+1}$ admits a v-set e-proper harmonious labelling. See examples shown in Fig.\ref{fig:graceful-euler-graphs}.

\begin{figure}[h]
\centering
\includegraphics[width=8.6cm]{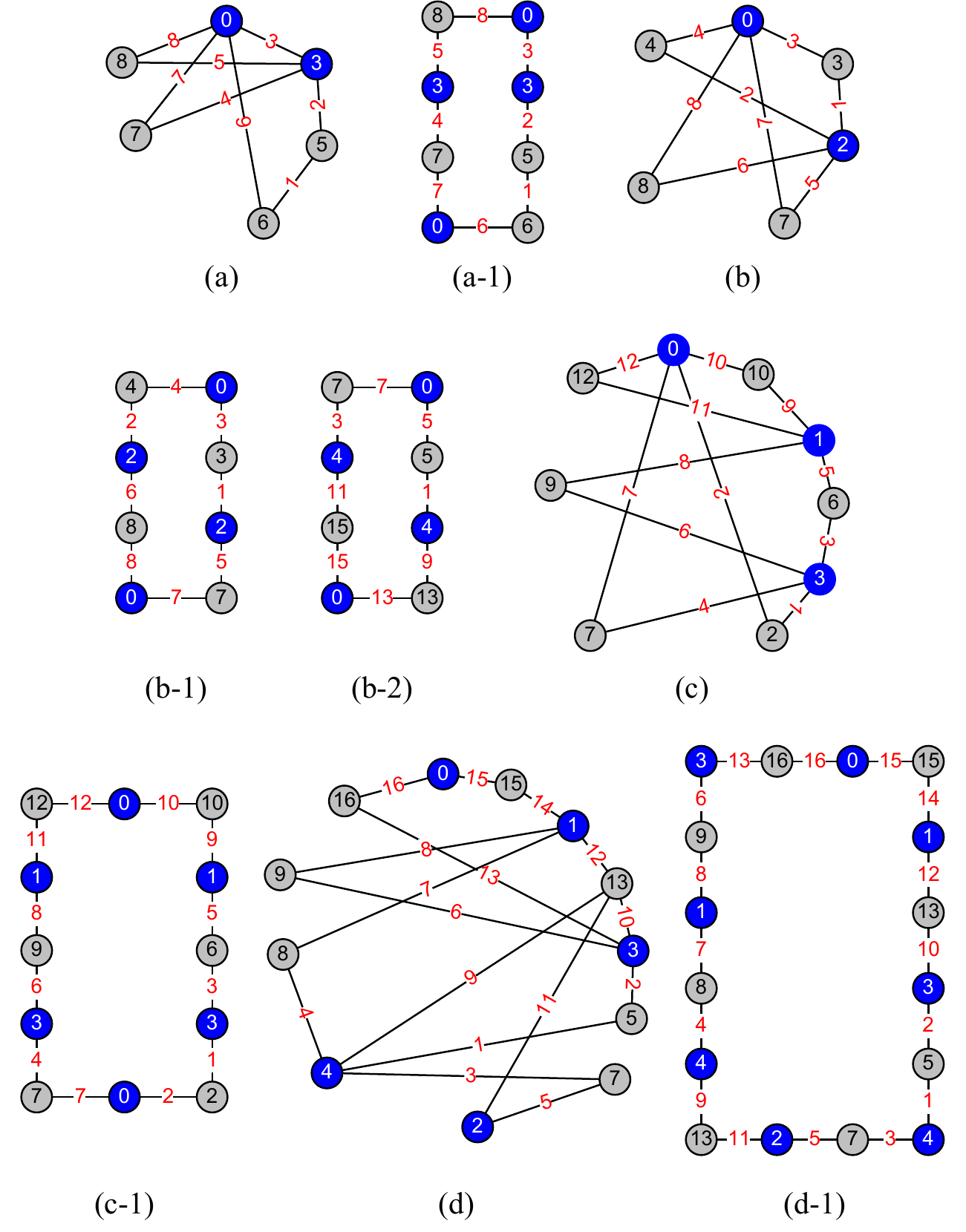}
\caption{\label{fig:graceful-euler-graphs}{\small Each cycle ($k$-1) admits a v-pseudo e-proper graceful labelling for $k=$a,b,c,d.}}
\end{figure}

\section{Conclusion}

In this paper, we have built up the mathematical model of Hanzis (simplified Chinese characters), called \emph{Hanzi-graphs}, and have made the topological graphic passwords of Hanzis by graph labellings, called \emph{Hanzi-gpws} for short, since it may be good for speaking or write conveniently Hanzis as passwords in communication information of Chinese people. In building up the mathematical models of Hanzis, we have shown some printed Hanzis because of there are differences between printed Hanzis.  Hanzis for mathematical models are based on GB2312-80. Hanzi-gpws introduced in Section 4 are constructed by two constituent parts:  Hanzi-graphs in Section 2 and graph labelling/colorings in Section 3. By the disconnected property of Hanzi-gpws we define a new type of graph labellings, called \emph{flawed graph labellings}. Joining disconnected graphs admitting flawed graceful/odd-graceful labellings by labelled edges can produce many pairwise non-isomorphic Hanzi-gpws, this is a good harvest for us in constructing Topsnut-gpws. One disconnected Hanzi-gpw admitting flawed graceful/odd-graceful labelling will distribute us a set of connected Topsnut-gpws admitting  graceful/odd-graceful labellings. Such disconnected Hanzi-gpws are easily used in safeguard of networks constructed by communities.

In the transformation from Hanzi-gpws to TB-paws in Section 5, we have provided several algorithms on Topsnut-matrices, adjacency e-value and ve-value matrices, such as, ALGORITHM-$x$ (or $1$-line-O-$k$) with $x\in \{$I, II, III, VI$\}$ and $k\in [1,4]$. Here, the e-value and ve-value matrices are new approaches in producing TB-paws. In this section 5, we have introduced ``analytic Hanzis'' in $xOy$-plane first, which can distribute more complex TB-paws from analytic Hanzi-matrices. Hanzi-gpws made by graph labellings with variables are useful in making TB-paws with longer bytes.

We, in Section 6, have explored the encryption of networks, a new topic in safety of networks. Especially, we have given some self-similar Hanzi-networks by algorithms, for example, FIBONACCI-VERTEX algorithm based on the vertex-planting operation. Similarity is related with the graph isomorphism problem, which is the computational problem of determining whether two finite graphs are isomorphic. The problem is not known to be solvable in polynomial time nor to be NP-complete, and therefore may be in the computational complexity class NP-intermediate. It is known that the graph isomorphism problem is in the low hierarchy of class NP, which implies that it is not NP-complete unless the polynomial time hierarchy collapses to its second level (Wikipedia). Based on simple Hanzis admitting graph labellings  used frequently, we can quickly set up self-similar Hanzi-networks by our vertex-coincided algorithm-$k$ and edge-coincided algorithm-$k$ with $k=$I, II, III, IV, and self-similar mixed-coincided algorithm, by these algorithms we get graphic semi-groups.

In the new topic of encrypting networks, we employee graphic groups to make passwords for large-scale networks in short time.  Our  every-zero graphic groups can ba used to encrypt every network/graph, namely ``graph-to-graph'', and such graphic groups run in encrypting network/graphs are supported well by coloring theory of graph theory, such as traditional colorings, distinguishing colorings etc. It is noticeable, graphic groups can be used in machine encryption, or AI  encryption for networks. We hope algebraic methods may solve some open problems of graph theory, for example, Graceful Tree Conjecture, and try to use algebraic methods for discovering possible regularity among labellings  all in confusion.

Motivated from the writing stroke order of Hanzis, we have studied directed Hanzi-gpws in Section 7. Directed graphs/networks are important in real life and scientific researches. For instance, ``graph network'' first proposed by Battaglia \emph{et al.} in \cite{Battaglia-27-authors-arXiv1806-01261v2} may be basic as strong AI, Graph NNs (Graph Neural Networks) and Relational Networks, also, may be the hotspot of the next Deep Learning and AI algorithm. We have worked on directed Hanzi-gpws, and obtained some results. Clearly, the investigation of directed Hanzi-gpws just is a beginning, one need to pay more hard attention on it.

Through investigation of  Hanzi-gpws we have met some mathematical problems in Section 8, such as, the number of Hanzi-graphs after decomposing a graph, determining graphs admitting $0$-rotatable set-ordered system of (odd-)graceful labellings, labelling a digraph  with a \emph{Hanzi labelling} (a new type of labellings based on Hanzis, but numbers and functions), finding graph labellings of self-similar Hanzi-networks.

\emph{No doubt, investigating Hanzi-gpws helps us to create new graph labellings for new Hanzi-gpws, and find new transformations for making more complex TB-paws, and think of new mathematical problems, as well as do more works on application of Hanzi-gpws.}

\emph{If there is a country in the world, every word of her can become a poem, a painting, then there must be Chinese characters}.

% use section* for acknowledgment
\section*{Acknowledgment}

The author, \emph{Bing Yao}, is delight for supported by the National Natural Science Foundation of China under grants 61163054, 61363060 and 61662066. Bing Yao, with his heart, thanks every member of Symposium of Hanzi Topological Graphic Passwords in the first semester of 2018-2019 academic year for their helpful discussions and part of their works, and part of members are supported by Scientific Research Project Of Gansu University under grants 2016A-067, 2017A-047 and 2017A-254.

% trigger a \newpage just before the given reference
% number - used to balance the columns on the last page
% adjust value as needed - may need to be readjusted if
% the document is modified later
%\IEEEtriggeratref{8}
% The "triggered" command can be changed if desired:
%\IEEEtriggercmd{\enlargethispage{-5in}}

% references section

% can use a bibliography generated by BibTeX as a . bbl file
% BibTeX documentation can be easily obtained at:
% http: //mirror. ctan. org/biblio/bibtex/contrib/doc/
% The IEEEtran BibTeX style support page is at:
% http: //www. michaelshell. org/tex/ieeetran/bibtex/
%\bibliographystyle{IEEEtran}
% argument is your BibTeX string definitions and bibliography database(s)
%\bibliography{IEEEabrv, . . /bib/paper}

\begin{thebibliography}{1}
\bibitem{256-bit-encryption}https://www.thesslstore.com/blog/what-is-256-bit-encryption/)
\bibitem{Barabasi-Albert1999} Albert-L\'{a}szl\'{o} Barab\'{a}si and Reka Albert. Emergence of scaling in random networks. \emph{Science} \textbf{286} (1999) 509-512.

\bibitem{M-E-J-Newman-SIAM-2003}M. E. J. Newman, The structure and function of complex networks, SIAM Review \textbf{45} (2003) 167-256.
\bibitem{Dorogovtsev-Goltsev-Mendes-2002}S. N. Dorogovtsev, A. V. Goltsev, J. F. F. Mendes. Pseudofractal scale-free web. Physical reviewer, 2002, (\textbf{65}), 066122-066125.
\bibitem{Humboldt-W-1999-1836}Humboldt, W. (1999/1836). On Language: On the diversity of human language construction and its influence on the mental development of the human species. Cambridge University Press.
\bibitem{Chomsky-N-1965}Chomsky, N. (1965). Aspects of the Theory of Syntax. MIT Press.
\bibitem{Battaglia-27-authors-arXiv1806-01261v2}Peter W. Battaglia, Jessica B. Hamrick, Victor Bapst, Alvaro Sanchez-Gonzalez, Vinicius Zambaldi, Mateusz Malinowski, Andrea Tacchetti, David Raposo, Adam Santoro, Ryan Faulkner, Caglar Gulcehre, Francis Song, Andrew Ballard, Justin Gilmer, George Dahl, Ashish Vaswani, Kelsey Allen, Charles Nash4, Victoria Langston, Chris Dyer, Nicolas Heess, Daan Wierstra, Pushmeet Kohli, Matt Botvinick, Oriol Vinyals, Yujia Li, Razvan Pascanu. Relational inductive biases, deep learning, and graph networks. arXiv:1806.01261v2 [cs. LG] 11 Jun 2018.
\bibitem{LeCun-Bengio-Hinton-2015}Y. LeCun, Y. Bengio, and G. Hinton. Deep learning. Nature, vol. \textbf{521}, no. 7553, p. 436, 2015.
\bibitem{Acharya-Hegde1990}B. D. Acharya and S. M. Hegde. Arithmetic graphs. \emph{J. Graph Theory}, \textbf{14} (1990), 275-299.
\bibitem{Bondy-2008} J. A. Bondy, U. S. R. Murty. Graph Theory. Springer London, 2008.
\bibitem{Gallian2018} Joseph A. Gallian. A Dynamic Survey of Graph Labeling. \emph{The electronic journal of
combinatorics}, Twenty-first edition, December 21 (2018), \# DS6. (502 pages, 2643 reference papers)
\bibitem{Harary-Palmer-1973}Harary F. and Palmer E. M. Graphical enumeration. Academic Press, 1973.


\bibitem{S-M-Hegde2000}S. M. Hegde. On $(k,d)$-graceful graphs. \emph{Journal of Combinatorics,
Information \& System Sciences}, Vol. \textbf{25} (1-4) (2000), 255-265.



\bibitem{Suo-Zhu-Owen-2005} Xiaoyuan Suo, Ying Zhu, G. Scott. Owen. Graphical Password: A Survey. In: Proceedings of Annual
Computer Security Applications Conference (ACSAC), Tucson, Arizona. IEEE (2005) 463-472. (10 pages, 38 reference
papers)
\bibitem{Biddle-Chiasson-van-Oorschot-2009}R. Biddle, S. Chiasson, and P. C. van Oorschot. Graphical passwords: Learning from the First Twelve Years. ACM Computing Surveys, \textbf{44} (4), Article 19:1-41. Technical Report TR-09-09, School of Computer Science, Carleton University, Ottawa, Canada. 2009. (25 pages, 145 reference papers)
\bibitem{Gao-Jia-Ye-Ma-2013}Haichang Gao, Wei Jia, Fei Ye and Licheng Ma. A Survey on the Use of Graphical Passwords in Security. Journal Of Software, Vol. \textbf{8} (7), July 2013, 1678-1698. (21 pages, 88 reference papers)

\bibitem{A-Rosa-1966}A. Rosa. On certain valuations of the vertices of a graph, Theory of Graphs (Internat.
Symposium, Rome, July 1966), Gordon and Breach, N. Y. and Dunod Paris (1967) 349-355.


\bibitem{Wang-Xu-Yao-2016} Hongyu Wang, Jin Xu, Bing Yao. Exploring New Cryptographical Construction Of Complex Network Data. IEEE First International Conference on Data Science in Cyberspace. IEEE Computer Society, (2016) 155-160.
\bibitem{Wang-Xu-Yao-Key-models-Lock-models-2016}Hongyu Wang, Jin Xu, Bing Yao. The Key-models And Their Lock-models For Designing New Labellings Of Networks. Proceedings of 2016 IEEE Advanced Information Management, Communicates, Electronic and Automation Control Conference (IMCEC 2016) 565-5568.
\bibitem{Wang-Xu-Yao-2017-Twin2017}Hongyu Wang, Jin Xu, Bing Yao. Twin Odd-Graceful Trees Towards Information Security. Procedia Computer Science \textbf{107} (2017)15-20, DOI: 10.1016/j.procs.2017.03.050





\bibitem{Sun-Zhang-Zhao-Yao-2017}Hui Sun, Xiaohui Zhang, Meimei Zhao and Bing Yao. New Algebraic Groups Produced By Graphical Passwords Based On Colorings And Labellings. ICMITE 2017, MATEC Web of Conferences \textbf{139}, 00152 (2017), DOI: 10. 1051/matecconf/201713900152


\bibitem{Yao-Mu-Sun-Zhang-Wang-Su-Ma-IAEAC-2018}Bing Yao, Yarong Mu, Hui Sun, Xiaohui Zhang, Hongyu Wang, Jing Su, Fei Ma. Algebraic Groups For Construction Of Topological Graphic Passwords In Cryptography. 2018 IEEE 3rd Advanced Information Technology, Electronic and Automation Control Conference (IAEAC 2018), 2211-2216.

\bibitem{Yao-Cheng-Yao-Zhao-2009}Bing Yao, Hui Cheng, Ming Yao and Meimei Zhao. A Note on Strongly Graceful Trees. Ars Combinatoria \textbf{92} (2009), 155-169.
\bibitem{Yao-Sun-Zhang-Mu-Wang-Xu-2018}Bing Yao, Hui Sun, Xiaohui Zhang, Yarong Mu, Hongyu Wang, Jin Xu. New-type Graphical Passwords Made By Chinese Characters With Their Topological Structures. 2018 2nd IEEE Advanced Information Management,Communicates, Electronic and Automation Control Conference (IMCEC 2018), 1606-1610.
\bibitem{YAO-SUN-WANG-SU-XU2018arXiv}Bing Yao, Hui Sun, Hongyu Wang, Jing Su, Jin Xu. Graph Theory Towards New Graphical Passwords In Information Networks. arXiv:1806.02929v1 [cs.CR] 8 Jun 2018
\bibitem{Yao-Sun-Zhang-Mu-Sun-Wang-Su-Zhang-Yang-Yang-2018arXiv}Bing Yao, Hui Sun, Xiaohui Zhang, Yarong Mu, Yirong Sun, Hongyu Wang, Jing Su, Mingjun Zhang, Sihua Yang, Chao Yang. Topological Graphic Passwords And Their Matchings Towards Cryptography. arXiv: 1808. 03324v1 [cs.CR] 26 Jul 2018.

\bibitem{Yao-Zhang-Sun-Mu-Sun-Wang-Wang-Ma-Su-Yang-Yang-Zhang-2018arXiv}Bing Yao, Xiaohui Zhang, Hui Sun, Yarong Mu, Yirong Sun, Xiaomin Wang, Hongyu Wang, Fei Ma, Jing Su, Chao Yang, Sihua Yang, Mingjun Zhang. Text-based Passwords Generated From Topological Graphic Passwords. arXiv: 1809.04727v1 [cs.IT] 13 Sep 2018.

\bibitem{Yao-Liu-Yao-2017}Bing Yao, Xia Liu and Ming Yao. Connections between labellings of trees. Bulletin of the Iranian Mathematical Society, ISSN: 1017-060X (Print) ISSN: 1735-8515 (Online), Vol. \textbf{43} (2017), 2, pp. 275-283.
\bibitem{Yao-Mu-Sun-Zhang-Wang-Su-2018} Bing Yao, Yarong Mu, Hui Sun, Xiaohui Zhang, Hongyu Wang, Jing Su. Connection Between Text-based Passwords and Topological Graphic Passwords. 2018 IEEE 4th Information Technology and Mechatronics Engineering Conference (2018), submitted.
\bibitem{Yao-Sun-Zhao-Li-Yan-2017}Bing Yao, Hui Sun, Meimei Zhao, Jingwen Li, Guanghui Yan. On Coloring/Labelling Graphical Groups For Creating New Graphical Passwords. (ITNEC 2017) 2017 IEEE 2nd Information Technology, Networking, Electronic and Automation Control Conference, (2017) 1371-1375.
\bibitem{Bing-Yao-Cheng-Yao-Zhao2009}Bing Yao, Hui Cheng, Ming Yao and Meimei Zhao. A Note on Strongly
Graceful Trees. \emph{Ars Combinatoria} \textbf{92} (2009), 155-169.

\bibitem{Yao-Yao-Cheng-2012}Bing Yao, Ming Yao, and Hui Cheng. On Gracefulness of Directed Trees with Short Diameters. Bulletin of the Malaysian Mathematical Sciences Society, 2012, (2) \textbf{35}(1). 133-146. WOS:000298904000012

\bibitem{Ma-Wang-Yao-2019} Fei Ma, Ping Wang and Bing Yao. Emergence of power law in the mean first-passage time for
random walks on Fibonacci tree as network models. submitted, 2019.


\bibitem{MU-YAO-2018-11}Yarong Mu, Bing Yao. On Disconnected Topological Graph Passwords For Information Security. 2018 2nd IEEE Advanced Information Management, Communicates, Electronic and Automation Control Conference (IMCEC 2018), 2109-2113.
\bibitem{MU-YAO-2018-22}Yarong Mu, Bing Yao. Exploring Topological Graph Passwords of Information Security By Chinese Culture. 2018 submitted
\bibitem{MU-ZHANG-YAO-2018}Yarong Mu, Xiaohui Zhang, Bing Yao. Designing Hanzi-Graphs Towards New-Type of Graphical Passwords With Applications. Mathematics In Practice And Theory (Chinese), 2018.

\bibitem{Mu-Yao-ITNEC2019}Yarong Mu, Bing Yao. Construction of Topological Graphic Passwords By Hanzi-gpws. submitted.
\bibitem{Mu-Sun-Zhang-Yao-2019}Yarong Mu, Yirong Sun, Mingjun Zhang, Bing Yao. Topological Graphic Passwords On Self-Similar Networks Made By Chinese Characters. submitted, 2019.
\bibitem{Wang-Ma-Yao2019-spltting-connectivity}Xiaomin Wang, Fei Ma, Bing Yao. On Divided-Type Connectivity of Graphs and Networks. submitted, 2019.
\bibitem{Wang-Wang-Yao2019-Euler-Split}Xiaomin Wang, Hongyu Wang, Bing Yao. Applying Divided Operations Towards New Labellings Of Euler's Graphs. submitted, 2019.


\bibitem{Mandelbrot-Benoit-B-1967}Mandelbrot, Benoit B. (5 May 1967). ``How long is the coast of Britain? Statistical self-similarity and fractional dimension''. Science. New Series. \textbf{156} (3775): 636-638. Bibcode:1967Sci. . . 156. . 636M. doi: 10.1126/science.156.3775.636. PMID 17837158. Retrieved 11 January 2016.



\bibitem{GB2312-80} ``GB2312-80 Encoding of Chinese characters'' cited from The Compilation Of National Standards For Character Sets And Information Coding, China Standard Press, 1998.


\bibitem{SUN-ZHANG-YAO-IMCEC-2018}Hui Sun, Xiaohui Zhang, Bing Yao. Construction Of New Graphical Passwords With Graceful-type Labellings On Trees. 2018 2nd IEEE Advanced Information Management, Communicates, Electronic and Automation Control Conference (IMCEC 2018), 1491-1494.
\bibitem{ZHANG-SUN-YAO-Liu-IMCEC-2018}Xiaohui Zhang, Hui Sun, Bing Yao, Xinsheng Liu. A Technique Based On The Module-K Super Graceful Labelling For Designing New-type Graphical Passwords. 2018 2nd IEEE Advanced Information Management,Communicates, Electronic and Automation Control Conference (IMCEC 2018), 1494-1499.
\bibitem{Zhang-Liu-Wang-2002-strong}Zhang Zhongfu, Liu Linzhong, Wang Jianfang. Adjacent strong edge coloring of graphs. Applied Mathematics Letters, 2002, \textbf{15}: 623-626.
\bibitem{ZHANG-SUN-YAO-IMCEC-2018}Xiaohui Zhang, Hui Sun, Bing Yao. On Topological Graphic Passwords Made By Twin Edge Module-k Odd-graceful Labelling. 2018 2nd IEEE Advanced Information Management, Communicates, Electronic and Automation Control Conference (IMCEC 2018), 2114-2118.
\bibitem{Zhang-Yao-Wang-Wang-Yang-Yang-2013} Jiajuan Zhang, Bing Yao, Zhiqian Wang, Hongyu Wang, Chao Yang, Sihua Yang. Felicitous Labellings of Some Network Models. Journal of Software Engineering and Applications, 2013, \textbf{6}, 29-32. DOI: 10. 4236/jsea. 2013. 63b007 Published Online March 2013 (http://www. scirp. org/journal/jsea)
\bibitem{Zhou-Yao-Chen-Tao2012}Xiangqian Zhou, Bing Yao, Xiang'en Chen and Haixia Tao. A proof to the
odd-gracefulness of all lobsters. \emph{Ars Combinatoria} \textbf{103} (2012), 13-18.
\bibitem{Zhou-Yao-Chen2013}Xiangqian Zhou, Bing Yao, Xiang'en Chen. Every Lobster Is Odd-elegant. \emph{Information Processing Letters} \textbf{113} (2013), 30-33.
\bibitem{Zhou-Yao-Cheng-2011}Xiangqian Zhou, Bing Yao, Hui Cheng. On 0-rotatable trees. Journal Of South China Normal University (Chinese), 2011, \textbf{4}, 54-57.
\end{thebibliography}
%
% <OR> manually copy in the resultant . bbl file
% set second argument of \begin to the number of references
% (used to reserve space for the reference number labels box)

%%\newpage

{\large \textbf{Appendix A.}}

\begin{flushleft}
{\small
\textbf{Table-1.} Stirling's approximation $n!=(\frac{n}{e})^n\sqrt{2n\pi}$

$${
\begin{split}
9!&=362,880\approx 2^{18.46913302}\approx 2^{18.5}\\
10!&=3,628,800\approx 2^{21.791061114717}\approx 2^{21.8}\\
11!&=39,916,800\approx 2^{25.2504927333542}\approx 2^{25.3}\\
12!&=479,001,600\approx 2^{28.8354552340754}\approx 2^{28.8}\\
13!&=6,227,020,800\approx 2^{32.5358949522165}\approx 2^{32.5}\\
14!&=87,178,291,200\approx 2^{36.3432498742741}\approx 2^{36.3}\\
15!&\approx 2^{40.2501404698826}\approx 2^{40}\\
20!&\approx 2^{61.0773839209062}\approx 2^{61}\\
30!&\approx 2^{107.709067341973}\approx 2^{107.7}\\
50!&\approx 2^{214.20813806359}\approx 2^{214}\\
100!&\approx 2^{524.76499329006}\approx 2^{524.8}\\
150!&\approx 2^{872.859506347079}\approx 2^{872.9}\\
160!&\approx 2^{945.664752154612}\approx 2^{945.7}\\
170!&\approx 2^{1019.36945292773}\approx 2^{1019}
\end{split}}$$
}
\end{flushleft}

\vskip 1cm

{\large \textbf{Appendix B.}}

\begin{flushleft}
{\footnotesize
\textbf{Table-2.} The numbers of trees of order $p$ \cite{Harary-Palmer-1973}.
\begin{tabular}{cll}
$p$&$t_p$&$T_p$\\
6&6&2\\
7&11&48\\
8&23&115\\
9&47&286\\
10&106&719\\
11&235&1,842\\
12&551&4,766\\
13&1,301&12,486\\
14&3,159&32,973\\
15&7,741&87,811\\
16&19,320&235,381\\
17&48,629&634,847\\
18&123,867&1,721,159\\
19&317,955&4,688,676\\
20&823,065&12,826,228\\
21&2,144,505&35,221,832\\
22&5,623,756&97,055,181\\
23&14,828,074&268,282,855\\
24&39,299,897&743,724,984\\
25&104,636,890&2,067,174,645\\
26&279,793,450&5,759,636,510\\
\end{tabular}
}

where $t_p$ is the number of trees of order $p$, and $T_p$ is the number of rooted trees of order $p$.
\end{flushleft}

\vskip 1cm

{\large \textbf{Appendix C.}}

\begin{figure}[h]
\centering
\includegraphics[width=8.6cm]{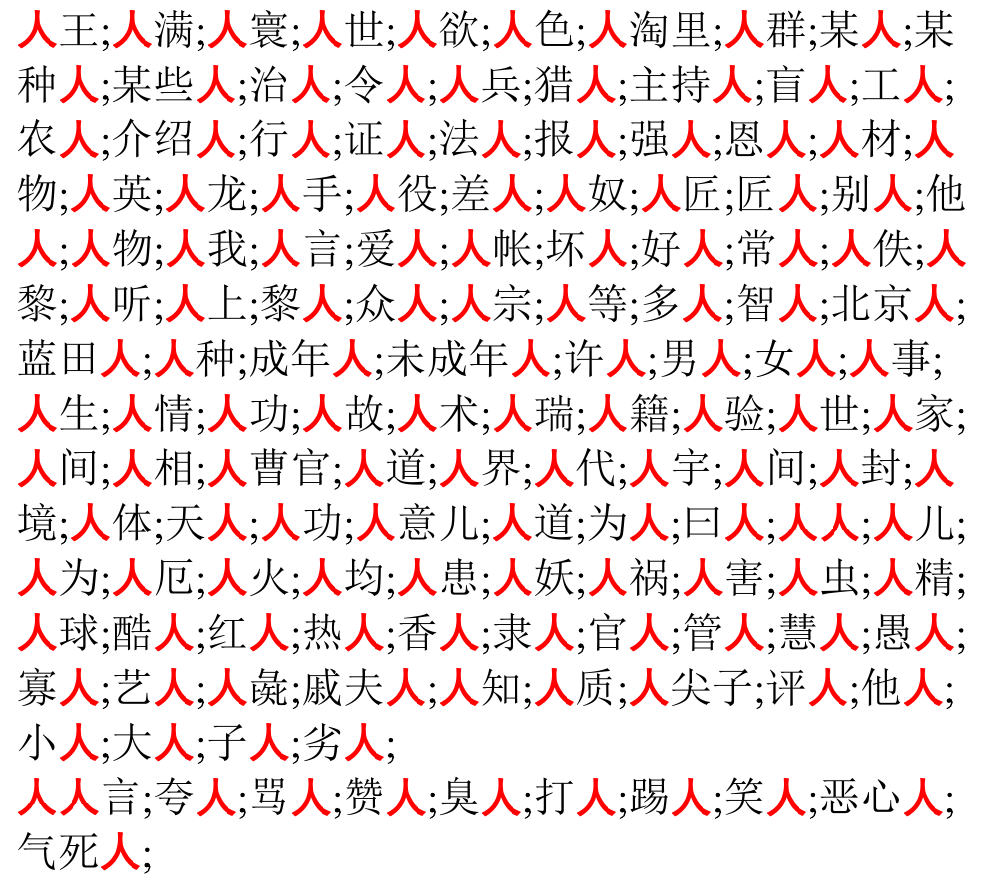}\\
\caption{\label{fig:man-word-group} {\small The word groups related with a Hanzi $H_{4043}$, also $H_{4043}=$man.}}
\end{figure}

\begin{figure}[h]
\centering
\includegraphics[width=8.6cm]{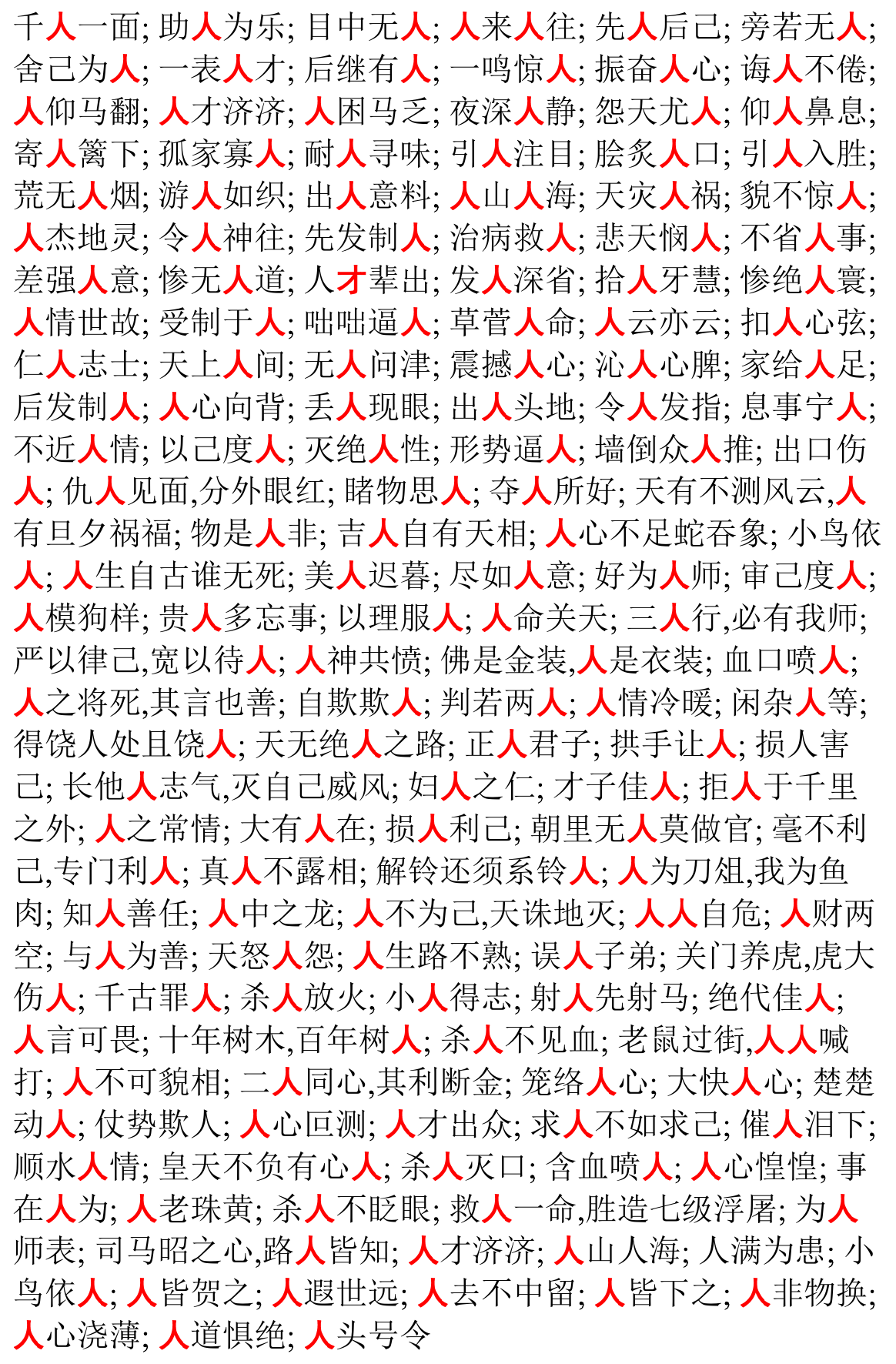}\\
\caption{\label{fig:man-chengyu} {\small The idioms related with a Hanzi $H_{4043}$.}}
\end{figure}

% that's all folks
\end{document}